\documentclass{Dissertate}

\usepackage{multirow}
\usepackage{float}
\usepackage{framed}
\usepackage{adjustbox}
\usepackage{longtable}

\usepackage{tikz}
\usepackage{tkz-fct}
\usepackage{tikz-cd}
\usetikzlibrary{matrix,arrows,decorations}

\makeatletter
\newcommand\thefontsize[2]{{#1 \noindent \f@size~pt: #2}}
\makeatother

\makeatletter
\def\nobreakhline{%
  \noalign{\ifnum0=`}\fi
    \penalty\@M
    \futurelet\@let@token\LT@@nobreakhline}
\def\LT@@nobreakhline{%
  \ifx\@let@token\hline
    \global\let\@gtempa\@gobble
    \gdef\LT@sep{\penalty\@M\vskip\doublerulesep}
  \else
    \global\let\@gtempa\@empty
    \gdef\LT@sep{\penalty\@M\vskip-\arrayrulewidth}
  \fi
  \ifnum0=`{\fi}%
  \multispan\LT@cols
     \unskip\leaders\hrule\@height\arrayrulewidth\hfill\cr
  \noalign{\LT@sep}%
  \multispan\LT@cols
     \unskip\leaders\hrule\@height\arrayrulewidth\hfill\cr
  \noalign{\penalty\@M}%
  \@gtempa}
\makeatother

\newtheoremstyle{definitionstyle}	
	{0.2cm}				
	{0.2cm}				
	{\it}						
	{}							
	{\it\bfseries}			
	{:}						
	{ }						
	{\thmname{#1}\thmnumber{~#2}\thmnote{~(#3)}}	

\newtheoremstyle{nameddefinitionstyle}	
	{\baselineskip\@plus.2\baselineskip\@minus.2\baselineskip}	
	{\baselineskip\@plus.2\baselineskip\@minus.2\baselineskip}	
	{}							
	{}							
	{\bfseries}			
	{:}						
	{ }						
	{\thmnote{#3}}	

\newtheoremstyle{framednameddefinitionstyle}	
	{0.2cm}				
	{0.2cm}				
	{\it}						
	{}							
	{\it\bfseries}			
	{:}						
	{ }						
	{\thmnote{#3}}	

\newtheoremstyle{theoremstyle}	
	{0.2cm}				
	{0.2cm}				
	{}							
	{}							
	{\bfseries}			
	{:}						
	{ }						
	{\thmname{#1}\thmnumber{~#2}\thmnote{~(#3)}}	

\newtheoremstyle{framedtheoremstyle}	
	{\baselineskip\@plus.2\baselineskip\@minus.2\baselineskip}	
	{\baselineskip\@plus.2\baselineskip\@minus.2\baselineskip}	
	{\sl}						
	{}							
	{\bfseries}			
	{:}						
	{ }						
	{\thmname{#1}\thmnumber{~#2}\thmnote{~(#3)}}	

\newtheoremstyle{proofstyle}
	{\baselineskip\@plus.2\baselineskip\@minus.2\baselineskip}	
	{\baselineskip\@plus.2\baselineskip\@minus.2\baselineskip}	
	{}							
	{}							
	{}							
	{:}						
	{ }						
	{\textsc{\thmname{#1}\thmnote{~#3}}}	

\theoremstyle{theoremstyle}

\newtheorem{prp}{Proposition}[chapter]
\newtheorem{thm}{Theorem}[chapter]
\newtheorem{cor}{Corollary}[chapter]

\newtheorem{rmk}{Remark}[chapter]

\theoremstyle{framedtheoremstyle}

\theoremstyle{nameddefinitionstyle}
\newtheorem{nameddef}{Definition}[chapter]

\theoremstyle{framednameddefinitionstyle}

\theoremstyle{proofstyle}

\theoremstyle{definitionstyle}

\newcommand{\fromto}{\rightarrow}

\newcommand{\xfromto}[1]{\xrightarrow{#1}}

\newcommand{\ZZZ}{\mathbb{Z}}

\newcommand{\RRR}{\mathbb{R}}
\newcommand{\CCC}{\mathbb{C}}

\renewcommand{\SS}{\mathbf{S}}

\newcommand{\A}{\mathcal{A}}

\newcommand{\K}{\mathcal{K}}
\newcommand{\T}{\mathcal{T}}

\DeclareMathOperator{\Map}{Map}
\DeclareMathOperator{\identity}{id}

\DeclareMathOperator{\image}{im}

\DeclareMathOperator{\pt}{pt}

\DeclareMathOperator{\Tor}{Tor}
\DeclareMathOperator{\Ext}{Ext}

\DeclareMathOperator{\kernel}{ker}

\newcommand{\coloneq}{\mathrel{\mathop:}=}
\newcommand{\eqcolon}{=\mathrel{\mathop:}}

\newcommand{\homotopic}{\simeq}

\newcommand{\isomorphic}{\cong}

\newcommand{\bars}[1]{\left| #1 \right|}
\newcommand{\paren}[1]{\left( #1 \right)}

\newcommand{\brackets}[1]{\left[ #1 \right]}
\newcommand{\braces}[1]{\left\{ #1 \right\}}

\newcommand{\bra}[1]{\left\langle #1\right|}
\newcommand{\ket}[1]{\left|#1\right\rangle}
\newcommand{\braket}[2]{\left\langle#1\middle|#2\right\rangle}

\def\ind\hspace{0.2in}

\newcommand{\SPT}{\operatorname{\mathcal{SPT}}}

\newcommand{\wSPT}{\operatorname{w\mathcal{SPT}}}

		\newcommand{\e}[1]{\begin{align}{#1}\end{align}}	
		
		\newcommand{\p}[2]{\frac{\partial #1}{\partial #2}}
		
		\newcommand{\q}[1]{Eq.\ (\ref{#1})}
		
		\newcommand{\s}[1]{Sec.\ \ref{#1}}
		\newcommand{\fig}[1]{Figure \ref{#1}}

\newcommand{\bx}{\boldsymbol{x}}

\newcommand{\bE}{\boldsymbol{E}}

\newcommand{\bB}{\boldsymbol{B}}

\newcommand{\bpm}{\begin{pmatrix}}
\newcommand{\epm}{\end{pmatrix}}

\newcommand{\bal}{\begin{align}}

\def\mod{{\rm\ mod\ }}

\def\p{\partial}

\def\widebar{\accentset{{\cc@style\underline{\mskip10mu}}}} 
\def\wideubar{\underaccent{{\cc@style\underline{\mskip10mu}}}} 

\begin{document}



\title{Classification and Construction of Topological Phases of Quantum Matter}
\author{Zhaoxi Xiong}

\advisor{Ashvin Vishwanath}

\committeeInternalOne{Subir Sachdev}
\committeeInternalTwo{Arthur Jaffe}



\degree{Doctor of Philosophy}
\field{Physics}
\degreeyear{2019}
\degreeterm{Spring}
\degreemonth{May}
\department{Physics}



\frontmatter
\setstretch{\dnormalspacing}


\setcounter{chapter}{0}  
\setlength{\footnotesep}{16pt}   

\chapter{Introduction}
\label{chap:introduction}

%
%
%
%
%
%
%
%
%
%

This thesis is about topological phases of strongly interacting quantum many-body systems. In 1980, it was discovered, at low temperatures and relatively strong magnetic field, that the Hall conductance of a 2D material is quantized in units of $\frac{e^2}{h}$ \cite{Klitzing_IQHE}:
\begin{equation}
\sigma_{xy} = \nu \frac{e^2}{h},
\end{equation}
where $\nu$ can only take integer values, forming a series of plateaus (see Figure \ref{fig:IQHE}). The quantization was insensitive to impurities and was extremely precise. In 1982, it was discovered, at even lower temperatures and for extremely clean samples, that $\nu$ can take quantized rational values, such as $\frac13$, and collective excitations with fractional electric charges, such as $-\frac{e}{3}$, were observed \cite{FQHE_original}. In 1988, it was proposed that the integral quantization of Hall conductance can occur even without any magnetic field, if we give complex values to the amplitudes of some hopping processes of electrons \cite{TKNN, Haldane1988}.

\begin{figure}
\centering
\includegraphics[width=3in]{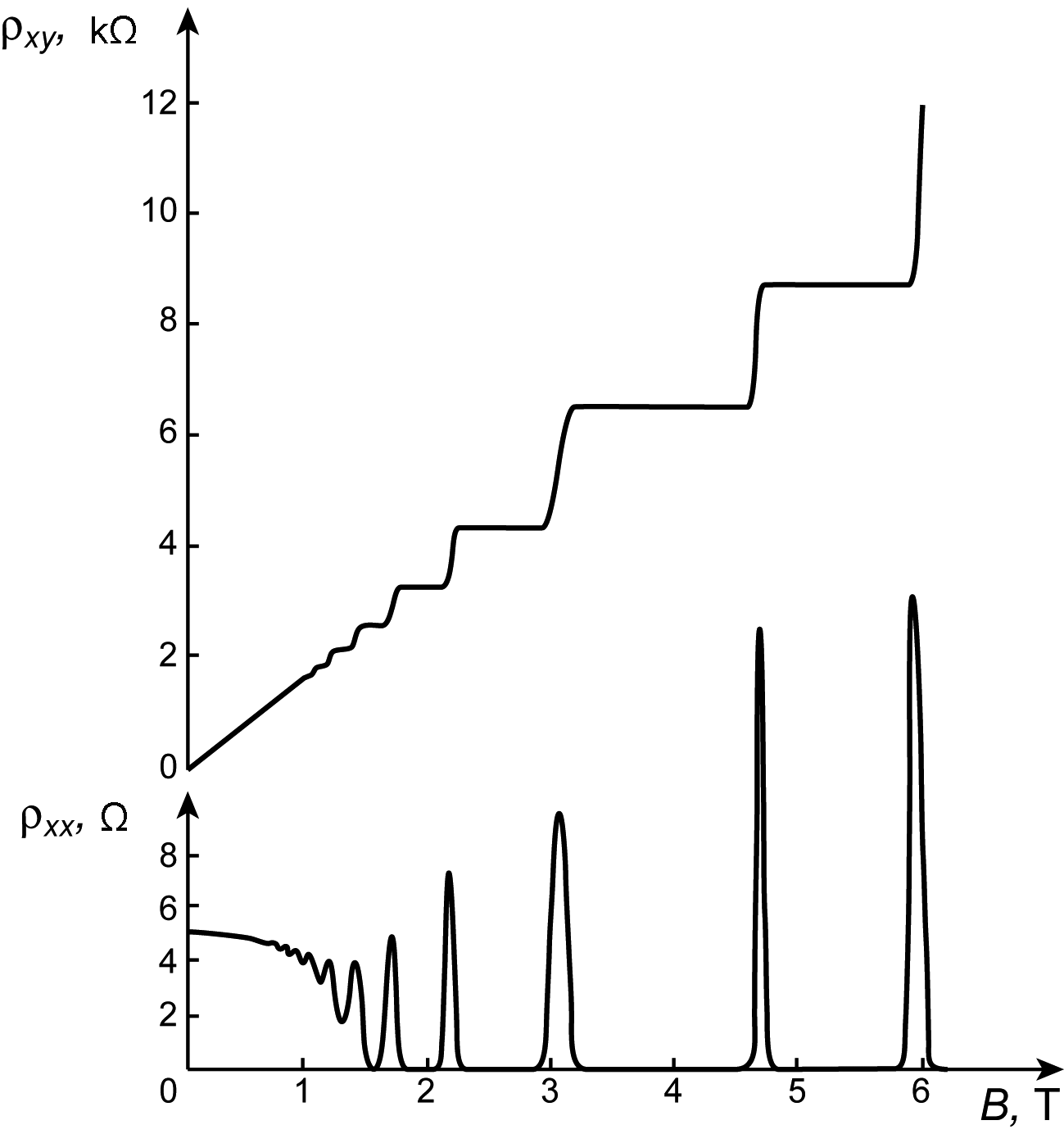}
\caption[The quantization of Hall resistance $\rho_{xy}$.]{The quantization of Hall resistance $\rho_{xy}$. Hall conductance is related to Hall resistance and longitudinal resistance by $\sigma_{xy} = - \frac{\rho_{xy}}{\rho_{xx}^2 + \rho_{xy}^2}$.}
\label{fig:IQHE}
\end{figure}

These phenomena are known as the integer quantum Hall effect, the fractional quantum Hall effect, and the anomalous quantum Hall effect, respectively. They stood in stark contrast to the physics at higher temperatures, where response functions like Hall conductance are continuous variables whose values depend on non-universal parameters of a system, such as impurity, geometry, or temperature. Theories were proposed to explain the mentioned phenomena \cite{laughlin1981, Laughlin_FQHE, TKNN, Haldane1988}, and it was discovered, in each case, that there was some ``topological invariant" for the ground state of a system that is mathematically pinned at discrete values. In particular, for the Bloch wavefunctions $e^{i \boldsymbol k \cdot \boldsymbol r} \ket{u_{\boldsymbol k}}$ of electrons in the anomalous quantum Hall effect, there is a topological invariant called the Chern number (the first Chern class) \cite{TKNN, Haldane1988},
\begin{equation}
C = \frac{i}{2 \pi} \int_0^{2\pi} dk_x \int_0^{2\pi} dk_y \paren{ \braket{\frac{\partial u_{\boldsymbol k}}{\partial k_x}}{\frac{\partial u_{\boldsymbol k}}{\partial k_y}} - \braket{\frac{\partial u_{\boldsymbol k}}{\partial k_y}}{\frac{\partial u_{\boldsymbol k}}{\partial k_x}} },
\end{equation}
that was tied to Hall conductance via
\begin{equation}
\sigma_{xy} = C \frac{e^2}{h}.
\end{equation}
For a simple two-level system, if we recognize the Brillouin zone as a torus and regard $\ket{u_{\boldsymbol k}} = a_{\boldsymbol k} \ket{\uparrow} + b_{\boldsymbol k} \ket{\downarrow}$ at each $\boldsymbol k$ as representing a point on the Bloch sphere, then the Chern number will count how many times $\ket{u_{\boldsymbol k}}$ wraps around the sphere as $\boldsymbol k$ moves around the torus (see Figure \ref{fig:chern}). Since a torus can only wrap around a sphere an integer number of times, the Chern number must be an integer and Hall conductance must be quantized.

\begin{figure}
\centering
\includegraphics[width=5in]{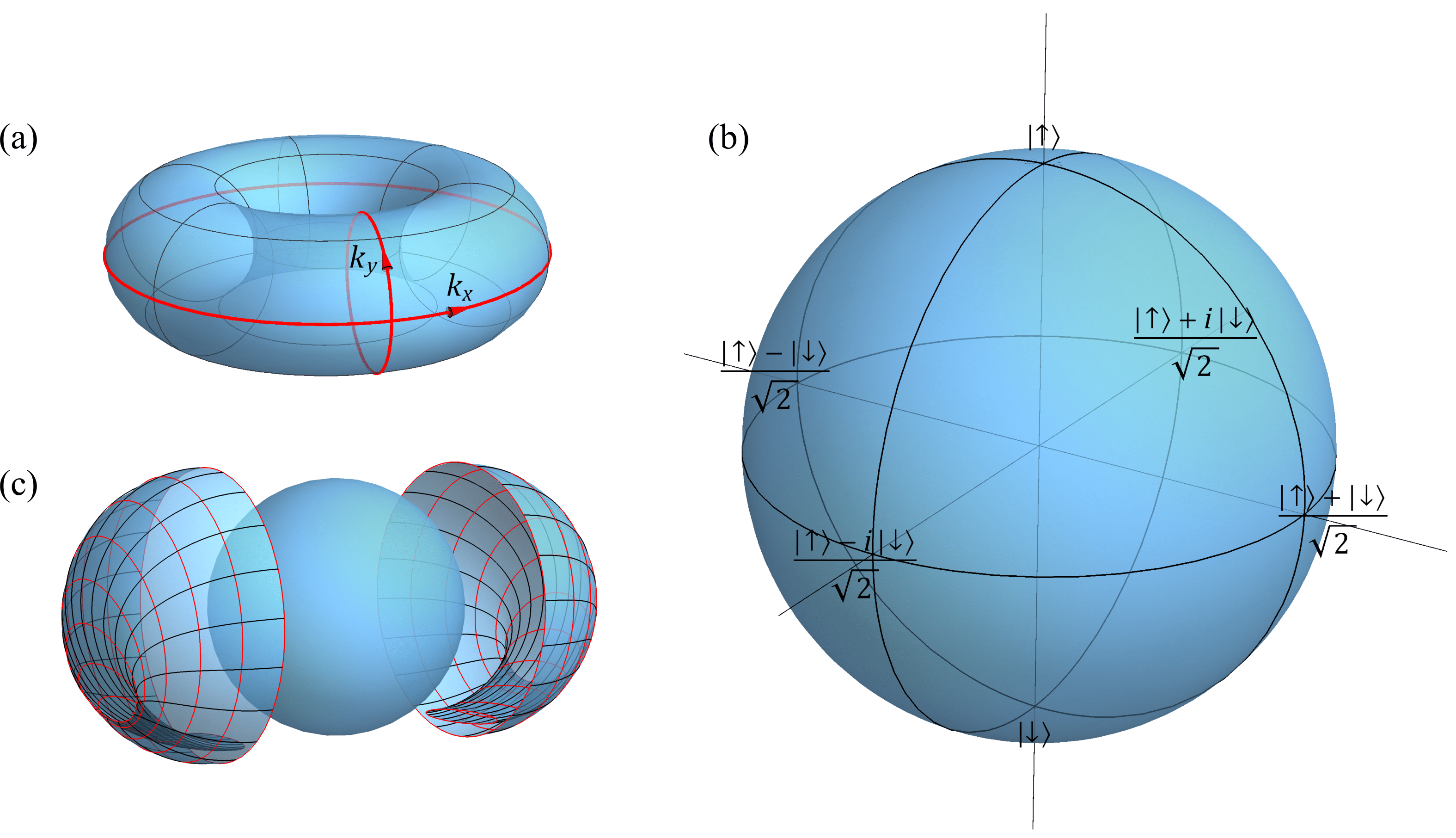}
\caption[Chern number as the degree of a map from the 2D Brillouin zone to the Bloch sphere.]{For a two-level system, the Chern number counts how many times (a) the Brillouin zone wraps around (b) the Bloch sphere. (c) A possible way to wrap a torus (dissected for easier viewing) around a sphere that corresponds to a Chern number of $C = 1$.}
\label{fig:chern}
\end{figure}

The connection between physics and topology in the quantum Hall effects was not a coincidence. Topology, in itself, is the study of the properties of space that are independent of continuous deformations, such as stretching, bending, or squeezing. Although a coffee cup has a very different geometry than a donut---one can hold liquid and the other cannot---they have the same topology since we can continuously deform one into the shape of the other while preserving the existence of a hole, which they both possess. In the case of Chern number above, the number of times the map $\ket{u_{\boldsymbol k}}$ wraps a torus around a sphere is independent of continuous deformations of the map $\ket{u_{\boldsymbol k}}$ itself. In the last several decades, countless other topological phenomena in condensed matter physics have been discovered. They are not restricted to itinerant electrons in periodic potentials, but also occur for other systems whose degrees of freedom are bosonic in nature, such as Mott insulators. The ground state of a spin-$1$ chain with a simple antiferromagnetic Heisenberg interaction,
\begin{equation}
\hat H = J \sum_{j = 1}^N \hat{\boldsymbol S}_j \cdot \hat{\boldsymbol S}_{j+1}, ~~~~ J > 0, \label{spin1AF}
\end{equation}
for example, was found to have emergent spin-$\frac12$ degrees of freedom at its boundaries \cite{AKLT, PhysRevLett.50.1153, Haldane_NLSM, Affleck_Haldane, Haldane_gap}. For Hamiltonian (\ref{spin1AF}), there is a gap between its ground state and the lowest excited state. For the more general Hamiltonian
\begin{equation}
\hat H = J \sum_{j = 1}^N \brackets{\hat{\boldsymbol S}_j \cdot \hat{\boldsymbol S}_{j+1} - \beta \paren{\hat{\boldsymbol S}_j \cdot \hat{\boldsymbol S}_{j+1}}^2 }, \label{spin1p}
\end{equation}
it was found, as $\beta$ is tuned away from $0$, that the spin-$\frac12$ degrees of freedom persist as long as the energy gap remains open. It would later be realized that this is true not only for deformation (\ref{spin1p}) but for any deformation that respects the $SO(3)$ spin-rotation symmetry, and has to do with the topological nature of symmetry fractionalization \cite{Wen_1d, Cirac}. It was shown that under renormalization, the ground state of a gapped spin chain with symmetry $G$ flows to a fixed-point state with boundary degrees of freedom that carry projective representations of $G$ (see Figure \ref{fig:AKLT}). Projective representations can be divided into equivalence classes: For $G = SO(3)$, half-integer spins belong to one equivalence class, and integer spins belong to the other. Because equivalence classes are rigid objects that do not change under continuous deformations---not even by reversing the renormalization flow---it naturally follows that one cannot eliminate the boundary spin-$\frac12$'s without breaking one of the assumptions (e.g.\,symmetry, or energy gap).

\begin{figure}
\centering
\includegraphics[width=2.5in]{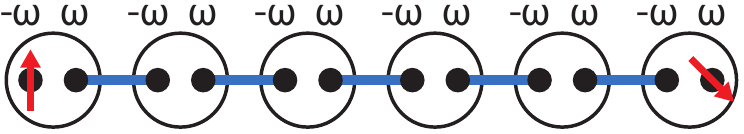}
\caption[The ground state of a gapped spin chain whose boundaries carry projective representations of the symmetry.]{Under renormalization, the ground state of a gapped spin chain with symmetry $G$ flows to a fixed-point state with boundary degrees of freedom that carry projective representations (red arrows) of $G$. The factor systems $-\omega$ and $\omega$ of the projective representations at the two edges must be conjugate to each other. In the bulk, fractionalized degrees of freedom that carry conjugate factor systems form inter-site maximally entangled states (blue line segments).}
\label{fig:AKLT}
\end{figure}

A 2D system with a nonzero Chern number is now known as a Chern insulator, and the phase of spin-$1$ chain represented by Hamiltonian (\ref{spin1AF}) is now known as the Haldane phase. Examples like these have inspired the now vibrant field of topological phases of quantum many-body systems. In the spirit of topology, the concept of a phase is now defined by whether two systems can be continuously deformed to each other. There are different variants of topological phases, depending on whether one allows for interactions and whether one imposes a symmetry. In this thesis, we will always allow for interactions. Thus, we are concerned with the topological phases of strongly interacting quantum many-body systems.

\section*{The theoretical challenge}

The subject of topological phases of strongly interacting quantum many-body systems is a difficult one. Several reasons can be named. First of all, the systems are quantum many-body systems. As is well-known, the dimensionality of a quantum many-body systems grows exponentially with system size. If an individual particle of the quantum system---a localized spin, an electron, etc.---has $m$ possible states, then to fully describe a many-body state, it will require $m^N$ complex numbers $\Psi_{i_1 i_2 \ldots i_N}$, where $N$ is the number of such individual particles the system contains:
\begin{equation}
\ket{\Psi} = \sum_{i_1, i_2, \ldots, i_N = 1}^m \Psi_{i_1 i_2 \ldots i_N} \ket{i_1 i_2 \ldots i_N}.
\end{equation}
This ``curse of dimensionality" makes a brute-force, exact diagonalization of Hamiltonians computationally expensive, which can usually only go up to an $N$ on the order of magnitude of $\sim 10^1$ \cite{zhang2010exact, weisse2008exact}, well below the Avogadro number $N_A = 6.022 \times 10^{23}$.

The second reason is that we allow for strong interactions. This means the Hamiltonian can contain not only bilinear terms in fermionic operators,
\begin{equation}
\hat c_\alpha^\dag \hat c_\beta, ~~ \hat c_\alpha^\dag \hat c_\beta^\dag, ~~ \hat c_\alpha \hat c_\beta, 
\end{equation}
but also quartic terms, sextic terms, octic terms, and so on. Without the higher-order terms, any Hamiltonian is exactly solvable, with a computational complexity polynomial in $N$ in the second-quantized basis. Such is the case for the band theory of non-interacting electrons in periodic potentials, where a single-particle picture suffices. Strong interactions render single-particle solutions invalid. For a strongly interacting system, some approximation is often needed. On the flip side, permitting interactions lets us study ``bosonic"---or spin---systems as a stand-alone problem. These are systems whose only relevant degrees of freedom are spins localized on different atomic sites. Individual spin operators are quadratic in fermionic operators,
\begin{equation}
\hat S^i = \frac12 \sum_{\alpha, \beta = \uparrow, \downarrow} \hat c_\alpha^\dag \sigma^i_{\alpha\beta} \hat c_\beta,
\end{equation}
for $i = x, y, z$, where $\sigma^i$ are the Pauli matrices, so any spin-spin interaction will necessarily involve higher-order terms in the fermionic operators.

The third reason topological phases of strongly interacting quantum many-body systems are difficult to study lies in the definition of topological phases. Conventionally, when one studies phase transitions, one focuses on qualitative changes to the behavior of a Hamiltonian as a small number of parameters are varied, such as the value of a coupling constant $g$, as in
\begin{equation}
\hat H = \hat H_0 + g \hat H_1.
\end{equation}
If a singularity occurs for some $g^*$, one says that there is a phase transition. In contrast, for topological phases, one considers all possible paths in the space of all Hamiltonians. Thus even if there is no ``direct" path between two Hamiltonians, as long as there is an ``indirect" path between them, one still considers them to be in the same phase. In this spirit---although these are thermal phases not topological phases---one would consider the liquid and gas phases of many substances to be the same phase, provided such substances have a supercritical fluid region allowing for an indirect path between the two phases (see Figure \ref{fig:CO2}). To determine two systems to be in distinct topological phases, one has to rule out all possible, arbitrarily complex paths between the two systems. This makes a direct approach to the problem impossible.

\begin{figure}
\centering
\includegraphics[width=3.5in]{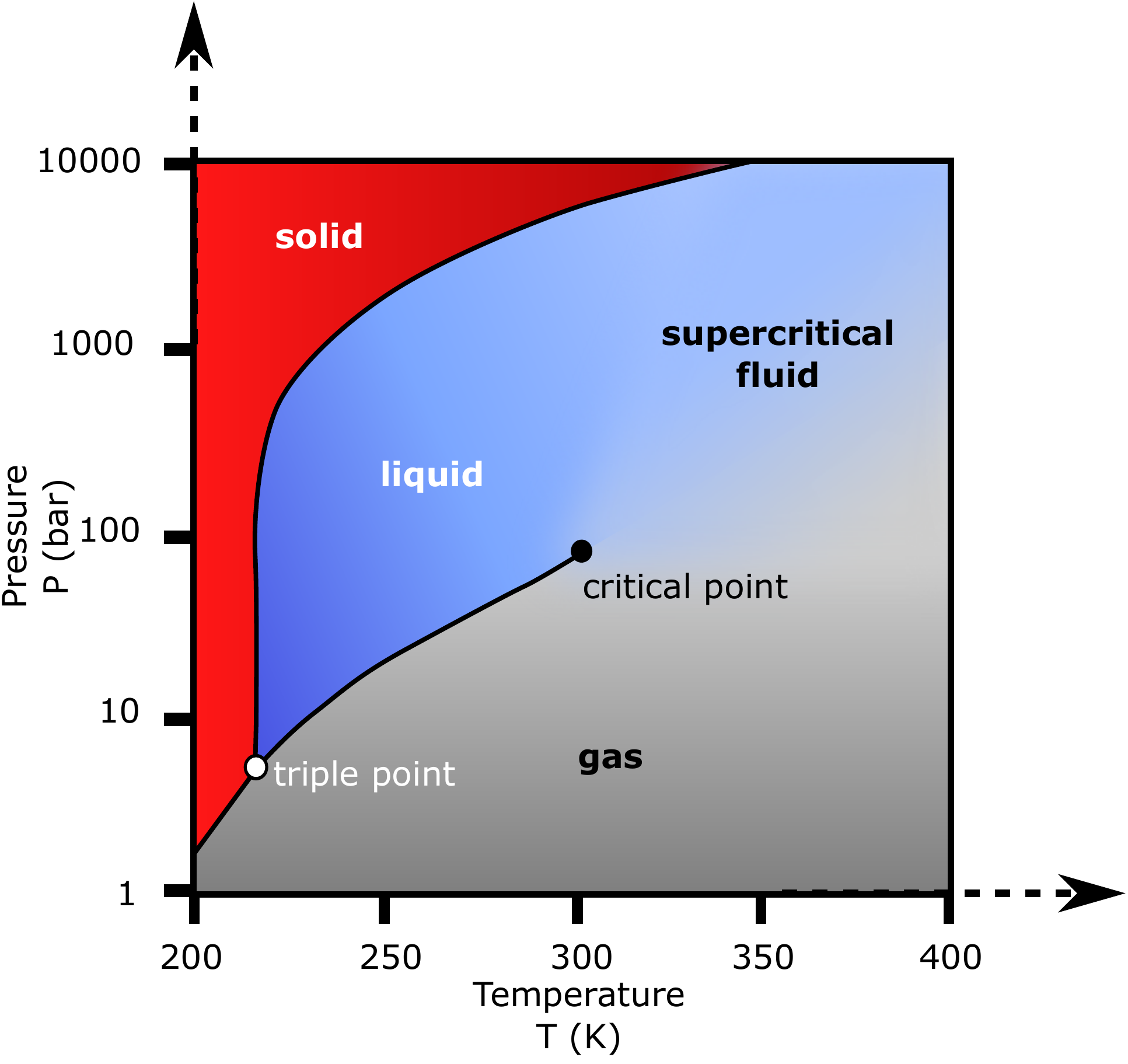}
\caption[The thermal phase diagram of carbon dioxide with respect to temperature and pressure.]{The thermal phase diagram of carbon dioxide with respect to temperature and pressure.}
\label{fig:CO2}
\end{figure}

The definition of topological phases imposes a condition on the Hamiltonians: an energy gap above a non-degenerate ground state (on a Euclidean space $\mathbb E^d$ or its compactification $\SS^d$), which is a blessing and a curse. It is a blessing because, as it turns out, this condition is effectively a recipe for a notion of ``zero-temperature phases" as it guarantees that the topological phase of a system depends on only its ground state. Note that, as well as accidental degeneracy, the condition prohibits any degeneracy due to spontaneous symmetry breaking, which is another simplification. The study of topological phases is essentially a study of the zeroth homotopy group $\pi_0$---the set of path-connected components---of the space of gapped systems, which is the coarsest information one can extract from a topological space. On the other hand, the gap condition is a curse because, when considering paths between two gapped systems, one now has to keep the gap open everywhere on the path. In general, it is nontrivial to determine whether a generic Hamiltonian is gapped or gapless. Since the separation of two topological phases requires in theory the consideration of all paths, the gap condition could sometimes create major hurdles.

\section*{Abstraction pays dividends}

\begin{figure}
\centering
\includegraphics[width=4.5in]{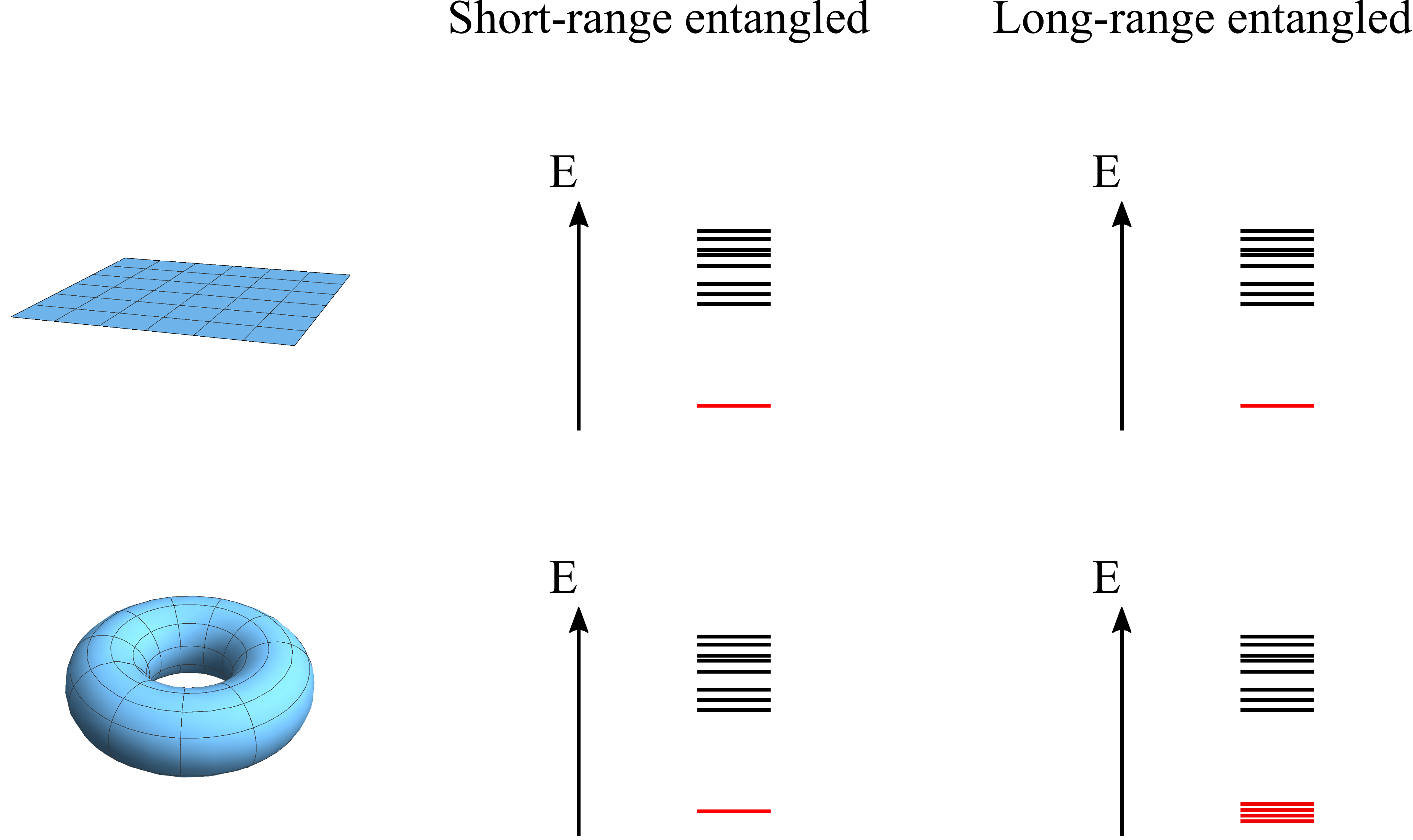}
\caption[Ground-state degeneracy of short-range entangled versus long-range entangled systems.]{Ground-state degeneracy of short-range entangled (left) versus long-range entangled systems (right) on spatial manifolds with trivial (top)---such as a Euclidean space $\mathbb E^d$ or its compactification $\SS^d$---versus nontrivial (bottom)---such as a $d$-torus $T^d$---topology.\label{fig:gsd}}
\end{figure}

In this thesis, we are specifically interested in the so-called \emph{symmetry protected topological phases}, or SPT phases \cite{Wen_Definition}, which are a subset of topological phases whose entanglement patterns are ``short-ranged." There are many manifestations of short-range entanglement. In contrast to ``long-range entangled"---or ``topologically ordered"---systems, which, when put on a spatial manifold with nontrivial topology, have topologically protected ground-state degeneracy, short-range entangled systems have a unique ground state on all spatial manifolds (see Figure \ref{fig:gsd}). While 2D long-range entangled systems have particle-like excitations above the ground state---or ``quasiparticles"---with anyonic statistics, 2D short-range entangled systems have no such anyonic excitations (see Figure \ref{fig:braiding}). This does not mean that SPT phases are trivial. Instead, SPT phases still exhibit anomalous boundaries, which either cannot be gapped out without breaking the symmetry, or can be symmetrically gapped out but still cannot exist as stand-alone lower-dimensional systems---SPT phases serve as the bulk that cancels boundary anomalies (see Figure \ref{fig:anomaly_inflow}). Typically, in discussing SPT phases, one assumes a symmetry $G$, fixes a dimension $d$, and considers $d$-dimensional SPT phases that have $G$ as a symmetry. The goal of the collection of works that make up this thesis is to classify and systematically construct bosonic and fermionic SPT phases for a general dimension $d$ and symmetry group $G$.

\begin{figure}
\centering
\includegraphics[width=3in]{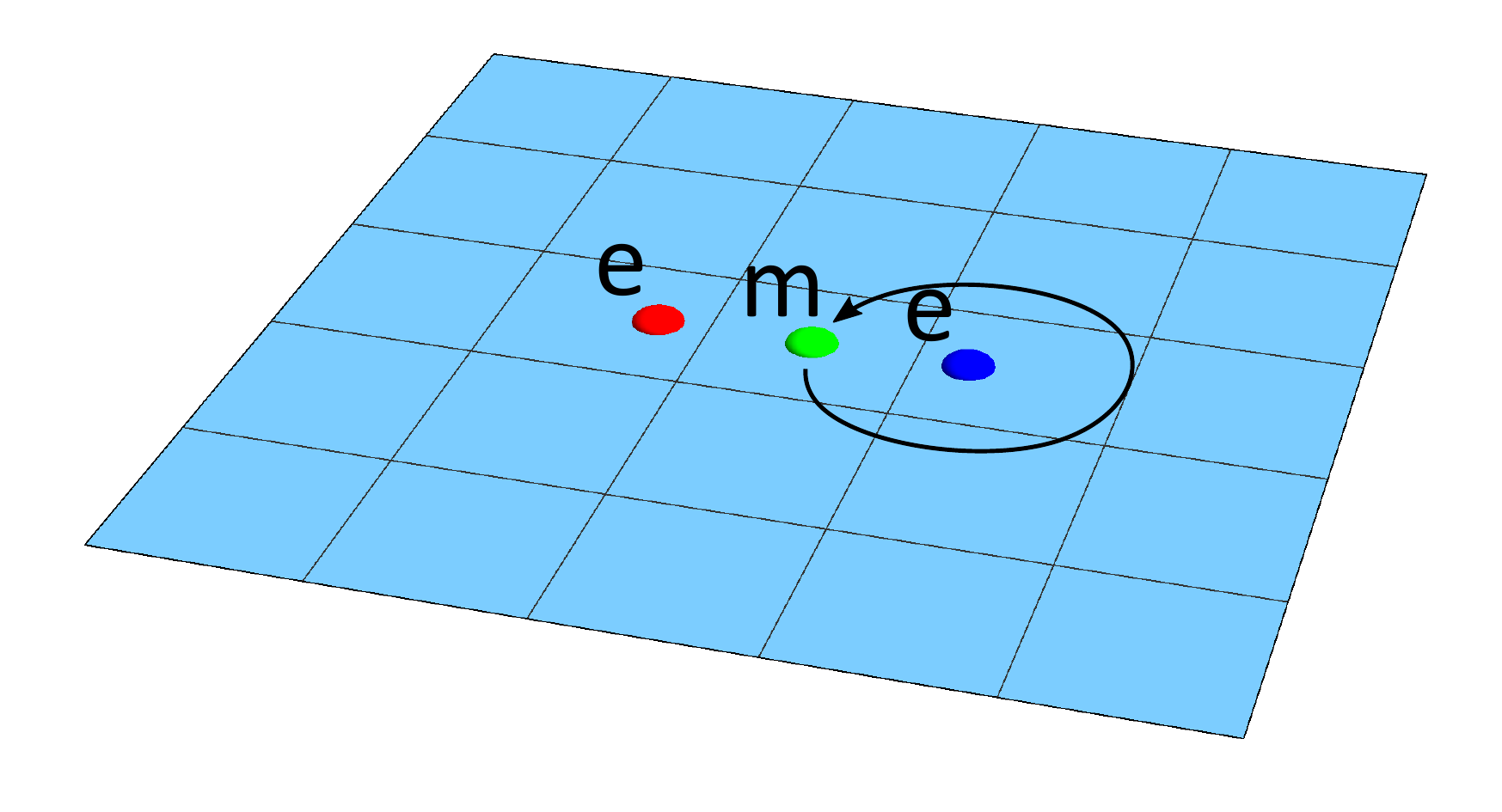}
\caption[Quasipartiples of a 2D long-range entangled system.]{In two dimensions, long-range entangled---or topologically ordered---systems have particle-like excitations above the ground state with anyonic statistics.\label{fig:braiding}}
\end{figure}

Before launching into a treatise on the classification and construction of SPT phases, let us take a step back and look at the studies of some of the other kinds of systems. First, suppose we do not allow interactions, so the only systems we can have are non-interacting fermionic systems. Non-interacting fermions in periodic potentials can be studied using band theory, if there is charge conservation. Suppose we further assume all bands are either completely filled or completely empty, so the systems we have are insulators. At each momentum $\boldsymbol k$, the filled states span a subspace of the vector space spanned by atomic orbitals. As $\boldsymbol k$ moves around the Brillouin zone, which is a $d$-torus $T^d$ for a $d$-dimensional system, the filled bands will trace out a vector bundle over $T^d$. This vector-bundle viewpoint has led to the $K$-theory classification of all topological insulators and superconductors in all dimensions for all ten symmetry classes that involve time-reversal, particle-hole, or chiral symmetries \cite{Kitaev_TI}, which revealed that the classification is periodic in the symmetry class and dimension $d \mod 8$ \cite{Hasan_Kane}.


\begin{figure}
\centering
\includegraphics[width=2.5in]{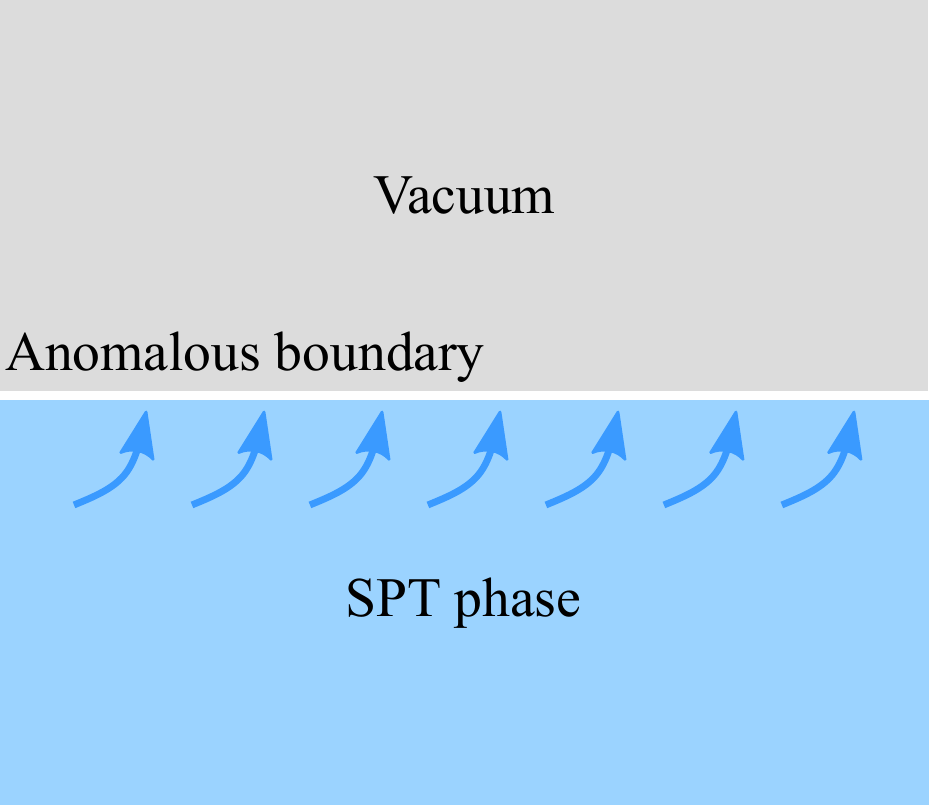}
\caption[Cancellation of boundary anomalies by SPT phases in the bulk.]{SPT phases serve as the bulk that cancel boundary anomalies.\label{fig:anomaly_inflow}}
\end{figure}

For a second example, let us reintroduce interactions and consider systems that are long-range entangled. Let us restrict to two dimensions. 2D long-range entangled systems have anyonic excitations above the ground state, which can be braided around each other, fused together into other anyons, or twisted on their own. Such a system can be abstracted into a \emph{modular tensor category} \cite{Kitaev_honeycomb}: the anyons become the \emph{objects} of the category, and physical operations like braiding, fusion, and twist become the \emph{morphisms} between the objects. This abstraction opened the way for the study of long-range entangled systems by the powerful machinery of tensor categories, topological quantum field theories, and modular functors \cite{bakalov2001lectures}. This has led to, among others, a classification of topological orders that have $\leq 4$ superselection sectors \cite{rowell2009classification}.

A classification using tensor categories is essentially a classification of topological phases modulo SPT phases. The classification of SPT phases themselves can be just as hard.
From the two examples above, we have learned that abstraction pays dividends. The question now is this: is there a general mathematical framework for the classification of SPT phases in the same way as there are tensor categories for topological orders and $K$-theory for topological insulators and superconductors?

\section*{A new framework for the classification and construction of SPT phases}

In this thesis, we will develop and apply a new, minimalist framework for the classification of SPT phases, in which we model the space of short-range entangled states by what are called \emph{$\Omega$-spectra} and the classification of SPT phases by what are called \emph{generalized cohomology theories}. This framework will give a unified view on the various proposed classifications of SPT phases \cite{Wen_Boson, Kapustin_Boson, Wen_Fermion, Kapustin_Fermion, Freed_SRE_iTQFT, Freed_ReflectionPositivity, wang2018towards, Huang_dimensional_reduction}, which are inequivalent from each other and are approximations to the unknown, true classification. This framework is based on the series of talks \cite{Kitaev_Stony_Brook_2011_SRE_1, Kitaev_Stony_Brook_2013_SRE, Kitaev_IPAM} by Kitaev. In our works \cite{Xiong, Xiong_Alexandradinata, Shiozaki2018, SongXiongHuang}, we developed the framework in full detail and demonstrated its power with a number of applications:
\begin{enumerate}
\item We classified and constructed 3D fermionic SPT phases in Wigner-Dyson classes A and AII with an additional glide symmetry. We showed that the so-called ``hourglass fermions" \cite{Ma_discoverhourglass, Hourglass, Cohomological} are robust to strong interactions. \label{application1}

\item We classified and constructed 3D bosonic SPT phases with space-group symmetries for all 230 space groups. Previously, only an incomplete classification had existed, which did not include phases built from the so-called $E_8$ state \cite{Kitaev_honeycomb, 2dChiralBosonicSPT, Kitaev_KITP}. \label{application2}

\item We gave a general rule for determining the classification of SPT phases with reflection symmetry given the classification of SPT phases without reflection symmetry. This is an application of the so-called Mayer-Vietoris long exact sequence. \label{application3}

\item We demonstrated that the structure of SPT phases with crystalline symmetries can in general be understood in terms of the so-called Atiyah-Hirzebruch spectral sequence. We suggested that there exist previously unconsidered higher-order gluing conditions and equivalence relations, which correspond to higher-order differentials in the Atiyah-Hirzebruch spectral sequence. \label{application4}
\end{enumerate}

The organization of this thesis is as follows. In Chapter \ref{chap:minimalist}, we will develop the aforementioned minimalist framework for the classification of SPT phases. In Secs.\,\ref{sec:glide} and \ref{sec:3D_beyond_group_cohomology} of Chapter \ref{chap:applications}, we will present applications \ref{application1} and \ref{application2} above, respectively. In Sec.\,\ref{sec:space_groups}, we will discuss the classifying spaces of space groups. In Secs.\,\ref{sec:Mayer-Vietoris} and \ref{sec:Atiyah-Hirzebruch} of Chapter \ref{chap:advanced}, we will present applications \ref{application3} and \ref{application4} above, respectively. In Chapter \ref{chap:conclusion}, we will conclude the thesis.


\chapter{A minimalist framework for the classification of SPT phases}
\label{chap:minimalist}

In this chapter, we develop the generalized cohomology framework for the classification of symmetry protected topological phases, or SPT phases. We will focus on bosonic SPT phases until Sec.\,\ref{subsec:fermionic_SPT}, at which point we will discuss the necessary changes for fermionic SPT phases. We will begin in Secs.\,\ref{sec:short-range entangled_SPT} by precisely defining the space of physical systems we are interested in, and then proceed to define the notion of topological phases. One can define topological phases as path-connected components of either the space of Hamiltonians or the space of ground states. In Sec.\,\ref{subsec:whe}, we will show, under reasonable assumptions, that the two spaces have the same set of components and are perhaps even weakly homotopy equivalent, so the definition of topological phases is unambiguous. We will identify a subset of the topological phases as the SPT phases. There are two distinct definitions of SPT phases: we will review the traditional one in Sec.\,\ref{subsec:traditional_definition} and give a modern definition, which we adopt in this thesis, in Sec.\,\ref{subsec:modern_definition}. In Sec.\,\ref{sec:short-range entangled_as_spectrum}, we make the point that the spaces $F_d$ of $d$-dimensional short-range entangled states form what is called an $\Omega$-spectrum. We will discuss the mathematical implications of this statement, which includes a relation between the homotopy groups of the spaces $F_d$ for different $d$. The physical reason why short-range entangled states form an $\Omega$-spectrum is given in Sec.\,\ref{sec:physical_meaning_of_spectrum}. We will see that the defining condition of an $\Omega$-spectrum can be interpreted by considering either the process of pumping a lower-dimensional short-range entangled state to the boundary in a cyclic adiabatic evolution, or a continuous spatial pattern of short-range entangled states. In Sec.\,\ref{sec:from_spectrum_to_classification}, we recall that every $\Omega$-spectrum defines what is called a generalized cohomology theory. We will then present the key statement of the generalized cohomology framework, known as the generalized cohomology hypothesis, which relates the classification of SPT phases to generalized cohomology theories. Finally, in Sec.\,\ref{sec:existing_proposals}, we will discuss how existing proposals for the classification of SPT phases fit into our general framework.

The discussions in this chapter are largely based on my work \cite{Xiong}. The generalization to antiunitary symmetries is based on App.\,A of my work \cite{Xiong_Alexandradinata} with A.\,Alexandradinata. The precise, general formulation where symmetries can be simultaneously spatial and internal, as at the beginning of Sec.\,\ref{sec:short-range entangled_SPT}, is new. In Sec.\,\ref{subsec:whe}, the argument that the space of Hamiltonians and the space of ground states have the same set of path-connected components and are perhaps even weakly homotopy equivalent is new, though I suspect this was the argument behind a similar claim in talk\,\cite{Kitaev_IPAM}, which was made without any elaboration. The precise, general formulation for fermionic systems, as in Sec.\,\ref{subsec:fermionic_SPT}, is based on App.\,A of work \cite{Xiong_Alexandradinata}. In any case, the formulation in Secs.\,\ref{sec:short-range entangled_SPT}, \ref{sec:from_spectrum_to_classification}, and \ref{subsec:fermionic_SPT} is done with more care than in previously published materials.

\section{Definition of short-range entangled states and SPT phases}
\label{sec:short-range entangled_SPT}

We will first define bosonic short-range entangled states and SPT phases, and consider the fermionic counterpart later. We fix a dimension $d$, a lattice\footnote{By a space group $\mathcal S$ we mean a crystallographic space group, i.e.\,a group of isometries of $\mathbb E^d$ that has a compact fundamental domain \cite{fundamental_domain}. By a lattice $\mathcal L$ we mean an $\mathcal S$-invariant subset of $\mathbb E^d$ that consists of one point from each cell of each dimension of an $\mathcal S$-CW structure on $\mathbb E^d$. For example, for $\mathbb E^1$ and the $\mathcal S$ generated by translation $x \mapsto x+1$ and reflection $x \mapsto -x$, if we choose the standard $\mathcal S$-CW structure, then $\mathcal L$ will be
\begin{equation}
\ZZZ \cup \braces{x + \frac12 \Big| x \in \ZZZ} \cup \braces{x + \delta | x \in \ZZZ} \cup \braces{x - \delta | x \in \ZZZ}
\end{equation}
for some $0 < \delta < \frac12$, which contains all reflection-symmetric points as well as representatives of non-reflection-symmetric 1-cells.} $\mathcal L$ with space group $\mathcal S$ defined on the Euclidean space $\mathbb E^d$, a symmetry group $G$, a homomorphism $\rho_{\rm spa}: G \fromto \mathcal S$ of $G$ into the space group, and a homomorphism $\phi_{\rm int}: G \fromto \braces{\pm 1}$. We consider all $d$-dimensional \emph{bosonic}, \emph{local}, \emph{gapped}, \emph{$G$-symmetric} Hamiltonians $\hat H$ defined on $\mathcal L$. The terminology is explained:
\begin{enumerate}
\item $\hat H$ being bosonic means that there is a finite-dimensional Hilbert space $\mathscr H$ associated to all sites\footnote{More generally, one can associate different finite-dimensional Hilbert spaces to different sites in $\mathcal L$ such that sites related by $\mathcal S$ have isomorphic Hilbert spaces. However, this does not make any difference in the classification of phases as long as we are allowed to embed Hilbert spaces into larger ones. It follows from our definition of a lattice that there are at most finitely many sites per fundamental domain. If these sites have Hilbert spaces $\mathscr H_i$ for $i = 1, 2, \ldots$, then we can embed them all into $\mathscr H \coloneq \bigoplus_i \mathscr H_i$.} in $\mathcal L$ and that $\hat H$ is a sum of tensor products of operators on different copies of $\mathscr H$.

\item $\hat H$ being local means that there is a $k$ such that the aforementioned tensor products have finite supports with linear size $< k$.

\item $\hat H$ being gapped means that $\hat H$ has a unique gapped ground state in the thermodynamic limit.

\item $\hat H$ being $G$-symmetric means that there is a representation $\rho_{\rm int}$ of the \emph{action groupoid} $G \ltimes_{\rho_{\rm spa}} \mathcal L$ by unitary or antiunitary operators on $\mathscr H$ such that $\hat H$ commutes with joint action of $\rho_{\rm spa}$ and $\rho_{\rm int}$. By definition, such a $\rho_{\rm int}$ is a map
\begin{equation}
\rho_{\rm int}: G \times \mathcal L \fromto \braces{\mbox{unitary or antiunitary operators on $\mathscr H$}}
\end{equation}
that satisfies
\begin{eqnarray}
\rho_{\rm int}(e, \bx) &=& \hat{\mathbb I}, \\
\rho_{\rm int}\paren{g_1, \rho_{\rm spa}(g_2) \bx} \rho_{\rm int}\paren{g_2, \bx} &=& \rho_{\rm int}\paren{g_1g_2, \bx},
\end{eqnarray}
for all $g_1, g_2 \in G$ and $\bx \in \mathcal L$, where $e \in G$ denotes the identity element. Specifically, we require $\rho_{\rm int}(g, \bx)$ to be antiunitary if and only if $\phi_{\rm int}(g) = -1$. The action $\hat H$ commutes with is explicitly
\begin{eqnarray}
\rho(g): \bigotimes_{\bx \in \mathcal L} \mathscr H_{\bx} &\fromto& \bigotimes_{\bx \in \mathcal L} \mathscr H_{\bx} \nonumber\\
\bigotimes_{\bx \in \mathcal L} \ket{\psi_{\bx}} &\mapsto& \bigotimes_{\bx \in \mathcal L} \rho_{\rm int}(g, \rho_{\rm spa}(g)^{-1} \bx) \ket{\psi_{\rho_{\rm spa}(g)^{-1} \bx}}.
\end{eqnarray}
\end{enumerate}
Physically, $\rho_{\rm spa}$ and $\rho_{\rm int}$ describe how $G$ permutes the lattice sites and how it acts internally, respectively. The $\bx$ in $\rho_{\rm int}(g, \bx)$ denotes the site to which $g$ is applied, since the internal representation can depend on it. If $\rho_{\rm spa}$ is chosen to be trivial, then we will have a symmetry action that is purely internal.
The homomorphism $\phi_{\rm int}$ keeps track of which elements of $G$ act unitarily on $\mathscr H$ and which antiunitarily. We may define another homomorphism $\phi: G \fromto \braces{\pm 1}$ such that
\begin{equation}
\phi(g) = \phi_{\rm int}(g) \phi_{\rm spa}(g), \label{phi_homomorphism}
\end{equation}
where $\phi_{\rm spa}(g) = -1$ if $\rho_{\rm spa}(g)$ is orientation-reversing and $1$ otherwise, which will be used later. Note that $\mathcal S$ is only the symmetry of the lattice on which $\hat H$ is defined. We did not define an action of $\mathcal S$ on $\mathscr H$, and $\hat H$ is not required to respect $\mathcal S$, but rather only $G$.

Within the space of $d$-dimensional bosonic, local, gapped, $G$-symmetric Hamiltonians on $\mathcal L$, we can consider paths, that is, one-parameter families of Hamiltonians
\begin{equation}
\hat H_t ~~~\mbox{for}~~~ 0 \leq t \leq 1,
\end{equation}
which have the same $\mathscr H$, $k$, and $\rho_{\rm int}$ for all values of $t$. Essentially, we would like to define the path-connected components of this space to be the $d$-dimensional bosonic $G$-equivariant topological phases. However, this would not be a very good idea because the notion of paths depends on $\mathscr H$, $k$, and $\rho_{\rm int}$. In general, we would like to be able to compare two systems with different $\mathscr H$, $k$, or $\rho_{\rm int}$ and state whether they are in the same phase. After all, $k$ is an upper bound that can be increased, and $\mathscr H$ depends in a real system on which degrees of freedom we choose to include. The proper way to do it is to consider a \emph{direct limit} of the spaces over $\mathscr H$, $k$, and $\rho_{\rm int}$. Practically, this means that to compare system $0$ which has $\mathscr H_0$, $k_0$, and $\rho_{\rm int, 0}$ and system $1$ which has $\mathscr H_1$, $k_1$, and $\rho_{\rm int, 1}$, we are allowed to relax the upper bound to some
\begin{equation}
k \geq k_0, k_1,
\end{equation}
and embed the Hilbert spaces into some
\begin{equation}
\mathscr H \supset \mathscr H_0, \mathscr H_1,
\end{equation}
such that
\begin{equation}
\left. \rho_{\rm int} \right| {\mathscr H_0} = \rho_{\rm int, 0}, ~~~ \left. \rho_{\rm int} \right| {\mathscr H_1} = \rho_{\rm int, 1},
\end{equation}
where $\rho_{\rm int}$ is a representation of $G \ltimes_{\rho_{\rm spa}} \mathcal L$ on $\mathscr H$. We modify $\hat H_0$ (resp.\,$\hat H_1$) to give states orthogonal to $\mathscr H_0$ (resp.\,$\mathscr H_1$) higher energy so that the ground state remains in $\mathscr H_0 \subset \mathscr H$ (resp.\,$\mathscr H_1 \subset \mathscr H$). We are then to find a path between $\hat H_0$ and $\hat H_1$ with respect to the common $k$, $\mathscr H$, and $\rho_{\rm int}$. We say systems $0$ and $1$ are in the same \emph{$d$-dimensional bosonic $G$-equivariant topological phase} if such $k$, $\mathscr H$, and $\rho_{\rm int}$ can be found so that there is a path between $\hat H_0$ and $\hat H_1$. This is a kind of ``stable equivalence" \cite{Wen_1d, Cirac}, analogous to how one is allowed to add trivial bands in the classification of topological insulators \cite{Kitaev_TI}.

\subsection{Weak homotopy equivalence between space of Hamiltonians and space of states}
\label{subsec:whe}

The space considered above is a space of Hamiltonians. It is easy to see the phase of a Hamiltonian depends only on its ground state. Suppose $\hat H_0$ and $\hat H_1$ shared a gapped ground state $\ket{\Psi}$ and both respected $\rho$ (after any necessary embedding). Then $\ket{\Psi}$ would also be the ground state of
\begin{equation}
\hat H_t = (1-t) \hat H_0 + t \hat H_1 \label{linear_interpolation}
\end{equation}
for all $0 \leq t \leq 1$, which respects $\rho$ and has a gap $\Delta_t \geq (1-t) \Delta_0 + t \Delta_1 > 0$, where $\Delta_0, \Delta_1 > 0$ are the gaps of $\hat H_0$ and $\hat H_1$, respectively. Therefore, it would be desirable to introduce a space of states, namely the space of ground states of $d$-dimensional bosonic, local, gapped, $G$-symmetric Hamiltonians on $\mathcal L$. This is a quotient space of the space of Hamiltonians, where Hamiltonians with the same ground state are identified.

We can speak of paths in the space of states like we did for the space of Hamiltonians and define the notation of phase accordingly. Obviously, if there is a path between two Hamiltonians, then there is also a path between their ground states. Conversely, it can be argued that if there is a path between two ground states, then there is also a path between the Hamiltonians. This depends on the assumption that the map from the space of Hamitonians to the space of ground states has the \emph{path lifting property}: given ground states $\ket{\Psi_t}$ for $0 \leq t \leq 1$ and a parent Hamiltonian $\hat H_0$ of $\ket{\Psi_0}$, there exists a path $\hat{\tilde H}_t$ for $0 \leq t \leq 1$ such that $\hat{\tilde H}_0 = \hat H_0$ and that $\hat{\tilde H}_t$ is a parent Hamiltonian of $\ket{\Psi_t}$ for all $0 \leq t \leq 1$. The endpoint $\hat{\tilde H}_1$ of the lift may not be $\hat H_1$, but since they share ground state $\ket{\Psi_1}$, there is a path between them by the previous paragraph; it follows that there is a path between $\hat H_0$ and $\hat H_1$ themselves. To my knowledge, there has not been a proof of the path lifting property in the literature,\footnote{Note that what Chen, Gu, and Wen argued in Ref.\,\cite{Wen_Definition} was something different, which is there is a path between parent Hamiltonians if and only if there is a ``local unitary evolution" that relates the ground states, that is,
\begin{equation}
\ket{\Psi_1} = \T \brackets{e^{- i \int_0^1 dt \hat{\tilde H}_t }} \ket{\Psi_0}
\end{equation}
for some $\hat{\tilde H}_t$ for $0 \leq t \leq 1$. Since the existence of a path between ground states does not a priori imply the existence of a local unitary evolution, the ``if" part of the claim in Ref.\,\cite{Wen_Definition} does not imply the ``if" part of our claim that there is a path between parent Hamiltonians if and only if there is a path between ground states.} except when some ansatz such as injective\footnote{The definition of ``injective" in Ref.\,\cite{Cirac} follows that of Ref.\,\cite{MPS_fundamental_theorem}. Note that sometimes, e.g.\,in Ref.\,\cite{MPS_geometry}, the word ``injective" is used to mean something weaker, which is called ``normal" in Ref.\,\cite{MPS_fundamental_theorem}.} matrix product states is assumed \cite{Cirac}. However, the path lifting property will hold as long as the map from the space of Hamiltonians to the space of ground states is a \emph{Serre fibration}, and either of these would seem like a reasonable assumption to make. In conclusion, under reasonable assumptions, the notion of \emph{$d$-dimensional bosonic $G$-equivariant topological phase} is unambiguous whether it is defined using ground states or Hamiltonians. We shall therefore focus on the space of states from now on.

In fact, if the map from the space of Hamiltonians to the space of ground states is a Serre fibration, then we can show that the two spaces are \emph{weakly homotopy equivalent}. Every Serre fibration $X \fromto Y \fromto Z$ admits a long exact sequence of homotopy groups
\begin{eqnarray}
\cdots \fromto
\pi_1\paren{X, x_0} \fromto \pi_1\paren{Y, y_0} \fromto \pi_1\paren{Z, z_0} \fromto
\pi_0\paren{X, x_0} \fromto \pi_0\paren{Y, y_0} \fromto \pi_0\paren{Z, z_0},
\end{eqnarray}
where $x_0$, $y_0$, and $z_0$ denote some choice of basepoints of $X$, $Y$, and $Z$, respectively, such that $y_0$ is a lift of $z_0$ and $x_0$ is the restriction of $y_0$. We take $Y$ to be the space of Hamiltonians and $Z$ to be the space of ground states. The fiber $X$ is the inverse image of a ground state---the basepoint $z_0$ of $Z$---with respect to the map from the space of Hamiltonians to the space of ground states. It is thus a space of Hamiltonians that share the same ground state. Given any $k \geq 0$ and any pointed map,
\begin{eqnarray}
\paren{\SS^k, s_0} &\fromto& \paren{X, x_0}, \\
s &\mapsto& \hat H_s,
\end{eqnarray}
from the $k$-sphere into the space of Hamiltonians, we have a nulhomotopy,
\begin{equation}
\hat{\tilde H}_{t, s} = (1 - t) \hat H_{s_0} + t \hat H_s,
\end{equation}
for $0 \leq t \leq 1$. If $\Delta_s > 0$ is the gap of $\hat H_s$ and $\ket{\Psi_0}$ is the common ground state of $\hat H_s$ for all $s \in \SS^k$, then $\hat{\tilde H}_{t, s}$ will also have $\ket{\Psi_0}$ as the ground state and have a gap
\begin{equation}
\Delta_{t, s} \geq (1 - t) \Delta_{s_0} + t \Delta_s > 0,
\end{equation}
for all $0 \leq t \leq 1$ and $s \in \SS^k$. This shows that $X$ is weakly contractible, that is,
\begin{equation}
\pi_k\paren{X, x_0} = 0, ~~\forall k.
\end{equation}
Returning to the long exact sequence, we thus see that the map from $Y$ to $Z$ induces isomorphisms
\begin{equation}
\pi_k\paren{Y, y_0} \isomorphic \pi_k\paren{Z, z_0}
\end{equation}
for all $k \geq 1$, and an injection $\pi_0\paren{Y, y_0} \fromto \pi_0\paren{Z, z_0}$. Since the map from the space of Hamiltonians to the space of ground states is by construction surjective, the injection $\pi_0\paren{Y, y_0} \fromto \pi_0\paren{Z, z_0}$ must be a bijection. Since this works for any choice of basepoints, it follows that the map is a weak homotopy equivalence from $Y$ to $Z$.

We are interested in a subset of the set of $d$-dimensional bosonic $G$-equivariant topological phases, called \emph{symmetry protected topological phases}, or SPT phases. If a state represents an SPT phase for some choice of $G$, then it is called a \emph{short-range entangled state}. There is often confusion surrounding the terms ``SPT phase" and ``short-range entangled state" because there are different versions to their definitions. In this thesis, we adopt the more modern definitions, which are based on the ``invertibility" of a phase \cite{Kitaev_Stony_Brook_2011_SRE_2, Kitaev_Stony_Brook_2013_SRE, Kapustin_Boson, Freed_SRE_iTQFT, Freed_ReflectionPositivity, McGreevy_sSourcery, Xiong}. Due to its historical significance, we shall nevertheless review the more traditional definitions, which are based on the ``trivializability" of a state \cite{Wen_Definition, Cirac}.

\subsection{Traditional definition of SPT phases}
\label{subsec:traditional_definition}

In the traditional definition \cite{Wen_Definition, Cirac}, without regard to symmetry, one defines a \emph{trivial product state} to be a state of the form
\begin{equation}
\bigotimes_{\bx \in \mathcal L} \ket{\psi_{\bx}} \in \bigotimes_{\bx \in \mathcal L} \mathscr H_{\bx},
\end{equation}
which has no entanglement between different sites. Any such state has a bosonic, local, gapped parent Hamiltonian. Furthermore, between any pair of such states, there is a path consisting of trivial product states. Then, one defines a bosonic \emph{short-range entangled state} to be the ground state of any bosonic, local, gapped Hamiltonian for which there is a path connecting it to any trivial product state.

Introducing symmetry, a $d$-dimensional bosonic $G$-equivariant topological phase is called an \emph{SPT phase} if any representative of this phase is a short-range entangled state, or in other words, if any representative of this phase can be connected to a trivial product state by breaking the symmetry. Thus, a generic SPT phase is nontrivial only due to the protection by the symmetry. This motivated the name ``symmetry protected topological phases."

\subsection{Modern definition of SPT phases}
\label{subsec:modern_definition}

In the modern definition \cite{Kitaev_Stony_Brook_2011_SRE_2, Kitaev_Stony_Brook_2013_SRE, Kapustin_Boson, Freed_SRE_iTQFT, Freed_ReflectionPositivity, McGreevy_sSourcery, Xiong}, one does not start with trivial product states and defines short-range entangled states and SPT phases in terms of them. Instead, one starts directly with the set of $d$-dimensional bosonic $G$-equivariant topological phases, and defines a subset of it to be the SPT phases. The resulting definition of SPT phases will be more relaxed, in that every SPT phase in the traditional sense is an SPT phase in the modern sense but that the converse may not hold. The modern definition depends on the ``invertibility" of a phase with respect to a ``stacking" operation, which we describe now.

Given any pair of $d$-dimensional bosonic, local, gapped, $G$-symmetric ground states defined on $\mathcal L$, the \emph{stacking} operation combines them into a new $d$-dimensional bosonic, local, gapped, $G$-symmetric ground state defined on $\mathcal L$. Suppose the two ground states $\ket{\Psi_0}$ and $\ket{\Psi_1}$ have on-site Hilbert spaces $\mathscr H_0$ and $\mathscr H_1$ and representations $\rho_{\rm int, 0}$ and $\rho_{\rm int, 1}$ respectively. Then the stacked state is by definition
\begin{equation}
\ket{\Psi} = \ket{\Psi_0} \otimes \ket{\Psi_1},
\end{equation}
which is defined on the same lattice $\mathcal L$ with on-site Hilbert space
\begin{equation}
\mathscr H = \mathscr H_0 \otimes \mathscr H_1,
\end{equation}
on which the symmetry acts according to
\begin{equation}
\rho_{\rm int} = \rho_{\rm int, 0} \otimes \rho_{\rm int, 1}.
\end{equation}
It is easy to see that $\ket{\Psi}$ is also a gapped ground state. If $\ket{\Psi_0}$ and $\ket{\Psi_1}$ had parent Hamiltonians $\hat H_0$ and $\hat H_1$, respectively, then $\hat H_0 \otimes \hat{\mathbb I} + \hat{\mathbb I} \otimes \hat H_1$ would be a parent Hamiltonian of $\ket{\Psi}$. One can picture the stacking operation as interlaying two systems without coupling them. The lattice sites from the two systems are combined to form the composite lattice sites of the stacked system. See Figure \ref{fig:stacking}. We will use `+' to denote the stacking operation, as in
\begin{equation}
a + b
\end{equation}
is the state obtained by stacking $a$ and $b$.

\begin{figure}
\centering
\includegraphics[width=0.55 \columnwidth]{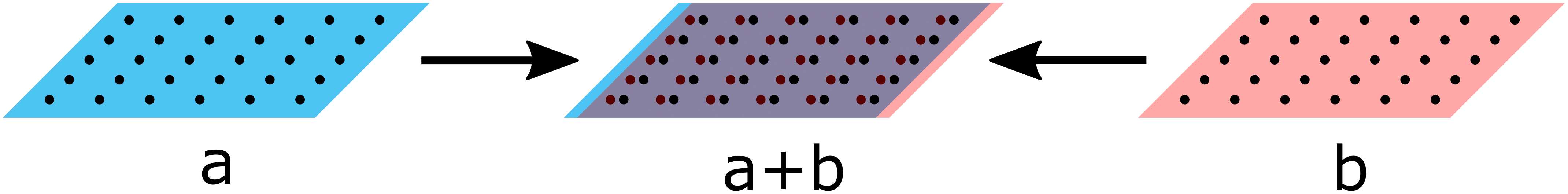}
\caption[The stacking operation.]{The stacking operation combines two $d$-dimensional bosonic, local, gapped, $G$-symmetric ground states defined on $\mathcal L$ into a new $d$-dimensional bosonic, local, gapped, $G$-symmetric ground states defined on $\mathcal L$.\label{fig:stacking}}
\end{figure}

Above we defined the stacking operation for states. It is easy to see that the stacking operation is well-defined at the level of phases, as a deformation of either $\ket{\Psi_0}$ or $\ket{\Psi_1}$ constitutes a legitimate deformation of $\ket{\Psi}$ as well. Stacking thus gives the set of $d$-dimensional bosonic $G$-equivariant topological phases the structure of a \emph{commutative monoid}. This means that we have
\begin{enumerate}
\item Associativity: for any phases $[a]$, $[b]$, and $[c]$, we have $([a] + [b]) + [c] = [a] + ([b] + [c])$.

\item Commutativity: for any phases $[a]$ and $[b]$, we have $[a] + [b] = [b] + [a]$.

\item Identity element: there is a phase $[0]$ such that for any phase $[a]$, we have $[0] + [a] = [a] + [0] = [a]$.
\end{enumerate}
Associativity is tautological. The identity element is represented by any trivial product state and the trivial action $\rho_{\rm int}$. As for commutativity, it follows from the fact that there is a path between $\ket{\Psi_0} \otimes \ket{\Psi_1}$ and $\ket{\Psi_1} \otimes \ket{\Psi_0}$ for any systems $\ket{\Psi_0}$ and $\ket{\Psi_1}$. More explicitly, we embed $\mathscr H_0 \otimes \mathscr H_1$ and $\mathscr H_1 \otimes \mathscr H_0$ into
\begin{equation}
\paren{\mathscr H_0 \otimes \mathscr H_1} \oplus \paren{\mathscr H_1 \otimes \mathscr H_0} \isomorphic \mathscr H_0 \otimes \mathscr H_1 \otimes \CCC^2. \label{commutativity_ancilla}
\end{equation}
We think of $\CCC^2$ as an ancillary qubit and consider the one-parameter family of Hamiltonians
\begin{equation}
\hat H_t = \hat H_0 \otimes \hat{\mathbb I} \otimes \hat{\mathbb I} + \hat{\mathbb I} \otimes \hat H_1 \otimes \hat{\mathbb I} - \hat{\mathbb I} \otimes \hat{\mathbb I} \otimes \hat P_t, \label{commutativity_path}
\end{equation}
where $\hat H_0$ and $\hat H_1$ are parent Hamiltonians of $\ket{\Psi_0}$ and $\ket{\Psi_1}$, respectively, and
\begin{equation}
\hat P_t = \sum_{\bx \in \mathcal L} \paren{\cos\frac{\pi t}{2} \ket{0}_{\bx} + \sin\frac{\pi t}{2} \ket{1}_{\bx}}\paren{\cos\frac{\pi t}{2} \bra{0}_{\bx} + \sin\frac{\pi t}{2} \bra{1}_{\bx}}.
\end{equation}
Obviously, this is a gapped, symmetric Hamiltonian for all $0 \leq t \leq 1$. At $t = 0$ and $t = 1$, the third term in $\hat H_t$ forces the ancillary qubit of every site be in the $\ket{0}$ and $\ket{1}$ states, respectively. Since these correspond to the first and second summands in the left-hand side of Eq.\,(\ref{commutativity_ancilla}), we see that Eq.\,(\ref{commutativity_path}) is a path between $\ket{\Psi_0} \otimes \ket{\Psi_1}$ and $\ket{\Psi_1} \otimes \ket{\Psi_0}$.

We shall now define a bosonic \emph{SPT phase} to be a $d$-dimensional bosonic $G$-equivariant topological phase that has an inverse with respect to the stacking operation. In other words, the set of bosonic SPT phases is the group of invertible elements of the commutative monoid above. Naturally, the $d$-dimensional bosonic SPT phases with symmetry $G$ form an abelian group, denoted
\begin{equation}
\SPT_b^d(G),
\end{equation}
or
\begin{equation}
\SPT_b^d(G, \phi)
\end{equation}
for completeness. We further define a bosonic \emph{short-range entangled state} to be a $d$-dimensional bosonic, local, gapped, $G$-symmetric ground state that represents an SPT phase for some choice of symmetry.

It has been argued that the modern definition of short-range entangled states is equivalent to the condition that the Hamiltonian has a unique ground state on all (sufficiently nice---oriented, framed, etc.) manifolds, not just the Euclidean space or its one-point compactification \cite{Kitaev_Stony_Brook_2011_SRE_2, Kitaev_Stony_Brook_2013_SRE, Kapustin_Boson, Freed_SRE_iTQFT, Freed_ReflectionPositivity, McGreevy_sSourcery}. In the special case of 2D, this means that the Hamiltonian has a unique ground state on the the torus, the connected sum of two tori, etc., as well as the plane or the sphere. This in turn is believed to be equivalent to the condition that the 2D system has no anyonic excitations. Furthermore, in any dimension, it has been argued that if a Hamiltonian has a unique ground state on all manifolds, then the orientation-reversed version of itself will represent the inverse of the phase, at least when $G$ is purely internal, that is, when $\rho_{\rm spa}$ is trivial \cite{McGreevy_sSourcery}. When $\rho_{\rm spa}$ is non-trivial, it is not well-understood how to generally construct the inverse \cite{Xiong_Alexandradinata}.

It is easy to see that a short-range entangled state in the traditional sense is also a short-range entangled state in the modern sense. Taking the remarks in the previous paragraph about orientation reversal for granted, we also deduce that an SPT phase in the traditional sense must be an SPT phase in the modern sense, at least when $G$ is purely internal.\footnote{In the absence of symmetry, if a state can be connected to a trivial product state, then it has an inverse up to homotopy. The tricky part is that the existence of an inverse does not imply the existence of a symmetric inverse. The orientation-reversal argument offers a canonical construction and solves this problem.} Table \ref{table:SPT_examples} lists some examples of states that are considered short-range entangled according to the modern definition but not the traditional one.

{\small
\begin{spacing}{\dnormalspacing}
\begin{longtable}[c]{lll}
\caption[Examples of short-range entangled states according to the modern definition.]{Examples of short-range entangled states according to the modern definition.\label{table:SPT_examples}} \\
\nobreakhline\hline
Fermionic or bosonic & Dimension & System \\
\nobreakhline
\endfirsthead
\caption[]{(Continued).} \\
\nobreakhline\hline
Fermionic or bosonic & Dimension & System \\
\nobreakhline
\endhead
\nobreakhline\hline
\endfoot
Fermionic & $0$ & An odd number of fermions \\
Fermionic & $1$ & The Majorana chain \cite{Majorana_chain} \\
Fermionic & $2$ & $\paren{p+ip}$-superconductors \cite{Volovik_p+ip, Read_p+ip, Ivanov_p+ip} \\
Bosonic & $2$ & The $E_8$-model \cite{Kitaev_honeycomb, 2dChiralBosonicSPT, Kitaev_KITP} \\
\end{longtable}
\end{spacing}
}

\section{Short-range entangled states as \texorpdfstring{$\Omega$}{Omega}-spectrum}
\label{sec:short-range entangled_as_spectrum}

We shall now denote the space of $d$-dimensional bosonic short-range entangled states by
\begin{equation}
F_d.
\end{equation}
It has been argued that $F_d$ for $d = 0, 1, 2, \ldots$ form an ``$\Omega$-spectrum" \cite{Kitaev_Stony_Brook_2011_SRE_1, Kitaev_Stony_Brook_2013_SRE, Kitaev_IPAM}. The physical arguments for why this is the case will be given in the next section. For now let us examine what it entails mathematically.

By definition, an \emph{$\Omega$-spectrum} \cite{Hatcher, Adams1, Adams2} is a family of pointed topological spaces $F_n$ indexed by $n \in \ZZZ$ such that
\begin{equation}
F_n \homotopic \Omega F_{n+1} \label{Omega_spectrum}
\end{equation}
for all $n \in \ZZZ$, that is, $F_n$ is pointed homotopy equivalent to the loop space of $F_{n+1}$. The \emph{loop space} $\Omega X$ of a pointed topological space $(X, x_0)$ is the space of loops in $X$ based at $x_0$. Namely, a point in $\Omega X$ is a loop in $X$ based at $x_0$. This implies that there is a bijection between the set of path-connected components of $\Omega X$ and the fundamental group of $X$, i.e.\,$\pi_0\paren{\Omega X} \isomorphic \pi_1(X)$. In general, there are isomorphisms between homotopy groups $\pi_k\paren{\Omega X} \isomorphic \pi_{k+1}(X)$ for all $k \geq 0$. Applied to $F_n \homotopic \Omega F_{n+1}$, this gives
\begin{equation}
\pi_k\paren{F_n} \isomorphic \pi_{k+1}\paren{F_{n+1}} \label{shift_in_homotopy}
\end{equation}
for all $k \geq 0$ and $n \in \ZZZ$. Notice an $\Omega$-spectrum is indexed by integers whereas the physical dimension $d$ is nonnegative. The claim at the beginning of the section says that the spaces of $d$-dimensional bosonic short-range entangled states form the nonnegative-degree portion of an $\Omega$-spectrum. Since any $F_{n+1}$ uniquely determines $F_n$ by Eq.\,(\ref{Omega_spectrum}), this can then be extended to negative degrees to form an $\Omega$-spectrum indexed by $\ZZZ$.

A consequence of Eq.\,(\ref{shift_in_homotopy}) is that all homotopy groups of all of the spaces $F_n$ for $n \in \ZZZ$ are determined by $\pi_k\paren{F_0}$ for $k \geq 0$ and $\pi_0\paren{F_d}$ for $d \geq 0$. $F_0$ is the space of $0$-dimensional bosonic short-range entangled states, so $\pi_k\paren{F_0}$ for $k \geq 0$ are easy to understand. $\pi_0\paren{F_d}$ is the set of path-connected components of the space of $d$-dimensional bosonic short-range entangled states, so it is the classification of $d$-dimensional SPT phases when the symmetry group $G$ is trivial. The current understanding of these homotopy groups is \cite{Kitaev_Stony_Brook_2011_SRE_1, Kitaev_Stony_Brook_2013_SRE}
\begin{enumerate}
\item $\pi_k\paren{F_0} = \ZZZ$ for $k = 2$ and $0$ for $k \neq 2$ (this follows from $F_0 = \CCC P^\infty$, the direct limit of rays in finite-dimensional Hilbert spaces);

\item $\pi_0\paren{F_d} = 0$ for $d = 0, 1, 3$ and $\ZZZ$ (generated by the $E_8$ state \cite{Kitaev_honeycomb, 2dChiralBosonicSPT, Kitaev_KITP}) for $d = 2$; $\pi_0\paren{F_d}$ for $d > 3$ is not well-understood. 
\end{enumerate}
This results in the table of homotopy groups in Table \ref{table:homotopy_groups}.

Homotopy groups contain important but not all information about a spectrum. In general, $F_n$ can be reconstructed from $\pi_k(F_n)$ and $k$-invariants via \emph{Postnikov towers} \cite{Hatcher}; see App.\,A of work \cite{Xiong} for an example. Alternatively, we will see in Sec.\,\ref{sec:from_spectrum_to_classification} that the classification of SPT phases is given by the twisted equivariant generalized cohomology theory $h^{\phi+\bullet}_G\paren{\pt}$ of $F_n$, and $h^{\phi+\bullet}_G\paren{\pt}$ can be computed directly via the \emph{Leray-Serre spectral sequence} if differentials in the spectral sequence can be determined; see work \cite{Shiozaki2018}.

{\small
\begin{spacing}{\dnormalspacing}
\begin{longtable}[c]{c|ccccccc}
\caption[Homotopy groups of spaces of bosonic short-range entangled states]{Homotopy groups of the spaces $F_{d}$ of $d$-dimensional bosonic short-range entangled states for $d \leq 3$.\label{table:homotopy_groups}} \\
\endfirsthead
\caption[]{(Continued).} \\
\endhead
\endfoot
$\pi_{>5}$ & $0$ & $0$ & $0$ & $0$ & $0$ & $0$ & $0$ \\
$\pi_{5}$ & $0$ & $0$ & $0$ & $0$ & $0$ & $0$ & $\ZZZ$ \\
$\pi_{4}$ & $0$ & $0$ & $0$ & $0$ & $0$ & $\ZZZ$ & $0$ \\
$\pi_{3}$ & $0$ & $0$ & $0$ & $0$ & $\ZZZ$ & $0$ & $0$ \\
$\pi_{2}$ & $0$ & $0$ & $0$ & $\ZZZ$ & $0$ & $0$ & $0$ \\
$\pi_{1}$ & $0$ & $0$ & $\ZZZ$ & $0$ & $0$ & $0$ & $\ZZZ$ \\
$\pi_{0}$ & $0$ & $\ZZZ$ & $0$ & $0$ & $0$ & $\ZZZ$ & $0$ \\
\nobreakhline 
~ & $F_{<-2}$ & $F_{-2}$ & $F_{-1}$ & ~$F_{0}$~ & ~$F_{1}$~ & ~$F_{2}$~ & ~$F_{3}$~ 
\end{longtable}
\end{spacing}
}

\section{Physical meaning of \texorpdfstring{$\Omega$}{Omega}-spectrum}
\label{sec:physical_meaning_of_spectrum}

We now explain the physical meaning of the $\Omega$-spectrum $(F_n)_{n\in \ZZZ}$ \cite{Kitaev_Stony_Brook_2011_SRE_1, Kitaev_Stony_Brook_2013_SRE, Kitaev_IPAM}. We have said that, for a given $d$, $F_d$ can be interpreted as the space of $d$-dimensional bosonic short-range entangled states. These spaces satisfy the pointed homotopy equivalence condition
\begin{equation}
F_d \homotopic \Omega F_{d+1}
\end{equation}
for all $d \geq 0$. Concretely, this means there exists a pair of maps
\begin{eqnarray}
&&f: F_d \fromto \Omega F_{d+1}, \\
&&g: \Omega F_{d+1} \fromto F_d,
\end{eqnarray}
such that both $f \circ g$ and $g \circ f$ are pointed homotopic to the identity. We interpret the basepoint of $F_d$ as a trivial $d$-dimensional bosonic short-range entangled state, i.e. a trivial product state. It remains to interpret the maps $f$ and $g$.

There are two interpretations for $f$ and $g$. One of them is based on pumping a $d$-dimensional state to the boundary of a $(d+1)$-dimensional state \cite{Kitaev_Stony_Brook_2011_SRE_1, Kitaev_Stony_Brook_2013_SRE}. The other is based on viewing a texture of $(d+1)$-dimensional states as a $d$-dimensional domain wall \cite{Kitaev_IPAM}. The two interpretations can be shown to be equivalent. In Sec.\,\ref{subsubsec:pumping}, we present the first interpretation. In Sec.\,\ref{subsubsec:domain_wall}, we present the second.

\subsection{Pumping interpretation\label{subsubsec:pumping}}

With $F_d$ being the space of $d$-dimensional short-range entangled states, $f$ will now be a map that sends a $d$-dimensional short-range entangled state $a$ to a one-parameter family of $(d+1)$-dimensional short-range entangled states (recall $\Omega F_{d+1}$ is the loop space of $F_{d+1}$). We denote this family of $(d+1)$-dimensional short-range entangled states by $f(a)_t$, for $0 \leq t \leq 1$, which satisfies $f(a)_0 = f(a)_1$. In the pumping interpretation \cite{Kitaev_Stony_Brook_2011_SRE_1, Kitaev_Stony_Brook_2013_SRE}, we set $f(a)_t$ to be the state shown in Figure \ref{fig:pumping}(a). It is obtained by putting copies of $a$ at $x = (2n+t)L$ and copies of the inverse $\bar a$ at $x = (2n-t)L$ for all $n\in \ZZZ$, where $x$ is the additional coordinate the $(d+1)$-dimensional system has compared to $d$-dimensional systems, and $L$ is a length scale much greater than the correlation length $\xi$. As $t$ increases, the separation between $a$'s and $\bar a$'s changes. The evolution of $f(a)_t$ with $t$ is illustrated in Figure \ref{fig:pumping}(b).

\begin{figure}
\centering
\includegraphics[width=5.5in]{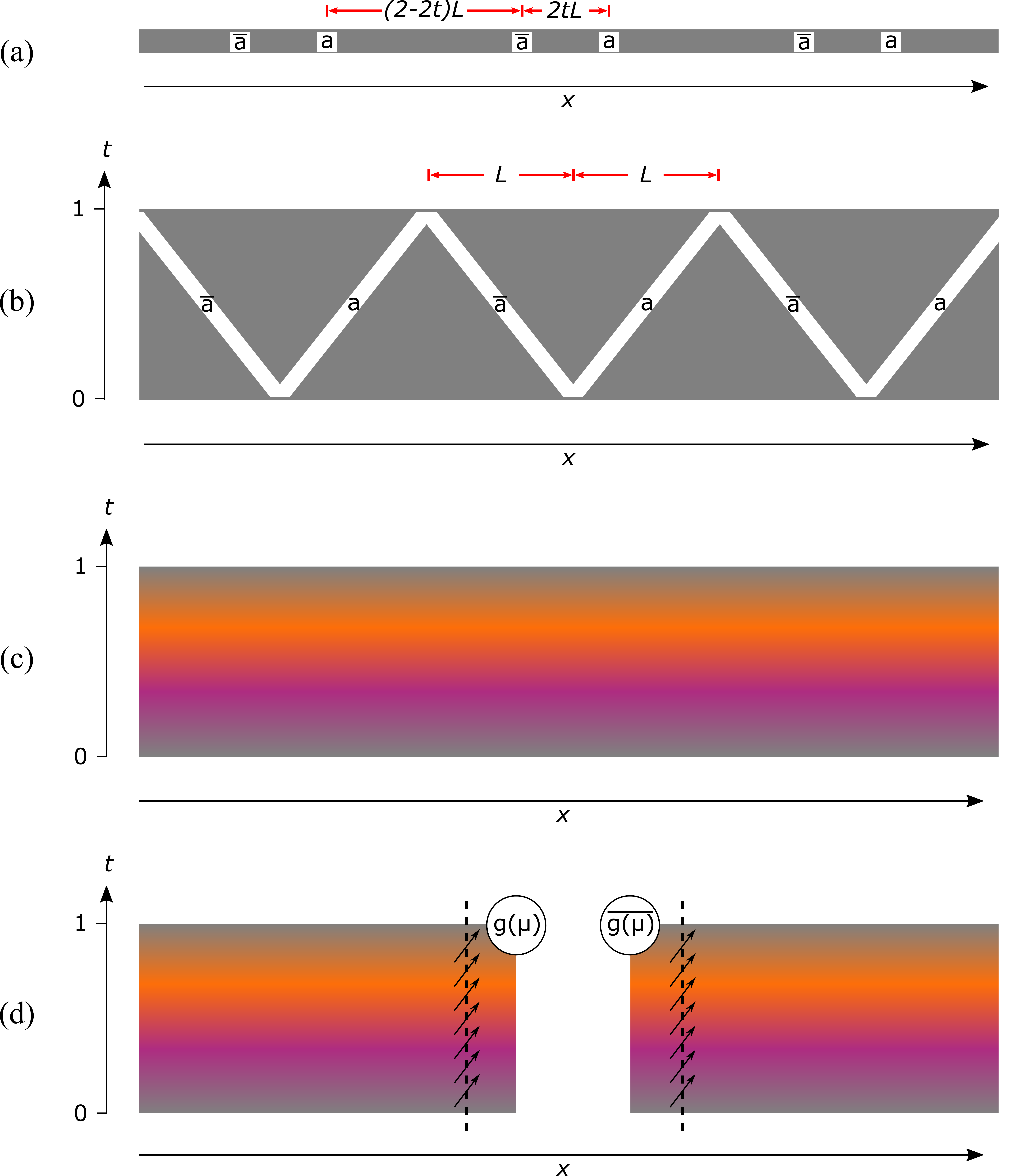}
\caption[The pumping interpretation of $\Omega$-spectrum.]{The pumping interpretation of $\Omega$-spectrum. (a) A $(d+1)$-dimensional short-range entangled state $f(a)_t$ constructed from a $d$-dimensional short-range entangled state $a$. (b) The evolution of $f(a)_t$ as $t$ varies from 0 to 1. (c) An arbitrary one-parameter family of $(d+1)$-dimensional short-range entangled states, $\mu(t)$ for $0 \leq t \leq 1$. (d) The pumping of a $d$-dimensional short-range entangled state to the boundary of a $(d+1)$-dimensional system that is cut open, in the adiabatic evolution defined by $\mu$.}
\label{fig:pumping}
\end{figure}

Conversely, given a one-parameter family of $(d+1)$-dimensional short-range entangled states, $\mu(t)$ with $0\leq t \leq 1$, the map $g$ will send it to a $d$-dimensional short-range entangled state, $g(\mu)$. To define this state, we take $t$ to be the time coordinate and regard $\mu(t)$ as defining an adiabatic evolution of a $(d+1)$-dimensional system. We will then set $g(\mu)$ to be the $d$-dimensional state that is pumped across the $(d+1)$-dimensional system in the adiabatic evolution. Put differently, if the $(d+1)$-dimensional system is cut open, the $g(\mu)$ will be the $d$-dimensional state that is pumped to the boundary of the $(d+1)$-dimensional system. This is illustrated in Figure \ref{fig:pumping}(c)(d).

In App.\,B of work \cite{Xiong}, it was shown that $f$ and $g$ are indeed pointed homotopy inverses of each other, that is, both $f\circ g$ and $g \circ f$ are pointed homotopic to the identity.

\subsection{Domain wall interpretation\label{subsubsec:domain_wall}}

In the domain wall interpretation \cite{Kitaev_IPAM}, $f$ is again a map that sends a $d$-dimensional short-range entangled state $a$ to a one-parameter family of $(d+1)$-dimensional short-range entangled states. This time, however, we will denote the family of $(d+1)$-dimensional short-range entangled states by $f(a)_x$ for $-\infty < x < \infty$, which satisfies $f(a)_{-\infty} = f(a)_{+\infty}$. We will set $f(a)_x$ to be the state shown in Figure \ref{fig:domain_wall}(a). This is equivalent to the state in Figure \ref{fig:pumping}(a) via the formal change of variable $t = \frac{1 - \tanh x}{2}$.

\begin{figure}
\centering
\includegraphics[width=6in]{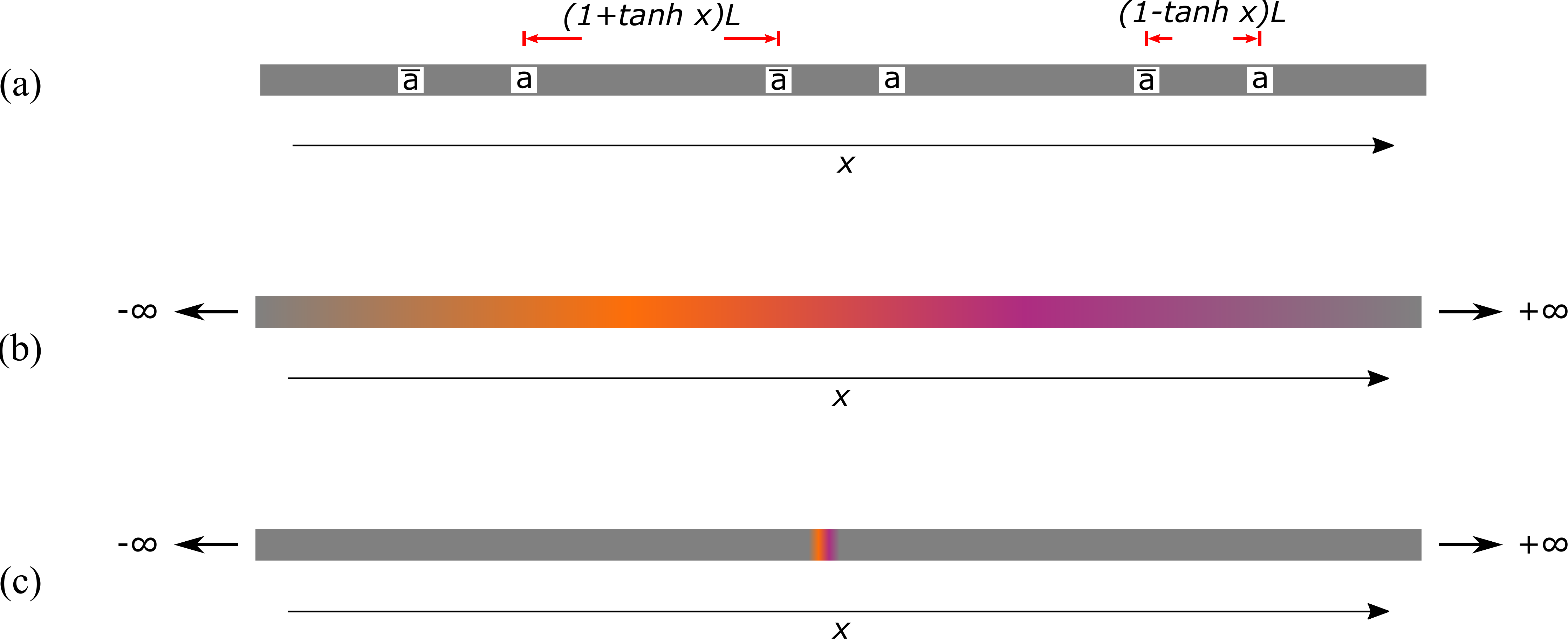}
\caption[The domain wall interpretation of $\Omega$-spectrum.]{The domain wall interpretation of $\Omega$-spectrum. (a) A $(d+1)$-dimensional short-range entangled state $f(a)_x$ constructed from a $d$-dimensional state $a$. (b) A texture of $(d+1)$-dimensional short-range entangled states, $\mu(x)$ for $-\infty < x < \infty$. (c) The $d$-dimensional domain wall obtained by squeezing the texture in (b).}
\label{fig:domain_wall}
\end{figure}

Conversely, given a one-parameter family of $(d+1)$-dimensional short-range entangled states, $\mu(x)$ for $-\infty < x < \infty$, the map $g$ will send it to a $d$-dimensional short-range entangled state, $g(\mu)$. This time, we will treat $x$ as a spatial coordinate and regard $\mu$ as defining a texture of $(d+1)$-dimensional short-range entangled states, which varies spatially with $x$ and is locally indistinguishable from the $(d+1)$-dimensional short-range entangled state $\mu(x)$ in the vicinity of $x$; see Figure \ref{fig:domain_wall}(b). To define $g(\mu)$, we will then squeeze the texture in the $x$-direction; see Figure \ref{fig:domain_wall}(c). This results in a $d$-dimensional domain wall within the $(d+1)$-dimensional system, and we set $g(\mu)$ to be the $d$-dimensional state that lives on the domain wall.

In App.\,B of work \cite{Xiong}, it was shown that $f$ and $g$ are indeed pointed homotopy inverses of each other, and that the definitions of $f$ and $g$ in the domain wall interpretation are equivalent to those in the pumping interpretation. Furthermore, under either interpretation, given any $\mu_1, \mu_2 \in \Omega F_{d+1}$, the concatenation of $\mu_1$ and $\mu_2$ as loops corresponds to the stacking of $\mu_1$ and $\mu_2$ as short-range entangled states \cite{Xiong}.

\section{From \texorpdfstring{$\Omega$}{Omega}-spectrum to classification of SPT phases}
\label{sec:from_spectrum_to_classification}

\subsection{Generalized cohomology theories}

Every $\Omega$-spectrum $\paren{F_n}_{n\in \ZZZ}$ defines a \emph{generalized cohomology theory} \cite{Adams1, Adams2}. A generalized cohomology theory is a sequence of contravariant functors $h^n$ indexed by $n\in \ZZZ$ that maps a CW pair $\paren{X, Y}$---i.e.\,a CW complex $X$ together with a subcomplex $Y$---to an abelian group, according to
\begin{equation}
h^n\paren{X, Y} \coloneq \brackets{\paren{X, Y}, \paren{F_n, \ast}}, \label{gc_of_a_pair}
\end{equation}
where
\begin{equation}
\brackets{\paren{X, Y}, \paren{F_n, \ast}} \coloneq \pi_0 \Map\paren{\paren{X, Y}, \paren{F_n, \ast}}
\end{equation}
is the set of homotopy classes of maps from $\paren{X, Y}$ to $\paren{F_n, \ast}$, where a map from $\paren{X, Y}$ to $\paren{F_n, \ast}$ is by definition a map $f: X \fromto F_n$ that sends all of $Y$ to the basepoint $\ast$ of $F_n$. Furthermore, if $G$ is a group and $\paren{X, Y}$ is a $G$-CW pair \cite{May}---i.e.\,a CW pair with a cellular $G$-action---then we can define the \emph{Borel equivariant generalized cohomology theory}, $h^n_G$ for $n \in \ZZZ$, according to
\begin{equation}
h^n_G\paren{X, Y} \coloneq h^n\paren{EG \times_G X, EG \times_G Y}, \label{equivariant_of_a_pair}
\end{equation}
where $EG$ is the contractible universal cover of the \emph{classifying space} $BG$ \cite{AdemMilgram},\footnote{The classifying space $BG$ of a discrete group $G$ is the unique topological space (up to homotopy equivalence) that satisfies $\pi_1(BG) \isomorphic G$ and $\pi_k(BG) = 0$ for all $k \neq 1$. The universal cover $EG$ of $BG$ is contractible. There is a natural $G$-action on $EG$ such that $EG/G = BG$. For non-discrete groups, provided $G$ is a sufficiently nice topological space, we have in general a $G$-principal bundle $G \fromto EG \fromto BG$ in which $EG$ is contractible. This also determines $BG$ up to homotopy equivalence, though in this case $\pi_k(BG) \isomorphic \pi_{k-1}(G)$ are not necessarily trivial for $k > 1$.} and $EG \times_G X$ denotes the quotient space $\paren{EG \times X} / G$. As a special case, if $X = \pt$ is a point and $Y = \emptyset$ is the empty set, then
\begin{equation}
h^n_G\paren{\pt, \emptyset} = h^n\paren{BG, \emptyset}. \label{equivariant_of_pt}
\end{equation}
Finally, it is customary to drop $Y$ in the notation when it is the empty set, as in
\begin{eqnarray}
h^n\paren{X} &\coloneq& h^n\paren{X, \emptyset}, \\
h^n_G\paren{X} &\coloneq& h^n_G\paren{X, \emptyset}.
\end{eqnarray}
Thus Eq.\,(\ref{equivariant_of_pt}) can be rewritten
\begin{equation}
h^n_G\paren{\pt} = h^n\paren{BG}. \label{equivariant_of_pt_drop_Y}
\end{equation}

Next, if $\phi: G \fromto \braces{\pm 1}$ is a homomorphism, then we can define the \emph{twisted equivariant generalized cohomology theory}, $h^{\phi + n}_G$ for $n\in \ZZZ$, according to
\begin{equation}
h^{\phi + n}_G\paren{X, Y} \coloneq h^{n + 1}_G \paren{X \times \tilde I, Y \times \tilde I \cup X \times \partial \tilde I}, \label{twisted_equivariant_of_a_pair}
\end{equation}
where $\tilde I$ denotes the unit interval $[0, 1]$ with the $G$-action $g.x = \frac12 + \phi(g)\paren{x - \frac12}$ for $x \in [0, 1]$. As a special case, if $X = \pt$ and $Y = \emptyset$, then
\begin{eqnarray}
h^{\phi + n}_G\paren{\pt} &=& h^{n + 1}_G \paren{\tilde I, \partial \tilde I} \nonumber\\
&=& h^{n + 1} \paren{ EG \times_G \tilde I, EG \times_G \partial \tilde I } \nonumber\\
&=& \brackets{ \paren{ EG \times_G \tilde I, EG \times_G \partial \tilde I }, \paren{F_{n + 1}, \pt} } \nonumber\\
&=& \brackets{ EG, \tilde\Omega F_{n + 1} }_G. \label{h_phi_n_G_pt}
\end{eqnarray}
In the last line, $\tilde \Omega F_{n + 1}$ denotes $\Omega F_{n + 1}$ with the $G$-action in which elements $g$ with $\phi(g) = -1$ reverse the direction of loops, and $\brackets{ EG, \tilde\Omega F_{n + 1} }_G$ is the set of homotopy classes of maps from $EG$ to $\tilde\Omega F_{n + 1}$ that preserve $G$-action. As another special case,
\begin{equation}
h^{\phi + n}_G\paren{X, Y} = h^n_G\paren{X, Y} ~~ \mbox{if $\phi$ is trivial.}
\end{equation}

Earlier we said that $h^n$ maps CW pairs to abelian groups. In fact, $h^n_G$ and $h^{\phi + n}_G$ also map their respective categories into the category of abelian groups. The abelian group structure of Eqs.\,(\ref{gc_of_a_pair})(\ref{equivariant_of_a_pair})(\ref{twisted_equivariant_of_a_pair}) comes from the fact that the $F_n$ or $F_{n+1}$ therein can be written $\Omega F_{n+1}$ or $\Omega F_{n+2}$, respectively, and for the loop spaces $\Omega F_{n+1}$ and $\Omega F_{n+2}$ we can induce the composition law of an abelian group by concatenating loops.

We referred to $h^n$ as contravariant functors because it is contravariant in $\paren{X, Y}$. More relevant to us is the fact that $h^n_G$ and $h^{\phi + n}_G$ are contravariant in $G$, which means the following. Given any pairs $\paren{G_1, \phi_1}$ and $\paren{G_2, \phi_2}$ and any homomorphism
\begin{equation}
f: G_1 \fromto G_2
\end{equation}
such that
\begin{equation}
\phi_2 \circ f = \phi_1,
\end{equation}
we have an induced homomorphism of abelian groups
\begin{equation}
f^*: h^{\phi_2 + n}_{G_2}\paren{X, Y} \fromto h^{\phi_1 + n}_{G_1}\paren{X, Y}. \label{f_ast_1}
\end{equation}
These induced homomorphisms satisfy
\begin{eqnarray}
\paren{f_2 \circ f_1}^* &=& f_1^* \circ f_2^*, \label{coherence_1} \\
\paren{\identity}^* &=& \identity. \label{coherence_2}
\end{eqnarray}

\subsection{Generalized cohomology hypothesis}
\label{subsec:GCH}

In Sec.\,\ref{subsec:modern_definition} we mentioned that the set of $d$-dimensional bosonic SPT phases with symmetry $G$ has the structure of an abelian group. We shall refer to this property the \emph{additivity} of SPT phases. SPT phases also have another property which we shall call \emph{functoriality}, which means the following. Given any group homomorphism $f: G_1 \fromto G_2$ satisfying $\phi_2 \circ f = \phi_1$ and a representation $\rho_2$ of $G_2$ with on-site Hilbert space $\mathscr H$, the composition $\rho_2 \circ f$ defines a representation of $G_1$ with the same on-site Hilbert space $\mathscr H$. This means $f$ can be used to convert $G_2$-symmetric systems to $G_1$-symmetric systems. Indeed, we can retain the same Hilbert space and Hamiltonian and simply replace $G_2$-actions by $G_1$-actions following this recipe. In the event that $f: G_1 \fromto G_2$ is an inclusion, this conversion process is a symmetry-forgetting process. In general, $f$ does not have to be either injective or surjective, and the conversion process is a symmetry forgetting followed by a symmetry relabeling. Either way, we get an induced homomorphism,
\begin{equation}
f^*: \SPT_b^d(G_2, \phi_2) \fromto \SPT_b^d(G_1, \phi_1). \label{f_ast_2}
\end{equation}
It is easy to see that these induced homomorphisms also satisfy the coherence relations (\ref{coherence_1})(\ref{coherence_2}).

Now, we are finally ready to relate the classification of bosonic SPT phases to generalized cohomology theories. This is the \emph{generalized cohomology hypothesis} stated in works \cite{Xiong, Xiong_Alexandradinata}.
\begin{nameddef}[Bosonic generalized cohomology hypothesis]
For the setup discussed at the beginning of Sec.\,\ref{sec:short-range entangled_SPT} and the definition of bosonic SPT phases in Sec.\,\ref{subsec:modern_definition}, there exists an $\Omega$-spectrum $\paren{F_n}_{n \in \ZZZ}$ such that the following is true. Given any $G$, $\phi$, denote by $h^{\phi + \bullet}_G$ the twisted equivariant generalized cohomology theory defined by $\paren{F_n}_{n \in \ZZZ}$. Then for all nonnegative integers $d$, groups $G$, and homomorphisms $\phi \fromto \braces{\pm 1}$, we have isomorphisms
\begin{equation}
\SPT_b^d \paren{G, \phi} \isomorphic h^{\phi + d}_G\paren{\pt}
\end{equation}
that are \emph{natural} in $G$ and $\phi$. Here, naturality means that the isomorphisms respect both additivity (i.e.\,the abelian group structure of both sides) and functoriality [i.e.\,the induced homomorphisms\,(\ref{f_ast_1})(\ref{f_ast_2})].
\end{nameddef}

Note that we used the homomorphism $\phi$ of Eq.\,(\ref{phi_homomorphism}), which depends on both whether $g$ is ``intrinsically" antiunitary and whether $g$ reverses the orientation of space, rather than the homomorphism $\phi_{\rm int}$, which depends only on whether $g$ is intrinsically antiunitary. In particular, this means that a reversal of spatial orientation has a similar effect to a reversal of temporal orientation. Both give a nontrivial twist to the classification of SPT phases.

In the interest of space, we will not get into the detailed arguments for the generalized cohomology hypothesis. For internal symmetries, the generalized cohomology hypothesis can be justified using the idea of decorated domain walls \cite{decorated_domain_walls}. This idea had already been used in many specific proposals for the classification of SPT phases, such as Refs.\,\cite{Wen_Fermion, wang2018towards}. For arbitrary generalized cohomology theories, a general construction was first sketched in Ref.\,\cite{Kitaev_IPAM}. It uses continuous patterns of short-range entangled states that fluctuate with a lattice gauge field of $G$. In work \cite{Xiong}, this construction was made cellular by considering discrete patterns of cellular decoration by invertible topological orders, which similarly fluctuate with a lattice gauge field of $G$. In addition, Ref.\,\cite{Gaiotto_Johnson-Freyd} considered the gauged version of SPT phases, and gave a general construction that involves decorating cells by defects, which, in contrast, do not fluctuate quantum mechamically because of the gauging.

The development for spatial symmetries was more recent. In our work \cite{Shiozaki2018}, we gave a general construction of SPT phases with spatial symmetries that works for arbitrary generalized cohomology theories, by decomposing the generalized cohomology theory using the Atiyah-Hirzebruch spectral sequence. Previously, there had been heuristic arguments \cite{ThorngrenElse, Keyserlingk_Floquet, Else_Floquet, Potter_Floquet} for why the classification of SPT phases does not depend on whether the symmetry is internal or spatial, tensor-network computations \cite{Jiang_sgSPT} and dimensional-reduction constructions \cite{Huang_dimensional_reduction} that showed explicitly that the group cohomology proposal of Ref.\,\cite{Wen_Boson} works for spatial as well as internal symmetries, and many specific examples \cite{Cirac, Wen_sgSPT_1d, SPt, You_sgSPT, Yoshida_sgSPT, Hsieh_sgSPT, Cho_sgSPT, Hermele_torsor} that confirm in special cases that the classification of SPT phases is the same for spatial and internal symmetries. A summary of the state of the matter around the time of publication of work \cite{Xiong} was given in Sec.\,6.7 of work \cite{Xiong}.

\section{Fermionic SPT phases}
\label{subsec:fermionic_SPT}

Let us now discuss the changes needed to make Secs.\,\ref{sec:short-range entangled_SPT}-\ref{sec:from_spectrum_to_classification} work for fermionic SPT phases. This was discussed in work \cite{Xiong_Alexandradinata} and uses the idea of Ref.\,\cite{Gaiotto_Johnson-Freyd}.
\begin{enumerate}
\item At the beginning Sec.\,\ref{sec:short-range entangled_SPT}, in addition to the other data, we need to fix some intrinsically fermionic symmetry group $G_f$ such that it contains the fermion-parity symmetry $\hat P_f$. We also fix a representation $\rho_f$ of $G_f$ on single fermions and thereby fix its representation on any many-body system, which we also denote by $\rho_f$. In many applications, one may want to choose $G_f = \ZZZ_2 = \braces{1, \hat P_f}$, or $G_f = \ZZZ_4 = \braces{1, \hat T, \hat T^2 = \hat P_f, \hat T^3}$ if there is a time-reversal symmetry $\hat T$ that squares to $\hat P_f$. $G_f$ may also include non-local symmetries such as particle-hole symmetry.

\item We then assume that the full symmetry group is of the form $G_f \times G$. This form is not as constraining as it seems as any full symmetry group $\tilde G$ can be written $G_f \times G$ if we choose $G_f$ to be the entire $\tilde G$ and $G$ to be the trivial group. However, since $G_f$ is a fixed parameter and $G$ is a variable parameter in the fermionic version of the generalized cohomology hypothesis, this does mean that there are certain symmetries that the hypothesis will not be able to peek inside.

\item The Hamiltonian being fermionic means that there is a finite collection $\A$ of fermionic modes,
\begin{equation}
\hat c_\alpha,~ \hat c_\alpha^\dag, ~~~\mbox{for } \alpha \in \A,
\end{equation}
associated to all sites in $\mathcal L$ and that $\hat H$ is a sum of products of $\hat c_\alpha$ or $\hat c_\alpha^\dag$ from possibly different sites. The on-site Hilbert space is
\begin{equation}
\mathscr H = \bigotimes_{\alpha \in \A} \CCC^2.
\end{equation}

\item The fermionic Hamiltonian being local means that there is a $k$ such that the aforementioned products have finite supports with linear size $< k$.

\item The Hamiltonian being $G$-symmetric means the same as in the bosonic case. In particular, $\rho_{\rm int}$ will still be a representation on $\mathscr H$, and $\phi_{\rm int}$, $\phi_{\rm spa}$, and $\phi$ will still be homomorphisms from $G$---and not $G_f \times G$---into $\braces{\pm 1}$. However, we now additionally require that the Hamiltonian respects $\rho_f$:
\begin{equation}
\hat H \rho_f(g_f) = \rho_f(g_f) \hat H,
\end{equation}
for all $g_f \in G_f$.

\item In Sec.\,\ref{subsec:modern_definition}, we will denote the abelian group of $d$-dimensional fermionic SPT phases by
\begin{equation}
\SPT_f^d \paren{G_f \times G, \phi}.
\end{equation}

\item In Sec.\,\ref{sec:short-range entangled_as_spectrum}, there will be one $\Omega$-spectrum $\paren{F_n^{G_f, \rho_f}}_{n\in \ZZZ}$ for each choice of $\paren{G_f, \rho_f}$.

\item In Sec.\,\ref{sec:from_spectrum_to_classification}, the fermionic version of the generalized cohomology hypothesis will be as follows.

\begin{nameddef}[Fermionic generalized cohomology hypothesis]
For any $\paren{G_f, \rho_f}$, there exists an $\Omega$-spectrum $\paren{F_n^{G_f, \rho_f}}_{n\in \ZZZ}$ such that, for all nonnegative integers $d$, groups $G$, and homomorphisms $\phi \fromto \braces{\pm 1}$, we have isomorphisms
\begin{equation}
\SPT_f^d \paren{G_f \times G, \phi} \isomorphic h^{\phi + d}_G\paren{\pt}
\end{equation}
that are natural in $G$ and $\phi$, where $h^{\phi + \bullet}_G$ denotes the $\phi$-twisted $G$-equivariant generalized cohomology theory defined by $\paren{F_n^{G_f, \rho_f}}_{n\in \ZZZ}$.
\end{nameddef}
\end{enumerate}

For $G_f = \braces{1, \hat P_f} \eqcolon \ZZZ_2^f$, the current understanding of the homotopy groups of $F_n^{G_f, \rho_f}$ is as in Table \ref{table:homotopy_groups_fermion}. Here, $\pi_2\paren{F_0} = \ZZZ$ for the same reason as in the bosonic case, and $\pi_0\paren{F_0} = \ZZZ_2$ is the fermion parity. $\pi_0\paren{F_1} = \ZZZ_2$ and $\pi_0\paren{F_2} = \ZZZ$ are generated by the Majorana chain \cite{Majorana_chain} and $p+ip$ superconductors \cite{Volovik_p+ip, Read_p+ip, Ivanov_p+ip}, respectively. $\pi_0\paren{F_d}$ for $d > 2$ is not well-understood.

{\small
\begin{spacing}{\dnormalspacing}
\begin{longtable}[c]{c|ccccccc}
\caption[Homotopy groups of spaces of fermionic short-range entangled states with $G_f = \ZZZ_2^f$.]{Homotopy groups of the spaces $F_{d}$ of $d$-dimensional fermionic short-range entangled states with $G_f = \ZZZ_2^f$ for $d \leq 3$.\label{table:homotopy_groups_fermion}} \\
\endfirsthead
\caption[]{(Continued).} \\
\endhead
\endfoot
$\pi_{>5}$ & $0$ & $0$ & $0$ & $0$ & $0$ & $0$ & $0$ \\
$\pi_{5}$ & $0$ & $0$ & $0$ & $0$ & $0$ & $0$ & $\ZZZ$ \\
$\pi_{4}$ & $0$ & $0$ & $0$ & $0$ & $0$ & $\ZZZ$ & $0$ \\
$\pi_{3}$ & $0$ & $0$ & $0$ & $0$ & $\ZZZ$ & $0$ & $\ZZZ_2$ \\
$\pi_{2}$ & $0$ & $0$ & $0$ & $\ZZZ$ & $0$ & $\ZZZ_2$ & $\ZZZ_2$ \\
$\pi_{1}$ & $0$ & $0$ & $\ZZZ$ & $0$ & $\ZZZ_2$ & $\ZZZ_2$ & $\ZZZ$ \\
$\pi_{0}$ & $0$ & $\ZZZ$ & $0$ & $\ZZZ_2$ & $\ZZZ_2$ & $\ZZZ$ & $?$ \\
\nobreakhline 
~ & $F_{<-2}$ & $F_{-2}$ & $F_{-1}$ & ~$F_{0}$~ & ~$F_{1}$~ & ~$F_{2}$~ & ~$F_{3}$~ 
\end{longtable}
\end{spacing}
}

Similar tables can be compiled for other choices of $G_f$ using homotopy groups appropriate to the choice of $G_f$. For example, if, apart from $\hat P_f$ or $U(1)$ charge conservation, $G_f$ is generated by no more than time reversal, charge conjugation, or their product, then the appropriate homotopy groups can be derived from the the ten-fold-way classification \cite{Hasan_Kane} of non-interacting fermionic systems, taking into proper account the following:
\begin{enumerate}
\item the reduction of non-interacting classification by interactions (e.g.\,Refs.\,\cite{Fidkowski_Kitaev_1, Fidkowski_Kitaev_2});

\item the existence of intrinsically interacting phases (e.g.\,Ref.\,\cite{WangChong_3DSPTAII}).
\end{enumerate}

\section{Existing classification proposals as examples of generalized cohomology theories}
\label{sec:existing_proposals}

In the literature there exist many proposals for the classification of SPT phases. Some focused on specific dimensions and symmetry groups, while others proposed classification for general dimensions and symmetries. We note that the various general proposals for the classification of SPT phases are examples of generalized cohomology theories \cite{Kitaev_Stony_Brook_2011_SRE_1, Kitaev_Stony_Brook_2013_SRE, Kitaev_IPAM}. These include the Borel group cohomology proposal \cite{Wen_Boson}, the cobordism proposal \cite{Kapustin_Boson}, Kitaev's proposal \cite{Kitaev_Stony_Brook_2011_SRE_1, Kitaev_Stony_Brook_2013_SRE}, and the Freed-Hopkins proposal \cite{Freed_SRE_iTQFT, Freed_ReflectionPositivity} in the bosonic case; and the group supercohomology proposal \cite{Wen_Fermion}, the spin and pin cobordism proposal \cite{Kapustin_Fermion}, Kitaev's proposal \cite{Kitaev_Stony_Brook_2013_SRE, Kitaev_IPAM}, and the Freed-Hopkins proposal \cite{Freed_SRE_iTQFT, Freed_ReflectionPositivity} in the fermionic case. Each of these proposals corresponds to a distinct generalized cohomology theory and a distinct $\Omega$-spectrum. In fact, these $\Omega$-spectra can be explicitly written down, which we briefly summarize in Table \ref{table:existing_proposals}.

{\footnotesize
\begin{spacing}{\dnormalspacing}
\begin{longtable}[c]{p{0.3\columnwidth}p{0.65\columnwidth}}
\caption[Existing proposals for the classification of SPT phases, and the $\Omega$-spectra that represent them.]{Existing proposals for the classification of SPT phases, and the $\Omega$-spectra that represent them. Here, $K(A,n)$ denotes the $n$-th \emph{Eilenberg-Mac Lane space} of $A$, $HA$ denotes the \emph{Eilenberg-Mac Lane spectrum} of $A$, $U$ denotes the infinite unitary group $U(\infty) = \bigcup_{i=1}^\infty U(i)$, $\CCC P^\infty = \bigcup_{i=1}^\infty \CCC P^i$ denotes the infinite projective space, and $\pi_i$ and $k_i$ denote the $i$-th homotopy group and the $i$-th \emph{$k$-invariant}, respectively. The $\CCC P^\infty$ in $F_0$, $\ZZZ_2$ in fermionic $F_0$, $\ZZZ_2$ in fermionic $F_1$, and $\ZZZ$ in bosonic $F_2$ have to do with Berry's phase, fermion parity, the Majorana chain, and the $E_8$-model, respectively (cf.\,Table \ref{table:SPT_examples}) \cite{Kitaev_Stony_Brook_2011_SRE_1, Kitaev_Stony_Brook_2013_SRE, Kitaev_IPAM}. More details of these proposals can be found in App.\,A of work \cite{Xiong}.\label{table:existing_proposals}} \\
\hline
\hline
Classification proposal & Spectrum \\
\hline
\endfirsthead
\caption[]{(Continued).} \\
\hline
\hline
Classification proposal & Spectrum \\
\hline
\endhead
\hline
\hline
\endfoot
Borel group cohomology as in Ref.\,\cite{Wen_Boson} & Shifted $H\ZZZ$:
$$F_d = \begin{cases} K\paren{\ZZZ, d+2}, & d\geq -2, \\\pt, & d<-2. \end{cases}$$ \par
In particular, $F_0 \homotopic \CCC P^\infty$. \\
Cobordism as in Ref.\,\cite{Kapustin_Boson} & Related to the Thom spectra $MSO$ and $MO$. See App.\,A of work \cite{Xiong}. \\
Kitaev's bosonic proposal \cite{Kitaev_Stony_Brook_2011_SRE_1, Kitaev_Stony_Brook_2013_SRE} & Constructed from physical knowledge. $F_d$ is uniquely determined in low dimensions:
$$F_d = \begin{cases} \CCC P^\infty, & d=0, \\ K(\ZZZ,3), & d = 1, \\K(\ZZZ, 4)\times \ZZZ, & d=2, \\K(\ZZZ, 5) \times \SS^1, & d=3. \end{cases}$$
See App.\,A of work \cite{Xiong}. \\
Group supercohomology as in Ref.\,\cite{Wen_Fermion} & ``Twisted product" of $H\ZZZ_2$ and shifted $H\ZZZ$. Works for $G_f = \ZZZ_2^f$ (cf.\,Sec.\,\ref{subsec:fermionic_SPT}). $F_d$ can be constructed as a \emph{Postnikov tower} with these data:
$$\pi_i(F_d) \isomorphic \begin{cases} \ZZZ_2, & i=d, \\\ZZZ, & i = d+2, \\0, & \text{otherwise}, \end{cases}$$
$$k_{d+1} = \beta \circ Sq^2,$$
where $Sq^2$ is the \emph{Steenrod square} and $\beta$ is the \emph{Bockstein homomrphism} associated with
$$0 \fromto \ZZZ \xfromto{2} \ZZZ \fromto \ZZZ_2 \fromto 0.$$
In particular, $F_0 \homotopic \CCC P^\infty\times \ZZZ_2$ and $F_1 \homotopic K(\ZZZ, 3) \times K(\ZZZ_2, 1)$.\\
Spin and pin cobordism as in Ref.\,\cite{Kapustin_Fermion} & Related to the Thom spectra $MSpin$, $MPin_+$, and $MPin_-$. \\
Kitaev's fermionic proposal \cite{Kitaev_Stony_Brook_2013_SRE, Kitaev_IPAM} & Constructed from physical knowledge.
$$F_0 = \CCC P^\infty \times \ZZZ_2$$
is uniquely determined, and $F_{d>0}$ are partially determined. See App.\,A of work \cite{Xiong}. This works for $G_f = \ZZZ_2^f$ (cf.\,Sec.\,\ref{subsec:fermionic_SPT}).
\end{longtable}
\end{spacing}
}

We highlight some of the differences in the predictions of different proposals. In the bosonic case with time-reversal symmetry, the Borel group cohomology proposal \cite{Wen_Boson} predicts a $\ZZZ_2$ classification in 3 spatial dimensions. In contrast, the cobordism proposal \cite{Kapustin_Boson} predicts a $\ZZZ_2^2$ classification, with the extra $\ZZZ_2$ conjectured to correspond to the ``3D $E_8$ phase" studied in Refs.\,\cite{3dBTScVishwanathSenthil, 3dBTScWangSenthil, 3dBTScBurnell}. Furthermore, there are certain phases in 6 dimensions that are nontrivial according to the Borel group cohomology proposal but are trivial according to the cobordism proposal. In the fermionic case, with a time-reversal symmetry that squares to the identity (as opposed to fermion parity), the group supercohomology proposal \cite{Wen_Fermion} predicts a $\ZZZ_4$ classification of SPT phases in 1 dimension. In contrast, the spin cobordism \cite{Kapustin_Fermion} proposal predicts a $\ZZZ_8$ classification. We thus see that these proposals are inequivalent.

In general, these $\Omega$-spectra can be thought of as good approximations to the true, unknown $\Omega$-spectrum that gives the complete classification of SPT phases. More precisely, existing bosonic proposals are approximations to the true bosonic $\Omega$-spectrum, whereas fermionic proposals are approximations to the fermionic counterpart.


\chapter{Applications}
\label{chap:applications}

In this chapter, we will present two applications of the minimalist framework that we developed in Chapter \ref{chap:minimalist}:
\begin{enumerate}
\item In Sec.\,\ref{sec:glide}, we will classify and construct 3D fermionic SPT phases in Wigner-Dyson classes A and AII with an additional glide symmetry, and show that the so-called ``hourglass fermions" \cite{Ma_discoverhourglass, Hourglass, Cohomological} are robust to strong interactions.

\item In Sec.\,\ref{sec:3D_beyond_group_cohomology}, we will classify and construct 3D bosonic SPT phases with space-group symmetries for all 230 space groups, which include phases beyond the previously proposed group-cohomology classification \cite{Huang_dimensional_reduction, ThorngrenElse}, which ignored all phases built from the so-called $E_8$ state \cite{Kitaev_honeycomb, 2dChiralBosonicSPT, Kitaev_KITP}. 
\end{enumerate}

Sec.\,\ref{sec:glide} was published in my work \cite{Xiong_Alexandradinata} with A.\,Alexandradinata. Sec.\,\ref{sec:3D_beyond_group_cohomology} was published in my work \cite{SongXiongHuang} with Hao Song and Shengjie Huang. A list of classifications of 3D bosonic SPT phases with space-group symmetries for all 230 space groups, which was also published in work \cite{SongXiongHuang}, appears in App.\,\ref{app:list_of_classifications} to this thesis.

\section{Classification and construction of 3D fermionic SPT phases with glide symmetry}
\label{sec:glide}


The recent intercourse between band theory, crystalline symmetries, and topology has been highly fruitful in both theoretical and experimental laboratories. The recent experimental discovery \cite{Ma_discoverhourglass} of hourglass-fermion surface states in the material class KHgSb \cite{Hourglass,Cohomological} heralds a new class of topological insulators (TIs) protected by glide symmetry \cite{ChaoxingNonsymm,unpinned,Shiozaki2015,Nonsymm_Shiozaki,Poyao_mobiuskondo,Ezawa_hourglass, shiozaki_review,singlediraccone} 
-- a reflection composed with a translation by half the lattice period. Despite the manifold successes of band theory, electrons are fundamentally interacting. To what extent are topological phases predicted from band theory robust to interactions?

In this work, we demonstrate that the question of (a) robustness to interactions is intimately linked to two other seemingly unrelated questions: (b) how glide-symmetric topological phases can be constructed by layering lower-dimensional topological phases, and (c) in what ways can the classification of topological phases be altered by the inclusion of glide symmetry.

In fact, question (c) is very close in spirit to the types of questions asked in a symmetry-based classification of solids: how many ways are there to combine discrete translational symmetry with rotations and/or reflections to form a space group -- the full symmetries of a crystalline solid? This has been recognized as a group extension problem, and its solution through group cohomology has led to the classification of 230 space groups of 3D solids \cite{hiller1986crystallography}. Here, we are proposing that the same mathematical structure ties together (a-c). More precisely, we are proposing a short exact sequence of abelian groups, which classify gapped, \emph{interacting} phases of matter, also known as symmetry-protected topological (SPT) phases \cite{SPT_origin, Wen_Definition}, that carries the information of (a-c).

\begin{figure}[t]
\centering
\includegraphics[width=3.5in]{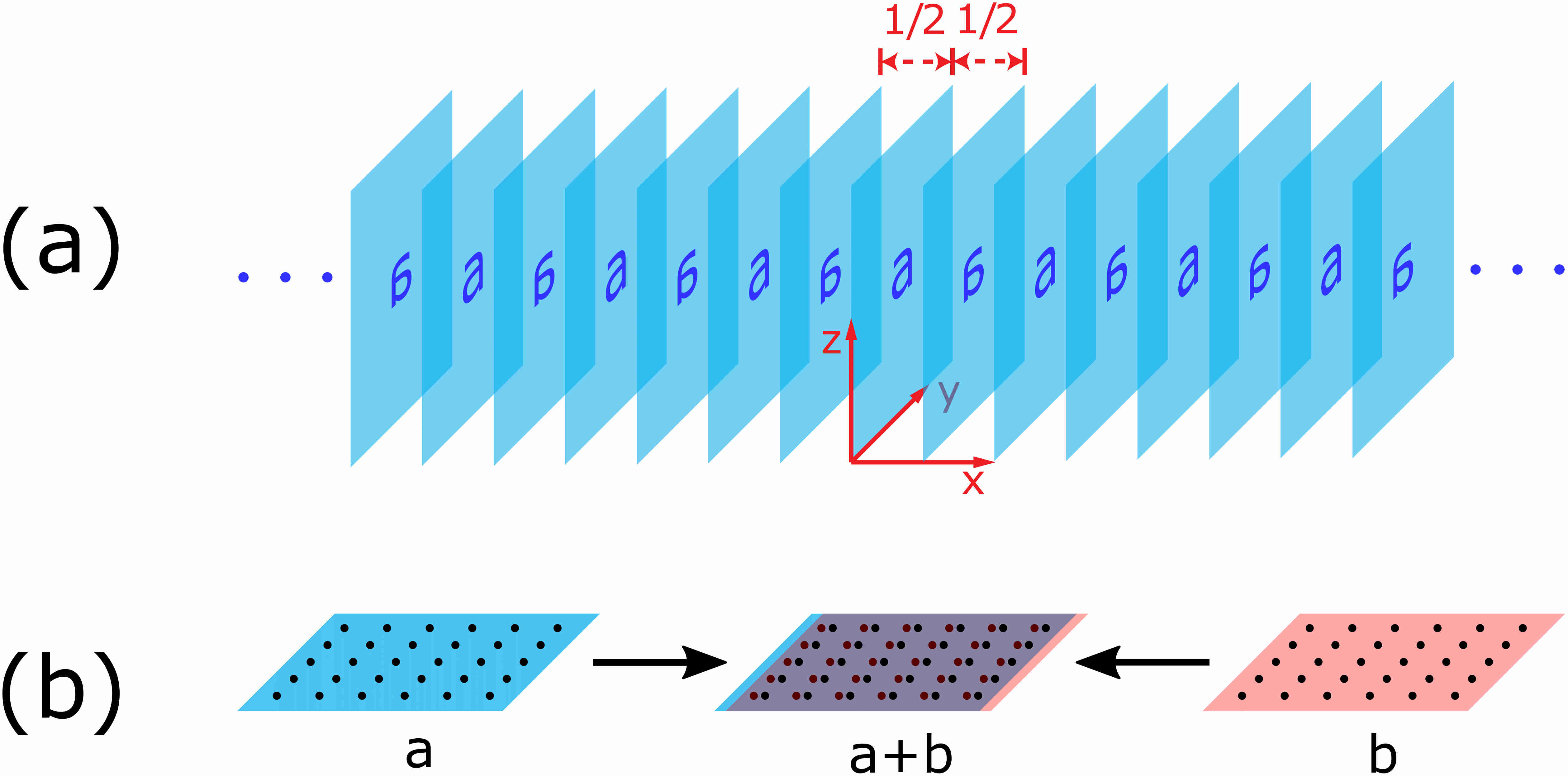}
\caption[The alternating-layer construction of SPT phases with glide symmetry.]{(a) The ``alternating-layer construction" repeats a given $G$-symmetric system $a$ and its mirror image in an alternating fashion to produce a one higher-dimensional system that respects glide symmetry in addition to $G$. (b) The ``stacking" operation combines \emph{two} systems $a$ and $b$ respecting a given symmetry into a new system of the \emph{same} dimension respecting the \emph{same} symmetry. Illustration is given for particular dimensions but the constructions are general.\label{fig:alternating_layer_construction}}
\end{figure}

In the symmetry class of hourglass fermions, i.e., spin-orbit-coupled solids with charge-conservation [$U(1)$], time-reversal ($\T$, with $\T^2 = -1$ on single fermions), and glide symmetries,\footnote{More precisely, the hourglass-fermion phase belongs to a nonsymmorphic space group which includes at least one glide symmetry.} there is a pair of consecutive maps between abelian groups classifying TIs in two and three dimensions,
\begin{equation}
\ZZZ_2 \xfromto{\times 2} \ZZZ_4 \xfromto{\rm mod~ 2} \ZZZ_2,\label{SES_AII}
\end{equation}
that can be viewed as a non-interacting analog of our short exact sequence.  The nontrivial element of the first $\ZZZ_2$, which distinguishes the two phases of 2D TIs that respect $\T$ and $U(1)$,  can be realized by a 2D quantum spin Hall (QSH) system \cite{kane2005B, kane2005A}.
By placing \emph{decoupled} copies of a QSH system on all planes of constant $x\in \ZZZ$, and its mirror image on all $x\in \ZZZ + 1/2$ planes, one constructs a 3D system that respects the glide symmetry $(x,y,z)\mapsto(x+1/2,-y,z)$. This ``alternating-layer construction" (see Figure\,\ref{fig:alternating_layer_construction}) takes one from the first $\ZZZ_2$ to $\ZZZ_4$ -- in particular the QSH phase to the hourglass fermion phase \cite{Nonsymm_Shiozaki,Ezawa_hourglass} -- where $\ZZZ_4$ distinguishes the four phases of 3D TIs that respect glide in addition to $\T$ and $U(1)$ \cite{Nonsymm_Shiozaki}.\footnote{Since there is no known generalization of the $\ZZZ_4$ invariant to disordered systems, we assume discrete translational symmetry in all three directions at first and quotient out phases that can be obtained by layering lower-dimensional phases in the $y$- and $z$-directions in the end \cite{Nonsymm_Shiozaki}. In the same spirit, both $\ZZZ_2$'s in Eq.\,(\ref{SES_AII}) are obtained after quotienting out layered phases; that is, they are the strong \cite{Kitaev_TI} classifications.\label{footnote:strong_phases}} By dropping the glide symmetry constraint, one can in turn view a 3D TI respecting glide, $\T$, and $U(1)$ as an element of the second $\ZZZ_2$, which is the strong classification of 3D TIs respecting $\T$ and $U(1)$ but not necessarily glide \cite{moore2007,fu2007b,Rahul_3DTI}. In this ``symmetry-forgetting" process, certain distinct classes in the $\ZZZ_4 = \braces{0,1,2,3}$ classification are identified: classes 0 and 2 (resp.\,1 and 3) can be connected to each other if glide symmetry is not enforced. This gives the second, ``glide-forgetting" map between $\ZZZ_4$ and $\ZZZ_2$ in Eq.\,(\ref{SES_AII}). We stress that glide forgetting only conceptually expands the space of allowed Hamiltonians by letting go of the glide constraint, and does not involve an immediate, actual perturbation to a particular system that is under consideration.

Our short exact sequence of abelian groups classifying SPT phases works in essentially the same manner as Eq.\,(\ref{SES_AII}) but with the non-interacting classification replaced by the classification of bosonic or fermionic SPT phases. 
In its full generality, the sequence applies to all symmetries $G$, including those that are represented \cite{Jiang_sgSPT, ThorngrenElse} antiunitarily, and all dimensions $d$, where the analog of glide is
\begin{equation}
(x_1, x_2, x_3, \ldots, x_d) \mapsto (x_1 + 1/2, -x_2, x_3, \ldots, x_d).\label{glide_definition}
\end{equation}
Writing $\ZZZ$ for the symmetry generated by Eq.\,(\ref{glide_definition}), the existence of the short exact sequence implies that all $d$-dimensional $\ZZZ\times G$-protected SPT phases must have order 1, 2, or 4 and that their classification must be a direct sum of $\ZZZ_4$'s {and/or} $\ZZZ_2$'s, where the order of an SPT phase is defined with respect to a ``stacking" operation (imagine interlaying two systems \emph{without} coupling them; see Figure\,\ref{fig:alternating_layer_construction}) `$+$' that makes the set of $d$-dimensional $G$-protected SPT phases into an abelian group.\footnote{In the literature on non-interacting fermionic phases, the stacking operation is also known as the {``Whitney sum,"} which refers to the direct sum of vector spaces corresponding to two sets of fermion-filled bands over the same Brillouin torus. Note that a direct sum of single-particle Hilbert spaces corresponds to a tensor product of Fock spaces or many-body Hilbert spaces.} The short exact sequence also implies that, in general, not all $d$-dimensional $G$-protected SPT phases have glide-symmetric representatives and that a necessary and sufficient condition for such representatives to exist is for the given $G$-protected SPT phase to square to the trivial phase, where the square of an SPT phase $[a]$ is by definition $2[a] \coloneq [a]+[a]$. Note that this implication is non-obvious because certain $G$-protected SPT phases are known to be incompatible with certain symmetries outside the group $G$: e.g., a Chern insulator that conserves charge is not compatible with time reversal. From the perspective gained through our short exact sequence, it is then not surprising, in the symmetry class of the hourglass-fermion phase, that there exist four non-interacting 3D phases in the presence of glide and that the nontrivial 3D $\ZZZ_2$ TI, which squares to the trivial phase, can be made glide-symmetric \cite{Cohomological}.

In fact, by combining our general result with the proposed complete classifications of 2D \cite{2dChiralBosonicSPT_erratum} and 3D \cite{WangChong_3DSPTAII} fermionic SPT phases in the Wigner-Dyson class AII, we can show that the complete classification of 3D fermionic SPT phases with an additional glide symmetry must be $\ZZZ_4 \oplus \ZZZ_2 \oplus \ZZZ_2$, such that the first summand can be identified with the $\ZZZ_4$ in Eq.\,(\ref{SES_AII}).
We will do so in two steps. First, we will argue that the hourglass-fermion phase is robust to interactions using a corollary of the general result and known arguments \cite{qi_spincharge,essin2009} for the robustness of QSH systems and 3D TIs without glide. Assuringly, the same conclusion was recently drawn in Ref.\,\cite{Lu_sgSPT} through the construction of an anomalous surface topological order. Then, we will show that the exactness of the sequence
\begin{equation}
0 \fromto \ZZZ_2 \fromto \mbox{?} \fromto \ZZZ_2 \oplus \ZZZ_2 \oplus \ZZZ_2 \fromto 0,
\end{equation}
where $0$ denotes the trivial group (also written $\ZZZ_1$),
is simply constraining enough that $\ZZZ_4\oplus\ZZZ_2\oplus\ZZZ_2$ is the \emph{unique} solution that is compatible with the robustness of the hourglass-fermion phase.

We will derive our general result within a bare-bones, minimalist framework that one of us developed \cite{Xiong} based on Kitaev's argument that the classification of SPT phases should carry the structure of a generalized cohomology theory \cite{Kitaev_Stony_Brook_2011_SRE_1, Kitaev_Stony_Brook_2013_SRE, Kitaev_IPAM}.
The framework assumes minimally that SPT phases form abelian groups satisfying certain axioms, and applies to all existing non-dimension-specific proposals for the classification of SPT phases \cite{Wen_Boson, Wen_Fermion, Kapustin_Boson, Freed_SRE_iTQFT,  Freed_ReflectionPositivity, Kitaev_Stony_Brook_2011_SRE_1, Kitaev_Stony_Brook_2013_SRE, Kitaev_IPAM}.
The axioms provide for the switching from one symmetry group to another and from one dimension to another, which makes the derivation of our short exact sequence possible.
The results mentioned above are far from an exhaustive list of implications of the short exact sequence, which we will elaborate upon in this work.

The rest of Sec.\,\ref{sec:glide} is organized as follows. In Sec.\,\ref{subsec:hourglass_fermions}, we will argue that the hourglass-fermion phase and its square roots are robust to interactions. In Sec.\,\ref{subsec:general_relations}, we will deal with SPT phases with glide more systematically. We will give a more precise definition of SPT phases, derive our general result, and explore its implications. In Sec.\,\ref{subsec:physical_picture}, we will break down the general result into individual statements and offer the physical intuition behind some of them. In Sec.\,\ref{subsec:applications}, we will apply the general result to 3D fermionic SPT phases in Wigner-Dyson classes A and AII, where the complete classification with glide will be derived. In Sec.\,\ref{subsec:computations}, we will discuss generalizations. Specifically, in Sec.\,\ref{sec:puretranslation}, we will discuss the difference between glide and pure translation. In Sec.\,\ref{subsec:bSPT_with_glide}, we will apply the general result to bosonic SPT phases. In Sec.\,\ref{subsec:temporalglide}, we will discuss generalized, spatiotemporal glide symmetries.

\subsection{Robustness of hourglass fermions to interactions\label{subsec:hourglass_fermions}}

Let us apply generalized cohomology to 3D spin-orbit-coupled, time reversal-invariant TIs and their interacting analogues. Due to spin-orbit coupling, time reversal necessarily carries half integer-spin representation, i.e.\,it squares to $-1$ on single fermions. With the addition of glide symmetry, the two well-known classes of 3D TIs, which are distinguished by a $\ZZZ_2$ index $\nu_0\in \braces{0,1}$ \cite{fu2007b,moore2007,Rahul_3DTI,Inversion_Fu}, subdivide into four classes distinguished by a $\ZZZ_4$ invariant $\chi \in \braces{0,1,2,3}$ \cite{Nonsymm_Shiozaki, AA_Z4}. Of the four classes, $\chi=2$ corresponds to the hourglass fermion phase. We will call the other two nontrivial phases the ``square roots" of the hourglass-fermion phase, since the $\ZZZ_4$ invariant adds under stacking and $1+1 \equiv 3 + 3 \equiv 2 \mod 4$. In this section we will argue
\begin{enumerate}[(i)]
\item that the square roots of the hourglass-fermion phase are robust to interactions, and \label{list_1_2}

\item that the hourglass-fermion phase is robust to interactions. \label{list_1_3}
\end{enumerate}
Note that (\ref{list_1_3}) implies (\ref{list_1_2}), for if the $\chi= 1$ or 3 phase was unstable to interactions, then so would two decoupled copies of itself, which represent the $\chi = 2$ phase. Nevertheless, we will dedicate a separate subsection to (\ref{list_1_2}), which can be justified by an independent magnetoelectric response argument, as a consistency check.

\subsubsection{Robustness of the square roots of the hourglass-fermion phase\label{subsubsec:robustness_square_roots}}

In this subsection, we argue that the square roots ($\chi = 1, 3$) of the hourglass-fermion phase are robust to interactions. More precisely, we argue that they cannot be connected to the trivial phase ($\chi=0$) by turning on interactions that preserve the many-body gap and the glide, $U(1)$, and $\T$ symmetries.

We begin by noting that the $\ZZZ_4$ and $\ZZZ_2$ classifications with and without glide symmetry are related as
\begin{equation}
\chi \equiv \nu_0 \mod 2,\label{glideforget}
\end{equation}
which is supported by the following heuristic argument.\footnote{This argument was first presented in Ref.\ \cite{Nonsymm_Shiozaki}} Recall that $\nu_0$ counts the parity of the number of surface Dirac fermions \cite{fu2007b,moore2007,Rahul_3DTI,Inversion_Fu} The glide symmetry assigns to each Dirac fermion a chirality according to its glide representation [compare Figures \ref{fig:z4surfacestates}(b) and (c)]. Unless symmetry is broken, we cannot \cite{Hourglass} fully gap out surface states that carry two positively-chiral Dirac fermions [\fig{fig:z4surfacestates}(d)]. However, two positively-chiral Dirac fermions can be deformed into two negatively-chiral fermions [\fig{fig:z4surfacestates}(d)$\rightarrow$(e)$\rightarrow$(f)$\rightarrow$(g)]. We thus expect four topologically distinct classes, which we distinguish by a $\ZZZ_4$ invariant $\chi$ that counts the number of chiral Dirac fermions mod 4. Since both $\chi$ and $\nu_0$ count the number of surface Dirac fermions, Eq.\,(\ref{glideforget}) follows. This argument was made with representatives of $\chi=\pm 1,\pm 2$ whose surface states are Dirac fermions situated at the Brillouin zone center (point $2$ in Figure\,\ref{fig:z4surfacestates}); more generally, the surface states form a nontrivial connected graph over the bent line $0123$ \cite{Nonsymm_Shiozaki,Hourglass}. This motivates a more general proof of \q{glideforget}, which we present in App.\,B of work \cite{Xiong_Alexandradinata}.

\begin{figure}[tb]
\centering
\includegraphics[width=10 cm]{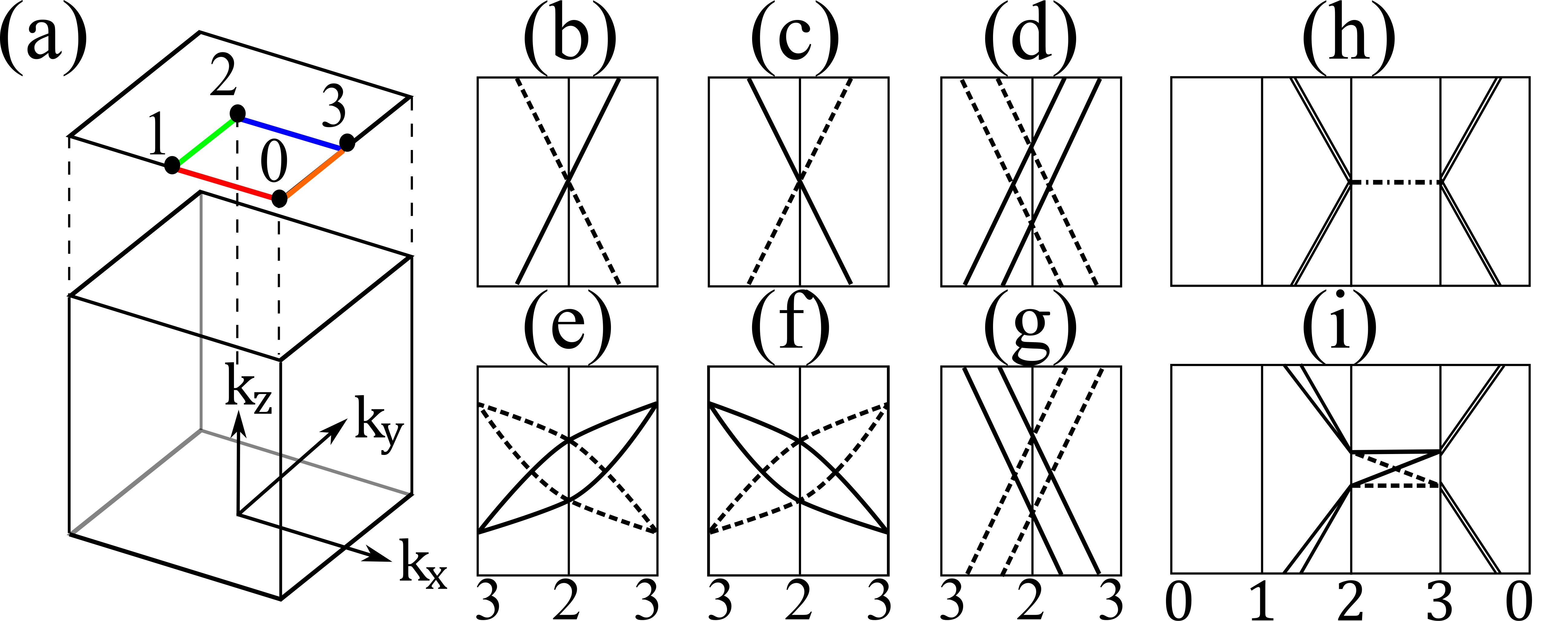}
\caption[Possible surface states of a glide-symmetric crystal in Wigner-Dyson class AII.]{ (a) Bottom: Brillouin 3-torus for a glide-symmetric crystal; top:  Brillouin 2-torus corresponding to the glide-symmetric surface. (b-g) Possible surface states on the glide-invariant line $323$; a surface band in the even (odd) representation of glide  is indicated by a solid (dashed) line. (h-i) Surface states on the high-symmetry line $01230$. Bands along $30$ are Kramers-degenerate owing to the composition of time reversal and glide \cite{Hourglass}. In (h), the additional degeneracy (two-fold along 12, four-fold along 23) originates from the alternating-layering construction; these degeneracies may be split by generic perturbations, as illustrated in (i).\label{fig:z4surfacestates}}
\end{figure}

The square roots of the hourglass-fermion phase have $\ZZZ_4$ invariant $\chi=1,3$, so they must correspond to the $\nu_0 = 1$ phase by Eq.\,(\ref{glideforget}). As can be argued from the quantization of magnetoelectric response \cite{Qi_Hughes_Zhang,Wilczek_axion}, which persists in the many-body case \cite{essin2009}, the $\nu_0=1$ phase is robust to interactions in the sense that it cannot be connected to the trivial phase by turning on interactions that preserve the many-body gap and the $U(1)$ and $\T$ symmetries. But we know that interactions that preserve glide in addition to $U(1)$ and $\T$ are a subset of those that preserve $U(1)$ and $\T$. If a system cannot be made trivial by turning on interactions that preserve $U(1)$ and $\T$, then it surely cannot be made trivial by turning on interactions that simultaneously preserve glide, $U(1)$, and $\T$, and our argument is complete.

While the minimalist framework did not enter the argument above, it will enter the argument for the robustness of the hourglass-fermion phase, which is a stronger claim and the subject of the next subsection.

\subsubsection{Robustness of the hourglass-fermion phase\label{subsubsec:robustness_hourglass_fermion_phase}}

In this subsection, we argue that the hourglass-fermion phase ($\chi=2$) is robust to interactions. More precisely, we argue that it cannot be connected to the trivial phase ($\chi =0$) by turning on interactions that preserve the many-body gap and the glide, $U(1)$, and $\T$ symmetries.

Let us represent the hourglass-fermion phase by a system obtained through the alternating-layer construction. More specifically, let us put copies of a QSH system on all planes of constant $x\in \ZZZ$ and its image under $y\mapsto -y$ on all $x\in \ZZZ+ 1/2$ planes, without turning on inter-plane coupling. To see this represents the hourglass-fermion phase, we recall that a QSH system and its mirror image have identically dispersing 1-dimensional Dirac fermions on the edge. When layered together as described, we obtain two degenerate surface Dirac fermions that do not disperse as functions of $k_x$. This is illustrated in \fig{fig:z4surfacestates}(h), where along the glide-symmetric line 23 we have a four-fold degeneracy originating from two degenerate Dirac points. Glide-symmetric interlayer coupling can only perturb the surface band structure into a connected graph like in \fig{fig:z4surfacestates}(i) \cite{Hourglass}, owing to a combination of the Kramers degeneracy and the monodromy \cite{connectivityMichelZak} of the representation of glide. The resultant connected graph over 0123 has the same topology as the hourglass-fermion phase. Since the bulk gap is maintained throughout the perturbation, the unperturbed system must be in the hourglass-fermion phase. A tight-binding model that demonstrates the construction has been devised by Ezawa \cite{Ezawa_hourglass}.

Next, let us recognize that the robustness of the hourglass-fermion phase to interactions is equivalent to its nontriviality as a 3D SPT phase protected by glide, $U(1)$, and $\T$. A corollary to our general result to be presented in Sec.\,\ref{subsec:general_relations} says that given any symmetry $G$, dimension $d$, and nontrivial $(d-1)$-dimensional $G$-protected SPT phase $[a]$, \emph{the $d$-dimensional $\ZZZ \times G$-protected SPT phase obtained from $[a]$ through the alternating-layer construction is trivial if and only if $[a]$ has a square root}. Since the hourglass-fermion phase can be obtained from the QSH phase through the alternating-layer construction, its robustness to interactions now boils down to the absence of a square root of the QSH phase.

To support the last claim, we employ a many-body generalization \cite{qi_spincharge,lee2008,aa2011} of the 2D $\ZZZ_2$ topological invariant $\Delta\in \braces{0,1}$ in Wigner-Dyson class AII. Following the approach of Ref.\ \cite{qi_spincharge}, which is closely related to a preceding pumping formulation \cite{fu2006} that generalizes the well-known Laughlin argument \cite{laughlin1981}, we define $\Delta$ to be the parity of the charge that is pumped toward a flux tube as half a quantum of \emph{spin flux} is threaded. This charge is quantized to be integers even in the many-body case. Moreover, it adds under stacking: if two systems $a$ and $b$ are stacked, then the charge pumped in the stacked system $a+b$ must be a sum of the individual systems. Now it is obvious that not only is a phase with odd $\Delta$ nontrivial, but it also cannot have any square root.

To recapitulate, we have argued for the robustness to interactions of all three nontrivial band insulators in the non-interacting $\ZZZ_4$ classification, by employing only a corollary to our general result. The full power of the general result will be manifest in \s{subsubsec:sanity_check} when we derive from it the complete classification of 3D SPT phases protected by glide, $U(1)$, and $\T$, which contains $\ZZZ_4$ as a subgroup.

\subsection{General relation between classifications with and without glide\label{subsec:general_relations}}

From now on we will be dealing with an \emph{arbitrary} symmetry $G$, and SPT phases protected by either $G$ or $G$ combined with a glide symmetry $\ZZZ$. Since our general results apply to both fermionic and bosonic SPT phases, we will often omit such adjectives as ``fermionic" and ``bosonic," with the understanding that $G$ denotes a full symmetry group, which contains fermion parity, in the fermionic case. The Wigner-Dyson class AII, to which hourglass fermions belong, corresponds to the fermionic case and a $G$ that is generated by charge conservation and a time reversal that squares to fermion parity---it is the unique non-split $U(1)$-extension of $\ZZZ_2$ for the non-trivial action of $\ZZZ_2$ on $U(1)$. We will present the main result of this section, a short exact sequence relating the classification of $G$- and $\ZZZ\times G$-protected SPT phases, in Sec.\,\ref{subsubsec:SES_classifications}. We will then explore its implications and derive some useful corollaries in Sec.\,\ref{subsubsec:corollaries}.
Before delving into the results, let us first clarify our terminology.

As demonstrated in Ref.\,\cite{McGreevy_sSourcery} by a worm hole array argument, the inverse of an SPT phase protected by an \emph{on-site} symmetry is given by its orientation-reversed version. That is, if $a$ is a system that represents an SPT phase $[a]$, then the orientation-reversed system $\bar a$ will represent the inverse SPT phase:
\begin{equation}
-[a] = [\bar a].
\end{equation}
The situation with non-on-site symmetries is more involved: with glide, the orientation-reversed version of either square root of the hourglass fermion phase is itself rather than the inverse, for instance.

Another useful notion is that of ``weakness with respect to glide."
Writing $\ZZZ$ for a glide symmetry, we say a $\ZZZ\times G$-protected SPT phase is weak with respect to glide if it becomes trivial under the glide-forgetting map:
\begin{equation}
\beta': \SPT^d\paren{\ZZZ \times G} \fromto \SPT^d\paren{G}. \label{beta}
\end{equation}
We shall denote the set of $d$-dimensional $\ZZZ\times G$-protected SPT phases that are weak with respect to glide by $\wSPT^d(\ZZZ \times G)$, or $\wSPT^d(\ZZZ\times G,\phi)$ for completeness. It is precisely the kernel of $\beta'$
\begin{equation}
\wSPT^d\paren{\ZZZ \times G} \coloneq \kernel \beta',
\end{equation}
which is a subgroup of the abelian group of $d$-dimensional $\ZZZ\times G$-protected SPT phases. Again, subscripts can be introduced to distinguish between bosonic and fermionic phases, as in $\wSPT^d_b$ or $\wSPT^d_f$, if necessary.

\subsubsection{Short exact sequence of classifications\label{subsubsec:SES_classifications}}

We now present the main result of the section, a short exact sequence that relates the classification of $d$-dimensional $\ZZZ \times G$-protected SPT phases $\SPT^d(\ZZZ \times G)$, the classification of $d$-dimensional $G$-protected SPT phases $\SPT^d(G)$, and the classification of $(d-1)$-dimensional $G$-protected SPT phases $\SPT^{d-1}(G)$, where $G$ is arbitrary and $\ZZZ$ is generated by a glide reflection.

Given any abelian group $A$, we write
\begin{equation}
2A \coloneq \braces{ 2a | a\in A }, \label{2A}
\end{equation}
for the subgroup of $A$ of those elements that have square roots, and we write $A/2A$ for the quotient of $A$ by $2A$. For example, $\ZZZ_n/ 2\ZZZ_n = \ZZZ_{\gcd(n,2)}$, where $\gcd$ stands for greatest common divisor; this even applies to $n= \infty$ if one defines $\gcd(\infty, 2) = 2$.

\begin{prp}[Short exact sequence of classifications]
Assume the generalized cohomology hypothesis. Let $d$ and $G$ be arbitrary and $\ZZZ$ be generated by a glide reflection. There is
a short exact sequence
\begin{equation}
0 \fromto \SPT^{d-1}(G) / 2\SPT^{d-1}(G) \xfromto{\alpha} \SPT^d\paren{\ZZZ\times G} \xfromto{\beta}  \{ [c] \in \SPT^d\paren{G} \big| 2[c] = 0 \} \fromto 0,\label{SES}
\end{equation}
where $\beta$ is the glide-forgetting map [same as $\beta'$ in Eq.\,(\ref{beta}) but with restricted codomain].
\label{prp:SES_classifications}
\end{prp}

\begin{figure}[t]
\centering
\includegraphics[height=1.5in]{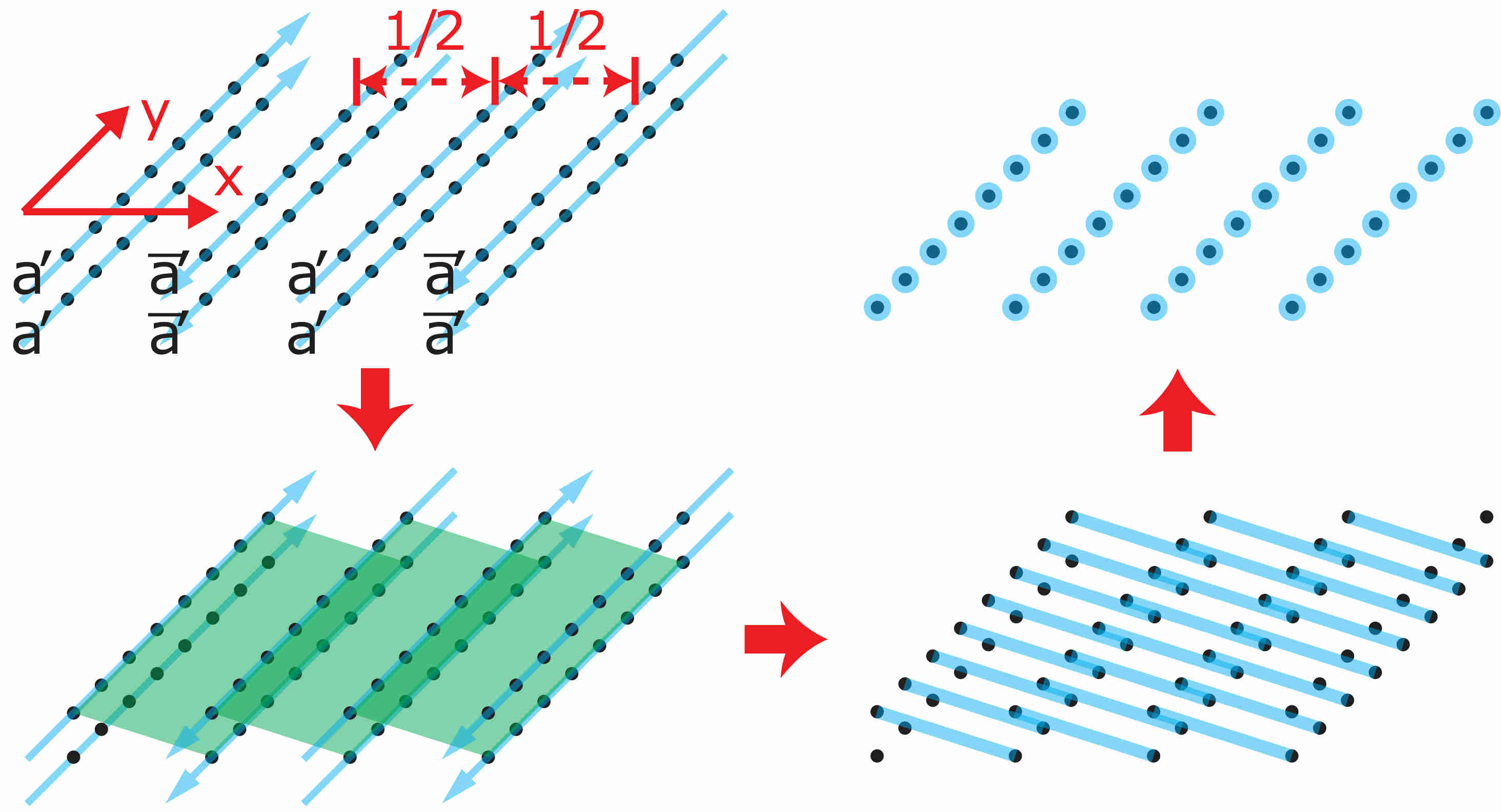}
\caption[Physical justification for the claim $\brackets{a} = \brackets{a'}+\brackets{a'} \Rightarrow \alpha'\paren{\brackets{a}} = 0$.]{Physical justification for the claim $[a] = [a']+[a'] \Rightarrow \alpha'([a]) = 0$, depicted for $d=2$, $(x,y)\mapsto (x+1/2, -y)$. Applying the alternating-layer construction to $a'+a'$ gives a 2D system (upper-left panel). By coupling an $a'$ (or $\overline{a'}$) in each $x\in \ZZZ$ (resp.\,$x\in \ZZZ + 1/2$) layer to an $\overline{a'}$ (resp.\,$a$) in the $x+1/2$ layer (lower-left panel), one can deform the ground state to a tensor product of individual states supported on diagonal pairs of sites (lower-right panel) (cf.\,Footnote\,\ref{footnote:diagonal_pairs}). A redefinition of sites then turns this into a tensor product of individual states supported on single sites, i.e.\,a trivial product state (upper-right panel) (cf.\,Footnote\,\ref{footnote:enlarged_Hilbert_space}). All deformations can be chosen to preserve $\ZZZ \times G$ and the gap.\label{fig:property_alpha'}}
\end{figure}

\begin{proof}
See App.\,C of work \cite{Xiong_Alexandradinata}.
\end{proof}

While the interpretation of $\beta$ in Proposition \ref{prp:SES_classifications} is clear from the proof in App.\,C of work \cite{Xiong_Alexandradinata},
the latter does not address the question as to what $\alpha$ means. Motivated by the discussions in Sec.\,\ref{subsubsec:alternating_layer_construction_triviality_glide_forgetting}, we shall posit that $\alpha$ is given by the alternating-layer construction introduced earlier (see Figure\,\ref{fig:alternating_layer_construction}).
More precisely, the alternating-layer construction defines a map
\begin{equation}
\alpha': \SPT^{d-1}(G) \fromto \SPT^d\paren{\ZZZ \times G}, \label{alpha'}
\end{equation}
with domain the abelian group of $(d-1)$-dimensional $G$-protected SPT phases.
By the physical argument in Figure\,\ref{fig:property_alpha'}, we must have $\alpha'([a]) = 0$
whenever there exists an $[a']\in \SPT^{d-1}(G)$ such that $[a] = [a'] + [a']$.\footnote{In general it is not necessary, or even possible, to go through this intermediate step. In the discussion of SPT phases, one always considers stable equivalence. That is, one allows for change of Hilbert spaces through, for instance, the introduction of ancillary lattices or a blocking of lattice sites \cite{Wen_1d, Cirac}. The deformation of a composite system into a system whose ground state is a tensor product may already involve such changes, in which case the intermediate picture of a tensor product of states supported on pairs of sites will no longer be accurate.\label{footnote:diagonal_pairs}}\textsuperscript{,}\footnote{More precisely, the transformation of states supported on pairs of sites into states supported on single sites involves first an enlargement of the Hilbert space and then a transfer of states in one sector of the enlarged Hilbert space to another \cite{Wen_1d, Cirac}.\label{footnote:enlarged_Hilbert_space}} This means $\alpha'$ can effectively be defined on the quotient $\SPT^{d-1}(G)/ 2\SPT^{d-1}(G)$, and we shall identify the induced map
\begin{equation}
\alpha: \SPT^{d-1}(G) / 2 \SPT^{d-1}(G) \fromto \SPT^d\paren{\ZZZ \times G} \label{alpha}
\end{equation}
as the $\alpha$ that appears in Proposition \ref{prp:SES_classifications}.
This interpretation of $\alpha$ is further supported by the various examples in Sec.\,\ref{subsec:applications} below and Ref.\,\cite{Lu_sgSPT}.%

We will see in Sec.\,\ref{subsec:physical_picture} that twice the glide reflection being orientation-preserving is closely related to the factors of 2 appearing in Proposition \ref{prp:SES_classifications}. The reader may have realized that $\SPT^{d-1}(G) / 2\SPT^{d-1}(G)$ and $\{ [c] \in \SPT^d\paren{G} | 2[c] = 0 \}$ can be expressed using the extension and torsion functors, as $\Ext^1( \ZZZ_2, \SPT^{d-1}(G) )$ and $\Tor_1 (\ZZZ_2, \SPT^d\paren{G})$, respectively, both of which are contravariant as they should be \cite{Xiong}.

\subsubsection{Implications of short exact sequence \label{subsubsec:corollaries}}

Let us explore the implications of Proposition \ref{prp:SES_classifications} and derive some useful corollaries from it.

First, the exactness of sequence (\ref{SES}) implies that $\image \alpha = \kernel \beta$ (this can be equivalently stated as $\image \alpha' = \kernel \beta'$ since, by definition, $\image \alpha = \image \alpha'$ and $\kernel \beta = \kernel \beta'$), which reproduces the result in Ref.\,\cite{Lu_sgSPT} that
\begin{cor}
A $\ZZZ\times G$-protected SPT phase is weak with respect to glide if and only if it can be obtained through the alternating-layer construction.
\end{cor}

\noindent Furthermore, since every element of an abelian group of the form $A/2A$ is either trivial or of order 2, any $\ZZZ \times G$-protected SPT phase that is weak with respect to glide -- and hence obtainable from an element of $\SPT^{d-1}(G) / 2 \SPT^{d-1}(G)$ -- must also be either trivial or of order 2. In fact, a necessary and sufficient condition for such a phase to be trivial (resp. has order 2) is that the $(d-1)$-dimensional $G$-protected SPT phase it comes from has a square root (resp. has no square root).\footnote{While two $(d-1)$-dimensional $G$-protected SPT phases may lead to the same $d$-dimensional $\ZZZ \times G$-protected SPT phase, this if-and-only-if condition is unambiguous because any two such phases must both have square roots or both have no square roots.} This follows from the exactness of sequence (\ref{SES}), which implies $\alpha$ is injective. 

The above necessary and sufficient condition allows us to classify $d$-dimensional $\ZZZ \times G$-protected SPT phase that are weak with respect to glide by classifying $(d-1)$-dimensional $G$-protected SPT phases instead:
\begin{cor}
\label{cor:classification_SPT_weak_wrt_glide}
There is an isomorphism
\begin{equation}
\wSPT^d\paren{\ZZZ \times G} \isomorphic \SPT^{d-1}(G) / 2\SPT^{d-1}(G). \label{GSPT}
\end{equation}
\end{cor}

\noindent This isomorphism was conjectured in Ref.\,\cite{Lu_sgSPT} for on-site $G$, based on studies of a number of fermionic and bosonic examples in the $d=3$ case. Unfortunately, the anomalous surface topological order argument used therein does not generalize to all dimensions. The minimalist framework allows us to confirm their conjecture in the general case -- in all dimensions and for all symmetries $G$, which do not even have to act in an on-site fashion.

We are concerned with all $\ZZZ\times G$-protected SPT phases, not just those that are weak with respect to glide. We know that by forgetting glide each $d$-dimensional $\ZZZ\times G$-protected SPT phase can be viewed as a $d$-dimensional $G$-protected SPT phase, but can all $d$-dimensional $G$-protected SPT phases be obtained this way? In other words, are all $d$-dimensional $G$-protected SPT phases compatible with glide? Our result indicates that the answer is in general no. This is because the exactness of sequence (\ref{SES}) implies $\beta$ is surjective, and inspecting the third term (from the left, excluding the initial 0) of sequence (\ref{SES}) one sees that
\begin{cor}[Compatibility with glide]
A necessary and sufficient condition for a $d$-dimensional $G$-protected SPT phase to be compatible with glide is that it squares to the trivial phase.
\end{cor}

In Corollary \ref{cor:classification_SPT_weak_wrt_glide} we saw that the classification of $d$-dimensional $\ZZZ \times G$-protected SPT phases \emph{that are weak with respect to glide} can be obtained from the classification of $(d-1)$-dimensional $G$-protected SPT phases. We now demonstrate that the classification of \emph{all} $d$-dimensional $\ZZZ \times G$-protected SPT phases is severely constrained, if not completely determined, by the classification of $(d-1)$- and $d$-dimensional $G$-protected SPT phases. Indeed, we recognize that the task of determining the second term of a short exact sequence of abelian groups from the first and third terms is nothing but an abelian group extension problem. It is well-known that abelian group extensions $0 \fromto A \fromto B \fromto C \fromto 0$ of $C$ by $A$ are classified, with respect to a suitable notion of equivalence, by the subgroup $H^2_{\rm sym}\paren{C;A}$ of $H^2\paren{C; A}$ of symmetric group cohomology classes.
To illustrate how $A$ and $C$ constrain $B$, let us take $A= C = \ZZZ_2$. In this case, $H^2_{\rm sym}\paren{\ZZZ_2;\ZZZ_2} \isomorphic \ZZZ_2$. The trivial and nontrivial elements of $H^2_{\rm sym}\paren{\ZZZ_2;\ZZZ_2}$ correspond to $B = \ZZZ_2 \oplus \ZZZ_2$ and $\ZZZ_4$, respectively, which are the only solutions to the abelian group extension problem.

Inspecting (\ref{SES}), we note that its first and third terms are such that their nontrivial elements all have order 2. Consequently, all nontrivial elements of the second term must have order 2 or 4. More precisely, we have
\begin{cor}[Quad-chotomy of phases]
\label{cor:quad-chotomy}
Each $\ZZZ\times G$-protected SPT phase is exactly one of the following:
\begin{enumerate}[(i)]
\item the unique trivial phase;

\item a nontrivial phase of order 2 that is weak with respect to glide;

\item a nontrivial phase of order 2 that is not weak with respect to glide;

\item a nontrivial phase of order 4 that is not weak with respect to glide per se but whose square is of type (ii).
\end{enumerate}
\end{cor}

\begin{proof}
See App.\,D of work \cite{Xiong_Alexandradinata}.
\end{proof}

In particular, this means that a nontrivial $\ZZZ\times G$-protected SPT phase that is weak with respect to glide can sometimes have square roots, and that such square roots, if exist, are never weak with respect to glide. On the other hand, a nontrivial $\ZZZ\times G$-protected SPT phase that is \emph{not} weak with respect to glide can \emph{never} have square roots. Without Proposition \ref{prp:SES_classifications}, these results would not have been obvious.

From the perspective of classification, it would be nice to have a statement about the explicit form of $\SPT^d(\ZZZ \times G)$. In App.\,D of work \cite{Xiong_Alexandradinata}, we prove that Corollary \ref{cor:quad-chotomy}, together with the fact that SPT phases form an abelian group, implies that $\SPT^d(\ZZZ \times G)$ can be written as a direct sum of $\ZZZ_4$'s and $\ZZZ_2$'s:

\begin{cor}[Direct-sum decomposition]
\label{cor:direct_sum_decomposition}
There is a direct sum decomposition,
\begin{equation}
\SPT^d(\ZZZ \times G) \isomorphic \paren{\bigoplus_i \ZZZ_4} \oplus \paren{\bigoplus_j \ZZZ_2 }.\label{directsumdecomp}
\end{equation}
$\ZZZ \times G$-protected SPT phases that correspond to 1 or 3 of any $\ZZZ_4=\braces{0,1,2,3}$ summand are never weak with respect to glide, whereas those that correspond to $2\in \ZZZ_4$ are always weak with respect to glide. The nontrivial element of a $\ZZZ_2$ summand, on the other hand, may or may not be weak with respect to glide.
\end{cor}

\subsection{Physical intuition behind general relation\label{subsec:physical_picture}}

Proposition \ref{prp:SES_classifications} was derived mathematically from the generalized cohomology hypothesis. Here we offer a complementary physical perspective. We will break down Proposition \ref{prp:SES_classifications} into 6 individual statements and explain them \emph{physically} where possible:
\begin{enumerate}[(i)]

\item $\image \alpha' \subset \kernel \beta'$;

\item $\image \alpha' \supset \kernel \beta'$;

\item $\alpha$ is well-defined;

\item $\alpha$ is injective;

\item $\image \beta' \subset \{ [c] \in \SPT^d\paren{G} \big| 2[c] = 0 \}$;

\item $\image \beta' \supset \{ [c] \in \SPT^d\paren{G} \big| 2[c] = 0 \}$.
\end{enumerate}
We will discuss (i) and (ii) in Sec.\,\ref{subsubsec:alternating_layer_construction_triviality_glide_forgetting}, (iii) and (iv) in Sec.\,\ref{subsubsec:square_root_triviality}, and (v) and (vi) in Sec.\,\ref{subsubsec:compatibility_glide_2_torsion}. Of the six statements, (i)(ii)(iii)(v)(vi) admit obvious physical explanation, whereas (iv) can be justified by physical examples. Although the proof of Proposition \ref{prp:SES_classifications} was rigorous and the generalized cohomology hypothesis can largely be justified on independent grounds, it is but reassuring that Proposition \ref{prp:SES_classifications} is consistent with one's physical intuition.

\subsubsection{Alternating-layer construction and triviality under glide forgetting\label{subsubsec:alternating_layer_construction_triviality_glide_forgetting}}

\begin{figure}[t]
\centering
\includegraphics[width=3.8in]{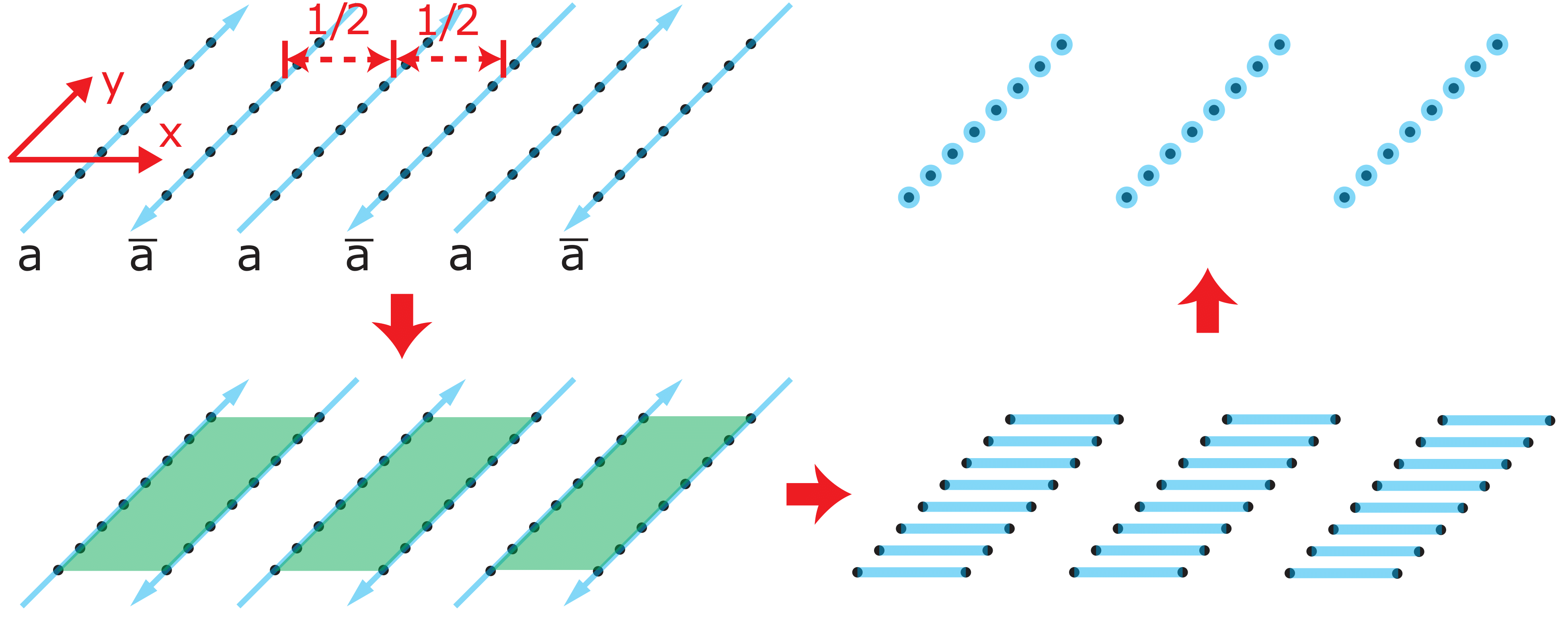}
\caption[Physical justification for the claim $\image \alpha' \subset \kernel \beta'$.]{Physical justification for the claim $\image \alpha' \subset \kernel \beta'$, depicted for $d=2$, $(x,y)\mapsto (x+1/2, -y)$. Applying the alternating-layer construction to $a$ gives a 2D system (upper-left panel). By coupling each $x\in \ZZZ$ layer to the $x+1/2$ layer (lower-left panel), one can deform the 2D system so as to have a ground state that is the tensor product of individual states supported on pairs of sites (lower-right panel) (cf. Footnote\,\ref{footnote:diagonal_pairs}). A blocking procedure then turns the latter into a tensor product of individual states supported on single sites (upper-right panel). All deformations can be chosen to preserve $G$ and the gap.\label{fig:image_alpha_subset_kernel_beta}}
\end{figure}

Given our interpretation of $\alpha'$ as the alternating-layer construction, $\image \alpha' {\subset} \kernel \beta'$ amounts to saying that \emph{the alternating-layer construction always produces $\ZZZ\times G$-protected SPT phases that are weak with respect to glide}. In other words, given a system $a$ representing a $(d-1)$-dimensional $G$-protected SPT phase, the $d$-dimensional system obtained from $a$ through the alternating-layer construction can always be trivialized when the glide symmetry constraint is relaxed. Indeed, given such a $d$-dimensional system, one can simply pair up neighboring layers and deform the pairs into trivial systems, as illustrated in Figure\,\ref{fig:image_alpha_subset_kernel_beta}.

The converse, $\kernel \beta' \subset \image \alpha'$, says that \emph{all $\ZZZ\times G$-protected SPT phases that are weak with respect to glide can be obtained through the alternating-layer construction}. That is, if a system $b$ representing a $\paren{\ZZZ \times G}$-protected SPT phase can be trivialized when the glide symmetry constraint is relaxed, then it can be deformed to a system obtained from the alternating-layer construction while preserving the glide symmetry. Indeed, an argument involving applying a symmetric, finite-depth quantum circuit to subregions of a glide-symmetric system has been devised in Ref.\,\cite{Lu_sgSPT} to justify this claim, assuming the lattice period is large compared to the correlation length.

We can view the above as physically motivating our identification of $\alpha'$ as the alternating-layer-construction map in the first place. We will soon be delighted to find out that this interpretation is consistent with the other statements as well.

\subsubsection{Square root in $(d-1)$ dimensions and triviality in $d$ dimensions \label{subsubsec:square_root_triviality}}

The well-definedness of $\alpha$ says, \emph{given a $(d-1)$-dimensional $G$-protected SPT phase $[a]$, that the $d$-dimensional $\ZZZ\times G$-protected SPT phase obtained from it through the alternating-layer construction is trivial whenever $[a]$ has a square root}. As mentioned in Sec.\,\ref{subsubsec:SES_classifications}, we can justify this claim using the physical argument in Figure\,\ref{fig:property_alpha'}. Ref.\,\cite{Lu_sgSPT} has also given an equivalent argument.

On the other hand, the injectivity of $\alpha$ says that \emph{the $d$-dimensional $\ZZZ\times G$-protected SPT phase obtained from $[a]$ through the alternating-layer construction is nontrivial whenever $[a]$ has no square root}.
Physically, this has been shown to be the case for a number of bosonic and fermionic systems for $d=3$ using the $K$-matrix construction \cite{Lu_sgSPT}. In general, one can attempt a construction of bulk invariants for the $d$-dimensional system in question, but a universal strategy that works for all $d$ seems lacking.

\subsubsection{Compatibility with glide and $\ZZZ_2$ torsion \label{subsubsec:compatibility_glide_2_torsion}}

Finally, let us argue that
\begin{equation}
\image \beta' \subset \{ [c] \in \SPT^d\paren{G} \big| 2[c] = 0 \},
\end{equation}
which amounts to saying that \emph{if a $G$-protected SPT phase has a glide-symmetric representative} (i.e.\,is compatible with glide), \emph{then it must square to the trivial phase} (i.e.\,belong to the $\ZZZ_2$ torsion subgroup). In other words, given a system $b$ representing a $d$-dimensional $\ZZZ \times G$-protected SPT phase, the stacked system $b+b$ can always be trivialized when the glide symmetry constraint is relaxed.

Let us begin by considering the stacked system $b + \bar b$, where $\bar b$ is the mirror image of $b$ under $y \mapsto -y$ [Figure\,\ref{fig:image_beta}(a)]. When glide is relaxed,
$\bar b$ serves as the inverse of $b$. This means one can deform $b + \bar b$ to a trivial product state [Figure\,\ref{fig:image_beta}(b)]. Now we translate $\bar b$ by $1/2$ in the $x$-direction. Then the stacked system $b + (\mbox{translated }\bar b)$ [Figure\,\ref{fig:image_beta}(c)] can also be deformed to a tensor product state, or more precisely, a tensor product of individual states that are supported on diagonal pairs of sites [Figure\,\ref{fig:image_beta}(d)] (cf.\,Footnote\,\ref{footnote:diagonal_pairs}). A redefinition of sites then turns the latter into a tensor product of individual states supported on single sites, that is, into a trivial product state [Figure\,\ref{fig:image_beta}(f)] (cf.\,Footnote\,\ref{footnote:enlarged_Hilbert_space}). To see that $b+b$ can be deformed to a trivial product state, we simply need to note that, being glide-symmetric, $b$ is the same as the translated $\bar b$.

\begin{figure}[t]
\centering
\includegraphics[width=3.8in]{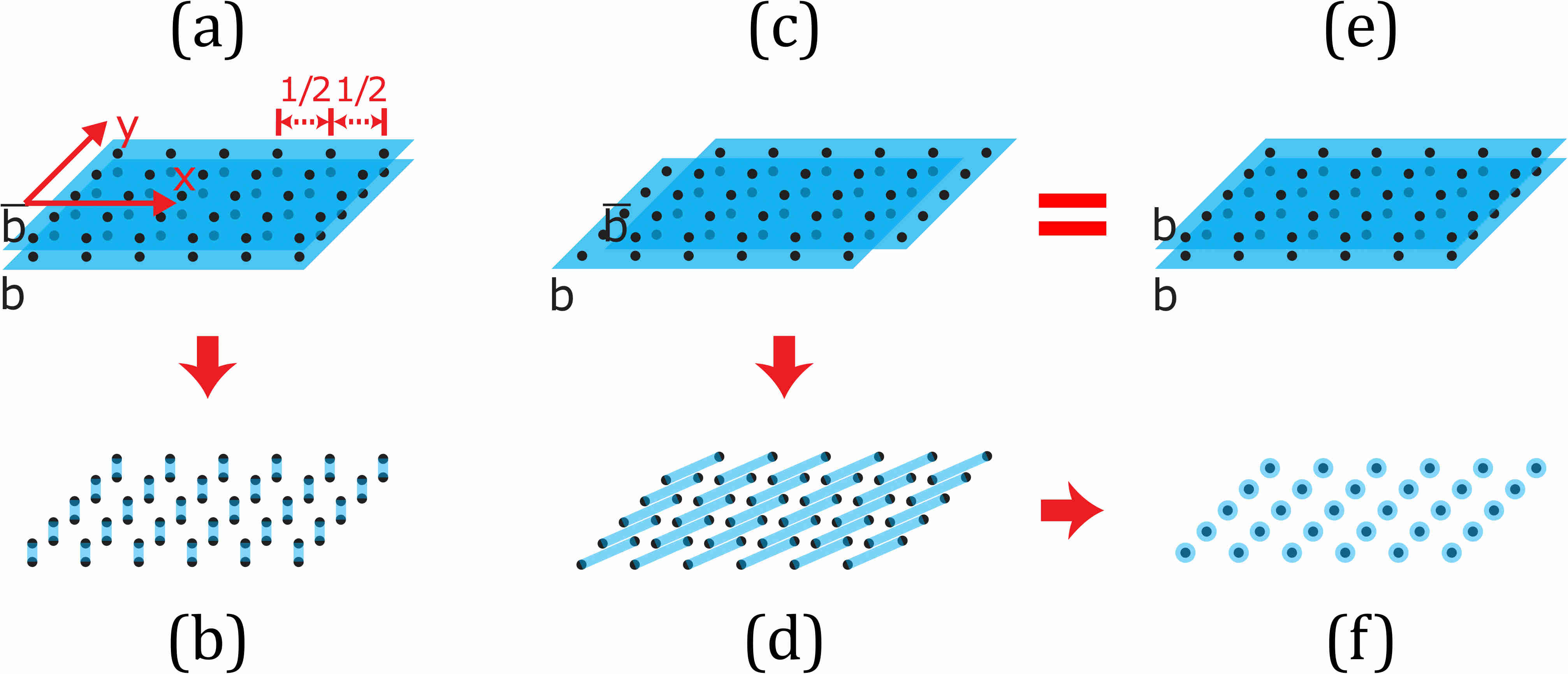}
\caption[Physical justification for the claim $\image \beta' \subset \{ \brackets{c} \in \SPT^d\paren{G} \big| 2\brackets{c} = 0 \}$.]{Physical justification for the claim $\image \beta' \subset \{ [c] \in \SPT^d\paren{G} \big| 2[c] = 0 \}$, depicted for $d=2$, $(x,y)\mapsto (x+1/2, -y)$.\label{fig:image_beta}}
\end{figure}

\begin{figure}[t]
\centering
\includegraphics[width=2.6in]{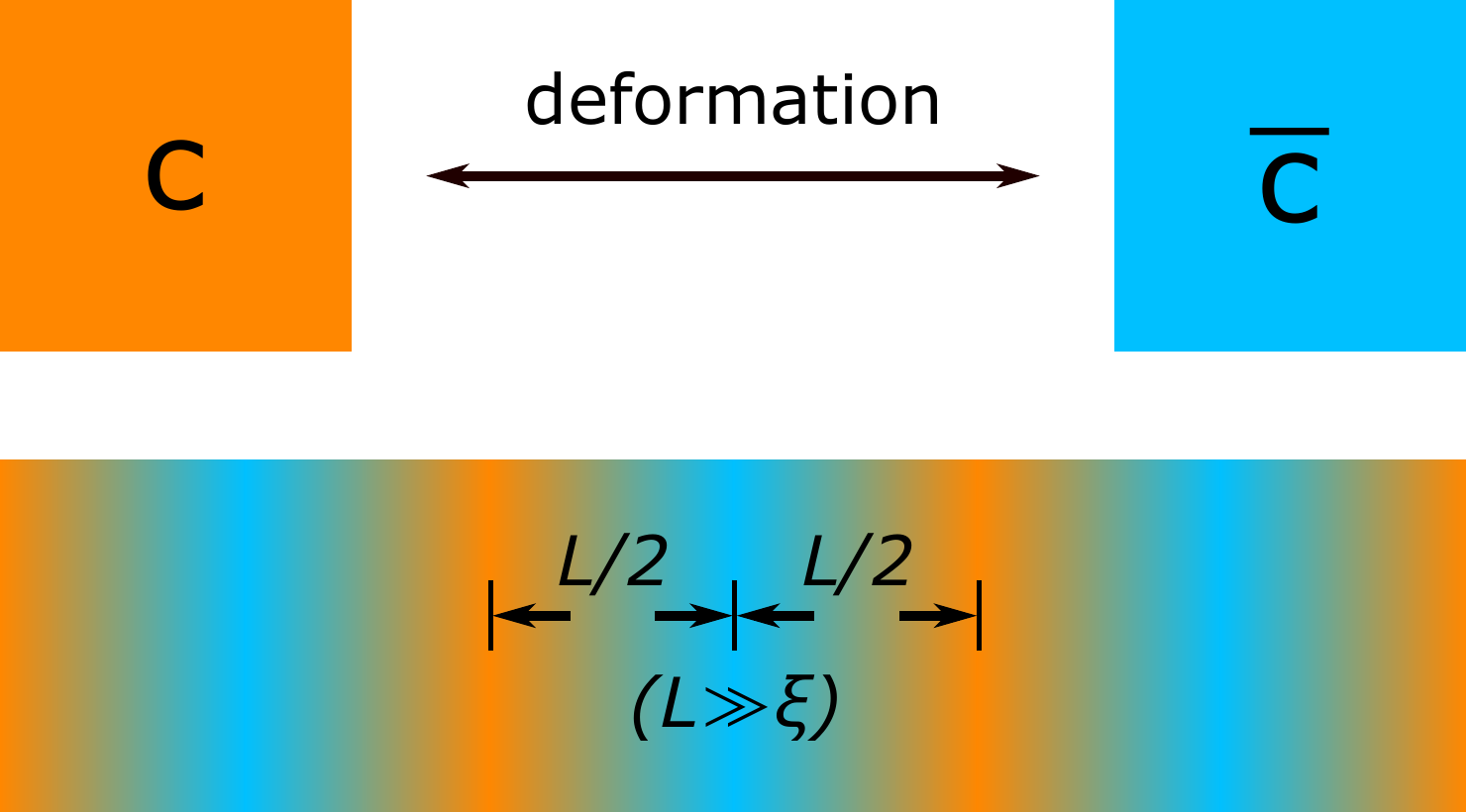}
\caption[Physical justification for the claim $\{ \brackets{c} \in \SPT^d\paren{G} \big| 2\brackets{c} = 0 \} \subset \image \beta'$.]{Physical justification for the claim $\{ [c] \in \SPT^d\paren{G} \big| 2[c] = 0 \} \subset \image \beta'$, depicted for $d=2$.\label{fig:lastresult}}
\end{figure}

The converse,
\begin{equation}
\{ [c] \in \SPT^d\paren{G} \big| 2[c] = 0 \} \subset \image \beta', \label{lastresult}
\end{equation}
says that \emph{if a $d$-dimensional $G$-protected SPT phase squares to the trivial phase, then it must have a glide-symmetric representative.} As a quick argument for this, we appeal to the empirical beliefs that (a) an SPT phase that squares to the trivial phase has a reflection-symmetric representative, and that (b) any SPT has a translation-invariant representative \cite{Xiong}, which are consistent with known examples. A case in point for (a) is the reflection-symmetric topological insulator Bi$_2$Se$_3$ ($\nu_0=1$ in class AII); (b) is exemplified by all experimentally realized band topological insulators. Now, suppose (a) and (b) can be compatibly realized in the same system, with the reflection axis ($x_2 \mapsto -x_2$) orthogonal to at least one translation direction ($x_1 \mapsto x_1 + 1/2$). Composing the two transformations, we see that the system is also invariant under the glide symmetry $(x_1,x_2,\ldots)\mapsto (x_1+1/2,-x_2,\ldots)$.

An alternative argument which does not depend on the belief (a) above is this. Suppose a system $c$ represents a $G$-protected SPT phase $[c]$ that squares to the trivial phase. The condition $2[c]=0$ is equivalent to the condition $[c] = -[c]$, or $[c] = \brackets{\bar c}$, where $\bar c$ is the orientation-reversed (say $x_2 \mapsto -x_2$) version of $c$. The last expression means that $c$ can be deformed to $\bar c$ without closing the gap or breaking the symmetry (see the upper panel of Figure\,\ref{fig:lastresult}). Let $\hat H(\lambda)$ be a family of translation-invariant Hamiltonians parameterized by $\lambda\in [0,1]$ that represents this deformation. Since $\lambda$ is a compact parameter ($[0,1]$ being closed and bounded), we expect the correlation length of $\hat H(\lambda)$ to be uniformly bounded by some finite $\xi$ \cite{hastings2006spectral}. Being translation-invariant, each $\hat H(\lambda)$ is a sum of the form
\begin{equation}
\hat H(\lambda) = \sum_{\boldsymbol x} \sum_i g_i(\lambda) \hat O^i_{\boldsymbol x},
\end{equation}
where $\boldsymbol x = \paren{x_1, x_2, \ldots}$ runs over all lattice points, the operators $\hat O^i$ have compact supports (of radii $r_i$), $\hat O^i_{\boldsymbol x}$ denotes the operator $\hat O^i$ centered at $\boldsymbol x$, the coupling constants $g_i(\lambda)$ depend on $\lambda$, and $g_i(\lambda)$ decay exponentially with $r_i$. Now, we construct a new Hamiltonian $\hat H'$ that modulates spatially at a scale $L$ much larger than $\xi$. This can be achieved by letting $\lambda$ vary with one of the coordinates, say $x_1$; for instance, we can set
\begin{equation}
\lambda = \frac{x_1}{L/2}
\end{equation}
for $x_1 \in \brackets{0,L/2}$. In the neighborhood of $x_1 = 0$ and $L/2$, the Hamiltonian $\hat H'$ will coincide with $\hat H(0)$ and $\hat H(1)$, respectively. This defines $\hat H'$ only in the strip $x_1 \in \brackets{0,L/2}$, but since $\hat H(1)$ is related to $\hat H(0)$ by $x_2 \mapsto -x_2$, we can place the reversed strip on $x_1 \in \brackets{L/2,L}$ and glue the two strips together. Iterating this process {\it ad infinitum} to create a superlattice, we will arrive at a Hamiltonian $\hat H'$ that is explicitly invariant under the glide transformation $\paren{x_1, x_2, \ldots} \mapsto \paren{x_1+L/2, -x_2, \ldots}$; see the lower panel of Figure\,\ref{fig:lastresult}. Due to the separation of scale $L \gg \xi$, we expect $\hat H'$ to be gapped. To see that $\hat H'$ represents the SPT phase $[c]$, we note that in the neighborhood of any $x_1 \in [0,L/2]$ (resp.\,$[L/2,L]$), there is some $\lambda$ for which $\hat H'$ is locally indistinguishable from $\hat H(\lambda)$ (resp.\,its orientation-reversed version), which represents $[c]$.

An explicit formula for $\hat H'$ can be given. Let $\hat M$ be the operator that implements the orientation-reversal $x_2 \mapsto -x_2$. Then we can write
\begin{equation}
\hat H' = \sum_{x_1} \hat H'_{x_1},
\end{equation}
where
\begin{equation}
\hat H'_{x_1} = \sum_{x_2, x_3, \ldots} \sum_i g_i\paren{\frac{x_1-nL}{L/2}} \hat O^i_{\boldsymbol x}
\end{equation}
for $x_1 \in \brackets{nL, (n+1/2)L}$, and
\begin{equation}
\hat H'_{x_1} = \sum_{x_2, x_3, \ldots} \sum_i g_i\paren{\frac{x_1-(n+1/2)L}{L/2}} \hat M \hat O^i_{\boldsymbol x} \hat M^{-1}
\end{equation}
for $x_1 \in \brackets{(n+1/2)L,(n+1)L}$. Here $n$ takes values in the integers.

\subsection{Complete classifications in Wigner-Dyson classes A and AII\label{subsec:applications}}

In this section, we will demonstrate that the predictions of Proposition \ref{prp:SES_classifications} are consistent with existing literature on the classification of free-fermion phases and their robustness to interactions. More importantly, we will use Proposition \ref{prp:SES_classifications} to deduce the putative complete classifications of fermionic SPT phases with glide from proposed complete classifications of fermionic SPT phases without glide. The latter is an abelian group extension problem, where knowing the first and third terms $A$ and $C$ of a short exact sequence of abelian groups,
\begin{equation}
0 \fromto A \fromto \mbox{?} \fromto C \fromto 0,
\end{equation}
one has to determine the second. For definiteness, we will first focus on $d=3$ and $G=U(1)$ (charge conservation only, Wigner-Dyson class A). Then, we will re-examine the symmetry class of the hourglass-fermion phase, where $d=3$ and $G$ is generated by $U(1)$ and $\T$ where $\T$ squares to fermion parity (charge conservation and time reversal, Wigner-Dyson class AII).

\subsubsection{Wigner-Dyson class A\label{subsubsec:classA}}

Let us set $d=3$ and $G=U(1)$, which corresponds to Wigner-Dyson class A. 2D free-fermion phases in this symmetry class are classified by the first Chern number ($C_1{\in} \ZZZ$), which is defined over the Brillouin torus \cite{TKNN} but can be generalized to the interacting or disordered case by considering the torus of twisted boundary conditions instead \cite{Niu_twistedBC}. Being robust to interactions and disorder and admitting no square root, a phase with odd Chern number represents a nontrivial element in the first term of sequence (\ref{SES}). This phase may be layered in an alternating fashion to form a 3D phase respecting an additional glide symmetry $\ZZZ$.

The non-interacting, clean limit\footnote{By a clean limit, we mean a system that has discrete translational symmetry in three independent directions.} of the resultant 3D phase was independently studied in Refs.\ \cite{unpinned} and \cite{Shiozaki2015}. The surface states have a characteristic connectivity over the surface Brillouin torus illustrated in \fig{fig:z2surfacestates}; these surface states have been described as carrying a M\"obius twist \cite{Shiozaki2015}, so we shall refer to this phase as the M\"obius-twist phase. It was concluded in both references that the non-interacting classification (class A \emph{with glide}, 3D) is $\ZZZ_2$, and a topological invariant was proposed ($\kappa\in \{0,1\}$) to distinguish the two phases.

\begin{figure}[t]
\centering
\includegraphics[width=8.2 cm]{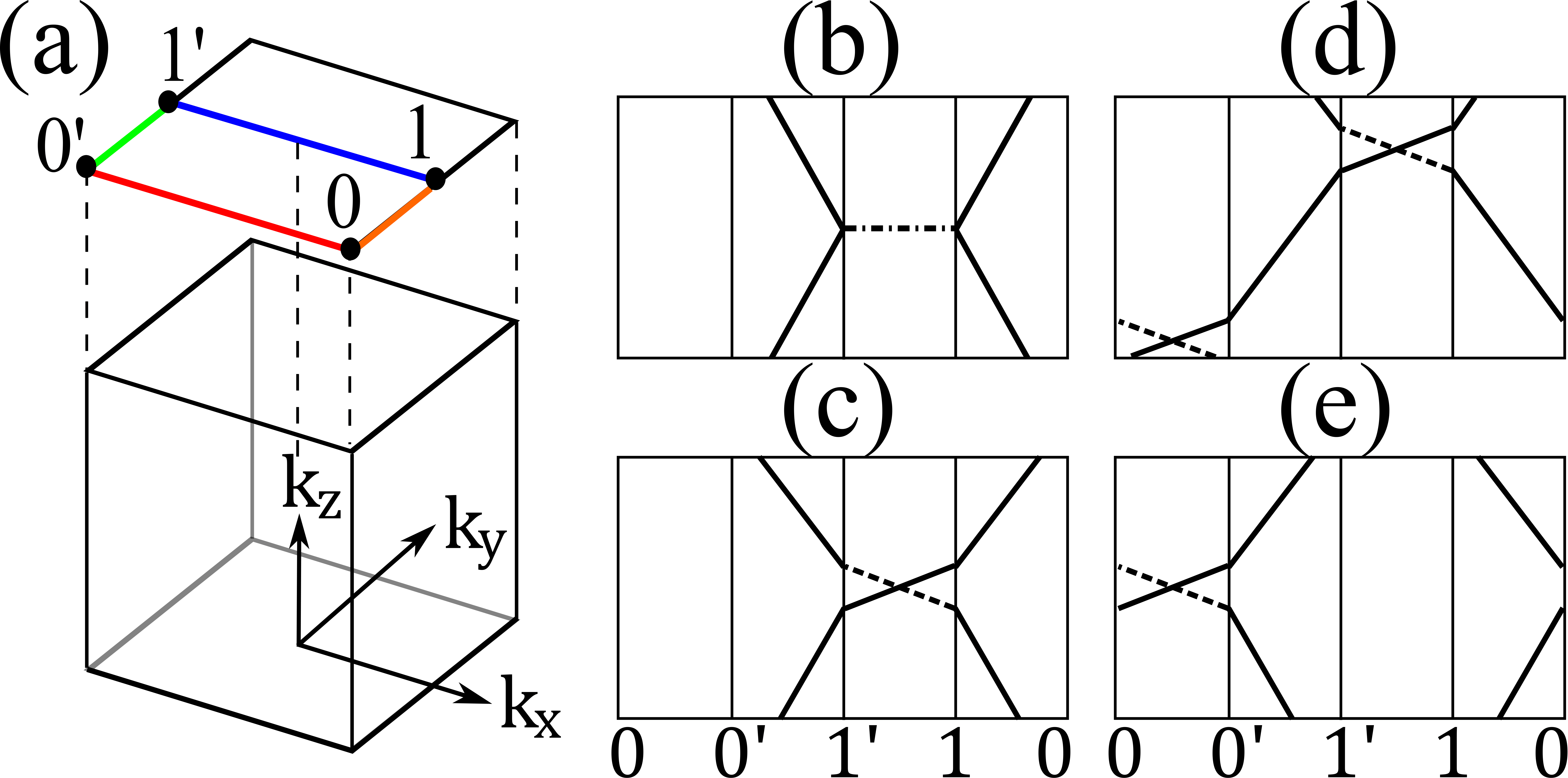}
\caption[Possible surface states of a glide-symmetric crystal in Wigner-Dyson class A.]{ (a) Bottom: Brillouin 3-torus for a glide-symmetric crystal; top:  Brillouin 2-torus corresponding to the glide-symmetric surface. (b-e) Surface states with a Mobius twist; a surface band in the even (odd) representation of glide  is indicated by a solid (dashed) line. (b-e) are representatives of the same phase, i.e., they are connected by symmetric deformations of the Hamiltonian that preserve the bulk gap. In (b), the solid-dashed line indicates a doubly-degenerate band  originating from the alternating-layer construction; this degeneracy may be split by generic perturbations, as illustrated in (c).\label{fig:z2surfacestates}}
\end{figure}

That the M\"obius-twist phase ($\kappa=1$) can be obtained from the alternating-layer construction as above is especially evident in the non-interacting limit, where we can utilize the connectivity of surface states as an argument, in conjunction with the bulk-boundary correspondence \cite{Cohomological}. The following may be viewed as the class-A analog of the argument presented in  \s{subsubsec:robustness_hourglass_fermion_phase} for class AII.
Let us begin with the $C_1{=}1$ phase with a single edge chiral mode; the mirror image of this Chern phase has $C_1{=}{-}1$ and a single edge chiral mode with opposite velocity -- an ``anti-chiral" mode for short. When layered together in the $x$-direction with vanishing interlayer coupling, we obtain a superposition of chiral and anti-chiral modes (in the shape of an X) that do not disperse with $k_x$, as illustrated in \fig{fig:z2surfacestates}(b). Note in particular the two-fold energy degeneracy along the glide-invariant line $1'1$, which originates from the intersection of chiral and anti-chiral modes. If we perturb the system with a glide-symmetric interlayer coupling, the degenerate two-band subspace is bound to split into a connected graph (in the shape of a M\"obius twist) over the glide-invariant line as in \fig{fig:z2surfacestates}(c), owing to the monodromy \cite{connectivityMichelZak} of the representation of glide. The topology of the graph over $1'100'$ then confirms that the system is characterized by $\kappa{=}1$.

One implication of our short exact sequence that goes beyond the aforementioned non-interacting works is that the 3D M\"obius-twist phase ($\kappa=1$) is robust to interactions. More precisely, it cannot be connected to the trivial phase ($\kappa=0$) by turning on interactions that preserve the many-body gap and $\ZZZ \times{ U(1)}$ symmetry. Indeed, we can utilize the same corollary as quoted in \s{subsubsec:robustness_hourglass_fermion_phase} and recognize that an insulator with odd Chern number admits no square root. Equivalently, we can view this as a direct consequence of the injectivity of the map $\alpha$ in Proposition \ref{prp:SES_classifications}.

An independent argument for the robustness of the M\"obius-twist phase under $\ZZZ \times U(1)$ may be obtained from the quantized magnetoelectric \cite{essin2009} bulk response. In this case, the quantization results from the glide symmetry, which maps the axion angle from $\theta \mapsto -\theta$.\footnote{One way to rationalize this is to apply the pseudo-scalar transformation behavior of $\bE\cdot \bB$.} Since $\theta$ is defined modulo $2\pi$ \cite{Wilczek_axion}, it is fixed to $0$ or $\pi$, for $\kappa=0$ or $1$, respectively. Another independent argument for the robustness is that the glide-symmetric surface of the M\"obius-twist phase allows for an anomalous topological order of the T-Pfaffian type, which cannot exist in pure 2D glide-symmetric systems \cite{Lu_sgSPT}.

When glide is forgotten, the non-interacting 3D classification in class A is trivial (i.e. there is only one phase). Hence the  M\"obius-twist phase is weak with respect to glide. We may also argue for its weakness by noting that, without glide, there is no obstruction to coupling and trivializing adjacent layers with opposite $C_1$. As a $\ZZZ \times U(1)$-protected SPT phase, the M\"obius-twist phase has order 2 because two copies of itself have $\kappa = 1+1 \equiv 0$. Thus the M\"obius-twist phase falls into category (ii) of Corollary \ref{cor:quad-chotomy}.

Let us now include interaction-enabled fermionic SPT phases\footnote{By an interaction-enabled fermionic SPT phase, we mean a phase that does not have a free-fermion representative.\label{footnote:interaction_enabled}} and demonstrate how the relation between classifications, as encapsulated in Proposition \ref{prp:SES_classifications}, can help us pin down the complete classification of 3D $\ZZZ \times U(1)$-protected fermionic SPT phases. It is believed, without glide, that the complete classifications of 2 and 3D $U(1)$-protected fermionic SPT phases are 
\begin{eqnarray}
\SPT^2_f\paren{  U(1) } &\isomorphic& \ZZZ{\oplus}\ZZZ, \\
\SPT^3_f\paren{  U(1) } &\isomorphic& 0,
\end{eqnarray}
respectively, where the first $\ZZZ$ is generated by the $C_1{=}1$ phase and the second $\ZZZ$ by the neutral $E_8$ phase \cite{Kitaev_honeycomb, 2dChiralBosonicSPT, 2dChiralBosonicSPT_erratum, Kitaev_KITP}. Inserting these into the short exact sequence of Proposition \ref{prp:SES_classifications}, we obtain an abelian group extension problem:
\e{0 \rightarrow \ZZZ_2 \oplus \ZZZ_2 \rightarrow\; ? \rightarrow 0 \rightarrow 0.\label{groupextension_A}
}
It is an elementary property of group extension that the group extension of the trivial group by any other group is unique. More generally, if $A=0$ or $C=0$ in a short exact sequence $0 \fromto A \fromto B \fromto C \fromto 0$, then $B\isomorphic C$ or $B\isomorphic A$, respectively. Either way, we conclude that there is a unique solution to Eq.\,(\ref{groupextension_A}), and the complete classification of 3D $\ZZZ \times U(1)$-protected fermionic SPT phases is
\e{\SPT^3_f\paren{\ZZZ {\times}  U(1) }\isomorphic \ZZZ_2 \oplus \ZZZ_2,\label{sptclassA}}
which is consistent with Corollaries \ref{cor:classification_SPT_weak_wrt_glide} and \ref{cor:direct_sum_decomposition}. This result goes beyond the previous work Ref.\,\cite{Lu_sgSPT} in that Ref.\,\cite{Lu_sgSPT} only classified SPT phases that are weak with respect to glide. Our result indicates that, in this case, the ``weak classification" is complete. In the next subsection, we will investigate a case where the weak classification is not complete. We will see that the complete classification can still be determined through our short exact sequence.

\subsubsection{Wigner-Dyson class AII\label{subsubsec:sanity_check}}

Let us set $d=3$ and $G$ to be generated by $U(1)$ and $\T$ where $\T$ squares to fermion parity, which corresponds to Wigner-Dyson class AII. As a group, $G$ is the unique non-split $U(1)$-extension of $\ZZZ_2$ for the non-trivial action of $\ZZZ_2$ on $U(1)$.

As mentioned in Sec.\,\ref{subsec:hourglass_fermions}, the free-fermion classification in this symmetry class is $\ZZZ_2$ without glide and $\ZZZ_4$ with glide. The hourglass-fermion phase has a $\ZZZ_4$ index $\chi = 2$ and 3D $\ZZZ_2$ index $\nu_0 = 0$, so it represents a $\ZZZ \times G$-protected SPT phase that is weak with respect to glide. It is still a nontrivial SPT phase, though, by the discussion in Sec.\,\ref{subsec:hourglass_fermions}. We commented in Sec.\,\ref{subsubsec:SES_classifications} that all such SPT phases must have order two, which is indeed the case because two copies of the hourglass-fermion phase will have a $\ZZZ_4$ index $\chi = 2 + 2 \equiv 0 \mod 4$. On the other hand, both square roots of the hourglass-fermion phase have a 3D $\ZZZ_2$ index $\nu_0 = 1$, so while they represent nontrivial $\ZZZ \times G$-protected SPT phases, they are not weak with respect to glide. As $\ZZZ \times G$-protected SPT phases they do not have order 2 because the hourglass-fermion phase is nontrivial. They have order 4 because four copies of either square root has a $\ZZZ_4$ index $\chi = 4 \times 1$ or $4 \times 3 \equiv 0 \mod 4$. We see that the hourglass-fermion phase and its square roots fall into categories (ii) and (iv) of Corollary \ref{cor:quad-chotomy}, respectively.

Let us now include interaction-enabled fermionic SPT phases (see Footnote\,\ref{footnote:interaction_enabled}) and demonstrate how, even though the classification without glide is nontrivial in both 2 and 3D, we can still deduce the complete 3D classification with glide using Proposition \ref{prp:SES_classifications}. It has been proposed that the complete classifications of 2D \cite{2dChiralBosonicSPT_erratum} and 3D \cite{WangChong_3DSPTAII} $G$-protected fermionic SPT phases, for the $G$ specified at the beginning of this subsection, are
\begin{eqnarray}
\SPT^2_f\paren{G} &\isomorphic& \ZZZ_2, \label{2DAII}\\
\SPT^3_f\paren{G} &\isomorphic& \ZZZ_2{\oplus}\ZZZ_2{\oplus}\ZZZ_2,
\end{eqnarray}
respectively, where the $\ZZZ_2$ in 2D is generated by the QSH phase, and the three $\ZZZ_2$'s in 3D are generated by a band insulator and two bosonic SPT phases, respectively. Inserting these into the short exact sequence in Proposition \ref{prp:SES_classifications}, we are led to the abelian group extension problem
\e{0 \rightarrow \ZZZ_2 \rightarrow\; ? \rightarrow \ZZZ_2\oplus\ZZZ_2\oplus\ZZZ_2 \rightarrow 0.\label{groupextension_AII}
}
The solution to this problem is not unique, as is evident from $H^2_{\rm sym}\paren{\ZZZ_2 \oplus \ZZZ_2 \oplus \ZZZ_2; \ZZZ_2} \isomorphic \ZZZ_2 \oplus \ZZZ_2 \oplus \ZZZ_2$. However, we know that the hourglass-fermion phase and its square roots are robust to interactions. We claim that, with this additional piece of information, a unique solution can be found.

Indeed, Corollary \ref{cor:direct_sum_decomposition} says that the unknown term must be a direct sum of $\ZZZ_4$'s and/or $\ZZZ_2$'s. We now show that there is exactly one $\ZZZ_4$ and two $\ZZZ_2$'s. By the remarks in Corollary \ref{cor:direct_sum_decomposition}, the only way for $\ZZZ\times G$-protected SPT phases of type (iv) of Corollary \ref{cor:quad-chotomy} to arise is for there to be a $\ZZZ_4$ summand. Since both square roots of the hourglass-fermion phase are of type (iv), there must be at least one $\ZZZ_4$. On the other hand, each $\ZZZ_4$ contains an SPT phase that is weak with respect to glide, which must arise from an independent non-trivial SPT phase in one lower dimensions through the alternating-layer construction.
Since the classification in one lower dimensions is given by a single $\ZZZ_2$, there can be at most one $\ZZZ_4$ in the second term of Eq.\,(\ref{groupextension_AII}).
As a result, there is exactly one $\ZZZ_4$. This $\ZZZ_4$ maps onto one of the three $\ZZZ_2$'s in the third term of Eq.\,(\ref{groupextension_AII}). To make the map surjective as required by exactness, we need two additional $\ZZZ_2$'s in the second term of Eq.\,(\ref{groupextension_AII}), whose nontrivial elements are of type (iii) of Corollary \ref{cor:quad-chotomy}.

In conclusion, the \emph{complete} classification of 3D $\ZZZ \times G$-protected fermionic SPT phases, for the $G$ specified at the beginning of this subsection, is
\e{\SPT^3_f\paren{\ZZZ {\times} G} \isomorphic \ZZZ_4\oplus \ZZZ_2 \oplus \ZZZ_2,\label{solution_AII}}
where without loss of generality we can identify the nontrivial elements of $\ZZZ_4$ with the hourglass-fermion phase and its square roots. This represents one key result of this work, which goes beyond the known classification of the \emph{subset} of SPT phases that are weak with respect to glide \cite{Lu_sgSPT}:
\begin{equation}
\wSPT^3_f\paren{\ZZZ {\times} G} \isomorphic \ZZZ_2.
\end{equation}
We may anyway verify that this weak classification, together with Eq.\,(\ref{2DAII}), is consistent with Corollary \ref{cor:classification_SPT_weak_wrt_glide}. We remark that while we used such physical terms as ``weak with respect to glide" in our argument above, we could have derived Eq.\,(\ref{solution_AII}) purely mathematically, by combining an explicit classification of abelian group extensions of $\ZZZ_2 \oplus \ZZZ_2 \oplus \ZZZ_2$ by $\ZZZ_2$ with the requirement that the extension contain an element of order 4.

\subsection{Generalizations\label{subsec:computations}}

\subsubsection{Pure translation versus glide reflection}
\label{sec:puretranslation}

A glide reflection is a translation followed by a reflection. In this section, we set out to answer two questions: (a) how has this additional reflection complicated the classification of SPT phases? The symmetry group generated by glide reflection contains a subgroup of pure translations. (b) What would happen if we relaxed glide symmetry to its translational subgroup? 

The first question, (a), can be answered by contrasting Proposition \ref{prp:SES_classifications} with the analogous result for pure translations \cite{Xiong}:
\begin{prp}
\label{prp:strong_weak}
Assume the generalized cohomology hypothesis. Let $\ZZZ$ be generated by a \emph{translation} and $G$ be arbitrary.\footnote{In Ref.\,\cite{Xiong}, $G$ was assumed to preserve spacetime orientation, but the proof of Proposition \ref{prp:SES_classifications} in App.\,C of work \cite{Xiong_Alexandradinata} can be easily adapted to show that this restriction was unnecessary.} There is a \emph{split} short exact sequence,
\begin{equation}
0 \fromto \SPT^{d-1}(G) \fromto \SPT^d\paren{\ZZZ\times G} \fromto  \SPT^d\paren{G} \fromto 0.\label{SES_translation}
\end{equation}
In particular, there is an isomorphism,
\begin{equation}
 \SPT^d\paren{\ZZZ\times G} \isomorphic \SPT^{d-1}(G)  \oplus \SPT^d\paren{G}.\label{strong_weak}
\end{equation}
\end{prp}

\noindent In sequence (\ref{SES_translation}), the first map is given by a layer construction, which is the same as the alternating-layer construction but without the orientation-reversal that occurs every other layer. The second map is given by forgetting translational symmetry. Unlike sequence (\ref{SES}), which has factors of 2 in the first and third terms, the sequence for pure translation does not contain any factors of 2. Starting from $\ZZZ\times G$-protected SPT phases, forgetting the translational symmetry gives us all $G$-protected SPT phases, so all $G$-protected SPT phases are compatible with translational symmetry, even if they do not square to the trivial phase.
Starting from a nontrivial $G$-protected SPT phase in one lower dimensions, applying the layer construction will always give us a nontrivial $\ZZZ\times G$-protected SPT phase, even if the lower-dimensional phase admits a square root. Furthermore, sequence (\ref{SES_translation}) is split. This means its second term is completely determined by the first and third terms, according to Eq.\,(\ref{strong_weak}), and there is no abelian group extension problem to solve. The orientation-reversing nature of glide reflections is responsible for all the complications in sequence (\ref{SES}). Various examples of $\SPT^d(\ZZZ \times G)$, with $\ZZZ$ generated by glide or pure translation, are juxtaposed in Tables \ref{table:bSPT}. 

Regarding the second question, (b), denoting the glide symmetry by $\ZZZ$, we know that if a $\ZZZ\times G$-protected SPT phase becomes trivial when $\ZZZ$ is relaxed to its translational subgroup, then it must become trivial when $\ZZZ$ is forgotten altogether. It turns out that the converse is also true:

\begin{prp}
Assume the generalized cohomology hypothesis. Let $\ZZZ$ be generated by a glide reflection and $G$ be arbitrary. If a $\ZZZ\times G$-protected SPT phase becomes trivial under glide forgetting, then it must already become trivial when $\ZZZ$ is relaxed to its translational subgroup.
\label{prp:glide_to_translation}
\end{prp}

\noindent This can be either proved mathematically from the generalized cohomology hypothesis as in App.\,E of work \cite{Xiong_Alexandradinata}, or argued physically as we proceed to do.  Let us denote the translational subgroup of the glide symmetry $\ZZZ$ by $2\ZZZ$ and introduce the map,
\begin{equation}
\gamma: \SPT^d\paren{\ZZZ \times G} \fromto \SPT^d\paren{2\ZZZ \times G},\label{gamma}
\end{equation}
given by relaxing glide to its translational subgroup. By the same argument as in Figure \ref{fig:image_alpha_subset_kernel_beta}, a system obtained through the alternating-layer construction can always be made trivial while preserving the translational symmetry. This means $\image \alpha \subset \kernel \gamma$.  On the other hand, the fact that a $\ZZZ\times G$-protected SPT phase that becomes trivial when $\ZZZ$ is relaxed to $2\ZZZ$ must become trivial when $\ZZZ$ is forgotten altogether shows that $\kernel \gamma \subset \kernel \beta$. Since we know that a $\ZZZ\times G$-protected SPT phase is weak with respect to glide if and only if it can be obtained through the alternating-layer construction, that is, $\image \alpha= \kernel \beta$, we must have $\image \alpha = \kernel \gamma = \kernel \beta$, whence the desired result follows.

\subsubsection{Bosonic phases with glide\label{subsec:bSPT_with_glide}}

In \s{subsec:applications}, we exemplified how one can utilize Proposition \ref{prp:SES_classifications} to deduce the classification of $\paren{\ZZZ\times G}$-protected fermionic SPT phases (with $\ZZZ$ generated by glide) from proposed classifications of $G$-protected fermionic SPT phases in the literature. The problem of identifying the correct classification was reduced an abelian group extension problem, which required very little technical work in comparison to deriving the classification from scratch. In this section, we apply the same principle to bosonic SPT phases for a variety of symmetries.

The input, $\SPT^d_b(G,\phi)$, of our computations will be given by a generalized cohomology theory $h$ that, in low dimensions, reads
\begin{eqnarray}
h^{\phi + 0}_G\paren{\pt} &=& H^{\phi + 2}_G\paren{\pt; \ZZZ} \label{bSPT_0}, \\
h^{\phi + 1}_G\paren{\pt} &=& H^{\phi + 3}_G\paren{\pt; \ZZZ} \label{bSPT_1}, \\
h^{\phi + 2}_G\paren{\pt} &=& H^{\phi + 4}_G\paren{\pt; \ZZZ} \oplus H^{\phi + 0}_G\paren{\pt; \ZZZ}, \label{bSPT_2} \\
h^{\phi + 3}_G\paren{\pt} &=& H^{\phi + 5}_G\paren{\pt; \ZZZ} \oplus H^{\phi + 1}_G\paren{\pt; \ZZZ}, \label{bSPT_3}
\end{eqnarray}
where $H^{\phi + \bullet}_G\paren{-; \ZZZ}$ denotes the $\phi$-twisted $G$-equivariant ordinary cohomology with coefficient $\ZZZ$. These expressions can be derived using a scheme due to Kitaev from a presumed classification of SPT phases \emph{without} symmetry. More specifically, we assume that bosonic $G$-protected SPT phases for trivial $G$ are classified by
\begin{eqnarray}
\SPT^{0,1,2,3}_b\paren{0} \isomorphic 0, ~0, ~\ZZZ, ~0,\label{bSPT_input}
\end{eqnarray}
in 0, 1, 2, and 3 dimensions, respectively, where the $\ZZZ$ in 2 dimensions is generated by the $E_8$ phase \cite{Kitaev_honeycomb, 2dChiralBosonicSPT, 2dChiralBosonicSPT_erratum, Kitaev_KITP}; this is consistent with the proposal reviewed in Ref.\,\cite{Wen_review_2016}, which goes up to 6 dimensions. As pointed out by Kitaev \cite{Kitaev_Stony_Brook_2011_SRE_1, Kitaev_Stony_Brook_2013_SRE, Kitaev_IPAM}, from the classification without symmetry one can reconstruct a not necessarily unique generalized cohomology theory $h$ which in turn will give one the classification for \emph{arbitrary} symmetries. In the case of Eq.\,(\ref{bSPT_input}), the reconstruction turns out to be unique in low dimensions, giving Eqs.\,(\ref{bSPT_0})-(\ref{bSPT_3}) \cite{Xiong}.

The output of our computations will be $\wSPT^d_b(\ZZZ \times G, \phi)$ for $d\leq 4$ and $\SPT^d_b(\ZZZ \times G, \phi)$ for $d\leq 3$, where $\ZZZ$ is generated by glide. These will be computed from $\SPT^d_b(G,\phi)$ using the correspondence (\ref{GSPT}) and the short exact sequence (\ref{SES}), respectively. We have summarized the results in Table \ref{table:bSPT}. As we can see, in most cases the short exact sequence (\ref{SES}) determines the classification of $d$-dimensional $\paren{\ZZZ\times G}$-protected bosonic SPT phases completely. The results for $\wSPT^3_b(\ZZZ \times G,\phi)$ are in agreement with Ref.\,\cite{Lu_sgSPT}.

{\footnotesize
\begin{longtable}[c]{cccccc}
\caption[Classification of bosonic SPT phases with glide reflection or translational symmetry.]{Classification of bosonic SPT phases with glide reflection or translational symmetry. $\SPT^d_b\paren{G}$ is computed from the proposal (\ref{bSPT_0})-(\ref{bSPT_3}), whence the next three columns are deduced using Eqs.\,(\ref{GSPT}), (\ref{SES}), and (\ref{strong_weak}), respectively. We have suppressed the $\phi$ in $\SPT^d_b(G, \phi)$, etc., for brevity. Abelian group extensions are in general not unique, accounting for the non-uniqueness of some entries. To distinguish glide reflection and translation, both of which generate the group $\ZZZ$, we explicitly substitute the words ``glide" and ``transl." for $\ZZZ$ in $\ZZZ \times G$. $\SPT^d_b(\text{glide} \times G)$ and $\SPT^d_b(\text{transl.} \times G)$ for $d = 4$ are left blank because they require $\SPT^4_b\paren{G}$ as an input, which we did not provide. The superscript $T$ in $\ZZZ_2^T$ indicates time reversal. In the last column, we give physical models corresponding to the generators of underlined summands, where ``$E_8$" stands for the $E_8$ model \cite{Kitaev_honeycomb, 2dChiralBosonicSPT, 2dChiralBosonicSPT_erratum, Kitaev_KITP}, ``BIQH" for bosonic integer quantum Hall \cite{3dBTScVishwanathSenthil, BIQH}, ``3D $E_8$" for the 3D $E_8$ model \cite{3dBTScVishwanathSenthil, 3dBTScWangSenthil, 3dBTScBurnell}, and ``Haldane" for the Haldane chain \cite{AKLT, PhysRevLett.50.1153, Haldane_NLSM, Affleck_Haldane, Haldane_gap}.\label{table:bSPT}} \\
\endfirsthead
\caption[]{(Continued).} \\
\endhead
\endfoot
\endlastfoot
\multicolumn{6}{l}{$(G,\phi)=0$:} \\
\nobreakhline\hline
$d$ & $\SPT^d_b(G)$ & $\wSPT^d_b(\text{glide} \times G)$ & $\SPT^d_b(\text{glide} \times G)$ & $\SPT^d_b(\text{transl.} \times G)$ & Comments \\
\hline
$0$ & $0$                        & $0$          & $0$         & $0$      & \\
$1$ & $0$                        & $0$          & $0$         & $0$      & \\
$2$ & $\underline{\ZZZ}$ & $0$          & $0$         & $\ZZZ$ & $E_8$ \\
$3$ & $0$                        & $\ZZZ_2$ & $\ZZZ_2$ & $\ZZZ$ & \\
$4$ & ~                           & $0$          & ~            & ~         & \\
\nobreakhline\hline \\[-1.5ex]
\multicolumn{5}{l}{$(G,\phi)=U(1)$:} & \\
\nobreakhline\hline
$d$ & $\SPT^d_b(G)$ & $\wSPT^d_b(\text{glide} \times G)$ & $\SPT^d_b(\text{glide} \times G)$ & $\SPT^d_b(\text{transl.} \times G)$ & Comments \\
\nobreakhline
$0$ & $\ZZZ$                                                      & $0$                              & $0$                             & $\ZZZ$                  & \\
$1$ & $0$                                                           & $\ZZZ_2$                     & $\ZZZ_2$                     & $\ZZZ$                  & \\
$2$ & $\underline{\ZZZ} \oplus \underline{\ZZZ}$ & $0$                              & $0$                             & $\ZZZ\oplus\ZZZ$  & BIQH, $E_8$ \\
$3$ & $0$                                                           & $\ZZZ_2\oplus \ZZZ_2$ & $\ZZZ_2\oplus \ZZZ_2$ & $\ZZZ\oplus \ZZZ$ & \\
$4$ & ~                                                              & $0$                              & ~                                & ~                          & \\
\nobreakhline\hline \\[-1.5ex]
\multicolumn{5}{l}{$(G,\phi)=\ZZZ_2^T$:} & \\
\nobreakhline\hline
$d$ & $\SPT^d_b(G)$ & $\wSPT^d_b(\text{glide} \times G)$ & $\SPT^d_b(\text{glide} \times G)$ & $\SPT^d_b(\text{transl.} \times G)$ & Comments \\
\nobreakhline
$0$ & $0$                                                 & $0$                               & $0$                             & $0$                              & \\
$1$ & $\ZZZ_2$                                        & $0$                               & $\ZZZ_2$                     & $\ZZZ_2$                     & \\
$2$ & $0$                                                 & $\ZZZ_2$                      & $\ZZZ_2$                     & $\ZZZ_2$                     & \\
$3$ & $\ZZZ_2 \oplus \underline{\ZZZ_2}$ & $0$                               & $\ZZZ_2\oplus \ZZZ_2$ & $\ZZZ_2\oplus \ZZZ_2$ & 3D $E_8$ \\
$4$ & ~                                                    & $\ZZZ_2 \oplus \ZZZ_2$ & ~                                 & ~                                & \\
\nobreakhline\hline \\[-1.5ex]
\multicolumn{5}{l}{$(G,\phi)=\ZZZ_{N<\infty}$:} & \\
\nobreakhline\hline
$d$ & $\SPT^d_b(G)$ & $\wSPT^d_b(\text{glide} \times G)$ & $\SPT^d_b(\text{glide} \times G)$ & $\SPT^d_b(\text{transl.} \times G)$ & Comments \\
\nobreakhline
$0$ & $\ZZZ_N$                                   & $0$                                             & $\ZZZ_{\gcd(N,2)}$                      & $\ZZZ_N$                   & \\
$1$ & $0$                                            & $\ZZZ_{\gcd(N,2)}$                     & $\ZZZ_{\gcd(N,2)}$                      & $\ZZZ_N$                   & \\
$2$ & $\ZZZ_N\oplus \underline{\ZZZ}$ & $0$                                             & $\ZZZ_{\gcd(N,2)}\oplus \ZZZ_2$ & $\ZZZ_N\oplus \ZZZ$   & $E_8$\\
$3$ & $0$                                            & $\ZZZ_{\gcd(N,2)}\oplus \ZZZ_2$ & $\ZZZ_{\gcd(N,2)} \oplus \ZZZ_2$ & $\ZZZ_N \oplus \ZZZ$ & \\
$4$ & ~                                               & $0$                                             & ~                                                 & ~                               & \\
\nobreakhline\hline \\[-1.5ex]
\multicolumn{5}{l}{$(G,\phi)=\ZZZ_2 \times \ZZZ_2$:} & \\
\nobreakhline\hline
$d$ & $\SPT^d_b(G)$ & $\wSPT^d_b(\text{glide} \times G)$ & $\SPT^d_b(\text{glide} \times G)$ & $\SPT^d_b(\text{transl.} \times G)$ & Comments \\
\nobreakhline
$0$ & $\ZZZ_2^2$                                   & $0$                                   & $\ZZZ_2^2$                                                                                              & $\ZZZ_2^2$ & \\
$1$ & $\underline{\ZZZ_2}$                      & $\ZZZ_2^2$                      & $\ZZZ_2 \oplus \ZZZ_4$ or $\ZZZ_2^3$                                                      & $\ZZZ_2^2\oplus\ZZZ_2$ & Haldane \\
$2$ & $\ZZZ_2^3\oplus \underline{\ZZZ}$ & $\ZZZ_2$                          & $\ZZZ_2^2 \oplus \ZZZ_4$ or $\ZZZ_2^4$                                                  & $\ZZZ_2\oplus\ZZZ_2^3\oplus \ZZZ$ & $E_8$ \\
$3$ & $\ZZZ_2^2$                                   & $\ZZZ_2^3 \oplus \ZZZ_2$ & $\ZZZ_2^2 \oplus \ZZZ_4^2$ or $\ZZZ_2^4 \oplus \ZZZ_4$ or $\ZZZ_2^6$ & $\ZZZ_2^3\oplus\ZZZ\oplus \ZZZ_2^2$ & \\
$4$ & ~                                                   & $\ZZZ_2^2$                      & ~                                                                                                              & ~ & \\
\nobreakhline\hline \\[-1.5ex]
\multicolumn{5}{l}{$(G,\phi)=SO(3)$:} & \\
\nobreakhline\hline
$d$ & $\SPT^d_b(G)$ & $\wSPT^d_b(\text{glide} \times G)$ & $\SPT^d_b(\text{glide} \times G)$ & $\SPT^d_b(\text{transl.} \times G)$ & Comments \\
\nobreakhline
$0$ & $0$                                     & $0$                              & $0$                              & $0$                                           & \\
$1$ & $\underline{\ZZZ_2}$           & $0$                              & $\ZZZ_2$                     & $\ZZZ_2$                                   & Haldane \\
$2$ & $\ZZZ\oplus \underline\ZZZ$ & $\ZZZ_2$                     & $\ZZZ_2$                     & $\ZZZ_2\oplus \ZZZ \oplus \ZZZ$ & $E_8$ \\
$3$ & $0$                                     & $\ZZZ_2\oplus \ZZZ_2$ & $\ZZZ_2\oplus \ZZZ_2$ & $\ZZZ \oplus \ZZZ$                     & \\
$4$ & ~                                        & $0$                              & ~                                 & ~                                              & \\
\nobreakhline\hline
\end{longtable}
}

\subsubsection{Spatiotemporal glide symmetry}\label{subsec:temporalglide}

In this work we have focused on spatial glide symmetry, but with the right definitions we expect the generalized cohomology hypothesis (hence also Proposition \ref{prp:SES_classifications}) to also work for generalized, spatiotemporal glide symmetries, as long as they commute with the symmetry $G$.
An example of spatiotemporal symmetries would be a translation followed by a time reversal, which has been considered by the authors of Ref.\,\cite{Teo_AF_TRS} under the name ``antiferromagnetic time-reversal symmetry" (AFTRS). 2D and 3D topological superconductors in Atland-Zirnbauer class D are classified by $\ZZZ$ and $0$, respectively, where the $\ZZZ$ in 2D is generated by spinless $p+ip$ superconductors. By putting a spinless $p+ip$ superconductor on all planes of constant $x\in \ZZZ$ and its \emph{time-reversed} version (time reversal squares to the identity in this case due to spinlessness) on all $x\in \ZZZ + 1/2$ planes without coupling, one creates a 3D system that respects the fermionc-parity $\ZZZ_2^f$ and AFTRS. Since spinless $p+ip$ superconductors are robust to interactions and admit no square root, our Proposition \ref{prp:SES_classifications} implies that this 3D system must represent a nontrivial $\text{AFTRS}\times\ZZZ_2^f$-protected fermionic SPT phase. Indeed, this was argued in Ref.\,\cite{Teo_AF_TRS} to be the case through the construction of surface topological orders and subsequent confinement of the classical extrinsic defects among the anyons. Since it is believed that {spinless} $p+ip$ superconductors generate the \emph{complete} classification of 2D fermionic SPT phases {with only fermion-parity symmetry}, this gives a putative $\ZZZ_2$ classification of 3D fermionic {$\text{AFTRS}\times\ZZZ_2^f$-protected} SPT phases that are weak with respect to glide. The proposal that 3D fermionic SPT phases {with only fermion-parity symmetry} have a trivial classification \cite{Kapustin_Fermion} would further imply that this $\ZZZ_2$ actually classifies \emph{all} 3D {$\text{AFTRS}\times\ZZZ_2^f$-protected} fermionic SPT phases.

\section{Classification and construction of 3D bosonic crystalline SPT phases beyond group cohomology}
\label{sec:3D_beyond_group_cohomology}

As the interacting generalization of topological insulators and superconductors \cite{kane2005A, kane2005B, Qi_Hughes_Zhang, Schnyder, Kitaev_TI, Hasan_Kane, Qi_Zhang, Chiu_Teo_Schnyder_Ryu}, symmetry protected topological (SPT) phases \cite{SPT_origin} have garnered considerable interest in the past decade \cite{pollmann2010, Wen_1d, Cirac, Wen_sgSPT_1d, Wen_2d, Wen_Boson, Wen_Fermion, levin2012, 2dChiralBosonicSPT, Wen_Boson, Kapustin_Boson, Kapustin_Fermion, Kapustin_equivariant, Freed_SRE_iTQFT, Freed_ReflectionPositivity, Kitaev_Stony_Brook_2011_SRE_1, Kitaev_Stony_Brook_2011_SRE_2, Kitaev_Stony_Brook_2013_SRE, Kitaev_IPAM, Husain, 2dChiralBosonicSPT, 2dChiralBosonicSPT_erratum, 3dBTScVishwanathSenthil, 3dBTScWangSenthil, 3dBTScBurnell, WangChong_3DSPTAII, 3dFTScWangSenthil_2, 3dFTScWangSenthil_2_erratum, 2dFermionGExtension, Lan_Kong_Wen_1, Lan_Kong_Wen_2, SOinfty, Else_edge, Jiang_sgSPT, ThorngrenElse, Wang_Levin_invariants, Wang_intrinsic_fermionic, Huang_dimensional_reduction, Lu_sgSPT}. The early studies of SPT phases focused on phases with \emph{internal} symmetries (\emph{i.e.}, symmetries that do not change the position of local degrees of freedom, such as Ising symmetry, $U(1)$ symmetry, and time reversal symmetry). Now it is slowly being recognized \cite{Kitaev_Stony_Brook_2011_SRE_1,Kitaev_Stony_Brook_2013_SRE, Kitaev_IPAM, Xiong, Gaiotto_Johnson-Freyd} that the classification of internal SPT phases naturally satisfies certain axioms which happen to define a well-known structure in mathematics called generalized cohomology \cite{Hatcher,DavisKirk,Adams1,Adams2}. In particular, different existing proposals for the classification of internal SPT phases are simply different examples of generalized cohomology theories.

Ref.\,\cite{Xiong} distilled the above observations regarding the general structure of the classification of SPT phases into a ``generalized cohomology hypothesis.'' It maintained that (a) there exists a generalized cohomology theory $h$ that correctly classifies internal SPT phases in all dimensions for all symmetry groups, and that (b) even though we may not know exactly what $h$ is, meaningful physical results can still be derived from the fact that $h$ is a generalized cohomology theory alone. Indeed, it can be shown, on the basis of the generalized cohomology hypothesis, that three-dimensional bosonic SPT phases with internal symmetry $G$ are classified by $H^{\phi + 5}_G \paren{\pt; \ZZZ}\oplus H^{\phi + 1}_G \paren{\pt; \ZZZ}$ \cite{Xiong, Gaiotto_Johnson-Freyd},\footnote{The direct sum of two abelian groups is the same as their direct product, but the direct sum notation $\oplus$ is more common for abelian groups in the mathematical literature.} where $H^{\phi + n}_G \paren{\pt; \ZZZ}$ denotes the $n$th twisted group cohomology group of $G$ with $\ZZZ$ coefficients, and $\phi$ emphasizes that
$g\in G$ acts non-trivially on coefficients by sending them to their inverses if $g$ reverses spacetime orientation and acts trivial otherwise. The first summand, $H^{\phi + 5}_G \paren{\pt; \ZZZ}$,\footnote{For $n=0,1,2,\cdots$ and a compact group $G$, $H^{\phi + n + 1}_G \paren{\pt; \ZZZ}$ is isomorphic to the ``Borel group cohmology'' $H_{\text{Borel}, \phi}^{n}(G;U(1))$ considered in Ref.~\cite{Wen_Boson}, which is simply $H^{\phi + n}_G \paren{\pt; U(1)}$ if $G$ is finite.} corresponds to the ``group cohomology proposal'' for the classification of SPT phases \cite{Wen_Boson}. The second summand, $H^{\phi + 1}_G \paren{\pt; \ZZZ}$, corresponds to phases beyond the group cohomology proposal, and are precisely the phases constructed in Ref.\,\cite{decorated_domain_walls} using decorated domain walls. Specifically, in Ref.\,\cite{decorated_domain_walls}, the domain walls were decorated with multiples of the $E_8$ state, which is a 2D bosonic state with quantized thermal Hall coefficient \cite{Kitaev_honeycomb,2dChiralBosonicSPT,2dChiralBosonicSPT_erratum,Kitaev_KITP}.

However, physical systems tend to crystallize. What is the classification of SPT phases if $G$ is a \emph{space-group} symmetry rather than internal symmetry? In the fermionic case, one can incorporate crystalline symmetries by imposing point-group actions on the Brillouin zone before activating interactions \cite{TCI_Fu}. This is obviously not applicable to bosonic systems due to the lack of a Brillouin zone. As a get-around, Refs. \cite{Hermele_torsor, Huang_dimensional_reduction} proposed to build bosonic crystalline SPT (cSPT) phases by focusing on high-symmetry points in the real space rather than momentum space. Concretely, on every high-symmetry line, plane, etc.\,of a $d$-dimensional space with space-group action by $G$, one can put an SPT phase of the appropriate dimensions with an internal symmetry equal to the stabilizer subgroup  of (any point on) the line, plane, etc. In particular, it was shown, for every element of $H^{\phi + 5}_G \paren{\pt; \ZZZ}$, that there is a 3D bosonic crystalline SPT phase with space group symmetry $G$ that one can construct. Curiously, the same mathematical object, $H^{\phi + 5}_G \paren{\pt; \ZZZ}$, is also the group cohomology proposal for the classification of 3D bosonic SPT phases with \emph{internal} symmetry $G$. A heuristic insight into this apparent correspondence between crystalline and internal SPT phases was provided in Ref.\,\cite{ThorngrenElse}, which drew an analogy between internal gauge fields and a certain notation of crystalline gauge fields. The correspondence was referred to therein as the ``crystalline equivalence principle.''

Just like for internal symmetries, the group cohomology proposal $H^{\phi + 5}_G \paren{\pt; \ZZZ}$ does not give the complete classification for crystalline symmetries either. In fact, in the block-state construction of Ref.\,\cite{Huang_dimensional_reduction}, $E_{8}$ state was excluded from being used as a building block for simplicity. Appealing to the crystalline equivalence principle, one might guess that the complete classification of 3D bosonic crystalline SPT phases with space group symmetry $G$ would be $H^{\phi + 5}_G \paren{\pt; \ZZZ}\oplus H^{\phi + 1}_G \paren{\pt; \ZZZ}$, since that is the classification when $G$ is internal. The recent work \cite{Shiozaki2018} gives us added confidence in this conjecture. In that work, the authors justified the extension of generalized cohomology theory to crystalline symmetries by systematically interpreting terms of a spectral sequence of a generalized cohomology theory as building blocks of crystalline phases. Related discussions along this direction can also be found in Ref.~\cite{Else2018, Song2018}. (To be precise, Ref.~\cite{Shiozaki2018} focused on generalized \emph{homology} theories, but it is highly plausible that that is equivalent to a generalized cohomology formulation via a Poincar\'e duality.)

In this work, we will conduct a thorough investigation into 3D bosonic cSPT phases protected by any space group symmetry $G$, dubbed $G$-SPT phases for short, and establish that their classification is indeed given by 
\begin{equation}
H^{\phi + 5}_G \paren{\pt; \ZZZ}\oplus H^{\phi + 1}_G \paren{\pt; \ZZZ}. \label{eq:classification}
\end{equation}
We will see that distinct embedding copies of the $E_8$ state in the Euclidean space $\mathbb E^3$ produce $G$-SPT phases with different $H^{\phi + 1}_G \paren{\pt; \ZZZ}$ labels. To obtain the classification of $G$-SPT phases and to understand its physical meaning, we will invoke three techniques: dimensional reduction, surface topological order, and explicit cellular construction. The combination of these techniques will establish that
\begin{enumerate}[(a)]
\item every 3D bosonic $G$-SPT phase can be mapped, via a homomorphism, to an element of $H^{\phi + 1}_G \paren{\pt; \ZZZ}$,
\item every element of $H^{\phi + 1}_G \paren{\pt; \ZZZ}$ can be mapped, via a homomorphism, to a 3D bosonic $G$-SPT phase,
\item the first map is a left inverse of the second, and
\item a 3D bosonic $G$-SPT phase comes from $H^{\phi + 5}_G \paren{\pt; \ZZZ}$ if and only if it maps to the trivial element of $H^{\phi + 1}_G \paren{\pt; \ZZZ}$.
\end{enumerate}
Therefore, with minor caveats such as the correctness of $H^{\phi + 5}_G \paren{\pt; \ZZZ}$ in classifying non-$E_8$-based phases and certain assumptions about the correlation length of short-range entangled (SRE) states, we will be providing an essentially rigorous proof that 3D bosonic cSPT phases are classified by $H^{\phi + 5}_G \paren{\pt; \ZZZ}\oplus H^{\phi + 1}_G \paren{\pt; \ZZZ}$. In addition, we will show that the value of $H^{\phi + 1}_G \paren{\pt; \ZZZ}$ can be easily read off from the international symbol for $G$, following this formula:
\begin{eqnarray}
H^{\phi + 1}_G \paren{\pt; \ZZZ}=\begin{cases}
\ZZZ^{k}, & \mbox{\ensuremath{G} preserves orientation},\\
\ZZZ^{k}\times\ZZZ_{2}, & \mbox{otherwise},
\end{cases}\label{formula!}
\end{eqnarray}
where $k=0$ if there is more than one symmetry direction listed in the international symbol, $k=3$ if the international symbol has one symmetry direction listed and it is $1$ or $\overline{1}$, and $k=1$ if the international symbol has one symmetry direction listed and it is not $1$ or $\overline{1}$.

Naturally, we expect that $H^{\phi + 5}_G \paren{\pt; \ZZZ}\oplus H^{\phi + 1}_G \paren{\pt; \ZZZ}$ also works for more general crystalline symmetries where $H^{\phi + 1}_G \paren{\pt; \ZZZ}$ can be easily computed in a similar way. In particular, for magnetic space groups $G$ (which include space groups as a subclass called type I), we still have $H^{\phi + 1}_G \paren{\pt; \ZZZ}=\ZZZ^{k}\times \ZZZ_{2}^{\ell}$ with $k\in\mathbb{Z}$ and $\ell\in\{0,1\}$ only depending on the associated magnetic point group. For magnetic group $G$ of type II or type IV, $H^{\phi + 1}_G \paren{\pt; \ZZZ}$ is simply $\mathbb{Z}_{2}$, while $k$ and $\ell$ can be read off from the associated magnetic point group type. See App.\,A of work \cite{SongXiongHuang} for detail.

The rest of Sec.\,\ref{sec:3D_beyond_group_cohomology} is organized as follows. In Sec.\,\ref{subsec:prediction}, we will explain how $H^{\phi + 5}_G \paren{\pt; \ZZZ}\oplus H^{\phi + 1}_G \paren{\pt; \ZZZ}$ arises from the generalized cohomology hypothesis. In Sec.~\ref{subsec:examples}, we will present examples of cSPT phases described by $H^{\phi + 1}_G \paren{\pt; \ZZZ}$ for select space groups. In Sec.~\ref{subsec:reduction_and_construction}, we will verify the general classification using physical arguments.

We will list the classifications of 3D bosonic cSPT phases for all 230 space groups in App.\,\ref{app:list_of_classifications}.

\subsection{General classification\label{subsec:prediction}}

To pave the way for us to generally construct and completely classify 3D $E_{8}$-based cSPT phases, let us make an prediction for what the classification of these phases might be using the generalized cohomology hypothesis \cite{Xiong, Xiong_Alexandradinata}. The generalized cohomology hypothesis was based on Kitaev's proposal \cite{Kitaev_Stony_Brook_2011_SRE_1,Kitaev_Stony_Brook_2013_SRE,Kitaev_IPAM} that the classification of SPT phases must carry the structure of generalized cohomology theories \cite{Hatcher,DavisKirk,Adams1,Adams2}. This proposal was further developed in Refs.\cite{Xiong, Xiong_Alexandradinata, Gaiotto_Johnson-Freyd}. The key idea here is that the classification of SPT phases can be encoded by a sequence $F_{\bullet}=\{F_{d}\}$ of topological spaces, 
\begin{equation}
F_{0},F_{1},F_{2},F_{3},F_{4},\ldots
\end{equation}
where $F_{d}$ is the space made up of all $d$-dimensional short-range entangled (SRE) states. It can be argued \cite{Kitaev_Stony_Brook_2011_SRE_1, Kitaev_Stony_Brook_2013_SRE, Kitaev_IPAM, Xiong, Gaiotto_Johnson-Freyd} that the spaces $F_{d}$ are related to each other: the $d$-th space is homotopy equivalent to the loop space \cite{Hatcher} of the $(d+1)$st space, 
\begin{equation}
F_{d}\homotopic\Omega F_{d+1}.\label{loop_relation}
\end{equation}
Physically, this says that there is a correspondence between $d$-dimensional SRE states and one-parameter families of $(d+1)$-dimensional SRE states.

To state how the sequence $F_{\bullet}$ determines the classification of SPT phases, let us introduce a homomorphism, 
\begin{equation}
\phi:G\fromto\braces{\pm1},\label{phi_hom}
\end{equation}
that tracks which elements of the symmetry group $G$ preserve the orientation of spacetime (mapped to $+1$) and which elements do not (mapped to $-1$). $G$-SPT phases (\emph{i.e.}, topological phase protected by symmetry $G$) in $d$ dimensions are classified by the Abelian group of $G$-SPT orders
\begin{equation}
\SPT^d\paren{G, \phi},
\end{equation}
whose addition operation is defined by stacking. It conjectured that the group structure of $\SPT^d\paren{G, \phi}$ can be obtained by computing the mathematical object
\begin{equation}
h^{\phi + d}_G\paren{\pt}\coloneq\brackets{EG,\Omega F_{d+1}}_{G}.\label{hypothesis_expression}
\end{equation}
Here, $EG$ is the total space of the universal principal $G$-bundle \cite{AdemMilgram}, and $\brackets{EG,\Omega F_{d+1}}_{G}$ denotes the set of deformation classes of $G$-equivariant maps from the space $EG$ to the space $\Omega F_{d+1}$. Explicitly, the generalized cohomology hypothesis states that we have an isomorphism
\begin{equation}
\SPT^d\paren{G, \phi} \isomorphic h^{\phi + d}_G\paren{\pt}. \label{iso_2}
\end{equation}

To compute (\ref{hypothesis_expression}), we note by definition that
the 0th homotopy group of $F_{d}$, 
\begin{equation}
\pi_{0}(F_{d}),
\end{equation}
(\emph{i.e.}, the set of connected components of $F_{d}$) classifies $d$-dimensional invertible topological orders (\emph{i.e.}, SPT phases without symmetry). In 0, 1, 2, and 3 dimensions, the classification of invertible topological orders is believed to be \cite{Kitaev_Stony_Brook_2011_SRE_1,Kitaev_Stony_Brook_2013_SRE,Kitaev_IPAM, Xiong, Gaiotto_Johnson-Freyd}
\begin{equation}
\pi_{0}(F_{0})=0,~\pi_{0}(F_{1})=0,~\pi_{0}(F_{2})=\ZZZ,~\pi_{0}(F_{3})=0,\label{classification_iTO}
\end{equation}
respectively, where the $\ZZZ$ in 2 dimensions is generated by the $E_{8}$ phase \cite{Kitaev_honeycomb,2dChiralBosonicSPT,2dChiralBosonicSPT_erratum,Kitaev_KITP}. Next, we note that 0-dimensional SRE states are nothing but rays in Hilbert spaces. These rays form the infinite-dimensional complex projective space \cite{Hatcher}, so 
\begin{equation}
F_{0}=\CCC P^{\infty}.\label{F0}
\end{equation}
Finally, we note, as a consequence of Eq.\,(\ref{loop_relation}), that the $(k+1)$st homotopy group of $F_{d+1}$ is the same as the $k$-th homotopy group of $F_{d}$ for all $k$ and $d$:
\begin{equation}
\pi_{k}(F_{d})\isomorphic\pi_{k+1}(F_{d+1}).
\end{equation}
This allows us to determine all homotopy groups of $F_{1},F_{2},F_{3}$ from the classification of invertible topological orders (\ref{classification_iTO}) and the known space of 0D SRE states (\ref{F0}). The results are shown in Table \ref{table:homotopy_groups}.


It turns out the homotopy groups in Table \ref{table:homotopy_groups} completely determine the space $F_{1}$, $F_{2}$, $F_{3}$ themselves \cite{Xiong}: 
\begin{eqnarray}
F_{1} & = & K(\ZZZ,3),\label{F1}\\
F_{2} & = & K(\ZZZ,4)\times\ZZZ,\label{F2}\\
F_{3} & = & K(\ZZZ,5)\times\mathbf{S}^{1},\label{F3}
\end{eqnarray}
where $K(\ZZZ,n)$ is the $n$th Eilenberg-MacLane space of $\ZZZ$ {[}defined by the property $\pi_{k}\paren{K(\ZZZ,n)}=\ZZZ$ for $k=n$ and $0$ otherwise{]} \cite{Hatcher}. Plugging Eqs.\,(\ref{F0})(\ref{F1})(\ref{F2})(\ref{F3}) into Eqs.\,(\ref{hypothesis_expression})(\ref{iso_2}), we arrive at the prediction
\begin{eqnarray}
\SPT^0\paren{G, \phi} & \isomorphic & H^{\phi + 2}_G \paren{\pt; \ZZZ},\\
\SPT^1\paren{G, \phi} & \isomorphic & H^{\phi + 3}_G \paren{\pt; \ZZZ},\\
\SPT^2\paren{G, \phi} & \isomorphic & H^{\phi + 4}_G \paren{\pt; \ZZZ}\oplus H^{\phi + 0}_G \paren{\pt; \ZZZ}, \label{2D_prediction}\\ 
\SPT^3\paren{G, \phi} & \isomorphic & H^{\phi + 5}_G \paren{\pt; \ZZZ}\oplus H^{\phi + 1}_G \paren{\pt; \ZZZ},\label{3D_prediction}
\end{eqnarray}
where $H^{\phi + n}_G \paren{\pt; \ZZZ}$ denotes the $n$th twisted group cohomology group of $G$ with coefficient $\ZZZ$ and twist $\phi$ \cite{AdemMilgram}. For finite or compact groups, we have $H^{\phi + 5}_G \paren{\pt; \ZZZ}\isomorphic H_{{\rm Borel}, \phi}^{4}(G;U(1))$; we identify this as the contribution from the group cohomology proposal \cite{Wen_Boson} to the 3D classification (\ref{3D_prediction}). The existence of $H^{\phi + 1}_G \paren{\pt; \ZZZ}$ in Eq.\,(\ref{3D_prediction}), on the other hand, can be traced back to the fact that $\pi_{0}(F_{2})=\ZZZ$ in Eq.\,(\ref{classification_iTO}); we identify it as the contribution of $E_{8}$-based phases \cite{Kitaev_honeycomb, 2dChiralBosonicSPT,2dChiralBosonicSPT_erratum, Kitaev_KITP} to the 3D classification. Therefore, we predict that 3D bosonic cSPT phases built from $E_{8}$ states are classified by
\begin{equation}
H^{\phi + 1}_G \paren{\pt; \ZZZ}
\end{equation}
(up to some non-$E_8$-based phases), where $G$ is the space group and $\phi:G\fromto\braces{\pm1}$ keeps track of which elements of $G$ preserve/reverse the orientation [see Eq.\,(\ref{phi_hom})].

$H^{\phi + 1}_G \paren{\pt; \ZZZ}$ can be intuitively thought of as the set of ``representations of $G$ in the integers.'' Explicitly, an element of $H^{\phi + 1}_G \paren{\pt; \ZZZ}$ is represented by a map (called a group 1-cocycle) 
\begin{equation}
\nu^{1}:G\fromto\ZZZ\label{cochain}
\end{equation}
satisfying the cocycle condition
\begin{equation}
\nu^{1}(g_{1}g_{2})=\nu^{1}(g_{1})+\phi(g_{1})\nu^{1}(g_{2})\label{cocycle_condition}
\end{equation}
for all $g_{1},g_{2}\in G$. Suppose $\phi$ was trivial for the moment (mapping all elements to $+1$). Then we would have $\nu^{1}(g_{1}g_{2})=\nu^{1}(g_{1})+\nu^{1}(g_{2})$, which is precisely the axiom $\rho(g_{1}g_{2})=\rho(g_{1})\rho(g_{2})$ for a representation $\rho$ of $G$, written additively as opposed to multiplicatively. Now if we allowed $\rho$ to be antilinear, then we would have the modified condition $\rho(g_{1}g_{2})=\rho(g_{1})\overline{\rho(g_{2})}$ (overline denoting complex conjugation) for all $g_{1}$'s that are represented antilinearly. The analogue of this for $\nu^{1}$ is precisely (\ref{cocycle_condition}). The cohomology group $H^{\phi + 1}_G \paren{\pt; \ZZZ}$ itself is defined to be the quotient of group 1-cocycles {[}as in Eqs.\,(\ref{cochain})(\ref{cocycle_condition}){]} by what is called group 1-coboundaries. In App.\,A of work \cite{SongXiongHuang}, we make the definition of group 1-coboundary explicit and show how one can easily read off $H^{\phi + 1}_G \paren{\pt; \ZZZ}$ from the international symbol of $G$. The list of classifications of 3D bosonic cSPT phases for all 230 space groups in App.\,A of work \cite{SongXiongHuang} is reproduced in App.\,\ref{app:list_of_classifications} of this thesis.

\subsection{Examples: \texorpdfstring{$P1$}{P1}, \texorpdfstring{$Pm$}{Pm}, \texorpdfstring{$Pmm2$}{Pmm2}, \texorpdfstring{$Pmmm$}{Pmmm}, \texorpdfstring{$P\bar 1$}{P bar 1}, and \texorpdfstring{$Pc$}{Pc}\label{subsec:examples}}

\subsubsection{Space group No.\,1 or $P1$}
\label{subsubsec:P1}

The space group $P1$ contains only translation symmetries. In proper coordinates, $P1$ is generated by 
\begin{align}
t_{x} & :\left(x,y,z\right)\mapsto\left(x+1,y,z\right),\label{eq:tx}\\
t_{y} & :\left(x,y,z\right)\mapsto\left(x,y+1,z\right),\label{eq:ty}\\
t_{z} & :\left(x,y,z\right)\mapsto\left(x,y,z+1\right).\label{eq:tz}
\end{align}
As an abstract group, it is isomorphic to $\mathbb{Z}\times\mathbb{Z}\times\mathbb{Z}\equiv\mathbb{Z}^{3}$.
Thus, we have 
\begin{align}
H^{\phi + 5}_{P1} \paren{\pt; \mathbb{Z}} & =0,\\
H^{\phi + 1}_{P1} \paren{\pt; \ZZZ} & =\mathbb{Z}^{3}.
\end{align}
The latter can be identified with layered $E_{8}$ states.

For instance, we can put a copy of $E_{8}$ state on each of the planes $y=\cdots,-2,-1,0,1,2,\cdots$ as in Figure \ref{fig:layerE8_y}. Since each $E_{8}$ state can be realized with translation symmetries $t_{x}$ and $t_{z}$ respected, the layered $E_{8}$ states can be made to respect all three translations and thus realize a 3D cSPT phase with $P1$ symmetry.

\begin{figure}
\centering
\includegraphics[width=0.5\columnwidth]{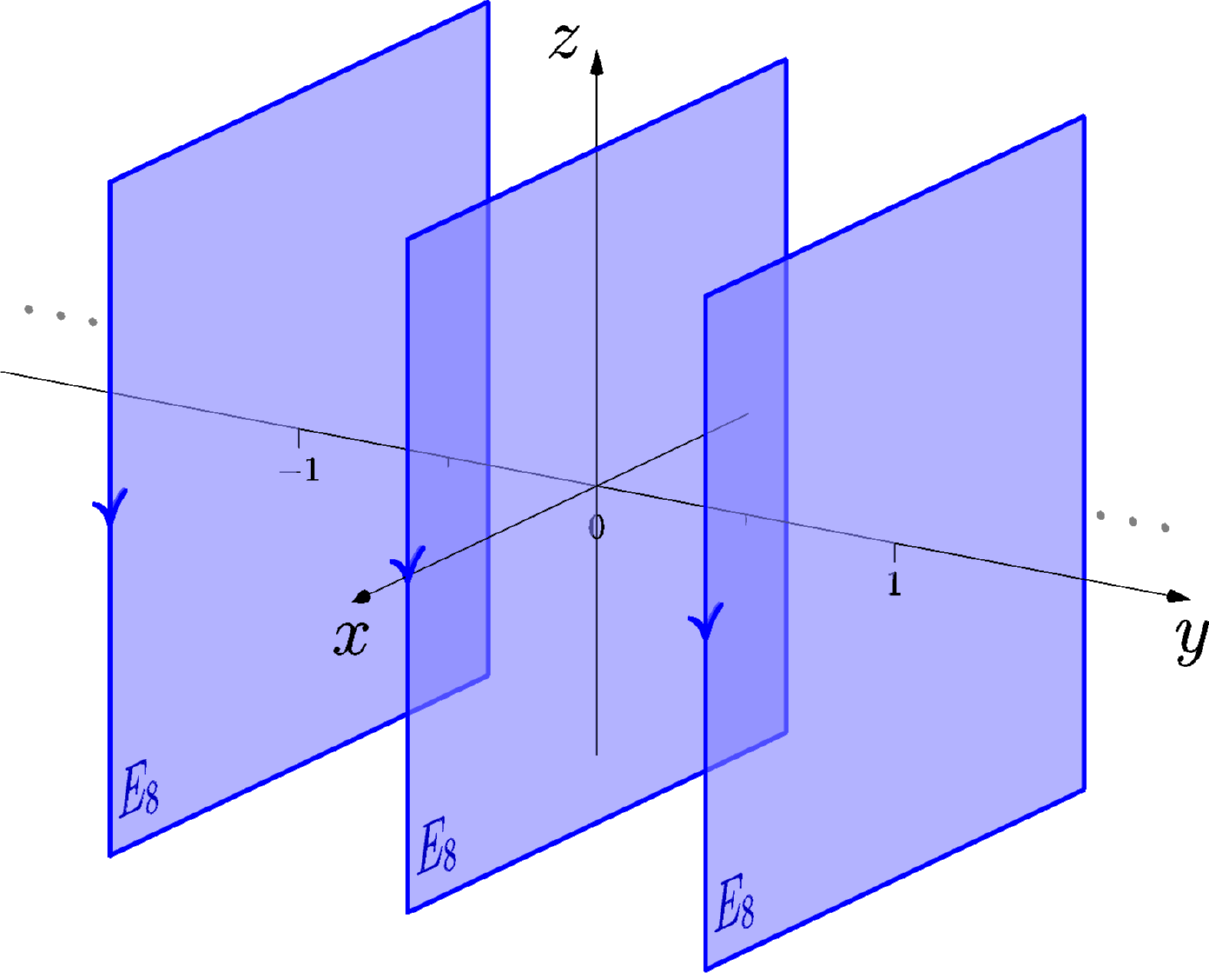}
\caption[Layered $E_8$ state.]{Layered $E_{8}$ state in the $t_{y}$ direction.}
\label{fig:layerE8_y} 
\end{figure}

To characterize this cSPT phase, we take a periodic boundary condition in $t_{y}$ direction requiring $t_{y}^{L_{y}}=1$. Such a procedure is called a compactification in the $t_{y}$ direction, which is well-defined for any integer $L_{y}\gg1$ in general for a gapped system with a translation symmetry $t_{y}$. The resulting model has a finite thickness in the $t_{y}$ direction and thus can be viewed as a 2D system extending in the $t_{x}$, $t_{z}$ directions. Further, neglecting the translation symmetries $t_{x}$ and $t_{z}$, we take an open boundary condition of the 2D system. Then its edge supports $8L_{y}$ co-propagating chiral boson modes (chiral central charge $c_{-}^{y}=8L_{y}$). The resulting quantized thermal Hall effect is proportional to $L_{y}$ and shows the nontriviality of the layered $E_{8}$ states as a cSPT phase with the $P1$ symmetry.

Actually, even without translation symmetries, we cannot trivialize such a system into a tensor product state with local unitary gates with a universal finite depth homogeneously in space. However, such a state is trivial in a weaker sense: if the system has correlation length less than $\xi>0$, then any ball region of size much larger than $\xi$ can be trivialized with correlation length kept smaller than $\xi$. We call such a state \emph{weakly trivial}.  In this work, the notion of cSPT phase includes all gapped quantum phases without emergent nontrivial quasiparticles and particularly these weakly trivial states.

A bosonic SPT system has $c_{-}^{y}=\gamma_{y}L_{y}$ with $\gamma_{y}$ a multiple of $8$ (\emph{i.e.}, $\gamma_{y}\in8\mathbb{Z}$) in general. To see $\gamma_{y}\in8\mathbb{Z}$, we notice that the net number of chiral boson modes along the interface between compactified systems of thicknesses $L_{y}$ and $L_{y}+1$ is $\gamma_{y}$. Then the absence of anyons in both sides implies $\gamma_{y}\in8\mathbb{Z}$ \cite{Kitaev_honeycomb}. It is obvious that $\gamma_{y}$'s are added during a stacking operation.

Clearly, this cSPT phase of layered $E_{8}$ states is invertible; its inverse is made of layered $\overline{E_{8}}$ states, where $\overline{E_{8}}$ denotes the chiral twin of $E_{8}$. The edge modes of $\overline{E_{8}}$ propagate in the direction opposite to those of $E_{8}$ and hence we have $\gamma_{y}=-8$ for the cSPT phase of the layered $\overline{E_{8}}$ states. Via the stacking operation, all possible $\gamma_{y}\in8\mathbb{Z}$ is generated by the cSPT phase of layered $E_{8}$ states and its inverse.

Analogously, we can define $\gamma_{x}$ (resp. $\gamma_{z}$) by compactifying a system in the $t_{x}$ (resp. $t_{z}$) direction. Thus, the cSPT phases with $P1$ symmetry are classified by $\frac{1}{8}\left(\gamma_{x},\gamma_{y},\gamma_{z}\right)\in\mathbb{Z}^{3}=H^{\phi + 1}_{P1} \paren{\pt; \ZZZ}$. Via the stacking operation, they can be generated by the three cSPT phases made of layered $E_{8}$ states in the $t_{x}$, $t_{y}$ and $t_{y}$ directions respectively and their inverses.

\subsubsection{Space group No.\,6 or $Pm$\label{subsubsec:Pm}}

The space group $Pm$ is generated by $t_{x}$, $t_{y}$, $t_{y}$ and a reflection 
\begin{equation}
m_{y}:\left(x,y,z\right)\mapsto\left(x,-y,z\right).
\end{equation}
Thus, the mirror planes are $y=\cdots,-1,-\frac{1}{2},0,\frac{1}{2},1,\cdots$ (\emph{i.e.}, integer $y$ planes and half-integer $y$ planes). For $Pm$, the group structure of $G$ and its subgroup $G_{0}$ of orientation preserving symmetries are
\begin{align}
G & =\mathbb{Z}\times\mathbb{Z}\times\left(\mathbb{Z}\rtimes\mathbb{Z}_{2}^{\phi}\right),\\
G_{0} & =\mathbb{Z}\times\mathbb{Z}\times\mathbb{Z},
\end{align}
where $\mathbb{Z}_{2}^{\phi}$ is the group generated by $m_{y}$ with the superscript $\phi$ emphasizing $\phi\left(m_{y}\right)=-1$.

Theorem~\ref{thm:H1_formula} in App.\,\ref{app:list_of_classifications} computes the second part of $H^{\phi + 1}_G \paren{\pt; \ZZZ}$ shown in Eq.~(\ref{3D_prediction}); the result is
\begin{gather}
H^{\phi + 1}_G \paren{\pt; \ZZZ}=\mathbb{Z}\times\mathbb{Z}_{2}.\label{eq:Pm_h1}
\end{gather}
The $\mathbb{Z}$ factor specifies $\frac{1}{8}\gamma_{y}$ of the cSPT phases compatible with $Pm$ symmetry; the reflection symmetry $m_{y}$ requires that $\gamma_{x}=\gamma_{z}=0$. Explicitly, putting an $E_{8}$ state on each integer $y$ plane produces a phase with $\gamma_{y}=1$.

\begin{figure}
\centering
\includegraphics[width=0.5\columnwidth]{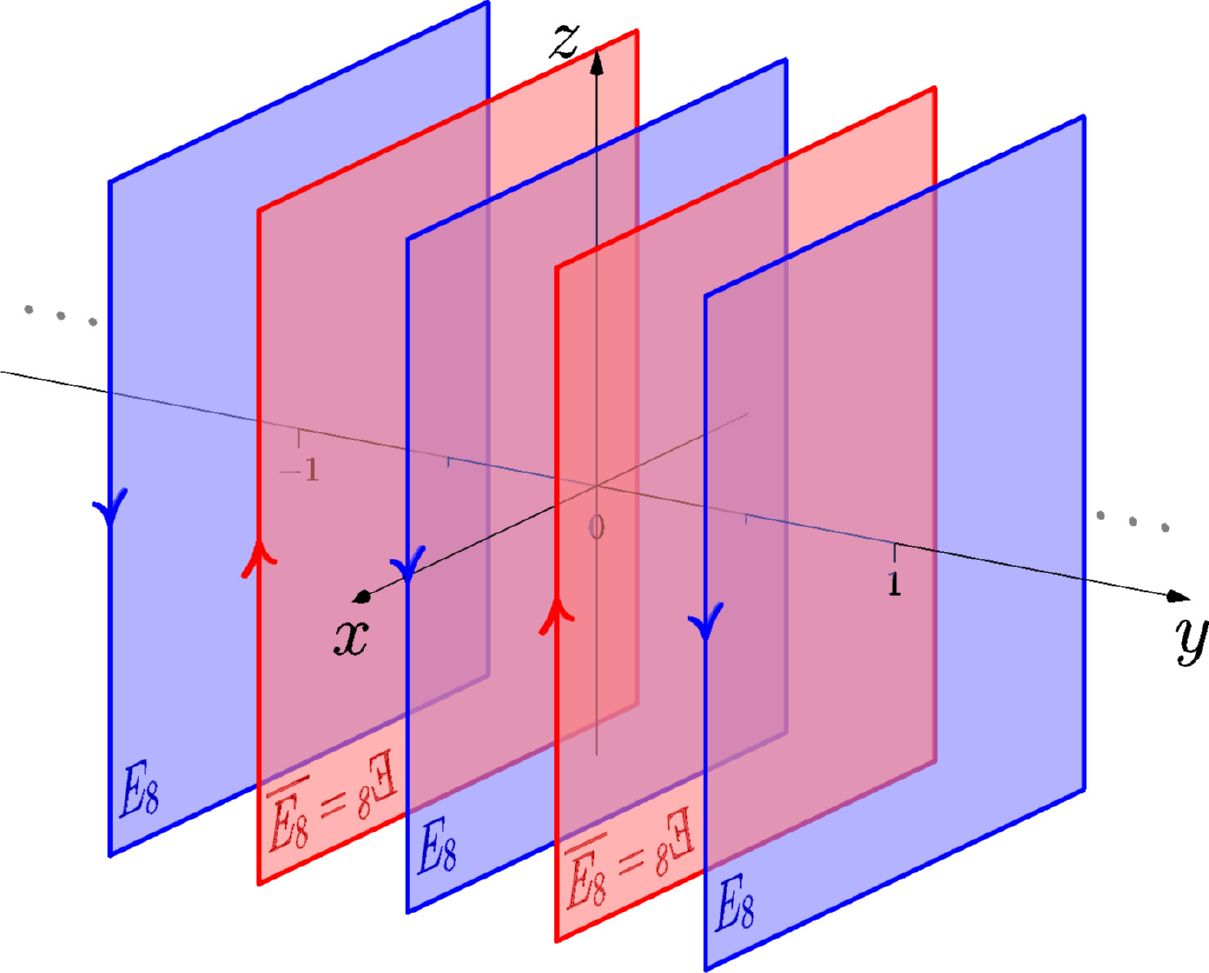}
\caption[Alternately layered $E_8$ state.]{Alternately layered $E_{8}$ state in the $t_{y}$ direction.}
\label{fig:layerE8_y2} 
\end{figure}

The factor $\mathbb{Z}_{2}$ in Eq.~(\ref{eq:Pm_h1}) is generated by the cSPT phase built by putting an $E_{8}$ state at each integer $y$ plane and its inverse $\overline{E_{8}}$ at each half-integer $y$ plane, as shown in Figure \ref{fig:layerE8_y2}. This construction was proposed and the resulting phases were studied in Ref.~\cite{Hermele_torsor}. In particular, the order of this cSPT phase is $2$, meaning that stacking two copies of such a phase produces a trivial phase. In fact, we will soon see that this generating phase can be realized by a model with higher symmetry like $Pmmm$ and is protected nontrivial by any single orientation-reversing symmetry.

In addition, the other part of $h_{\phi}^{3}\left(G,\phi\right)$ is 
\begin{eqnarray}
H^{\phi + 5}_{G} \paren{\pt; \mathbb{Z}} &=& H^{\phi + 5}_{\mathbb{Z}\rtimes\mathbb{Z}_{2}^{\phi}} \paren{\pt; \mathbb{Z}} \oplus \left[H^{\phi + 4}_{\mathbb{Z}\rtimes\mathbb{Z}_{2}^{\phi}} \paren{\pt; \mathbb{Z}} \right]^{2} \oplus H^{\phi + 3}_{\mathbb{Z}\rtimes\mathbb{Z}_{2}^{\phi}} \paren{\pt; \mathbb{Z}} \nonumber\\
&=& \mathbb{Z}_{2}^{2}\oplus0\oplus\mathbb{Z}_{2}^{2}.\label{eq:Pm_h5}
\end{eqnarray}
The three summands correspond to the phases built from group cohomology SPT phases in 2, 1, and 0 spatial dimensions respectively in Ref.~\cite{Huang_dimensional_reduction}.

For the reader's convenience, let us review the construction briefly here. The four cSPT phases corresponding to first summand $\mathbb{Z}_{2}^{2}$ in Eq.~(\ref{eq:Pm_h5}) can be built by putting 2D SPT states with Ising symmetry on mirror planes. We notice that each reflection acts as an Ising symmetry (\emph{i.e.}, a unitary internal symmetry of order 2) on its mirror plane. Moreover, the space group $Pm$ contains two families of inequivalent mirrors (\emph{i.e.}, integer $y$ planes and half-integer $y$ planes). For each family, there are two choices of 2D SPT phases with Ising symmetry, classified by $H^{4}_{\ZZZ_2} \paren{\pt; \ZZZ}\cong H^{3}_{\ZZZ_2} \paren{\pt; U(1)}\cong\mathbb{Z}_{2}$. Thus, we realize four cSPT phases with group structure $\mathbb{Z}_{2}^{2}$. Further, noticing the translation symmetries within each mirror, we can put, to every unit cell of the mirror, a 0D state carrying eigenvalue $\pm1$ of the corresponding reflection; 0D states with Ising symmetry are classified by their symmetry charges and are formally labeled by $H^{2}_{\ZZZ_2} \paren{\pt; \ZZZ}\cong H^{1}_{\ZZZ_2} \paren{\pt; U(1)}=\mathbb{Z}_{2}$. They produce another $\mathbb{Z}_{2}$ factor of cSPT phases associated with each inequivalent family of mirrors. Thus, we get another four cSPT phases labeled by the last summand $\mathbb{Z}_{2}^{2}$ in Eq.~(\ref{eq:Pm_h5}).

\begin{figure}
\centering
\includegraphics[width=0.5\columnwidth]{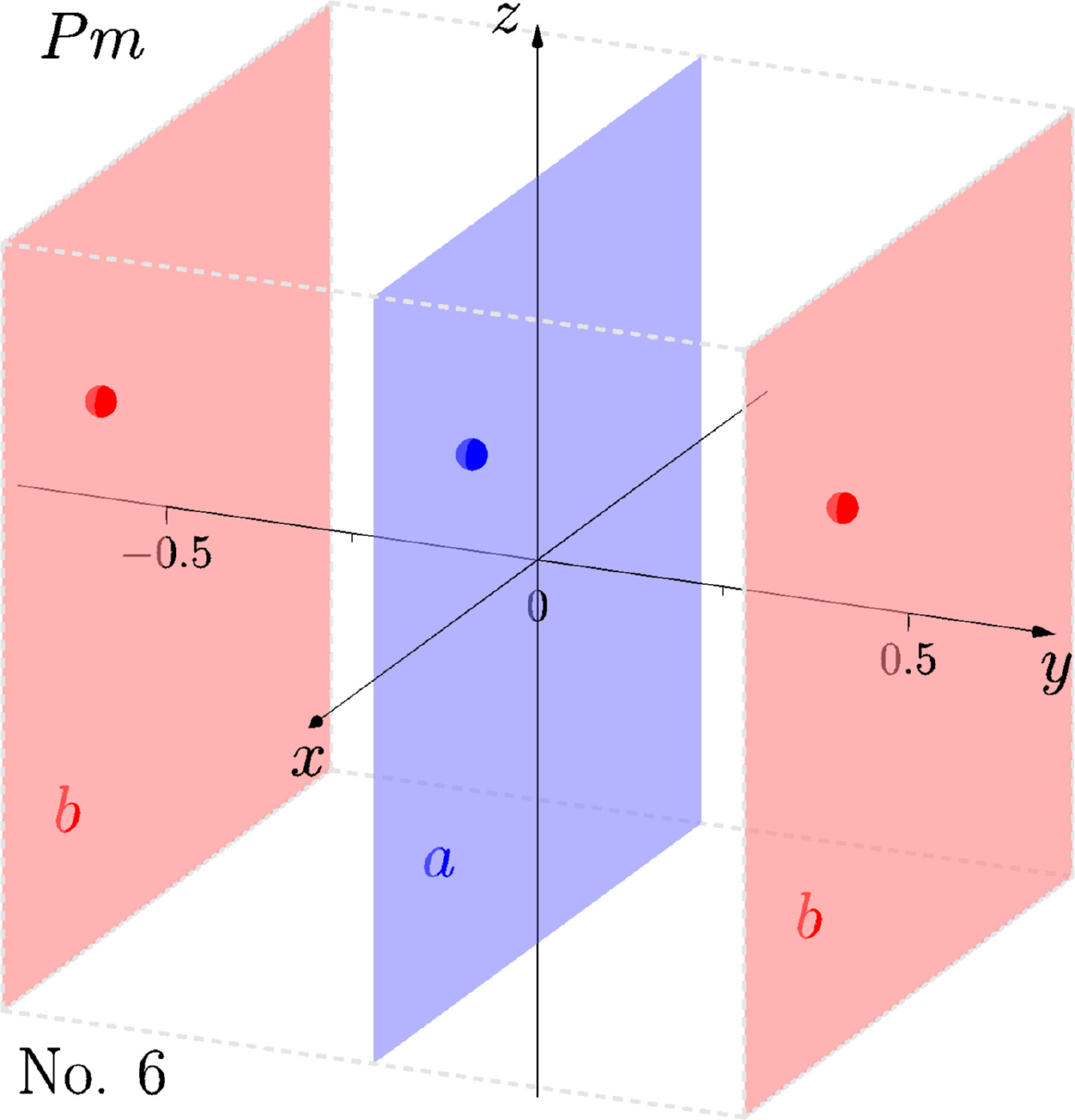}
\caption[The two inequivalent families of mirror planes of space group $Pm$.]{For space group No.~6 (\emph{i.e.}, $G=Pm$), there are two inequivalent family of mirrors (blue and red), whose Wyckoff position labels are $a$ and $b$ respectively in Ref.~\cite{ITA2006}. The cSPT phases corresponding to $H^{\phi + 5}_G \paren{\pt; \ZZZ}$ are built by assigning to the mirrors 2D SPT phases protected by Ising symmetry and $\mathbb{Z}_{2}$ point charges (blue and red dots). }
\label{fig:Pm_SPT} 
\end{figure}

To summarize, this construction of $H^{\phi + 5}_G \paren{\pt; \ZZZ}$ phases for the space group $Pm$ can be presented by Figure \ref{fig:Pm_SPT}. It can be generalized to all the other space groups: to each Wyckoff position, we assign group cohomology SPT phases protected by its site symmetry. Further technical details can be found in Ref.~\cite{Huang_dimensional_reduction}.

\subsubsection{Space group No.\,25 or $Pmm2$}
\label{subsubsec:Pmm2}

In the following, we will focus on developing a universal construction of $H^{\phi + 1}_G \paren{\pt; \ZZZ}$ phases. To get motivated, let us look at more examples.

The space group $G=Pmm2$ is generated by $t_{x}$, $t_{y}$, $t_{y}$, and two reflections 
\begin{align}
m_{x}: & \left(x,y,z\right)\mapsto\left(-x,y,z\right),\\
m_{y}: & \left(x,y,z\right)\mapsto\left(x,-y,z\right).
\end{align}
This time, Theorem~\ref{thm:H1_formula} in App.\,\ref{app:list_of_classifications}
tells us that 
\begin{gather}
H^{\phi + 1}_G \paren{\pt; \ZZZ}=\mathbb{Z}_{2}.\label{eq:Pmm2_h1}
\end{gather}
There is no $\mathbb{Z}$ factor any more as expected, because $m_{x}$ and $m_{y}$ together require $\gamma_{x}=\gamma_{y}=\gamma_{z}=0$.

The cSPT phase with $Pmm2$ symmetry generating $H^{\phi + 1}_G \paren{\pt; \ZZZ}$ is actually compatible with a higher symmetry $Pmmm$. Hence let us combine the study on this phase with the discussion of the space group $Pmmm$ below.

\subsubsection{Space group No.\,47 or $Pmmm$\label{subsubsec:cSPT_Pmmm}}

The space group $Pmmm$ is generated by $t_{x}$, $t_{y}$, $t_{y}$, and three reflections 
\begin{align}
m_{x}: & \left(x,y,z\right)\mapsto\left(-x,y,z\right),\label{eq:mx}\\
m_{y}: & \left(x,y,z\right)\mapsto\left(x,-y,z\right),\label{eq:my}\\
m_{z}: & \left(x,y,z\right)\mapsto\left(x,y,-z\right).\label{eq:mz}
\end{align}
For $G=Pmmm$, Theorem~\ref{thm:H1_formula} in App.\,\ref{app:list_of_classifications} gives
\begin{gather}
H^{\phi + 1}_G \paren{\pt; \ZZZ}=\mathbb{Z}_{2}.\label{eq:Pmmm_h1}
\end{gather}
There is no $\mathbb{Z}$ factor as in the case of $Pmm2$ above.

The cSPT phase generating $H^{\phi + 1}_G \paren{\pt; \ZZZ}=\mathbb{Z}_{2}$ can be constructed as in Figure \ref{fig:sg47_E8}. Step 1: we partition the 3D space into cuboids of size $\frac{1}{2}\times\frac{1}{2}\times\frac{1}{2}$. Each such cuboid works as a \emph{fundamental domain}, also know as an \emph{asymmetric unit} in crystallography \cite{ITA2006}; it is a smallest simply connected closed part of space from which, by application of all symmetry operations of the space group, the whole of space is filled. Every orientation-reversing symmetry relates half of these cuboids (blue) to the other half (red). Step 2: we attach an $\mathfrak{e_{f}}\mathfrak{\mathfrak{m}_{f}}$ topological state to the surface of each cuboid from inside with all symmetries in $G$ respected. An $\mathfrak{e_{f}m_{f}}$ topological state hosts three anyon species, denoted $\mathfrak{e_{f}}$, $\mathfrak{m_{f}}$, and $\boldsymbol{\varepsilon}$, all with fermionic self-statistics. Such a topological order can be realized by starting with a $\nu=4$ integer quantum Hall state and then coupling the fermion parity to a $\mathbb{Z}_{2}$ gauge field in its deconfined phase \cite{Kitaev_honeycomb}. This topological phase exhibits net chiral edge modes under an open boundary condition. Step 3: there are two copies of $\mathfrak{e_{f}m_{f}}$ states at the interface of neighbor cuboids and we condense $\left(\mathfrak{e_{f}},\mathfrak{e_{f}}\right)$ and $\left(\mathfrak{m_{f}},\mathfrak{m_{f}}\right)$ simultaneously without breaking any symmetries in $G$, where $\left(\mathfrak{e_{f}},\mathfrak{e_{f}}\right)$ (resp. $\left(\mathfrak{m_{f}},\mathfrak{m_{f}}\right)$) denotes the anyon formed by pairing $\mathfrak{e_{f}}$ (resp. $\mathfrak{m_{f}}$) from each copy. After condensation, all the other anyons are confined, resulting in the desired cSPT phase, denoted $\mathcal{E}$.

\begin{figure}
\centering
\includegraphics[width=0.5\columnwidth]{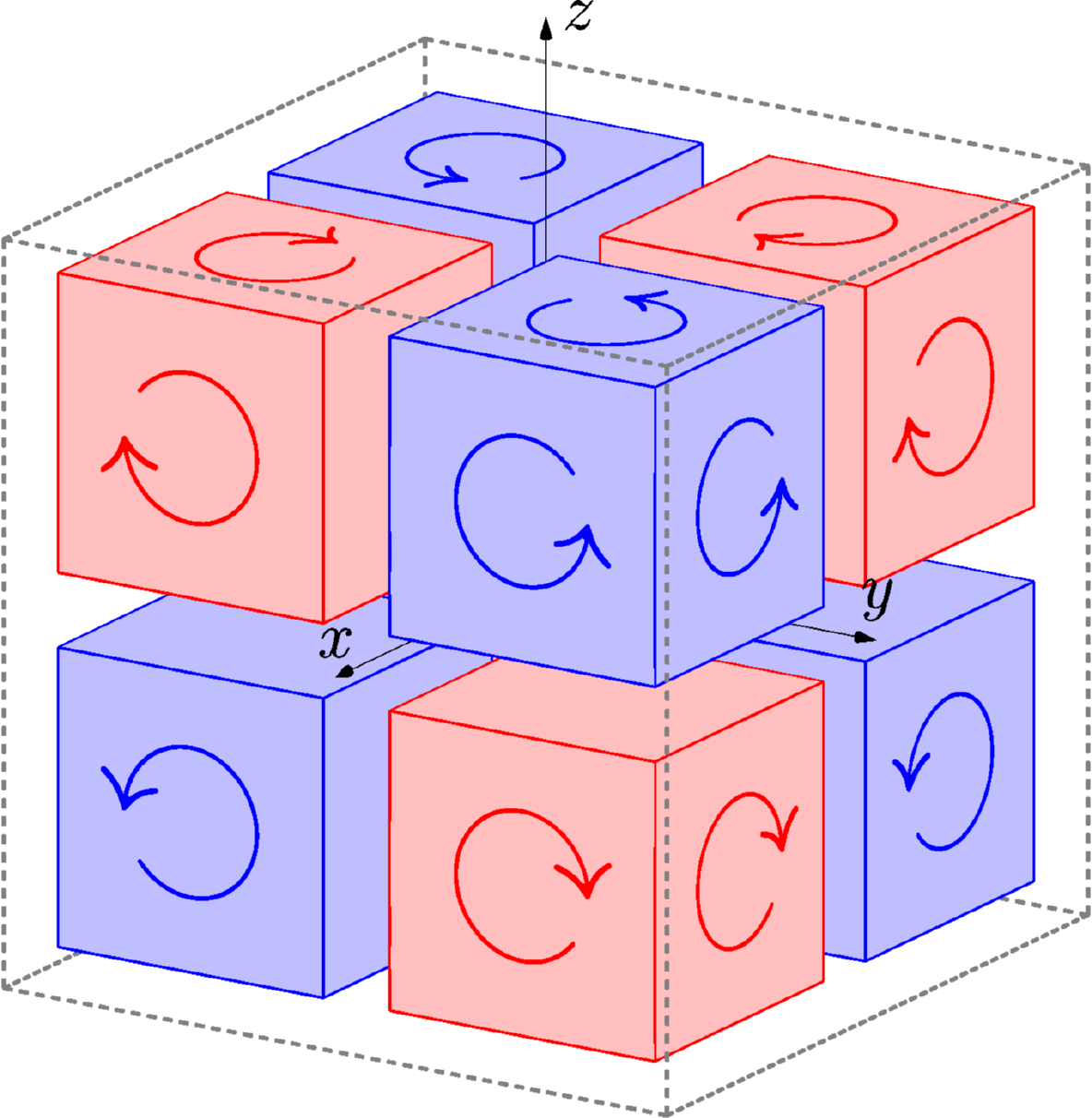}
\caption[Construction of a beyond-cohomology SPT phase with $Pmmm$ symmetry using the $\mathfrak{e_{f}m_{f}}$ state.]{A unit cell for the space group $Pmmm$ is partitioned into eight cuboids (red and blue), each of which works as a fundamental domain. A chiral $\mathfrak{e_{f}m_{f}}$ state is put on the surface of each fundamental domain; its chirality is indicated by the arrowed arcs such that all symmetries are respected. On each interface between two fundamental domains, there are two copies of $\mathfrak{e_{f}m_{f}}$ states. Let $\left(\mathfrak{e_{f}},\mathfrak{e_{f}}\right)$ (resp. $\left(\mathfrak{m_{f}},\mathfrak{m_{f}}\right)$) denote the anyon formed by pairing $\mathfrak{e_{f}}$ (resp. $\mathfrak{m_{f}}$) from each copy. Then the anyons $\left(\mathfrak{e_{f}},\mathfrak{e_{f}}\right)$ and $\left(\mathfrak{m_{f}},\mathfrak{m_{f}}\right)$ can be condensed, leading to a model that generates (via the stacking operation) cSPT phases corresponding to $H^{\phi + 1}_G \paren{\pt; \ZZZ}$ for $G=Pmmm$.}
\label{fig:sg47_E8}
\end{figure}

To see that $\mathcal{E}$ generates $H^{\phi + 1}_G \paren{\pt; \ZZZ}=\mathbb{Z}_{2}$, we need to check that $\mathcal{E}$ is nontrivial and that $2\mathcal{E}$ (\emph{i.e.}, two copies of $\mathcal{E}$ stacking together) is trivial. First, we notice that any orientation-reversing symmetry (\emph{i.e.}, a reflection, a glide reflection, an inversion, or a rotoinversion) $g\in G$ is enough to protect $\mathcal{E}$ nontrivial. Let us consider an open boundary condition of the model, keeping only the fundamental domains enclosed by the surface shown in ~\ref{fig:sg47_E8_STO}. Then the construction in Figure \ref{fig:sg47_E8} leaves a surface of the $\mathfrak{e_{f}m_{f}}$ topological order respecting $g$. If the bulk cSPT phase is trivial, then the surface topological order can be disentangled in a symmetric way from the bulk by local unitary gates with a finite depth. However, this strictly 2D system of the $\mathfrak{e_{f}m_{f}}$ topological order is chiral, wherein the orientation-reversing symmetry $g$ has to be violated. This contradiction shows that the bulk is nontrivial by the protection of $g$.

\begin{figure}
\noindent\begin{minipage}[t]{0.5\columnwidth}%
\centering
\includegraphics[width=0.8\columnwidth]{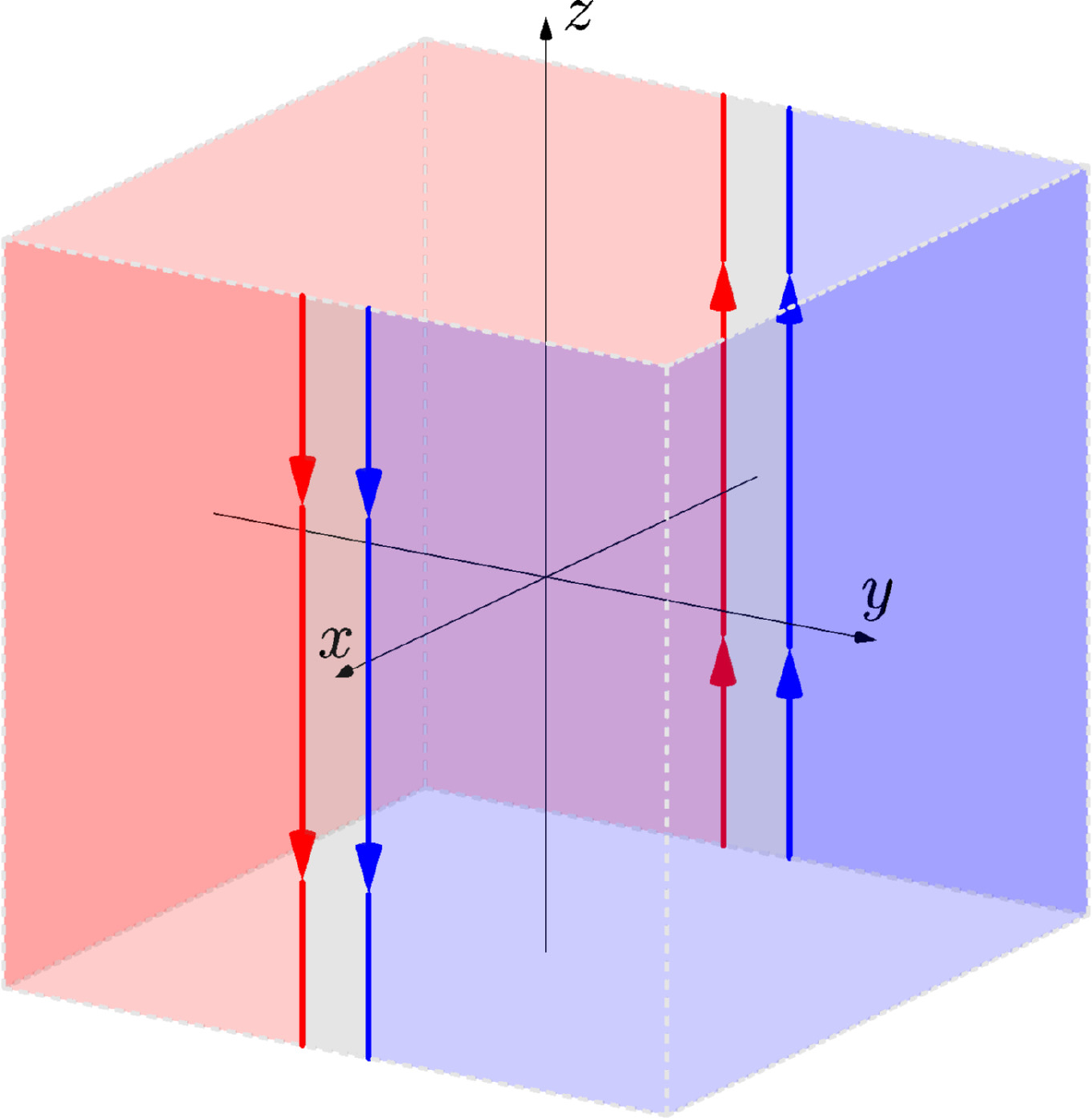}

(a)%
\end{minipage}
\begin{minipage}[t]{0.5\columnwidth}%
\centering
\includegraphics[width=0.8\columnwidth]{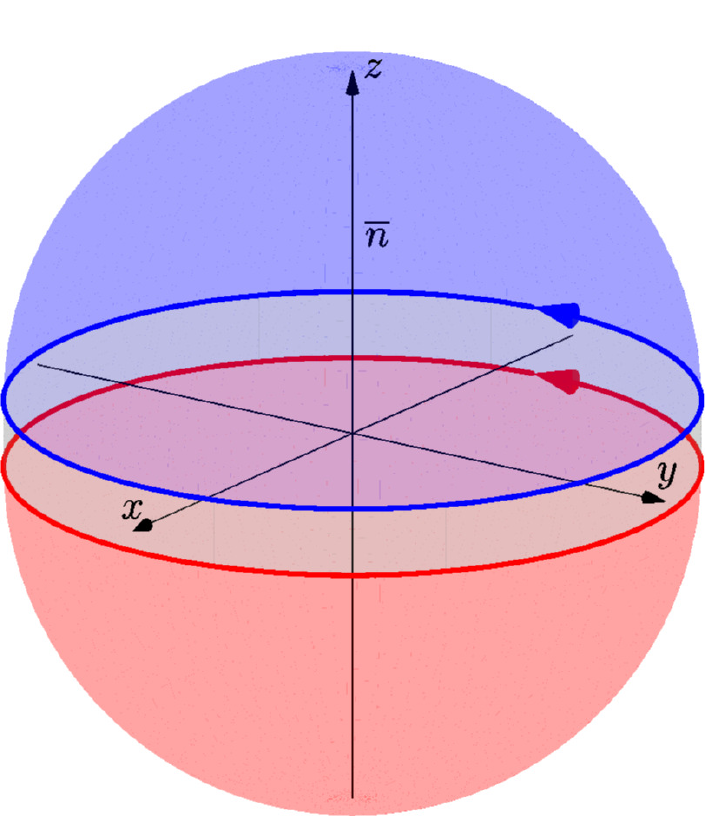}

(b)%
\end{minipage}
\caption[Incompatibility of the $\mathfrak{e_{f}m_{f}}$ topological order with orientation-reversing symmetry in strictly 2D systems.]{Two 2D systems (red and blue) of the $\mathfrak{e_{f}m_{f}}$ topological order cannot be glued together into a \emph{gapped} state respecting any orientation-reversing symmetry; their chiral edge modes must propagate in the same direction because of the symmetry. (a) illustrates cases with an rotoinversion axis $\overline{n}$ (including the spacial cases with an inversion center or a mirror plane). (b) illustrates cases with a glide reflection plane such as $c:\left(x,y,z\right)\protect\mapsto\left(x,-y,z+1/2\right)$. }
\label{fig:sg47_E8_STO}
\end{figure}

To better understand the incompatibility of a strictly 2D system of the $\mathfrak{e_{f}m_{f}}$ topological order with any orientation-reversing symmetry $g$, we view the 2D system as a gluing result of two regions related by $g$ as in Figure \ref{fig:sg47_E8_STO}. The $\mathfrak{e_{f}m_{f}}$ topological order implies net chiral edge modes for each region. Further, $g$ requires that edge modes from the two region propagate in the same direction at their 1D interface. Thus, a gapped gluing is impossible, which shows the non-existence of the $\mathfrak{e_{f}m_{f}}$ topological order compatible with $g$ in a strictly 2D system. Therefore, the surface $\mathfrak{e_{f}m_{f}}$ topological order respecting $g$ proves the non-triviality of the bulk cSPT phase $\mathcal{E}$.

On the other hand,  $2\mathcal{E}$ is equivalent to the model obtained by attaching an $E_{8}$ state to the surface of each fundamental domain from inside in a symmetric way. To check this, let us trace back the construction of $2\mathcal{E}$: we start with attaching two copies of $\mathfrak{e_{f}m_{f}}$ states on the surface of each fundamental domain from inside. Then let us focus a single rectangle interface between two cuboids. There are four copies of $\mathfrak{e_{f}m_{f}}$ states along it, labeled by $i=1,2$ for one side and $i=3,4$ for the other side. The symmetries relate $i=1\leftrightarrow i=4$ and $i=2\leftrightarrow i=3$ separately. If we further condense $(\mathfrak{e}_{\mathfrak{f}}^{1},\mathfrak{e}_{\mathfrak{f}}^{4})$, $(\mathfrak{m}_{\mathfrak{f}}^{1},\mathfrak{m}_{\mathfrak{f}}^{4})$, $(\mathfrak{e}_{\mathfrak{f}}^{2},\mathfrak{e}_{\mathfrak{f}}^{3})$, and $(\mathfrak{m}_{\mathfrak{f}}^{2},\mathfrak{m}_{\mathfrak{f}}^{3})$, then we get the local state of $2\mathcal{E}$ near the rectangle. Here $\mathfrak{e}_{\mathfrak{f}}^{i}$ (resp. $\mathfrak{m}_{\mathfrak{f}}^{i}$) denotes the $\mathfrak{e_{f}}$ (resp. $\mathfrak{m_{f}}$) particle from the $i^{th}$ copy of $\mathfrak{e_{f}m_{f}}$ state and $(\mathfrak{e}_{\mathfrak{f}}^{i},\mathfrak{e}_{\mathfrak{f}}^{j})$ (resp. $(\mathfrak{m}_{\mathfrak{f}}^{i},\mathfrak{m}_{\mathfrak{f}}^{j})$) is the anyon obtained by pairing $\mathfrak{e}_{\mathfrak{f}}^{i}$, $\mathfrak{e}_{\mathfrak{f}}^{j}$ (resp. $\mathfrak{m}_{\mathfrak{f}}^{i}$, $\mathfrak{m}_{\mathfrak{f}}^{j}$). However, we may alternately condense $(\mathfrak{e}_{\mathfrak{f}}^{1},\mathfrak{e}_{\mathfrak{f}}^{2})$, $(\mathfrak{m}_{\mathfrak{f}}^{1},\mathfrak{m}_{\mathfrak{f}}^{2})$, $(\mathfrak{e}_{\mathfrak{f}}^{3},\mathfrak{e}_{\mathfrak{f}}^{4})$, and $(\mathfrak{m}_{\mathfrak{f}}^{3},\mathfrak{m}_{\mathfrak{f}}^{4})$. The resulting local state is two copies of $E_{8}$ states connecting the same environment in a symmetric gapped way. Thus, the two local states produced by different condensation procedures have the same edge modes with the same symmetry behavior and hence are equivalent. Therefore, $2\mathcal{E}$ is equivalent to the model constructed by attaching an $E_{8}$ state to the surface of each fundamental domain from inside. Since the later can be obtained by blowing an $E_{8}$ state bubble inside each fundamental domain, it (and hence $2\mathcal{E}$) is clearly trivial.

Thus, we have shown that $\mathcal{E}$ is nontrivial and that $2\mathcal{E}$ is trivial. Therefore, the cSPT phases generated by $\mathcal{E}$ have the $\mathbb{Z}_{2}$ group structure, which holds for any non-orientation-preserving subgroup of $Pmmm$ such as $Pm$ and $Pmm2$.

In particular, for $Pm$, the model constructed in Figure \ref{fig:layerE8_y2} actually presents the same cSPT phase as $\mathcal{E}$ constructed in Figure \ref{fig:sg47_E8}. To see this, we could blow an $E_{8}$ state bubble inside cuboids centered at $\frac{1}{4}\left(1,\pm1,-1\right)+\mathbb{Z}^{3}$ and $\frac{1}{4}\left(-1,\pm1,1\right)+\mathbb{Z}^{3}$, with chirality opposite to those indicated by the arrowed arcs shown on the corresponding cuboids in Figure \ref{fig:sg47_E8}. This relates $\mathcal{E}$ to alternately layered $E_{8}$ states; however, each reflection does not acts trivially on the resulting $E_{8}$ layer on its mirror as the model in Figure \ref{fig:layerE8_y2}. To further show their equivalence, we look at a single $E_{8}$ layer at $y=0$ for instance. Since it may be obtained by condensing $(\mathfrak{e_{f}},\mathfrak{e_{f}})$ and $(\mathfrak{m_{f}},\mathfrak{m_{f}}$) in a pair of $\mathfrak{e_{f}m_{f}}$ topological states attached to the mirror, it hence can connect the $\mathfrak{e_{f}m_{f}}$ surface shown in Figure \ref{fig:sg47_E8_STO} in a gapped way with the reflection $m_{y}$ respected. On the other hand, we know that an $E_{8}$ state put at $y=0$ with trivial $m_{y}$ action can also connect to this surface in gapped $m_{y}$-symmetric way \cite{Hermele_torsor} and is thus equivalent to the corresponding $E_{8}$ layer just mentioned with nontrivial $m_{y}$ action. As a result, the models constructed in Figures \ref{fig:layerE8_y2} and \ref{fig:sg47_E8} realize the same cSPT phase with $Pm$ symmetry.

\subsubsection{Space group No.\,2 or $P\overline{1}$}

Let us explain the role of inversion symmetry by the example of space group $P\overline{1}$, which is generated by $t_{x}$, $t_{y}$, $t_{y}$, and an inversion
\begin{equation}
\overline{1}:\left(x,y,z\right)\mapsto\left(-x,-y,-z\right).
\end{equation}
For $G=P\overline{1}$, Theorem~\ref{thm:H1_formula} in App.\,\ref{app:list_of_classifications} tells us that 
\begin{equation}
H^{\phi + 1}_G \paren{\pt; \ZZZ}=\mathbb{Z}^{3}\times \mathbb{Z}_{2}.\label{eq:sg2_E8}
\end{equation}
The factor $\mathbb{Z}^{3}$ specifies $\frac{1}{8}\left(\gamma_{x},\gamma_{y},\gamma_{z}\right)$ as in the case of $P1$. For instance, the phase labeled by $\frac{1}{8}\left(\gamma_{x},\gamma_{y},\gamma_{z}\right)=\left(0,1,0\right)$ can be constructed by putting a copy of $E8$ state on each of the planes $y=\cdots,-2,-1,0,1,2,\cdots$ with the symmetries $t_{x},t_{y},t_{z}$ and $\overline{1}$ respected.

On the other hand, we can generate the factor $\mathbb{Z}_{2}$ in Eq.~(\ref{eq:sg2_E8}) by the phase constructed in Figure \ref{fig:sg47_E8}. In particular, we have shown that this phase is nontrivial under the protection any orientation-reversing symmetry, like the inversion symmetry here, in Sec.~\ref{subsubsec:cSPT_Pmmm}.

\subsubsection{Space group No.\,7 or $Pc$}

Finally, we explain the role of glide reflection symmetry by the example of space group $Pc$, which is generated by $t_{x}$, $t_{y}$, $t_{y}$, and a glide reflection
\begin{equation}
c:\left(x,y,z\right)\mapsto\left(x,-y,z+1/2\right).
\end{equation}
For $G=Pc$, Theorem~\ref{thm:H1_formula} in App.\,\ref{app:list_of_classifications} tells us that
\begin{equation}
H^{\phi + 1}_G \paren{\pt; \ZZZ}=\mathbb{Z}\times\mathbb{Z}_{2}.
\end{equation}
The factor $\mathbb{Z}$ specifies $\frac{1}{8}\left(0,\gamma_{y},0\right)$; the glide reflection symmetry $c$ requires $\gamma_{x}=\gamma_{z}=0$.

To construct the cSPT phase that generates the summand $\mathbb{Z}_{2}$, we consider the space group $G'$ generated by translations $t_{x}$, $t_{y}$, $t_{z}^{\prime}:\left(x,y,z\right)\mapsto\left(x,y,z+\frac{1}{2}\right)$ together with reflections $m_{x}$, $m_{y}$, $m_{z}$ in Eqs.~(\ref{eq:mx}-\ref{eq:mz}). Obviously, $G\subset G'$ and $G'$ is a space group of type $Pmmm$. Then the model constructed as in Figure \ref{fig:sg47_E8} with respect to $G'$ generates this  $\mathbb{Z}_{2}$ factor of cSPT phases protected by space group $G$. Particularly, the glide reflection $c\in G$ is enough to protect the corresponding cSPT phase nontrivial by the argument in Sec.~\ref{subsubsec:cSPT_Pmmm}.

\subsection{Verification of general classification by physical arguments\label{subsec:reduction_and_construction}}

In the above examples, we have presented cSPT phases by lower dimensional short-range entangled (SRE) states. Such a representation for a generic cSPT phase can be obtained by a dimensional reduction procedure; it is quite useful for constructing, analyzing, and classifying cSPT phases \cite{Hermele_torsor, Huang_dimensional_reduction}. Below, let us review the idea of dimension reduction and illustrate how to build cSPT phases with lower dimensional SRE states in general. More emphasis will be put on the construction of cSPT phases involving $E_{8}$ states, which have not been systematically studied for all space groups in the literature.

Given a space group $G$, we first partition the 3D euclidean space $\mathbb{E}^{3}$ into fundamental domains accordingly. A \emph{fundamental domain}, also know as an \emph{asymmetric unit} in crystallography \cite{ITA2006}, is a smallest simply connected closed part of space from which, by application of all symmetry operations of the space group, the whole of space is filled. Formally, the partition is written as
\begin{equation}
\mathbb{E}^{3}=\bigcup_{g\in G}g\mathcal{F},
\end{equation}
where $\mathcal{F}$ is a fundamental domain and $g\mathcal{F}$ its image under the action of $g\in G$. If $g$ is not the identity of space group $G$, then by definition $\mathcal{F}$ and $g\mathcal{F}$ only intersect in their surfaces at most. In general, $\mathcal{F}$ can be chosen to be a convex polyhedron: the Dirichlet-Voronoi cell of a point $\mathsf{P}\in\mathbb{E}^{3}$ to its $G$-orbit, with $\mathsf{P}$ chosen to have a trivial stabilizer subgroup (this is always possible by discreteness of space groups) \cite{fundamental_domain}.

\begin{figure}
\centering
\includegraphics[width=0.5\columnwidth]{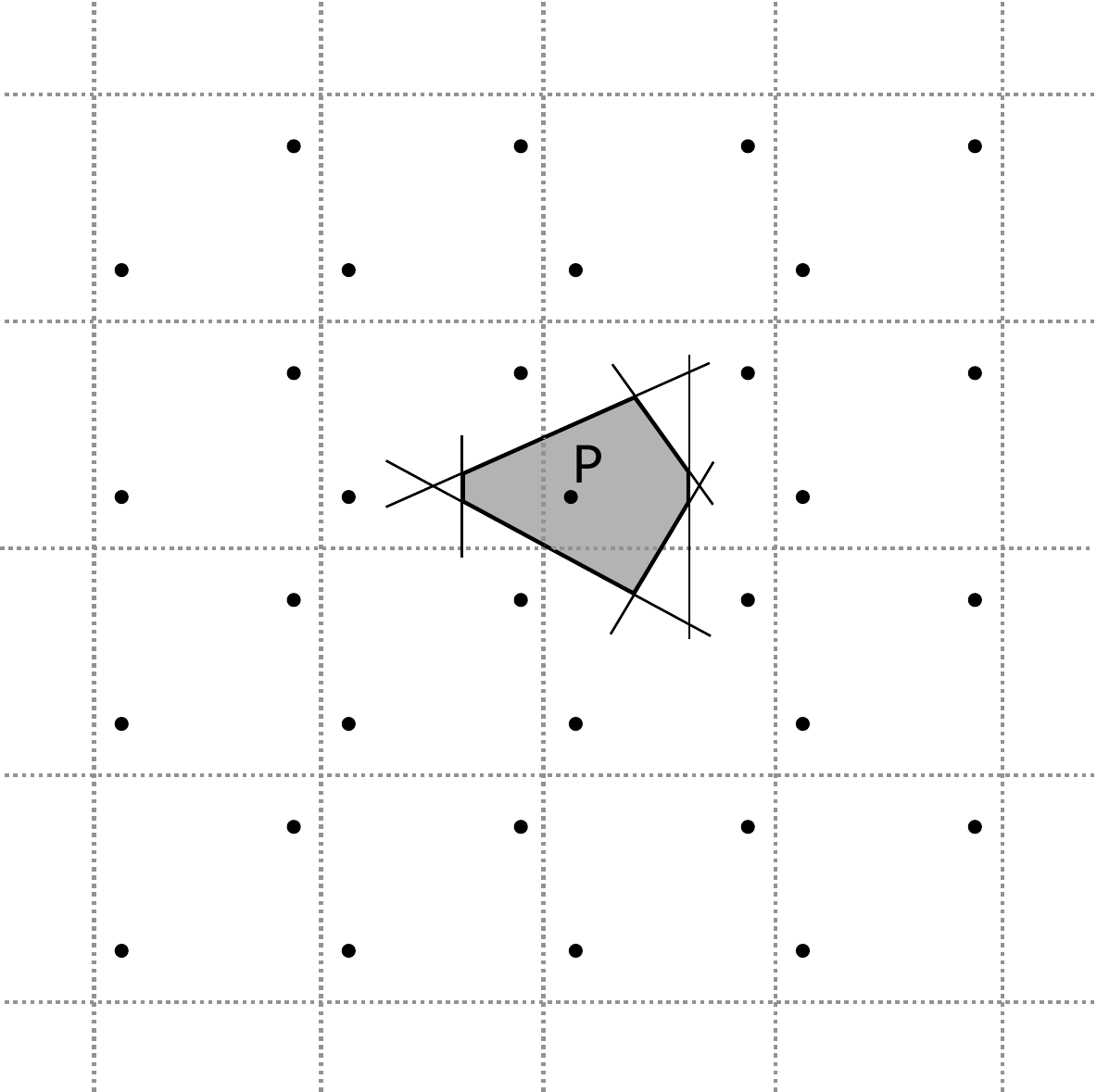}
\caption[The construction of Dirichlet-Voronoi cells.]{The Dirichlet-Voronoi cell (dark region) of a point $\mathsf{P}\in\mathbb{E}^{2}$ to its $G$-orbit (black dots), where $G$ is generated by translations and an in-plane two-fold rotation.}
\label{fig:p2_DV.pdf}
\end{figure}

The above definition and general construction of fundamental domain works for space groups in any dimensions. Let us take a lower dimensional case for a simple illustration: Figure \ref{fig:p2_DV.pdf} shows a fundamental domain given by the Dirichlet-Voronoi cell construction for wallpaper group No.~2, which is generated by translations and a two-fold rotation. Clearly, the choice of fundamental domain is often not unique as in this case; a regular choice of fundamental domain for each wallpaper group and 3D space group is available in the International Tables for Crystallography \cite{ITA2006}.

To use terminology from simplicial homology \cite{Elements_AT}, we further partition the fundamental domain $\mathcal{F}$ (resp. $g\mathcal{F}$) into tetrahedrons $\left\{ \varsigma_{\alpha}\right\} _{\alpha=1,2,\cdots,N}$ (resp. $\left\{ g\varsigma_{\alpha}\right\} _{\alpha=1,2,\cdots,N}$) such that $X=\mathbb{E}^{3}$ becomes a \emph{$G$-simplicial} complex. In a $G$-simplicial complex, each of its simplices is either completely fixed or mapped onto another simplex by $g$, $\forall g\in G$. Clearly, all internal points of every simplex share the same site symmetry.\footnote{The \emph{site symmetry} of a point $p$ is the group of symmetry operations under which $p$ is not moved, \emph{i.e.}, $G_{p}\coloneqq\left\{ g\in G|gp=p\right\} $.} To encode the simplex structure, let $\Delta^{n}\left(X\right)$ be the set of $n$-simplices (\emph{i.e.}, vertices for $n=0$, edges for $n=1$, triangles for $n=2$, and tetrahedrons for $n=3$) and $X_{n}$ the $n$-skeleton of $X$ (\emph{i.e.}, the subspace made of all $k$-simplices of $X$ for $k\leq n$).

\subsubsection{Dimensional reduction of cSPT phases}
\label{susbusbsec:dimensional_reduction}

Topological phases of matter should admit a topological quantum field theory (TQFT) description, whose correlation functions do not depend on the metric of spacetime. Thus, it is natural to conjecture the following two basic properties of topological phases (including cSPT phases). (1) Each phase can be presented by a state $\left|\psi\right\rangle $ with arbitrary short correlation length. (2) In each phase, any two states $\left|\psi_{0}\right\rangle $ and $\left|\psi_{1}\right\rangle $ with correlation length shorter than $r>0$ can be connected a path of states $\left|\psi_{\tau}\right\rangle $ (parameterized by $\tau\in\left[0,1\right]$) whose correlation length is shorter than $r$ for all $\tau$. The conjecture is satisfied by all topological states investigated in this work, allowing the dimensional reduction procedure described below. Its rigorous proof in a reasonable setting, however, remains an interesting question and goes beyond the scope of this work. If the conjecture holds in general for cSPT phases, our classification in this work will be complete. Otherwise, we would miss the cSPT phases where the two properties fail.

Given any cSPT phase for space group $G$, let us now describe the dimensional reduction procedure explaining why it can be built by lower dimensional states in general. To start, as conjectured, we can present this cSPT phase by a state $\left|\Psi\right\rangle $ with correlation length $\xi$ much smaller than the linear size of the fundamental domain $\mathcal{F}$. The short-range correlation nature implies that $\left|\Psi\right\rangle $ is the ground state of some gapped Hamiltonian $H$ whose interaction range is $\xi$ as well. The local part of $H$ inside $\mathcal{F}$ thus describes a 3D SRE state. It is believed that all 3D SRE states are trivial or weakly trivial (e.g. layered $E_{8}$ states). Thus, we can continuously change $H$ into a trivial Hamiltonian inside $\mathcal{F}$ (except within a thin region near the boundary of $\mathcal{F}$) keeping the correlation length of its ground state smaller than $\xi$ all the time. Removing the trivial degrees of freedom, we are left with a system on the $2$-skeleton $X_{2}$ of $\mathbb{E}^{3}$.

Still, the reduced system host no nontrivial excitations. Thus, there is an SRE state on each 2-simplex (\emph{i.e.}, triangle) $\tau_{\alpha}$, indexed by $\alpha$, of $X_{2}$. In particular, it could be $q_{2}\left(\alpha\right)$ copies of $E_{8}$ states (without specifying symmetry), where $q_{2}\left(\alpha\right)\in\mathbb{Z}$ with sign specifying the chirality. These data may be written collectively as a formal sum $q_{2}=\sum_{\alpha}q_{2}\left(\alpha\right)\tau_{\alpha}$, which may contain infinitely many terms as $X=\mathbb{E}^{3}$ is noncompact. On each edge $\ell$, the chiral modes from all triangles connecting to $\ell$ have to cancel in order for the system to be gapped. In terms of the simplicial boundary map $\partial$, we thus have that $\partial q_{2}\coloneqq\sum q_{2}\left(\alpha\right)\partial\tau_{\alpha}$ equals 0. In general, let $C_{k}\left(X\right)$ the set of such formal sums of $k$-simplices of $X$. Naturally, $C_{k}\left(X\right)$ has an Abelian group structure; $C_{-1}\left(X\right)$ is taken to be the trivial group. For any integer $k\geq0$, the boundary map $\partial_{k}$ (or simply $\partial$)$:C_{k}\left(X\right)\rightarrow C_{k-1}\left(X\right)$ is a group homomorphism and let $B_{k}\left(X\right)$ (resp. $Z_{k-1}\left(X\right)$) denote its kernel (resp. image). Thus, the $E_{8}$ state configuration on $X_{2}$ is encoded by $q_{2}\in B_{2}\left(X\right)$.

It clear that $\partial_{k+1}\circ\partial_{k}=0$ and hence $Z_{k}\left(X\right)\subseteq B_{k}\left(X\right)$. As $C_{k}\left(X\right)$ contains \emph{infinite} sums of simplices, it is not a standard group of $k$-chains. Instead, it can be viewed as $\left(3-k\right)$-cochains on the \emph{dual polyhedral decomposition} (also called \emph{dual block decomposition} \cite{Elements_AT}) of $X$. Thus, $B_{k}\left(X\right)/Z_{k}\left(X\right)$ should be understood as $\left(3-k\right)^{th}$ cohomology $H^{3-k}\left(X\right)$ rather than $k^{th}$ homology of $X$. Since $X=\mathbb{E}^{3}$ is contractible, its cohomology groups are the same as a point: $H^{n}\left(X\right)$ is $\mathbb{Z}$ for $n=0$ and trivial for $n>0$. Thus, $B_{2}\left(X\right)=Z_{2}\left(X\right)$ and hence any gapped $E_{8}$ state configuration $q_{2}\in B_{2}\left(X\right)$ can be expressed as $q_{2}=\partial q_{3}$ for some $q_{3}\in C_{3}\left(X\right)$. Also, it is clear that $B_{3}\left(X\right)$ is generated by the sum of all $3$-simplices with the right-handed orientation, which is simply denoted by $X$ as well.

As $q_{2}=\partial q_{3}$ is symmetric under $G$, we have $\partial\left(gq_{3}\right)=q\partial q_{3}=\partial q_{3}$ and hence 
\begin{equation}
gq_{3}=q_{3}+\nu^{1}\left(g\right)X,\label{eq:mu1}
\end{equation}
for some $\nu^{1}\left(g\right)\in\mathbb{Z}$. Clearly, $\nu^{1}\left(e\right)=0$ for the identity element $e\in G$. The consistent condition $(gh)q_{3}=g(hq_{3})$ requires that 
\begin{equation}
\nu^{1}(gh)=\phi(g)\nu^{1}(h)+\nu^{1}(g),
\end{equation}
\emph{i.e.}, the cocycle condition for $Z_{\phi}^{1}(G;\mathbb{Z})$. Definitions of group cocycles $Z_{\phi}^{1}(G;\mathbb{Z})$ as well as coboundaries $B_{\phi}^{1}(G;\mathbb{Z})$ and cohomologies $H^{\phi + 1}_G \paren{\pt; \ZZZ}$ are given in App.\,A of work \cite{SongXiongHuang}. Thus, $\nu^{1}$ is a normalized 1-cocycle. Moreover, we notice that $q_{3}$ is not uniquely determined by the $E_{8}$ state configuration $q_{2}$; solutions to $\partial q_{3}=q_{2}$ may differ by a multiple of $X$. According to Eq.~(\ref{eq:mu1}), the choice change $q_{3}\rightarrow q_{3}+\nu^{0}X$ leads to $\nu^{1}\left(g\right)\rightarrow\nu^{1}\left(g\right)+\left(d\nu^{0}\right)\left(g\right)$, where $\nu^{0}\in\mathbb{Z}$ and $\left(d\nu^{0}\right)\left(g\right)\coloneqq\phi\left(g\right)\nu^{0}-\nu^{0}$. Thus, $\nu^{1}$ is only specified up to a $1$-coboundary.  Therefore, any $G$-symmetric model (with correlation length much shorter than simplex size) on $X_{2}$ defines a cohomology group element $[\nu^{1}]\in H^{\phi + 1}_G \paren{\pt; \ZZZ}$.

By Lemma 1 in App.\,A of work \cite{SongXiongHuang}, each $[\nu^{1}]\in H^{\phi + 1}_G \paren{\pt; \ZZZ}$ can be parameterized by $\nu^{1}\left(t_{v_{1}}\right),\nu^{1}\left(t_{v_{2}}\right),\nu^{1}\left(t_{v_{3}}\right)\in\mathbb{Z}$ (together with $\nu^{1}\left(r\right)\pmod2\in\mathbb{Z}_{2}$ if $G$ is non-orientation-preserving), where $t_{v_{1}}$, $t_{v_{2}}$, $t_{v_{3}}$ are three elementary translations that generate the translation subgroup and $r$ is an orientation-reversing symmetry. Let us explain their physical meaning by examples. For the model in Figure \ref{fig:layerE8_y}, we could pick 
\begin{equation}
q_{3}=-\sum_{i,j,k\in\ZZZ}j\left(t_{x}^{i}t_{y}^{j}t_{z}^{k}\mathcal{F}\right)\label{eq:q3_translations}
\end{equation}
with fundamental domain $\mathcal{F}=\left[0,1\right]\times\left[0,1\right]\times\left[0,1\right]$. To use the terminology of simplicial homology, $\mathcal{F}$ can be partitioned into tetrahedrons and be viewed as a formal sum of them with the right-handed orientation. Moreover, $t_{x}^{i}t_{y}^{j}t_{z}^{k}\mathcal{F}$ denote the translation result of $\mathcal{F}$. It is straightforward to check that $\partial q_{3}$ equals the sum of 2-simplices (oriented toward the positive $y$ direction according to the right-hand rule) on integer $y$ planes; thus, $\partial q_{3}$ corresponds to the model in Figure \ref{fig:layerE8_y}. It is also clear that $t_{v}q_{3}-q_{3}=v^{y} X$ for a translation by $v=\left(v^{x},v^{y},v^{z}\right)$; hence $\nu^{1}\left(t_{x}\right)=0$, $\nu^{1}\left(t_{y}\right)=1$, and $\nu^{1}\left(t_{z}\right)=0$. Thus, we get a physical interpretation of $\nu^{1}\left(t_{v}\right)$: if the model is compactified in the $v$ direction such that $t_{v}^{L}=1$, then it is equivalent to $L\nu^{1}\left(t_{v}\right)$ copies of $E_{8}$ states as a 2D system. Clearly, models with different $\nu^{1}\left(t_{v}\right)$ on any translation symmetry $t_{v}$ must present distinct cSPT phases.

As another example, for $G=Pmmm$, $H^{\phi + 1}_G \paren{\pt; \ZZZ}=\mathbb{Z}_{2}$ with element $[\nu^{1}]$ parameterized by $\nu^{1}\left(r\right)\pmod2$ on any orientation-reversing element $r$ of $G$; here $\nu^{1}\left(t_{v}\right)$ on any translation $t_{v}\in G$ is required to be zero by symmetry. The $E_{8}$ state configuration of the model in Figure \ref{fig:sg47_E8} can be encoded by $q_{3}=\sum_{g\in G_{0}}g\mathcal{F}$ (\emph{i.e.}, the sum of cuboids colored blue), where $\mathcal{F}=[0,\frac{1}{2}]\times[0,\frac{1}{2}]\times[0,\frac{1}{2}${]} is a fundamental domain\footnote{To use the terminology of simplicial homology, $\mathcal{F}$ can be partitioned into tetrahedrons and be viewed as a formal sum of them with the right-handed orientation.} and $G_{0}$ is the orientation-preserving subgroup of $G$. For any orientation-reversing symmetry $r\in G$ (e.g. $m_{x}$, $m_{y}$, and $m_{z}$ in Eqs.~(\ref{eq:mx}-\ref{eq:mz})), it is clear that $rq_{3}=q_{3}-\gamma_{X}$ and hence $\nu^{1}\left(r\right)=-1$.

For any space group $G$, let $\SPT^3\paren{G, \phi}$ be the set of $G$-SPT phases, which forms a group under the stacking operation. we will see in Sec.~\ref{subsubsec:H1_invariance} that the dimensional reduction procedure actually gives a well-define group homorphism $\mathfrak{D}: \SPT^3\paren{G, \phi}\rightarrow H^{\phi + 1}_G \paren{\pt; \ZZZ}$. In particular, $[\nu^1]\in H^{\phi + 1}_G \paren{\pt; \ZZZ}$ is a well-defined invariant, independent of the dimensional reduction details, for each generic $G$-SPT phase. Conversely, we will show in Sec.~\ref{subsubsec:Construction} that a $G$-symmetric SRE state, denoted $\left|[\nu^{1}]\right\rangle $,  can always be constructed to present a $G$-SPT phase corresponding to each $[\nu^1]\in H^{\phi + 1}_G \paren{\pt; \ZZZ}$, resulting in a group homomorphism $\mathfrak{C}:H^{\phi + 1}_G \paren{\pt; \ZZZ} \rightarrow \SPT^3\paren{G, \phi} $.

Then a generic $G$-symmetric SRE state $\left|\Psi\right\rangle $ is equivalent (as a $G$-SPT phase) to $\left|[\nu^{1}]\right\rangle \otimes\left|\Psi_{2}\right\rangle $ with $[\nu^{1}]\in H^{\phi + 1}_G \paren{\pt; \ZZZ}$  specified by $\left|\Psi\right\rangle $ via the dimensional reduction and $\left|\Psi_{2}\right\rangle \coloneqq \left|-[\nu^{1}]\right\rangle \otimes\left|\Psi\right\rangle $ obtained by stacking $\left|-[\nu^{1}]\right\rangle$ (an inverse of $\left|[\nu^{1}]\right\rangle$) with $\left|\Psi\right\rangle $. By dimensional reduction, we may represent $\left|\Psi_{2}\right\rangle$ by a $G$-symmetry SRE state (with arbitrarily short correlation length) on $X_{2}$, whose $E_8$ state configuration specifies $0\in H^{\phi + 1}_G \paren{\pt; \ZZZ}$. Actually, $\left|\Psi_{2}\right\rangle$ can be represented without using $E_{8}$ states: since $E_{8}$ state configuration of $\left|\Psi_{2}\right\rangle $ is given by $q_{3}$ satisfying $gq_{3}=q_{3}$, we can reduce $q_{2}=\partial q_{3}$ to $q_{2}=0$ by blowing $-q_{3}\left(\alpha\right)$ copies of $E_{8}$ state bubbles in a process respecting $G$ symmetry. Explicitly, as in the example of $Pm$ in Sec.~\ref{subsubsec:Pm}, such phases equipped with the stacking operation form the summand $H^{\phi + 5}_G \paren{\pt; \ZZZ}$ in Eq.~(\ref{3D_prediction}) and can be built with lower dimensional group cohomology phases \cite{ThorngrenElse,Huang_dimensional_reduction}. Let us briefly describe how to decompose $\left|\Psi_{2}\right\rangle $ (with $q_{2}=0$) into these lower dimensional components.

Without $E_{8}$ states, $\left|\Psi_{2}\right\rangle $ (represented on $X_{2}$) can only have nontrivial 2D phases on 2-simplices in mirror planes, where each point has a $\mathbb{Z}_{2}$ site symmetry effectively working as an Ising symmetry protecting 2D phases classified by $H^{3}_{\ZZZ_2} \paren{\pt; U(1)}=\mathbb{Z}_{2}$. As the system is symmetric and gapped, the 2D states associated with all 2-simplices in the same mirror plane have to be either trivial or nontrivial simultaneously. Thus, each inequivalent mirror plane contributes a $\mathbb{Z}_{2}$ factor to cSPT phase classification. Conversely, a reference model of 2D nontrivial phase protected by the $\mathbb{Z}_{2}$ site symmetry on a mirror plane can be constructed by sewing 2-simplices together with all symmetries respected. Adding such a model to each mirror plane with nontrivial state in $\left|\Psi_{2}\right\rangle $ resulting in a state $\left|\Psi_{1}\right\rangle $ which may be nontrivial only along the 1-skeleton $X_{1}$. Only 1D states lie in the intersection of inequivalent mirror planes can be nontrivial projected by the site symmetry $C_{nv}$ for $n=2,4,6$. A reference nontrivial model, which generates the 1D phases protected by the $C_{nv}$ site symmetry and classified by $H^{2}_{C_{nv}} \paren{\pt; U(1)}=\mathbb{Z}_{2}$ for $n$ even, can be constructed by connecting states on 1-simplices in a gapped and symmetric way. Adding such reference states to the $C_{nv}$ axes where $\left|\Psi_{1}\right\rangle $ is nontrivial, we get a state $\left|\Psi_{0}\right\rangle $ which may be nontrivial only on the 0-skeleton $X_{0}$. Explicitly, $\left|\Psi_{0}\right\rangle $ is a tensor product of trivial degrees of freedom and some isolated 0D states centered at the 0-simplices (i.e., vertices) of $X$ carrying nontrivial one-dimensional representation of its site symmetry (\emph{i.e.}, site symmetry charges). However, some different site symmetry charge configurations may be changed into each other by charge splitting and fusion; they are not in 1-1 correspondence with cSPT phases, whose classification and characterization are studied in Ref.~\cite{Huang_dimensional_reduction}.

Putting all the above ingredient together, we get that a generic cSPT state $\left|\Psi\right\rangle $ can be reduced to the stacking of $\left|[\nu^{1}]\right\rangle $, 2D nontrivial states protected by $\mathbb{Z}_{2}$ site symmetry on some mirror planes, 1D nontrivial states protected by $C_{nv}$ site symmetry on some $C_{nv}$ axes with $n$ even, and 0D site symmetry charges. As we have mentioned, the cSPT phases built without using $E_{8}$ states have a group structure $H^{\phi + 5}_G \paren{\pt; \ZZZ}$ (\emph{i.e.}, the first summand in Eq.~(\ref{3D_prediction})) under stacking operation \cite{Huang_dimensional_reduction,ThorngrenElse}. Next, we will focus on the models with ground states $\left|[\nu^{1}]\right\rangle $ for $[\nu^{1}]\in H^{\phi + 1}_G \paren{\pt; \ZZZ}$.

\subsubsection{$H^{\phi + 1}_G \paren{\pt; \ZZZ}$ as a cSPT phase invariant\label{subsubsec:H1_invariance}}

The cSPT phase generating $H^{\phi + 1}_G \paren{\pt; \ZZZ}=\mathbb{Z}_{2}$ can be constructed as in Figure \ref{fig:sg47_E8}. Step 1: we partition the 3D space into cuboids of size $\frac{1}{2}\times\frac{1}{2}\times\frac{1}{2}$. Each such cuboid works as a \emph{fundamental domain}, also know as an \emph{asymmetric unit} in crystallography \cite{ITA2006}; it is a smallest simply connected closed part of space from which, by application of all symmetry operations of the space group, the whole of space is filled. Every orientation-reversing symmetry relates half of these cuboids (blue) to the other half (red). Step 2: we attach an $\mathfrak{e_{f}}\mathfrak{\mathfrak{m}_{f}}$ topological state to the surface of each cuboid from inside with all symmetries in $G$ respected. An $\mathfrak{e_{f}m_{f}}$ topological state hosts three anyon species, denoted $\mathfrak{e_{f}}$, $\mathfrak{m_{f}}$, and $\boldsymbol{\varepsilon}$, all with fermionic self-statistics. Such a topological order can be realized by starting with a $\nu=4$ integer quantum Hall state and then coupling the fermion parity to a $\mathbb{Z}_{2}$ gauge field in its deconfined phase \cite{Kitaev_honeycomb}. This topological phase exhibits net chiral edge modes under an open boundary condition. Step 3: there are two copies of $\mathfrak{e_{f}m_{f}}$ states at the interface of neighbor cuboids and we condense $\left(\mathfrak{e_{f}},\mathfrak{e_{f}}\right)$ and $\left(\mathfrak{m_{f}},\mathfrak{m_{f}}\right)$ simultaneously without breaking any symmetries in $G$, where $\left(\mathfrak{e_{f}},\mathfrak{e_{f}}\right)$ (resp. $\left(\mathfrak{m_{f}},\mathfrak{m_{f}}\right)$) denotes the anyon formed by pairing $\mathfrak{e_{f}}$ (resp. $\mathfrak{m_{f}}$) from each copy. After condensation, all the other anyons are confined, resulting in the desired cSPT phase, denoted $\mathcal{E}$.

To see that $\mathcal{E}$ generates $H^{\phi + 1}_G \paren{\pt; \ZZZ}=\mathbb{Z}_{2}$, we need to check that $\mathcal{E}$ is nontrivial and that $2\mathcal{E}$ (\emph{i.e.}, two copies of $\mathcal{E}$ stacking together) is trivial. First, we notice that any orientation-reversing symmetry (\emph{i.e.}, a reflection, a glide reflection, an inversion, or a rotoinversion) $g\in G$ is enough to protect $\mathcal{E}$ nontrivial. Let us consider an open boundary condition of the model, keeping only the fundamental domains enclosed by the surface shown in ~\ref{fig:sg47_E8_STO}. Then the construction in Figure \ref{fig:sg47_E8} leaves a surface of the $\mathfrak{e_{f}m_{f}}$ topological order respecting $g$. If the bulk cSPT phase is trivial, then the surface topological order can be disentangled in a symmetric way from the bulk by local unitary gates with a finite depth. However, this strictly 2D system of the $\mathfrak{e_{f}m_{f}}$ topological order is chiral, wherein the orientation-reversing symmetry $g$ has to be violated. This contradiction shows that the bulk is nontrivial by the protectionof $g$.

To better understand the incompatibility of a strictly 2D system of the $\mathfrak{e_{f}m_{f}}$ topological order with any orientation-reversing symmetry $g$, we view the 2D system as a gluing result of two regions related by $g$ as in Figure \ref{fig:sg47_E8_STO}. The $\mathfrak{e_{f}m_{f}}$ topological order implies net chiral edge modes for each region. Further, $g$ requires that edge modes from the two region propagate in the same direction at their 1D interface. Thus, a gapped gluing is impossible, which shows the non-existence of the $\mathfrak{e_{f}m_{f}}$ topological order compatible with $g$ in a strictly 2D system. Therefore, the surface $\mathfrak{e_{f}m_{f}}$ topological order respecting $g$ proves the non-triviality of the bulk cSPT phase $\mathcal{E}$.

On the other hand,  $2\mathcal{E}$ is equivalent to the model obtained by attaching an $E_{8}$ state to the surface of each fundamental domain from inside in a symmetric way. To check this, let us trace back the construction of $2\mathcal{E}$: we start with attaching two copies of $\mathfrak{e_{f}m_{f}}$ states on the surface of each fundamental domain from inside. Then let us focus a single rectangle interface between two cuboids. There are four copies of $\mathfrak{e_{f}m_{f}}$ states along it, labeled by $i=1,2$ for one side and $i=3,4$ for the other side. The symmetries relate $i=1\leftrightarrow i=4$ and $i=2\leftrightarrow i=3$ separately. If we further condense $(\mathfrak{e}_{\mathfrak{f}}^{1},\mathfrak{e}_{\mathfrak{f}}^{4})$, $(\mathfrak{m}_{\mathfrak{f}}^{1},\mathfrak{m}_{\mathfrak{f}}^{4})$, $(\mathfrak{e}_{\mathfrak{f}}^{2},\mathfrak{e}_{\mathfrak{f}}^{3})$, and $(\mathfrak{m}_{\mathfrak{f}}^{2},\mathfrak{m}_{\mathfrak{f}}^{3})$, then we get the local state of $2\mathcal{E}$ near the rectangle. Here $\mathfrak{e}_{\mathfrak{f}}^{i}$ (resp. $\mathfrak{m}_{\mathfrak{f}}^{i}$) denotes the $\mathfrak{e_{f}}$ (resp. $\mathfrak{m_{f}}$) particle from the $i^{th}$ copy of $\mathfrak{e_{f}m_{f}}$ state and $(\mathfrak{e}_{\mathfrak{f}}^{i},\mathfrak{e}_{\mathfrak{f}}^{j})$ (resp. $(\mathfrak{m}_{\mathfrak{f}}^{i},\mathfrak{m}_{\mathfrak{f}}^{j})$) is the anyon obtained by pairing $\mathfrak{e}_{\mathfrak{f}}^{i}$, $\mathfrak{e}_{\mathfrak{f}}^{j}$ (resp. $\mathfrak{m}_{\mathfrak{f}}^{i}$, $\mathfrak{m}_{\mathfrak{f}}^{j}$). However, we may alternately condense $(\mathfrak{e}_{\mathfrak{f}}^{1},\mathfrak{e}_{\mathfrak{f}}^{2})$, $(\mathfrak{m}_{\mathfrak{f}}^{1},\mathfrak{m}_{\mathfrak{f}}^{2})$, $(\mathfrak{e}_{\mathfrak{f}}^{3},\mathfrak{e}_{\mathfrak{f}}^{4})$, and $(\mathfrak{m}_{\mathfrak{f}}^{3},\mathfrak{m}_{\mathfrak{f}}^{4})$. The resulting local state is two copies of $E_{8}$ states connecting the same environment in a symmetric gapped way. Thus, the two local states produced by different condensation procedures have the same edge modes with the same symmetry behavior and hence are equivalent. Therefore, $2\mathcal{E}$ is equivalent to the model constructed by attaching an $E_{8}$ state to the surface of each fundamental domain from inside. Since the later can be obtained by blowing an $E_{8}$ state bubble inside each fundamental domain, it (and hence $2\mathcal{E}$) is clearly trivial.

Thus, we have shown that $\mathcal{E}$ is nontrivial and that $2\mathcal{E}$ is trivial. Therefore, the cSPT phases generated by $\mathcal{E}$ have the $\mathbb{Z}_{2}$ group structure, which holds for any non-orientation-preserving subgroup of $Pmmm$ such as $Pm$ and $Pmm2$.

In particular, for $Pm$, the model constructed in Figure \ref{fig:layerE8_y2} actually presents the same cSPT phase as $\mathcal{E}$ constructed in Figure \ref{fig:sg47_E8}. To see this, we could blow an $E_{8}$ state bubble inside cuboids centered at $\frac{1}{4}\left(1,\pm1,-1\right)+\mathbb{Z}^{3}$ and $\frac{1}{4}\left(-1,\pm1,1\right)+\mathbb{Z}^{3}$, with chirality opposite to those indicated by the arrowed arcs shown on the corresponding cuboids in Figure \ref{fig:sg47_E8}. This relates $\mathcal{E}$ to alternately layered $E_{8}$ states; however, each reflection does not acts trivially on the resulting $E_{8}$ layer on its mirror as the model in Figure \ref{fig:layerE8_y2}. To further show their equivalence, we look at a single $E_{8}$ layer at $y=0$ for instance. Since it may be obtained by condensing $(\mathfrak{e_{f}},\mathfrak{e_{f}})$ and $(\mathfrak{m_{f}},\mathfrak{m_{f}}$) in a pair of $\mathfrak{e_{f}m_{f}}$ topological states attached to the mirror, it hence can connect the $\mathfrak{e_{f}m_{f}}$ surface shown in Figure \ref{fig:sg47_E8_STO} in a gapped way with the reflection $m_{y}$ respected. On the other hand, we know that an $E_{8}$ state put at $y=0$ with trivial $m_{y}$ action can also connect to this surface in gapped $m_{y}$-symmetric way \cite{Hermele_torsor} and is thus equivalent to the corresponding $E_{8}$ layer just mentioned with nontrivial $m_{y}$ action. As a result, the models constructed in Figures \ref{fig:layerE8_y2} and \ref{fig:sg47_E8} realize the same cSPT phase with $Pm$ symmetry.

Let $\SPT^3\paren{G, \phi}$ be the set of cSPT phases with space group $G$ symmetry in $d=3$ spatial dimension. The stacking operation equips $\SPT^3\paren{G, \phi}$ with an Abelian group structure. Below, let us prove that the dimensional reduction procedure defines a group homomorphism
\begin{equation}
\mathfrak{D}:\SPT^3\paren{G, \phi}\rightarrow H^{\phi + 1}_G \paren{\pt; \ZZZ}.\label{eq:D}
\end{equation}
To check the well-definedness of $\mathfrak{D}$, we notice the following two facts. 

\begin{prp}
Let $\mathcal{M}$ be a model on $X_{2}$ with $\nu^{1}\left(t_{v_{1}}\right)$ specified by its $E_{8}$ configuration, where $t_{v_{1}}$ is the translation symmetry along a vector $v_{1}\in\mathbb{R}^{3}$. Compactifying $\mathcal{M}$ such that $t_{v_{1}}^{L}=1$ results in a 2D system with the invertible topologicial order of $L\nu^{1}\left(t_{v_{1}}\right)$ copies of $E_{8}$ states.
\end{prp}

\begin{proof}
Let us only keep the subgroup $H\subseteq G$ of symmetries generated by three translations $t_{v_{1}}$, $t_{v_{2}}$, and $t_{v_{3}}$ along linearly independent vectors $v_{1}$, $v_{2}$, and $v_{3}$ respectively. For convenience, we now use the coordinate system such that $t_{v_{1}}$, $t_{v_{2}}$, $t_{v_{3}}$ work as $t_{x}$, $t_{y}$, $t_{z}$ in Eqs.~(\ref{eq:tx}-\ref{eq:tz}). Let $\mathcal{F}'=\left[0,1\right]\times\left[0,1\right]\times\left[0,1\right]$. Then $\mathcal{F}'$ is both a fundamental domain and a unit cell for $H$. The original triangulation $X$ partitions $\mathcal{F}'$ into convex polyhedrons, which can be further triangulated resulting in a finer $H$-simplex complex $X'$. Clearly, $q_{3}$ can be viewed as an element of $C_{3}\left(X'\right)$ as well. Moreover, putting $\nu^{1}\left(t_{v_{1}}\right)$ (resp. $\nu^{1}\left(t_{v_{2}}\right)$, $\nu^{1}\left(t_{v_{3}}\right)$) copies of $\overline{E_{8}}$ states on each of integer $x$ (resp. $y$, $z$) planes produces a $H$-symmetric model $\mathcal{M}'$ with $E_{8}$ configuration encoded by $q_{3}^{\prime}$ satisfying $t_{v_{j}}q_{3}^{\prime}=q_{3}^{\prime}+\nu^{1}\left(t_{v_{j}}\right)X,\forall j=1,2,3$. Then $q_{3}-q_{3}^{\prime}$ is invariant under the action of $H$. Thus, the original model $\mathcal{M}$ can be continuously changed into $\mathcal{M}'$ through a path of $H$-symmetric SRE states. In particular, it reduces to a 2D system with the invertible topologicial order of $L\nu^{1}\left(t_{v_{1}}\right)$ copies of $E_{8}$ states, when the system is compactified such that $t_{v_{1}}^{L}=1$.
\end{proof}

\begin{rmk}
This compactification procedure provides an alternate interpretation of $\nu^{1}\left(t_{v_{1}}\right)$, which is now clearly independent of the dimensional reduction details and invariant under a continuous change of cSPT states. Thus, $\nu^{1}\left(t_{v_{1}}\right)$ is well-defined for a cSPT phase.
\end{rmk}

\begin{prp}
Suppose that $G$ contains an orientation-reversing symmetry $r$. Given two $G$-symmetric models $\mathcal{M}_{a}$ and $\mathcal{M}_{b}$ on $G$-simplicial complex structures $X_{a}$ and $X_{b}$ of $\mathbb{E}^{3}$ respectively, let $\nu_{a}^{1}\left(r\right)\pmod2$ and $\nu_{b}^{1}\left(r\right)\pmod2$ be specified by their $E_{8}$ configurations. If $\mathcal{M}_{a}$ and $\mathcal{M}_{b}$ are in the same $G$-SPT phase, then $\nu_{a}^{1}\left(r\right)=\nu_{b}^{1}\left(r\right)\pmod2$.
\end{prp}

\begin{proof}
The triangulation of $X_{a}$ partitions each tetrahedron of $X_{b}$ into convex polyhedrons, which can be further triangulated resulting a simplicial complex $X$ finer than both $X_{a}$ and $X_{b}$. Naturally, both $\mathcal{M}_{a}$ and $\mathcal{M}_{b}$ can be viewed as models on $X$ for the convenience of making a comparison; the values of  $\left[\nu_{a}^{1}\right]$ and $\left[\nu_{b}^{1}\right]$ are clearly invariant when computed on a finer triangulation.
	
To make a proof by contradiction, suppose $\nu_{a}^{1}\left(r\right)\neq\nu_{b}^{1}\left(r\right)\pmod2$. Without loss of generality, we may assume $\nu_{a}^{1}\left(r\right)\pmod2=1$ and $\nu_{b}^{1}\left(r\right)\pmod2=0$. To compare the difference between $\mathcal{M}_{a}$ and $\mathcal{M}_{b}$, let us construct an inverse of $\mathcal{M}_{b}$ with symmetry $r$ respected. Since $\nu_{b}^{1}\left(r\right)\pmod2=0$, the $E_{8}$ configuration of $\mathcal{M}_{b}$ can be described by $\partial q_{3}^{b}$ with $q_{3}^{b}\in C_{3}\left(X\right)$ satisfying $rq_{3}^{b}=q_{3}^{b}$. Let $\overline{\mathcal{M}_{b}^{\prime}}$ be a model obtained by attaching a bubble of $-q_{3}^{b}\left(\varsigma\right)$ copies of $E_{8}$ states to $\partial\varsigma,\forall\varsigma\in\Delta_{3}\left(X\right)$ from its inside. Then adding $\overline{\mathcal{M}_{b}^{\prime}}$ to $\mathcal{M}_{b}$ cancels its $E_{8}$ configuration; the resulting state may only have group cohomology SPT phases left on simplices and hence has an inerse, denoted $\overline{\mathcal{M}_{b}^{\prime\prime}}$. Let $\overline{\mathcal{M}_{b}}=\overline{\mathcal{M}_{b}^{\prime}}\oplus\overline{\mathcal{M}_{b}^{\prime\prime}}$ (\emph{i.e.}, the stacking of $\overline{\mathcal{M}_{b}^{\prime}}$ and $\overline{\mathcal{M}_{b}^{\prime\prime}}$). Then $\overline{\mathcal{M}_{b}}$ clearly is an inverse of $\mathcal{M}_{b}$ with symmetry $r$ respected. Since $\mathcal{M}_{a}$ and $\mathcal{M}_{b}$ are in the same $G$-SPT phase, $\overline{\mathcal{M}_{b}}$ is also an inverse of $\mathcal{M}_{a}$ with symmetry $r$ respected. Thus, $\mathcal{M}_{a}\oplus\overline{\mathcal{M}_{b}}$ admits a $r$-symmetric surface.  Moreover, since $\nu_{a}^{1}\left(r\right)\pmod2=1$, the $E_{8}$ configuration of $\mathcal{M}_{a}\oplus\overline{\mathcal{M}_{b}}$ can be described by $\partial q_{3}$ for some $q_{3}\in C_{3}\left(X\right)$ satisfying $rq_{3}=q_{3}-X$.

To proceed, we pick an $r$-symmetric region $Y$ with surface shown in Figure \ref{fig:sg47_E8_STO}. For convenience, $Y$ is chosen to be a subcomplex of $X$ (\emph{i.e.}, union of simplices in $X$). In addition, $r$ can be an inversion, a rotoinversion, a reflection, and a glide reflection in three dimensions. Let $\Pi$ be a plane passing the inversion/rotoinversion center, the mirror plane, and the glide reflection plane respectively. Clearly, $r\Pi=\Pi$. Adding $\Pi$ to the triangulation of $Y$, some tetrahedrons are divided in convex polyhedrons, which can be further triangulated resulting in a finer triangulation of $Y$. Below, we treat $Y$ as a simplicial complex specified by the new triangulation. By restriction, $q_{3}$ can be viewed as an element of $C_{3}\left(Y\right)$ satisfying $rq_{3}=q_{3}-Y$. Let $\mathcal{M}$ be a model on $Y$ with an $r$-symmetric SRE surface and a bulk identical to $\mathcal{M}_{a}\oplus\overline{\mathcal{M}_{b}}$. Thus, $\partial q_{3}$ describes the bulk $E_{8}$ configuration of $\mathcal{M}$.

On the other hand, putting an $r$-symmetric $E_{8}$ state on $\Pi$ gives a configuration described by $\partial q_{3}^{\prime}$ with $q_{3}^{\prime}\left(\varsigma\right)$ equal to $1$ for all 3-simplex $\varsigma$ on side of $\Pi$ and $0$ on the other side; accordingly, $rq_{3}^{\prime}=q_{3}^{\prime}-Y$ and hence $r\left(q_{3}^{\prime}-q_{3}\right)=q_{3}^{\prime}-q_{3}$. Thus, adding a bubble of $q_{3}^{\prime}\left(\varsigma\right)-q_{3}\left(\varsigma\right)$ copies of $E_{8}$ states to $\partial\varsigma,\forall\varsigma\in\Delta_{3}\left(Y\right)$ from inside gives an $r$-symmetric SRE model, denoted $\mathcal{M}'$. By construction, the bulk $E_{8}$ configuration of $\mathcal{M}'$ get concentrated on $\Pi$; more precisely, $\mathcal{M}'$ can be viewed as a gluing result of the two half surfaces (red and blue) in Figure \ref{fig:sg47_E8_STO} and an $E_{8}$ state on $\Pi$. As $\partial Y$ hosts no anyons, each half surface has $8n$ co-propagating chiral boson modes along its boundary, where $n\in\mathbb{Z}$. Due to the symmetry $r$, they add up to $16n$ co-propagating chiral boson modes, which cannot be canceled by boundary modes of the $E_{8}$ state on $\Pi$. This implies that $\mathcal{M}'$ cannot be gapped, which contradicts that $E_{8}$ is SRE and hence disproves our initial assumption. Therefore, $\nu_{a}^{1}\left(r\right)=\nu_{b}^{1}\left(r\right)\pmod2$ for any two models $\mathcal{M}_{a}$ and $\mathcal{M}_{b}$ in the same $G$-symmetric cSPT phase.
\end{proof}

By Lemma 1 in App.\,A of work \cite{SongXiongHuang}, $\left[\nu^{1}\right]\in H^{\phi + 1}_G \paren{\pt; \ZZZ}$ is specified by $\nu^{1}\left(t_{v_{1}}\right)$, $\nu^{1}\left(t_{v_{2}}\right)$, $\nu^{1}\left(t_{v_{3}}\right)$ together with $\nu^{1}\left(r\right)\pmod2\in\mathbb{Z}_{2}$ if $G$ is non-orientation-preserving, where $t_{v_{1}}$, $t_{v_{2}}$, $t_{v_{3}}$ are three linearly independent translation symmetries and $r$ is an orientation-reversing symmetry. Combining the above two facts, we get that any two models (probably on different simplicial complex structures of $\mathbb{E}^{3}$) in the same $G$-symmetric cSPT phase must determine a unique $\left[\nu^{1}\right]\in H^{\phi + 1}_G \paren{\pt; \ZZZ}$. Thus, $\mathfrak{D}$ in Eq.~(\ref{eq:D}) is well-defined, independent of the details of the dimensional reduction. Moreover, it clearly respects the group structure.

\subsubsection{Construction of $H^{\phi + 1}_G \paren{\pt; \ZZZ}$ cSPT phases\label{subsubsec:Construction}}

To show that the group structure of $\SPT^3\paren{G, \phi}$ (\emph{i.e.}, $G$-SPT phases in $d$ spatial dimensions) is $H^{\phi + 5}_G \paren{\pt; \ZZZ}\oplus H^{\phi + 1}_G \paren{\pt; \ZZZ}$, we analyze three group homomorphisms $\mathfrak{I}$, $\mathfrak{C}$, and $\mathfrak{D}$, which can be organized as
\begin{equation} 
\begin{tikzcd}
H^{\phi + 5}_G \paren{\pt; \ZZZ}   \arrow[r, "\mathfrak{I}", hook] & 
\SPT^3\paren{G, \phi}	 \arrow[r, "\mathfrak{D}", twoheadrightarrow] &
H^{\phi + 1}_G \paren{\pt; \ZZZ}. \arrow[l, "\mathfrak{C}", bend right, swap]
\end{tikzcd}
\label{eq:extension}
\end{equation}
A hooked (resp. two-head) arrow is used to indicate that $\mathfrak{I}$ is injective (resp. $\mathfrak{D}$ is surjective). The map $\mathfrak{I}$ is an inclusion identifying $H^{\phi + 5}_G \paren{\pt; \ZZZ}$ as a subgroup of $\SPT^3\paren{G, \phi}$; Refs.~\cite{Huang_dimensional_reduction,ThorngrenElse} have shown that $H^{\phi + 5}_G \paren{\pt; \ZZZ}$ classifies the $G$-symmetric cSPT phases built with lower dimensional group cohomology phases protected site symmetry. We have defined $\mathfrak{D}$ via dimensional reduction. Clearly, $\mathfrak{D}$ maps all $G$-SPT phases labeled by $H^{\phi + 5}_G \paren{\pt; \ZZZ}$ to $0\in H^{\phi + 1}_G \paren{\pt; \ZZZ}$. Conversely, the $E_{8}$ configuration of any model on $X_{2}$ with $\left[\nu^{1}\right]=0$ can be described by $\partial q_{3}$ with $q_{3}\in C_{3}\left(X\right)$ satisfying $gq_{3}=q_{3},\forall g\in G$. Thus, it is possible to grow a bubble of $-q_{3}\left(\varsigma\right)$ of $E_{8}$ states inside $\varsigma$ to its boundary for all $\varsigma\in\Delta_{3}\left(X\right)$ in a $G$-symmetric way, canceling all $E_{8}$ states on 2-simplicies. Thus, any phase with $\left[\nu^{1}\right]=0$ can be represented by a model built with group cohomology phases only. Formally, the image of $\mathfrak{I}$ equals the kernel of $\mathfrak{D}$. Below, we will define the group homomophism $\mathfrak{C}$ by constructing a $G$-symmetric SRE state representing a $G$-SPT phase with each $\left[\nu^{1}\right]\in H^{\phi + 1}_G \paren{\pt; \ZZZ}$ and show that $\mathfrak{D}\circ\mathfrak{C}$ equals the identity map on $H^{\phi + 1}_G \paren{\pt; \ZZZ}$, which implies the surjectivity of $\mathfrak{D}$ and further
\begin{equation}
\SPT^3\paren{G, \phi}\cong H^{\phi + 5}_G \paren{\pt; \ZZZ}\oplus H^{\phi + 1}_G \paren{\pt; \ZZZ}
\end{equation}
by the splitting lemma in homological algebra. 

Suppose that $t_{v_{1}}$, $t_{v_{2}}$, and $t_{v_{3}}$ generate the translation subgroup of $G$. Let $G_{0}$ be the orientation-preserving subgroup of $G$, \emph{i.e.}, $G_{0}\coloneqq\left\{ g\in G\;|\;\phi\left(g\right)=1\right\} $. By Lemma 1 in App.\,A of work \cite{SongXiongHuang}, each $[\nu^{1}]\in H^{\phi + 1}_G \paren{\pt; \ZZZ}$ can be parameterized by
\begin{equation}
\nu^{1}\left(t_{v_{1}}\right),\nu^{1}\left(t_{v_{2}}\right),\nu^{1}\left(t_{v_{3}}\right)\in\mathbb{Z}
\end{equation}
together with $\nu^{1}\left(r\right)\pmod2\in\mathbb{Z}_{2}$ if $G$ is non-orientation-preserving (\emph{i.e.}, $G_{0}\neq G$), where $r$ is an orientation-reversing symmetry. Below, we construct models for generators of $H^{\phi + 1}_G \paren{\pt; \ZZZ}$ and then the model corresponding to a generic $[\nu^{1}]$ will be obtained by stacking.

For a non-orientation-preserving space group $G$, let $\nu_{r}^{1}$ be a 1-cocycle satisfying $\nu_{r}^{1}\left(g\right)=0,\forall g\in G_{0}$ and $\nu_{r}^{1}\left(r\right)\pmod2=1$; the corresponding $[\nu_{r}^{1}]\in H^{\phi + 1}_G \paren{\pt; \ZZZ}$ clearly has order 2. Let us first construct a $G$-symmetric SRE state with $[\nu_{r}^{1}]\in H^{\phi + 1}_G \paren{\pt; \ZZZ}$, which is done by recasting the illustrative construction in Figure \ref{fig:sg47_E8} in a general setting. We attach an bubble of $\mathfrak{e_{f}}\mathfrak{m_{f}}$ topological state to $\partial\mathcal{F}$ from inside and duplicate it at $\partial\left(g\mathcal{F}\right)$ by all symmetries $g\in G$. Then there are two copies of $\mathfrak{e_{f}}\mathfrak{m_{f}}$ topological states on each 2-simplex, which is always an interface between two fundamental domains $g_{1}\mathcal{F}$ and $g_{2}\mathcal{F}$ for some $g_{1},g_{2}\in G$. Let $\left(\mathfrak{e_{f}},\mathfrak{e_{f}}\right)$ (resp. $\left(\mathfrak{m_{f}},\mathfrak{m_{f}}\right)$) denote the anyon formed by pairing $\mathfrak{e_{f}}$ (resp. $\mathfrak{m_{f}}$) from each copy. Condensing $\left(\mathfrak{e_{f}},\mathfrak{e_{f}}\right)$ and $\left(\mathfrak{m_{f}},\mathfrak{m_{f}}\right)$ on all 2-simplices results in a $G$-symmetric SRE state, denoted $\left|\Psi_{r}\right\rangle $. Its $E_{8}$ configuration can be described by $\partial q_{3}$ with $q_{3}=\sum_{g\in G_{0}}g\mathcal{F}$; there is a (resp. no) $E_{8}$ state on the interface between $g_{1}\mathcal{F}$ and $g_{2}\mathcal{F}$ with $\phi\left(g_{1}\right)\neq\phi\left(g_{2}\right)$ (resp. $\phi\left(g_{1}\right)=\phi\left(g_{2}\right)$). Clearly, $rq_{3}=q_{3}-X$ and hence $\nu^{1}\left(r\right)\pmod2=1$. In addition, with only $G_{0}$ respected, the $E_{8}$ configuration can be canceled by growing a bubble of an $\overline{E_{8}}$ state inside each $g\mathcal{F}$ for $g\in G_{0}$. Thus, as a $G_{0}$-symmetric cSPT phase, $\left|\Psi_{r}\right\rangle $ corresponds to $0\in H^{\phi + 1}_{G_0} \paren{\pt; \ZZZ}$; in particular, $\nu^{1}\left(t_{v}\right)=0$ for any translation $t_{v}$. Therefore, $\left|\Psi_{r}\right\rangle $ is a $G$-symmetric SRE state mapped to $[\nu_{r}^{1}]\in H^{\phi + 1}_G \paren{\pt; \ZZZ}$ by $\mathfrak{D}$. Moreover, the order of $\left|\Psi_{r}\right\rangle $ is also 2 in $\SPT^3\paren{G, \phi}$ by the same argument as the one in Sec.~\ref{subsubsec:cSPT_Pmmm} showing that stacking two copies of the models in Figure \ref{fig:sg47_E8} gives a trivial cSPT phase.

Below, the construction of $\mathfrak{C}$ will be completed case by case based on the number of symmetry directions in the international (Hermann-Mauguin) symbol of $G$.

\paragraph{$G$ with more than one symmetry direction}

For $G$ with more than one symmetry direction, if $G$ is orientation-preserving, then $H^{\phi + 1}_G \paren{\pt; \ZZZ}=0$ and hence the definition of $\mathfrak{C}$ is obvious: $\mathfrak{C}$ maps $0\in H^{\phi + 1}_G \paren{\pt; \ZZZ}$ to the trivial phase in $G\text{-SPT}$. Clearly, $\mathfrak{D}\circ\mathfrak{C}$ is the identity.

On the other hand, if $G$ does not preserve the orientation of $\mathbb{E}^{3}$, then $H^{\phi + 1}_G \paren{\pt; \ZZZ}=\mathbb{Z}_{2}$ with $[\nu_{r}^{1}]$ the only nontrivial element. Since the order of $\left|\Psi_{r}\right\rangle $ is 2 in $\SPT^d\paren{G, \phi}$, the group homomorphism $\mathfrak{C}$ can be specified by mapping $[\nu_{r}^{1}]$ to the cSPT phase presented by $\left|\Psi_{r}\right\rangle $. By construction, $\mathfrak{D}\circ\mathfrak{C}$ equals the identity on $H^{\phi + 1}_G \paren{\pt; \ZZZ}$.

\paragraph{$G=P1$ and $P\overline{1}$}

For $G=P1$ and $P\overline{1}$, we pick a coordinate system such that $t_{v_{1}}$, $t_{v_{2}}$, $t_{v_{3}}$ work as $t_{x}$, $t_{y}$, $t_{z}$ in Eqs.~(\ref{eq:tx}-\ref{eq:tz}) (and such that the origin is an inversion center\footnote{Every orientation-reversing symmetry of $P\overline{1}$ is an inversion.} of $r$ for $P\overline{1}$). Let $[\nu_{i}^{1}]$ be an element of $H^{\phi + 1}_G \paren{\pt; \ZZZ}$ presented by a 1-cocycle satisfying $\nu_{i}^{1}\left(t_{v_{j}}\right)=\delta_{ij}$ and $\nu_{i}^{1}\left(r\right)\pmod2=1$ for $i=1,2,3$. As an example, a $G$-symmetric SRE state $\left|\Psi_{2}\right\rangle $ for $[\nu_{2}^{1}]$ can be constructed by putting an $E_{8}$ state $\left|E_{8}^{y}\right\rangle $ on the plane $y=0$ with the symmetries $t_{v_{1}},t_{v_{3}}$ (as well as the inversion $r:\left(x,y,z\right)\mapsto-\left(x,y,z\right)$ for $P\overline{1}$) respected and its translation image $t_{v_{2}}^{n}\left|E_{8}^{y}\right\rangle $ on planes $y=n$ for $n\in\mathbb{Z}$. In particular, since there is a single $E_{8}$ layer passing the inversion center of $r$ in the case $G=P\overline{1}$, the $E_{8}$ configuration of $\left|\Psi_{2}\right\rangle $ determines $\nu_{2}^{1}\left(r\right)\pmod2=1$. A $G$-symmetric SRE state $\left|\overline{\Psi}_{2}\right\rangle $, inverse to $\left|\Psi_{2}\right\rangle $ in $\SPT^3\paren{G, \phi}$, can be obtained by replacing each $E_{8}$ layer by its inverse. Clearly, $\left|\overline{\Psi}_{2}\right\rangle $ is mapped to $-[\nu_{y}^{1}]$ by $\mathfrak{D}$. Analogously, we construct a $G$-symmetric SRE state $\left|\Psi_{i}\right\rangle $ for $[\nu_{i}^{1}]$ and its inverse $\left|\overline{\Psi}_{i}\right\rangle $ for $i=1,3$ as well.

Noticing that each $[\nu^{1}]\in H^{\phi + 1}_G \paren{\pt; \ZZZ}$ can uniquely be expressed as $\sum_{i=1}^{3}n_{i}[\nu_{i}^{1}]$ for $G=P1$ and $n_{r}[\nu_{r}^{1}]+\sum_{i=1}^{3}n_{i}[\nu_{i}^{1}]$ for $G=P\overline{1}$ with $n_{r}=0,1$ and $n_{1},n_{2},n_{3}\in\mathbb{Z}$, we can define $\mathfrak{C}\left([\nu^{1}]\right)$ as the phase presented by the $G$-symmetric SRE state $\left|\Psi_{r}\right\rangle ^{\otimes n_{r}}\otimes\left|\Psi_{1}\right\rangle ^{\otimes n_{1}}\otimes\left|\Psi_{2}\right\rangle ^{\otimes n_{2}}\otimes\left|\Psi_{3}\right\rangle ^{\otimes n_{3}}$, where $\left|\Psi_{r}\right\rangle ^{\otimes n_{r}}$ is needed only for $G=P\overline{1}$ and $\left|\Psi\right\rangle ^{\otimes n}$ denotes the stacking of $\left|n\right|$ copies of $\left|\Psi\right\rangle $ (resp. its inverse $\left|\overline{\Psi}\right\rangle $) for $n\geq0$ (resp. $n<0$). In particular, $\left|\Psi\right\rangle ^{\otimes0}$ denotes any trivial state. By construction, $\mathfrak{C}$ is a group homomorphism and $\mathfrak{D}\circ\mathfrak{C}$ is the identity on $H^{\phi + 1}_G \paren{\pt; \ZZZ}$.

\paragraph{G with exactly one symmetry diretion other than $1$ or $\overline{1}$
	\label{subsubsec:Construction1}}

For $G$ with exactly one symmetry direction other than $1$ or $\overline{1}$ (e.g. $I\overline{4}$, $P2$, $P2_{1}/c$, $Pm$), we pick a Cartesian coordinate system whose $z$ axis lies along the symmetry direction. With each point represented by a column vector of its coordinates $u=\left(x,y,z\right)^{\mathsf{T}}\in\mathbb{R}^{3}$, each $g\in G$ is represented as $u\mapsto\phi\left(g\right)R_{z}\left(\varphi\right)u+w$ with $w=\left(w^{x},w^{y},w^{z}\right)^{\mathsf{T}}\in\mathbb{R}^{3}$ and an orthogonal matrix
\begin{equation}
R_{z}\left(\varphi\right)\coloneqq\left(\begin{array}{ccc}
\cos\varphi & -\sin\varphi & 0\\
\sin\varphi & \cos\varphi & 0\\
0 & 0 & 1
\end{array}\right)
\end{equation}
describing a rotation about the $z$ axis. Clearly, $w^{z}$ is independent of the origin position for $g\in G_{0}$, where $G_{0}$ is the orientation-preserving subgroup of $G$. Let $\mathbb{A}$ be the collection of $w^{z}$ of all orientation-preserving symmetries. Then $\mathbb{A}=\kappa\mathbb{Z}$ for some positive real number $\kappa$. Pick $h\in G_{0}$ with $w^{z}=\kappa$.

If $G_{0}\neq G$, then $r\in G-G_{0}$ may be an inversion, a rotoinversion, a reflection, or a glide reflection. Clearly, there is a plane $\Pi$ perpendicular to the $z$ direction satisfying $r\Pi=\Pi$. If $G_{0}=G$, we just pick $\Pi$ to be any plane perpendicular to the $z$ direction. For convenience, we may choose the coordinate origin on $\Pi$ such that $w^{z}=0$ for $r$. In such a coordinate system,  $w^{z}\in\kappa\mathbb{Z}$ for all $g\in G$ even if $\phi\left(g\right)=-1$. Moreover, $G_{\Pi}\coloneqq\left\{ g\in G|g\Pi=\Pi\right\} $ contains symmetries with $w^{z}=0$. It is a wallpaper group for $\Pi$, which preserves the orientation of $\Pi$.

Putting a $G_{\Pi}$-symmetric $E_{8}$ layer $\left|E_{8}^{z}\right\rangle $ on $\Pi$ and its duplicate $h^{n}\left|E_{8}^{z}\right\rangle $ on $h^{n}\Pi$, we get an $h$-symmetric SRE state $\left|\Psi_{z}\right\rangle \coloneqq\cdots\otimes h^{-1}\left|E_{8}^{z}\right\rangle \otimes\left|E_{8}^{z}\right\rangle \otimes h\left|E_{8}^{z}\right\rangle \otimes\cdots$. For all $f\in G_{\Pi}$, $fh^{n}\left|E_{8}^{z}\right\rangle =h^{\phi\left(f\right)n}f_{n}\left|E_{8}^{z}\right\rangle =h^{\phi\left(f\right)n}\left|E_{8}^{z}\right\rangle $, where $f_{n}=h^{-\phi\left(f\right)n}gh^{n}$ is an element of $G_{\Pi}$ and hence leaves $\left|E_{8}^{z}\right\rangle $ invariant. Thus, $\left|\Psi_{z}\right\rangle $ is $G_{\Pi}$-symmetric and hence $G$-symmetric, as each $g=G$ can be expressed as $h^{m}f$ with $m\in\mathbb{Z}$ and $f\in G_{\Pi}$. A $G$-symmetric SRE state $\left|\overline{\Psi}_{z}\right\rangle $, inverse to $\left|\Psi_{z}\right\rangle $, can be obtained by replacing each $E_{8}$ by its inverse.

Let $[\nu_{z}^{1}]\in H^{\phi + 1}_G \paren{\pt; \ZZZ}$ (presented by a 1-cocyle $\nu_{z}^{1}$) be the image of $\left|\Psi_{z}\right\rangle $ under $\mathfrak{D}$. By construction, $\nu_{z}^{1}\left(g\right)=w^{z}/\kappa$ on each orientation-preserving symmetry $g:u\mapsto R_{z}\left(\varphi\right)u+\left(w^{x},w^{y},w^{z}\right)$. When $G_{0}\neq G$, a single $E_{8}$ layer on $\Pi$ implies the nonexistence of $r$-symmetric $q_{3}$ with $\partial q_{3}$ describing the $E_{8}$ configuration of $\left|E_{8}^{z}\right\rangle $ and hence $\nu_{z}^{1}\left(r\right)\pmod2=1$. Clearly, $\left|\overline{\Psi}_{z}\right\rangle $ is mapped to $-[\nu_{z}^{1}]$ by $\mathfrak{D}$.

Further, we notice that every $[\nu^{1}]\in H^{\phi + 1}_G \paren{\pt; \ZZZ}$ can be uniquely expressed as $n_{h}[\nu_{z}^{1}]$ (resp. $n_{r}[\nu_{r}^{1}]+n_{h}[\nu_{z}^{1}]$) when $G_{0}=G$ (resp. $G_{0}\neq G$); comparing values on $r$ and $h$ gives $n_{h}=\nu^{1}\left(h\right)\in\mathbb{Z}$ and $n_{r}=\nu^{1}\left(r\right)+\nu^{1}\left(h\right)\pmod2\in\left\{ 0,1\right\} $. Thus, a map $\mathfrak{C}$ can be defined by $n_{h}[\nu_{z}^{1}]\mapsto\left|\Psi_{z}\right\rangle ^{\otimes n_{h}}$ (resp. $n_{r}[\nu_{r}^{1}]+n_{h}[\nu_{z}^{1}]\mapsto\left|\Psi_{r}\right\rangle ^{\otimes n_{r}}\otimes\left|\Psi_{z}\right\rangle ^{\otimes n_{h}}$), where $\left|\Psi\right\rangle ^{\otimes n}$ denotes the stacking of $\left|n\right|$ copies of $\left|\Psi\right\rangle $ (resp. its inverse $\left|\overline{\Psi}\right\rangle $) for $n\geq0$ (resp. $n<0$). In particular, $\left|\Psi\right\rangle ^{\otimes0}$ denotes any trivial state. Since $\left|\Psi_{r}\right\rangle $ has order 2 in $\SPT^3\paren{G, \phi}$, $\mathfrak{C}$ is a group homomorphism. By construction, $\mathfrak{D}\circ\mathfrak{C}$ is clearly the identity on $H^{\phi + 1}_G \paren{\pt; \ZZZ}$.

\chapter{More advanced techniques}
\label{chap:advanced}

In this chapter, we will present some more advanced applications of the minimalist framework that we developed in Chapter \ref{chap:minimalist}:
\begin{enumerate}
\item In Sec.\,\ref{sec:Mayer-Vietoris}, we will derive a general relation between 
the classification of SPT phases with reflection symmetry and the classification of SPT phases without reflection symmetry. It will be shown, for any $G$, that
\begin{equation}
\SPT^d\paren{G \times \paren{\ZZZ \rtimes \ZZZ_2^R}, \phi}, ~~ \SPT^d\paren{G \times \ZZZ_2^R, \phi}, ~~ \SPT^d\paren{G, \phi},
\end{equation}
fit into a long exact sequence, where $\ZZZ$ is generated by a translation (mapped to $1$ under $\phi$) and $\ZZZ_2^R$ is generated by a reflection (mapped to $-1$ under $\phi$). This follows from the Mayer-Vietoris sequence of any generalized cohomology theory.

\item In Sec.\,\ref{sec:Atiyah-Hirzebruch}, we will show that the structure of SPT phases with crystalline symmetries can in general be understood in terms of an Atiyah-Hirzebruch spectral sequence. We will show that different terms on the $E^1$-page of the spectral sequence correspond to SPT phases of different dimensions living on cells of different dimensions with effective internal symmetries. We will propose that there exist previously unconsidered higher-order gluing conditions and equivalence relations, which correspond to higher differentials in the spectral sequence.
\end{enumerate}

Before delving into these applications, we will discuss the construction of classifying spaces and the classification of SPT phases for the special case of space-group symmetries in Sec.\,\ref{sec:space_groups}. We will see that it is possible in this case to write the classification as a generalized cohomology theory equivariant with respect to the point group rather than the space group. This will be useful for Secs.\,\ref{sec:Mayer-Vietoris} and \ref{sec:Atiyah-Hirzebruch}.

The results in Sec.\,\ref{sec:space_groups} are based on private communications with Kiyonori Gomi. The results of Sec.\,\ref{sec:Atiyah-Hirzebruch} are based on Secs.\,III-III A and IV-IV E of my work \cite{Shiozaki2018} with Ken Shiozaki and Kiyonori Gomi. The results in Sec.\,\ref{sec:Mayer-Vietoris} appeared in a minimal form in Sec.\,III B of 
work \cite{Shiozaki2018}, and, if were to be published independently from it, would be joint among the author of this dissertation, Ken Shiozaki, and Kiyonori Gomi. Work \cite{Shiozaki2018} used a homological rather than cohomological formulation for the classification of SPT phases. In the homological formulation, orientation-reversing symmetries do not make the generalized homology twisted. We suspect that the two formulations are equivalent by some sort of Poincar\'e duality. In Sec.\,\ref{sec:Mayer-Vietoris}, we will use the cohomological formulation developed in Chapter \ref{chap:minimalist} of this thesis. In Sec.\,\ref{sec:Atiyah-Hirzebruch}, we will use the homological formulation following work \cite{Shiozaki2018}.

\section{Classifying spaces of space groups}
\label{sec:space_groups}

Consider a $d$-dimensional space group $G$ with translational group $\Pi$ and point group $P$. They fit into a short exact sequence,
\begin{equation}
1 \fromto \Pi \fromto G \xfromto{\pi} P \fromto 1.
\end{equation}
Recall that $G$ is a subgroup of $\RRR^d \rtimes O(d)$, $\Pi$ is the subgroup of $G$ of pure translations, and $P$ is the image of $G$ under the projection $\pi: \RRR^d \rtimes O(d) \fromto O(d)$. We can write the elements of $\RRR^d \rtimes O(d)$ and hence elements of $G$ as $(v, p)$, with $v\in \RRR^d$ and $p\in P$. According to Ref.\,\cite{Gomi_comm}, the classifying space of $G$ can be constructed as follows. Choose any lift
\begin{equation}
a: P \fromto \RRR^d \label{lift}
\end{equation}
such that $\paren{ a(p), p } \in G$; this does not have to be a homomorphism. We know that the classifying space $B \Pi$ of $\Pi$ is the $d$-torus $T^d = \RRR^d / \ZZZ^d$. We define an action of $P$ on $B \Pi$ by
\begin{eqnarray}
p.[v] = [a(p) + pv],
\end{eqnarray}
for $v \in \RRR^d$. We have the usual $P$-action on the contractible universal cover $EP$ of the classifying space $BP$ of $P$. Then the classifying space $BG$ of $G$ can be constructed as
\begin{equation}
BG = B\Pi \times_{P} EP, \label{BS_1}
\end{equation}
where $B\Pi \times_P EP$ denotes the quotient space $\paren{B \Pi \times EP} / P$. It is not hard to show that this is the same space as
\begin{equation}
BG = E\Pi \times_G EP, \label{BS_2}
\end{equation}
where $G$ acts on $E \Pi = \RRR^d$ according to the space-group action and on $EP$ by first projecting $G$ to $P$. This fits into a fiber bundle,
\begin{equation}
\begin{tikzcd}
G \arrow{r} & E\Pi \times EP \arrow{d} \\
& E\Pi \times_G EP
\end{tikzcd}, \label{ES}
\end{equation}
which is indeed universal because $EP \times E\Pi$ is contractible. The above construction works for both symmorphic and nonsymmorphic space groups. For symmorphic space groups, $a$ in Eq.\,(\ref{lift}) can be chosen to be trivial, i.e.\,$a(p) = 0$ for all $p\in P$. For nonsymmorphic space groups, $a$ has to be nontrivial.

Incidentally, there is another fiber bundle,
\begin{equation}
\begin{tikzcd}
B\Pi \arrow{r} & B\Pi \times_{P} EP \arrow{d} \\
& BP
\end{tikzcd},
\end{equation}
which separates the point group from the translational group. It is not a universal fiber bundle.

The classification of SPT phases with space-group symmetries has a special structure too. By the generalized cohomology hypothesis, the classification of SPT phases with space-group symmetry $G$ can be written
\begin{equation}
\SPT^d\paren{G, \phi} \isomorphic h^{\phi + d}_G\paren{\pt}, \label{BS_form}
\end{equation}
where
\begin{equation}
\phi: G \fromto \braces{\pm 1}
\end{equation}
is the homomorphism in Eq.\,(\ref{phi_homomorphism}). For simplicity, we consider the case where the internal representation of translations is unitary, i.e.\,$\phi_{\rm int}(v) = 1$ for all $v \in \Pi$. Since translations do not reverse the orientation, i.e.\,$\phi_{\rm spa}(v) = 1$, we then have $\phi(v) = \phi_{\rm int}(v) \phi_{\rm spa}(v) = 1$. By Εq.\,(\ref{h_phi_n_G_pt}), we have
\begin{equation}
\SPT^d\paren{G, \phi} \isomorphic h^{d + 1}\paren{ EG \times_G \tilde I, EG \times_G \partial \tilde I }.
\end{equation}
From fiber bundle (\ref{ES}), we see that a model for $EG$ is $E\Pi \times EP$. Thus 
\begin{equation}
\SPT^d\paren{G, \phi} \isomorphic h^{d + 1}\paren{ \paren{E\Pi \times EP} \times_G \tilde I, \paren{E\Pi \times EP} \times_G \partial \tilde I }.
\end{equation}
Since translations act on $\tilde I$ trivially, by the comment above Eq.\,(\ref{BS_2}), we can rewrite the above as
\begin{eqnarray}
\SPT^d\paren{G, \phi} &\isomorphic& h^{d + 1}\paren{ \paren{B\Pi \times EP} \times_P \tilde I, \paren{B\Pi \times EP} \times_P \partial \tilde I } \nonumber \\
&\isomorphic& h^{d + 1}\paren{ \paren{B\Pi \times \tilde I} \times_P EP, \paren{B\Pi \times \partial \tilde I} \times_P EP },
\end{eqnarray}
By definitions (\ref{equivariant_of_a_pair})(\ref{twisted_equivariant_of_a_pair}), this is exactly
\begin{equation}
\SPT^d\paren{G, \phi} \isomorphic h^{\phi + d}_P \paren{B\Pi}. \label{c+n_form}
\end{equation}
Thus the classification of SPT phases with space-group symmetry $G$ can not only be written as an equivariant theory with respect to $G$ as in Eq.\,(\ref{BS_form}), but also as an equivariant theory with respect to $P$ as in Eq.\,(\ref{c+n_form}).

Next, suppose that $\phi$ is trivial, i.e.\,each $p \in P$ either preserves orientation and has a unitary internal representation, or reverses orientation and has an antiunitary internal representation. In this case,
\begin{eqnarray}
\SPT^d\paren{G} &\isomorphic& h^{d}_G\paren{\pt}, \\
&\isomorphic& h^{d}_P \paren{B\Pi}. \label{c+n_form_trivial_phi}
\end{eqnarray}
By Eq.\,(\ref{equivariant_of_pt_drop_Y}) and Eq.\,(\ref{BS_1}), we can also write explicitly
\begin{equation}
\SPT^d\paren{G} \isomorphic h^d\paren{B\Pi \times_{P} EP}. \label{fiber_bundle_form}
\end{equation}

It can be shown that, for any $G$ and $\phi$, $h^{\phi + \bullet}_G$ and $h^\bullet_G$ satisfy the axioms of generalized cohomology theories just like $h^\bullet$ does \cite{Gomi_comm}. Another noteworthy point is that, because $\RRR^d$ with the action of a space group $G$ is equivariantly contractible, Eq.\,(\ref{BS_form}) can also be written
\begin{equation}
\SPT^d\paren{G, \phi} \isomorphic h^{\phi + d}_G\paren{\RRR^d}. \label{BS_form_Rd}
\end{equation}
In particular, if $\phi$ is trivial, then
\begin{equation}
\SPT^d\paren{G} \isomorphic h^{d}_G\paren{\RRR^d}. \label{BS_form_Rd_trivial_phi}
\end{equation}

\section{Mayer-Vietoris sequence for SPT phases with reflection symmetry}
\label{sec:Mayer-Vietoris}

Consider the space group generated by a translation $x \mapsto x+1$ and a reflection $x \mapsto -x$. As a group, this is $\ZZZ \rtimes \ZZZ_2$. To indicate that the generator of $\ZZZ_2$ is a reflection, we will write $\ZZZ \rtimes \ZZZ_2^R$. We will consider this symmetry not just for 1D systems, but for higher-dimensional systems as well, where the generators of $\ZZZ \rtimes \ZZZ_2^R$ are represented by
\begin{eqnarray}
\paren{x_1, x_2, \ldots, x_d} &\mapsto& \paren{x_1 + 1, x_2, \ldots, x_d}, \\
\paren{x_1, x_2, \ldots, x_d} &\mapsto& \paren{-x_1, x_2, \ldots, x_d}.
\end{eqnarray}
We assume that the internal representation is unitary for both generators. According to Sec.\,\ref{sec:space_groups}, the classification of SPT phases with $\ZZZ \rtimes \ZZZ_2^R$ symmetry can be written
\begin{equation}
\SPT^d\paren{\ZZZ \rtimes \ZZZ_2^R, \phi} \isomorphic h^{\phi + d}_{\ZZZ_2^R}\paren{B \ZZZ} \isomorphic h^d_{\ZZZ_2^R}\paren{\SS^1},
\end{equation}
where we recall the circle $\SS^1$ is a model for $B \ZZZ$. In fact, it is not difficult to show, given any other symmetry group $G$, that
\begin{equation}
\SPT^d\paren{G \times \paren{\ZZZ \rtimes \ZZZ_2^R}, \phi} \isomorphic h^{\phi + d}_{G \times \ZZZ_2^R}\paren{\SS^1}. \label{MV_GZZ2}
\end{equation}
This is true even if $\phi$ is nontrivial on $G$, as long as $G$ commutes with $\ZZZ \rtimes \ZZZ_2^R$. Similarly, we have
\begin{eqnarray}
\SPT^d\paren{G, \phi} &\isomorphic& h^{\phi + d}_{G}\paren{\pt}, \label{MV_G} \\
\SPT^d\paren{G \times \ZZZ_2^R, \phi} &\isomorphic& h^{\phi + d}_{G \times \ZZZ_2^R}\paren{\pt}. \label{MV_GZ2}
\end{eqnarray}
In what follows, we will derive a general relation among the classifications (\ref{MV_GZZ2})(\ref{MV_G})(\ref{MV_GZ2}) of SPT phases with $G \times \paren{\ZZZ \rtimes \ZZZ_2^R}$, $G$, and $G \times \ZZZ_2^R$ symmetries.

\subsection{Mayer-Vietoris sequence for \texorpdfstring{$\ZZZ \rtimes \ZZZ_2^R$}{Z rtimes Z2R} symmetry}
\label{subsec:MVLES}

According to Sec.\,\ref{sec:space_groups}, $h^{\phi + \bullet}_{G \times \ZZZ_2^R}$ satisfies the axioms of generalized cohomology theories. Any generalized cohomology theory $h^\bullet$ admits a Mayer-Vietoris long exact sequence. Namely, given a CW complex $X$ and two subcomplexes $U$ and $V$ such that $U \cup V = X$, there is a long exact sequence
\begin{equation}
\cdots \fromto h^{n-1}\paren{U \cap V} \fromto h^n\paren{X} \fromto h^n\paren{U} \oplus h^n\paren{V} \fromto h^n\paren{U \cap V} \fromto \cdots \label{MV}
\end{equation}

We now apply Eq.\,(\ref{MV}) to $h^{\phi + \bullet}_{G \times \ZZZ_2^R}\paren{\SS^1}$. First, we note that, if we view $\SS^1$ as the unit circle on the complex plane, then the generator of $\ZZZ_2^R$ will act on it by complex conjugation. We consider the decomposition $\SS^1 = U \cup V$ of $\SS^1$ shown in Figure \ref{fig:decompose_circle}. This yields the following long exact sequence:
\begin{equation}
\cdots \fromto h^{\phi + n-1}_{G \times \ZZZ_2^R}\paren{\pt \sqcup \pt} \fromto h^{\phi + n}_{G \times \ZZZ_2^R}\paren{\SS^1} \fromto h^{\phi + n}_{G \times \ZZZ_2^R}\paren{\pt} \oplus h^{\phi + n}_{G \times \ZZZ_2^R} \paren{\pt} \fromto h^{\phi + n}_{G \times \ZZZ_2^R}\paren{\pt \sqcup \pt} \fromto \cdots \label{second_to_last_MV_LES_h}
\end{equation}
In Eq.\,(\ref{second_to_last_MV_LES_h}), we have used the fact that $U$ and $V$ are both equivariantly contractible and that $U \cap V$ deformation retracts equivariantly to the two-point space $\pt \sqcup \pt$. Furthermore, it can be shown, because $\ZZZ_2^R$ permutes the two points in $\pt \sqcup \pt$, that we have
\begin{equation}
h^{\phi + n}_{G \times \ZZZ_2^R}\paren{\pt \sqcup \pt} \isomorphic h^{\phi + n}_{G} \paren{\pt}.
\end{equation}
Thus we can rewrite the long exact sequence (\ref{second_to_last_MV_LES_h}) as
\begin{equation}
\cdots \fromto h^{\phi + n-1}_{G}\paren{\pt} \fromto h^{\phi + n}_{G \times \ZZZ_2^R}\paren{\SS^1} \fromto h^{\phi + n}_{G \times \ZZZ_2^R}\paren{\pt} \oplus h^{\phi + n}_{G \times \ZZZ_2^R} \paren{\pt} \fromto h^{\phi + n}_{G}\paren{\pt} \fromto \cdots \label{last_MV_LES_h}
\end{equation}

\begin{figure}
\centering
\includegraphics[width=3in]{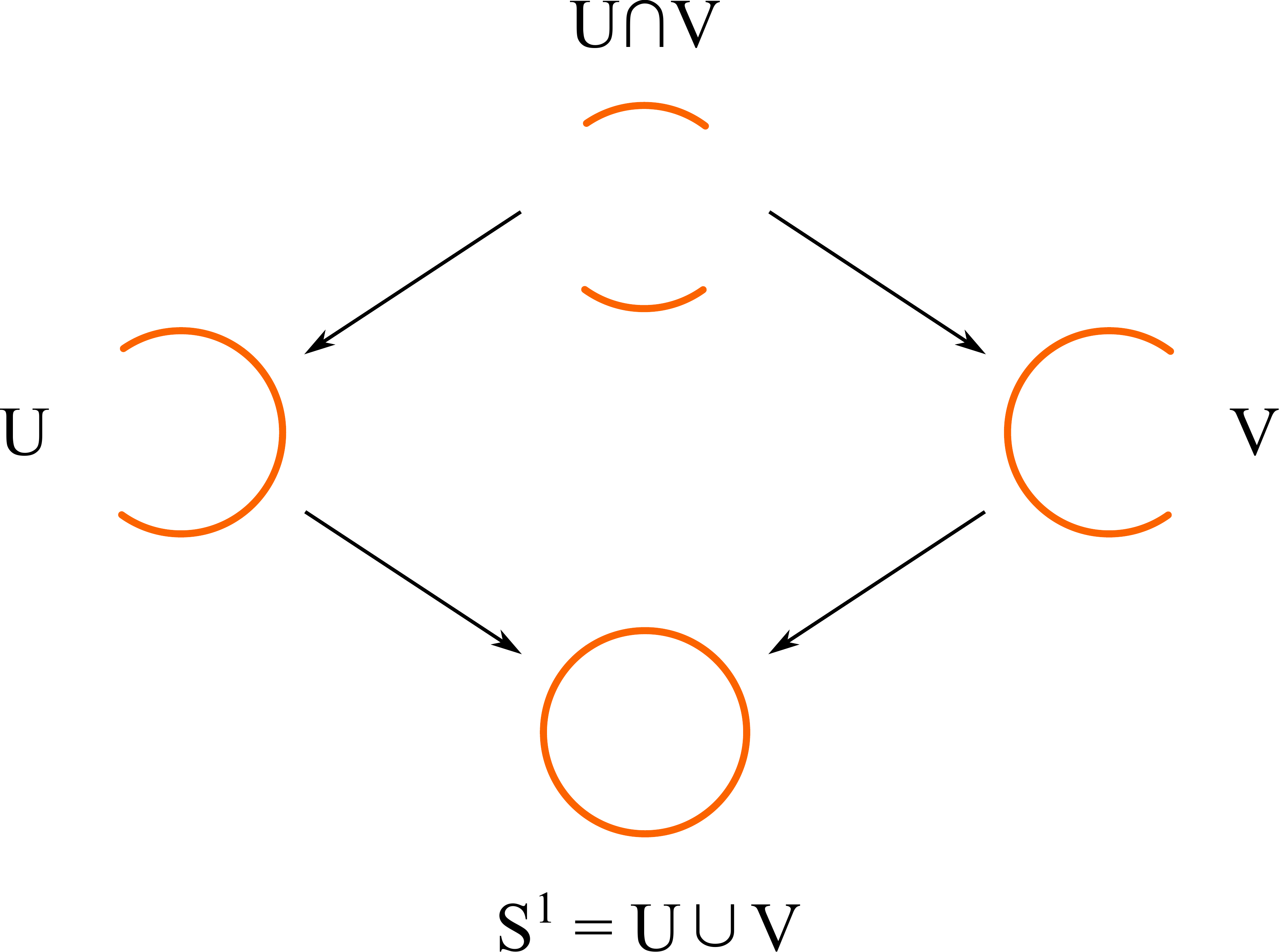}
\caption[A decomposition of the circle $\SS^1 = U \cup V$.]{A decomposition of the circle $\SS^1 = U \cup V$.}
\label{fig:decompose_circle}
\end{figure}

From Eqs.\,(\ref{MV_GZZ2})(\ref{MV_G})(\ref{MV_GZ2}), we see that the terms in the long exact sequence (\ref{last_MV_LES_h}) are precisely the classifications of SPT phases with symmetries $G$, $G \times \paren{\ZZZ \rtimes \ZZZ_2^R}$, $G \times \ZZZ_2^R$, $G$ in one higher dimension, and so forth. Thus we obtain the long exact sequence
\begin{eqnarray}
\cdots ~~ \xfromto{\gamma_1\oplus\gamma_2} ~~ \SPT^{d-1}\paren{G, \phi} &\xfromto{~~\alpha~~}& \SPT^d\paren{G \times\paren{\ZZZ \rtimes \ZZZ_2^R}, \phi} \nonumber\\
&\xfromto{\beta_1\times\beta_2}& \SPT^d\paren{G \times \ZZZ_2^R, \phi} \oplus \SPT^d\paren{G \times \ZZZ_2^R, \phi} \nonumber\\
&\xfromto{\gamma_1\oplus\gamma_2}& \SPT^d\paren{G, \phi} ~~ \xfromto{~~\alpha~~} ~~ \cdots, \label{LES_physical}
\end{eqnarray}
which ties the classifications of SPT phases with $G$, $G \times \paren{\ZZZ \rtimes \ZZZ_2^R}$, and $G \times \ZZZ_2^R$ symmetries together.

\subsection{Physical meaning of Mayer-Vietoris sequence}

Let us interpret the long exact sequence (\ref{LES_physical}) physically. We will interpret the homomorphisms $\alpha$, $\beta_1$, $\beta_2$, $\gamma_1$, and $\gamma_2$ in Sec.\,\ref{sec:physical_hom_MV}, and present the physical intuition behind the exactness of sequence (\ref{LES_physical}) in Sec.\,\ref{sec:physical_exact_MV}.

\subsubsection{Physical meaning of homomorphisms}
\label{sec:physical_hom_MV}

We propose the following physical interpretation of the homomorphisms $\alpha$, $\beta_1$, $\beta_2$, $\gamma_1$, and $\gamma_2$ in the long exact sequence (\ref{LES_physical}).

\paragraph{The connecting map $\alpha$:} The connecting map $\alpha$ will be the double-layer construction shown in Figure \ref{fig:double_layer_construction}, which turns a $(d-1)$-dimensional, $G$-symmetric SPT phase $[a]$ into a $d$-dimensional, $G \times \paren{\ZZZ \rtimes \ZZZ_2^R}$-symmetric SPT phase $[b]$. This construction consists of putting copies of $a$ on all $x_1 \in \mu + \ZZZ$ hyperplanes and its mirror image $\tilde a$ under $x_1 \mapsto -x_1$ on all $x_1 \in -\mu + \ZZZ$ hyperplanes, where $\mu$ is an arbitrary offset. It is obvious that $[b]$ does not depend on the choice of $\mu$ or representative $a$ and that the construction gives a homomorphism.

\begin{figure}
\centering
\includegraphics[width=3in]{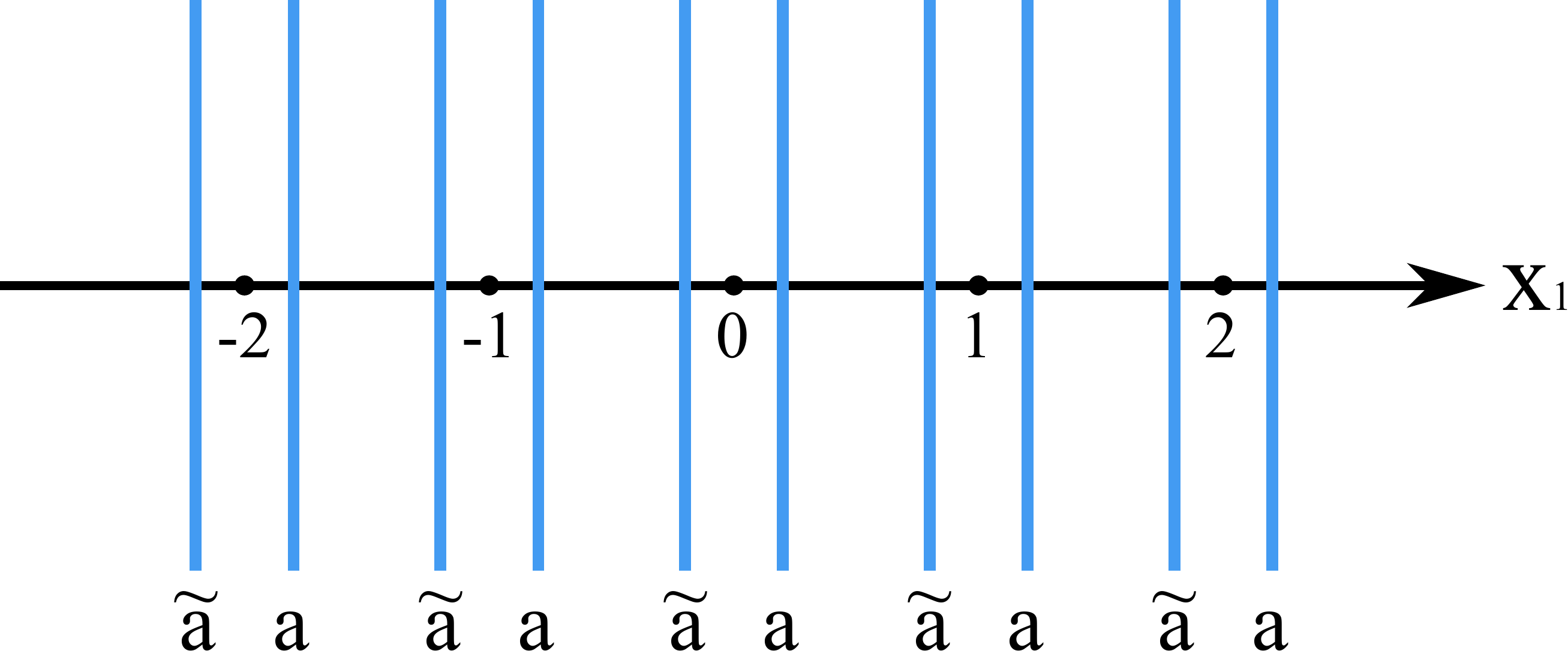}
\caption[The double-layer construction.]{The double-layer construction, which turns a $(d-1)$-dimensional, $G$-symmetric SPT phase $[a]$ into a $d$-dimensional, $G \times \paren{\ZZZ \rtimes \ZZZ_2^R}$-symmetric SPT phase $[b]$.}
\label{fig:double_layer_construction}
\end{figure}

\paragraph{The induced maps $\beta_1$ and $\beta_2$:} The maps $\beta_1$ and $\beta_2$ will be translation-forgetting maps with respect to two inequivalent reflection centers (see Figure \ref{fig:reflection_centers}). For $\beta_1$, we forget all symmetries in $G \times \paren{\ZZZ \rtimes \ZZZ_2^R}$ except $G$ and $x_1 \mapsto -x_1$. For $\beta_2$, we forget all symmetries in $G \times \paren{\ZZZ \rtimes \ZZZ_2^R}$ except $G$ and $x_1 \mapsto -x_1 + 1$. These are homomorphisms by definition.

\begin{figure}
\centering
\includegraphics[width=3in]{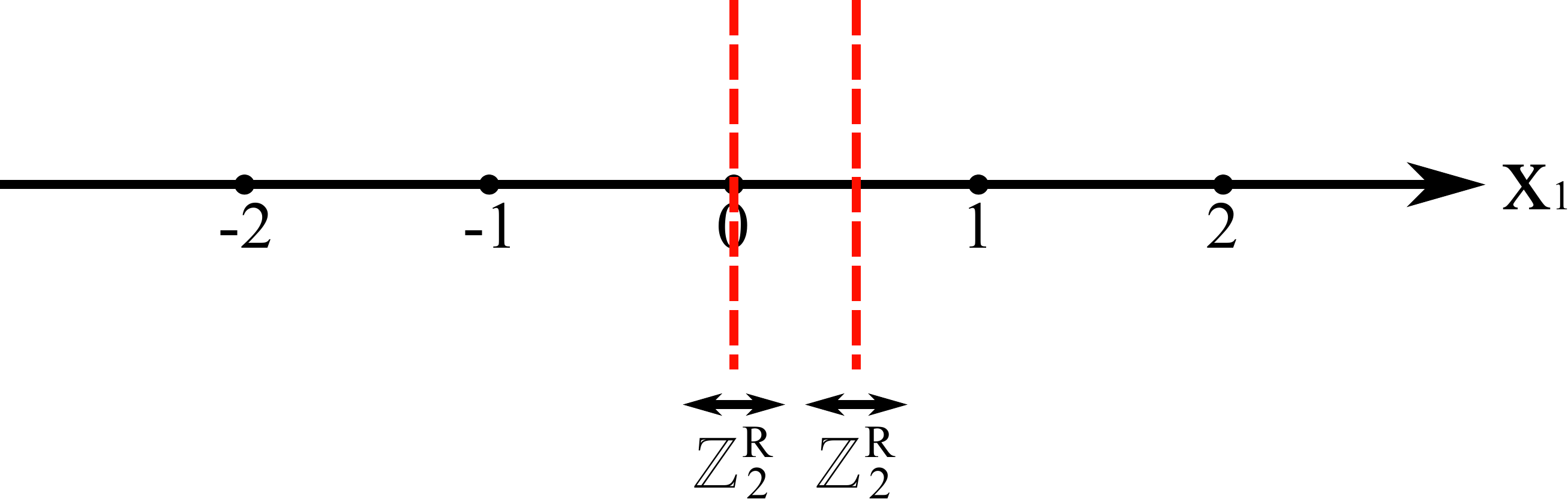}
\caption[Two inequivalent reflection centers of a $G \times \paren{\ZZZ \rtimes \ZZZ_2^R}$-symmetric system.]{Two inequivalent reflection centers, $x_1 = 0$ and $x_1 = \frac12$, of a $G \times \paren{\ZZZ \rtimes \ZZZ_2^R}$-symmetric system.}
\label{fig:reflection_centers}
\end{figure}

\paragraph{The induced maps $\gamma_1$ and $\gamma_2$:} The map $\gamma_1$ will be the reflection-forgetting map. That is, one forgets $x_1 \mapsto -x_1$ and retains only $G$. The map $\gamma_2$ will be the reflection-forgetting map followed by the inversion map. That is, one first forgets $x_1 \mapsto -x_1+1$ and then maps the resulting $G$-symmetric SPT phase to its inverse. It follows that the direct sum $\gamma_1 \oplus \gamma_2$ will map a pair $\paren{[b_1], [b_2]}$ to the difference between the two phases as $G$-symmetric SPT phases. This is the $G$-symmetric SPT phase represented by the system $b_1 + \bar b_2$, where $\bar b_2$ denotes a system that represents the inverse of $[b_2]$ as a $G$-symmetric SPT phase. Such a $\bar b_2$ can be chosen to be the orientation-reversed version of $b_2$ if $G$ is purely internal.

\subsubsection{Physical intuition behind exactness}
\label{sec:physical_exact_MV}

We will now present the physical intuition behind the exactness of sequence (\ref{LES_physical}).

\paragraph{Intuition behind $\image \alpha \subset \kernel \paren{\beta_1 \times \beta_2}$:}

This inclusion amounts to the statement that $(\beta_1\times\beta_2) \circ \alpha = 0$. To understand this, let us take any $(d-1)$-dimensional, $G$-symmetric system $a$ and apply the double-layer construction to it. Let us denote the resulting $d$-dimensional, $G \times \paren{\ZZZ \rtimes \ZZZ_2^R}$-symmetric system by $b$. If we forget the translation symmetry with respect to either reflection center, then $b$ will represent the trivial phase as we can push all copies of $a$ and $\tilde a$ to $\pm \infty$ while preserving the reflection symmetry. It follows that $\beta_1 \circ \alpha$ and $\beta_2 \circ \alpha$ are both zero.

\paragraph{Intuition behind $\kernel \paren{\beta_1 \times \beta_2} \subset \image \alpha$:}

To understand this inclusion, let us take any $G \times \paren{\ZZZ \rtimes \ZZZ_2^R}$-symmetric system $b$ and suppose it represents the trivial $G \times \ZZZ_2^R$-symmetric SPT phase with respect to both reflection centers. This means there is a finite-depth quantum circuit $U_A$ (respectively $U_B$) respecting $G$ and $x_1 \mapsto -x_1$ (resp.\,$x_1 \mapsto -x_1 + 1$) that would continuously deform $b$ to a trivial state. Now, we will apply shifted copies of $U_A$ to the subregions of $b$ defined by
\begin{eqnarray}
n-\mu < x_1 < n+\mu,
\end{eqnarray}
for $n\in \ZZZ$, and shifted copies of $U_B$ to the subregions of $b$ defined by
\begin{equation}
n+\mu < x_1 < n+1-\mu,
\end{equation}
for $n\in \ZZZ$, as in Figure \ref{fig:quantum_circuit}. The overall circuit respects $G \times \paren{\ZZZ \rtimes \ZZZ_2^R}$, and it will trivialize $b$ everywhere except for slabs situated near $x_1 = n \pm \mu$. The slabs at $x_1 = n+\mu$ for different integers $n$ are identical, which we will denote by $a$. The slabs at $x_1 = n-\mu$ will then be copies of the mirror image $\tilde a$ of $a$ under $x_1 \mapsto -x_1$. Note that both $a$ and $\tilde a$ are $G$-symmetric. Thus, we have $G \times \paren{\ZZZ \rtimes \ZZZ_2^R}$-symmetrically deformed $b$ into the system obtained by applying the double-layer construction to $a$. Since this works for any $b$, it follows that $\kernel \paren{\beta_1 \times \beta_2} \subset \image \alpha$.

\begin{figure}
\centering
\includegraphics[width=3.5in]{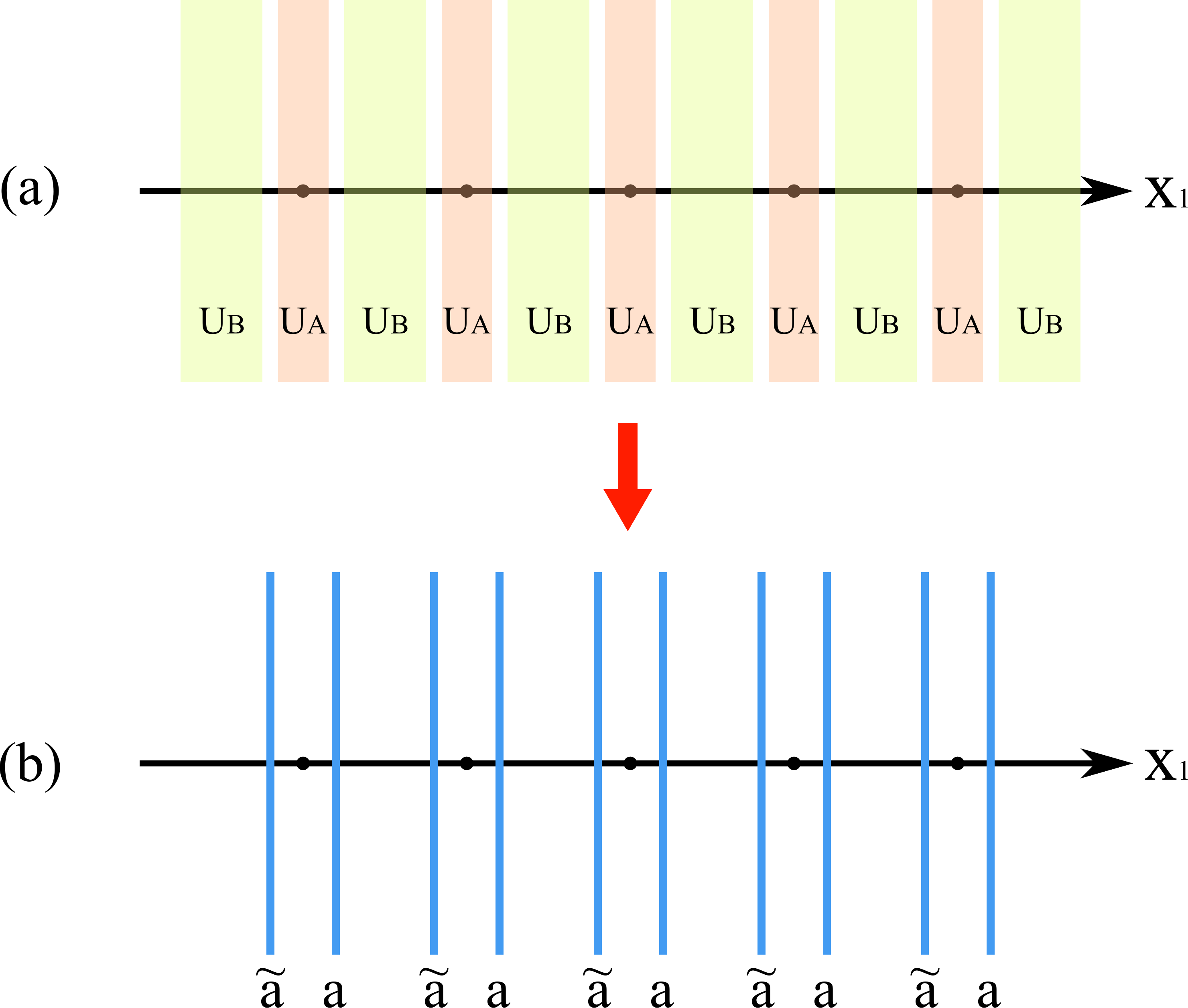}
\caption[Application of copies a finite-depth quantum circuit $U_A$ and copies of a finite-depth quantum circuit $U_B$ to subregions of a $G \times \paren{\ZZZ \rtimes \ZZZ_2^R}$-symmetric system.]{A $G \times \paren{\ZZZ \rtimes \ZZZ_2^R}$-symmetric application of copies a finite-depth quantum circuit $U_A$, which respects $G$ and $x_1 \mapsto -x_1$, and copies of a finite-depth quantum circuit $U_B$, which respects $G$ and $x_1 \mapsto -x_1 + 1$, to subregions of a $G \times \paren{\ZZZ \rtimes \ZZZ_2^R}$-symmetric system.}
\label{fig:quantum_circuit}
\end{figure}

Similar ideas have been used in Refs.\,\cite{Hermele_torsor, Huang_dimensional_reduction, Lu_sgSPT, Xiong_Alexandradinata}.

\paragraph{Intuition behind $\image\paren{\beta_1 \times \beta_2} \subset \kernel \paren{\gamma_1 \oplus \gamma_2}$:}

To understand this inclusion, let us take any $G \times\paren{\ZZZ \rtimes \ZZZ_2^R}$-symmetric system $b$. The statement is equivalent to $\gamma_1\paren{ \beta_1 ([b])} = - \gamma_2\paren{\beta_2([b])} = 0$, but this is obvious because both $\gamma_1 \circ \beta_1$ and $- \gamma_2 \circ \beta_2$ correspond to forgetting all symmetries but $G$.

\paragraph{Intuition behind $\kernel \paren{\gamma_1 \oplus \gamma_2} \subset \image\paren{\beta_1 \times \beta_2}$:}

To understand this inclusion, let us take any pair $\paren{[c_1], [c_2]}$ of $d$-dimensional $G \times \ZZZ_2^R$-symmetric SPT phases, with reflection centers $x_1 = 0$ and $x_1 = \frac12$, respectively. We assume that
\begin{equation}
\gamma_1([c_1]) + \gamma_2([c_2]) = 0. \label{gamma_beta_assumption}
\end{equation}
We will show that there exists a $d$-dimensional, $G \times \paren{\ZZZ \rtimes \ZZZ_2^R}$-symmetric system $b$ such that
\begin{equation}
[c_1] = \beta_1([b]), ~~ [c_2] = \beta_2([b]). \label{gamma_beta_proposition}
\end{equation}

Assumption (\ref{gamma_beta_assumption}) implies that $c_1$ and $c_2$ represent the same $G$-symmetric SPT phase upon forgetting reflection with respect to their respective reflection centers. To prove statement (\ref{gamma_beta_proposition}), the idea is to consider the difference between $c_1$ and $c_2$. To that end, we shall take it for granted that $c_2$ can be chosen to respect translation as well as the reflection about $x_1 = \frac12$. Thus $c_2$ is reflection-symmetric with respect to $x_1 = 0$ as well as $x_1 = \frac12$. If we view $c_2$ as representing a $G \times \ZZZ_2^R$-symmetric SPT phase with reflection center $x_1 = 0$, then $[c_1]$ and $[c_2]$ must differ by a $G \times \ZZZ_2^R$-symmetric SPT phase $[c_3]$ with reflection center $x_1 = 0$ such that $[c_3]$ becomes trivial upon reflection forgetting. Thus, without loss of generality, we can assume
\begin{equation}
c_1 = c_2 + c_3,
\end{equation}
where $c_3$ is some $G \times \ZZZ_2^R$-symmetric system with reflection center $x_1 = 0$ that represents the trivial $G$-symmetric SPT phase upon reflection forgetting.

Because $c_3$ represents the trivial $G$-symmetric SPT phase upon reflection forgetting, there must exist a finite-depth quantum circuit $U_3$ respecting $G$ that would continuously deform deform $c_3$ to a trivial state. By applying $U_3$ to the subregion of $c_3$ defined by
\begin{equation}
x_1 > 0
\end{equation}
and the image $\overline{U_3}$ of $U_3$ under $x_1 \mapsto -x_1$ to the subregion of $c_3$ defined by 
\begin{equation}
x_1 < 0
\end{equation}
[Figure \ref{fig:c_3}(a)], we can reduce $c_3$, in a $G \times \ZZZ_2^R$-symmetric manner, to a $(d-1)$-dimensional, $G \times \ZZZ_2^R$-symmetric system $a$ that lives on $x_1 = 0$ [Figure \ref{fig:c_3}(b)] \cite{Hermele_torsor}. By further bringing in copies of $a$ from $\pm \infty$ again in a $G \times \ZZZ_2^R$-symmetric manner, we will obtain a $d$-dimensional, $G \times \paren{\ZZZ \times \ZZZ_2^R}$-symmetric system that has one copy of $a$ on each $x_1 \in \ZZZ$ plane [Figure \ref{fig:c_3}(c)]. Thus, without loss of generality, we can assume $c_3$ consists of copies of $a$ on all $x_1 \in \ZZZ$ planes, for some $(d-1)$-dimensional $(d-1)$-dimensional system $a$ that is symmetric under $G$ and $x_1 \mapsto - x_1$.

\begin{figure}
\centering
\includegraphics[width=3.5in]{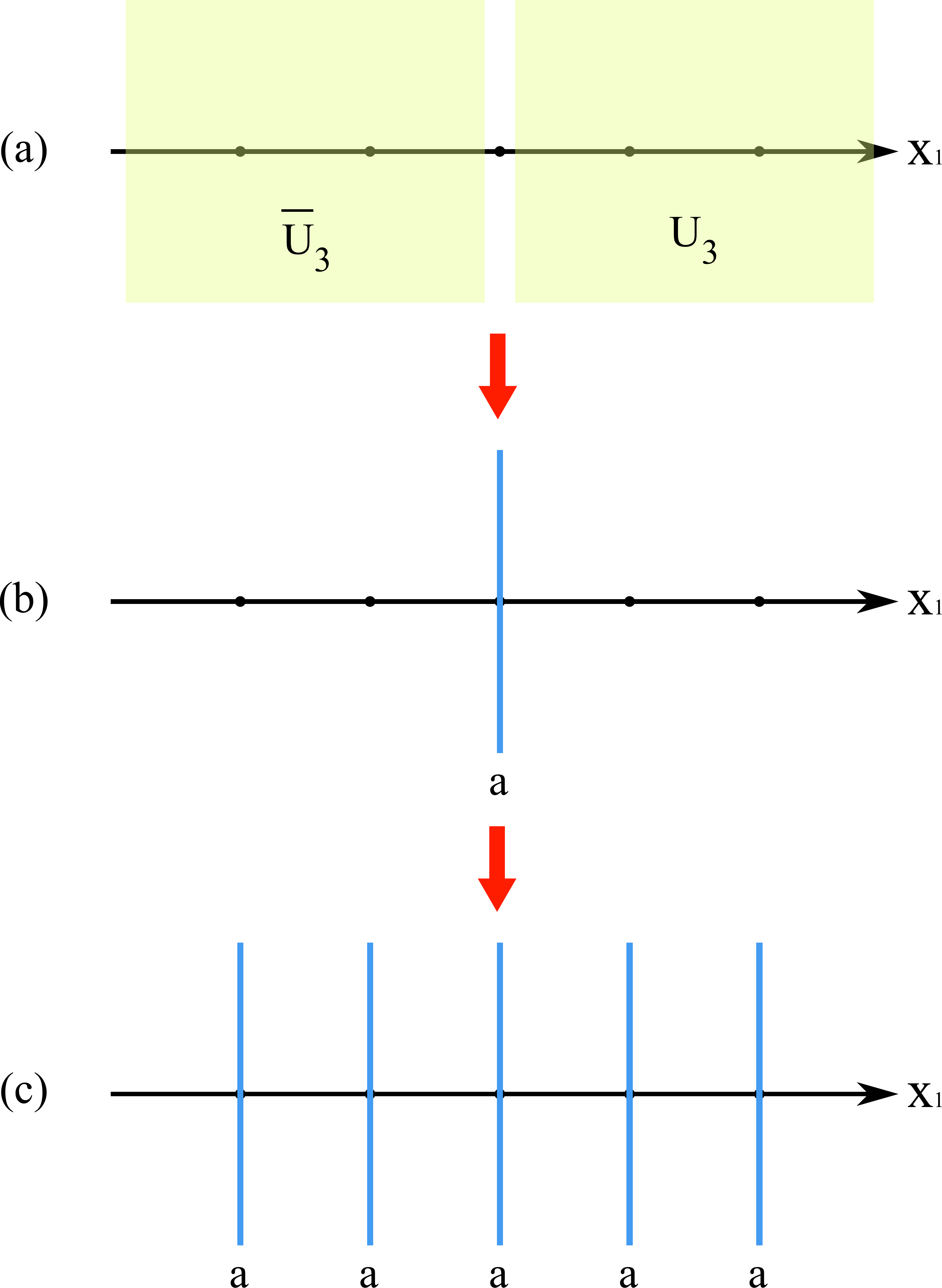}
\caption[Reduction of $c_3$ to the mirror plane and subsequent introduction of additional layers of $a$.]{Reduction of $c_3$ to the mirror plane and subsequent introduction of additional layers of $a$.}
\label{fig:c_3}
\end{figure}

Now, to complete the argument, we set
\begin{equation}
b = c_1.
\end{equation}
By construction, both $c_2$ and $c_3$ and hence $b = c_1 = c_2 + c_3$ have the full $G \times \paren{\ZZZ \rtimes \ZZZ_2^R}$ symmetry. Upon forgetting translation with respect to the reflection center $x_1 = 0$, $b$ will represent the same $G \times \ZZZ_2^R$-symmetric SPT phase as $c_1$ does because $b = c_1$. This gives $\beta_1\paren{[b]} = [c_1]$. On the other hand, upon forgetting translation with respect to the reflection center $x_1 = \frac12$, $b$ will represent the same $G \times \ZZZ_2^R$-symmetric SPT phase as $c_2$ does because the copies of $a$ in $c_3$ can be pushed to $\pm \infty$ in a $G \times \ZZZ_2^R$-symmetric fashion. This gives $\beta_2\paren{[b]} = [c_2]$.

\paragraph{Intuition behind $\image\paren{\gamma_1\oplus\gamma_2} \subset \kernel \alpha$:}

This inclusion amounts to the statements $\alpha \circ \gamma_1 = 0$ and $\alpha \circ \gamma_2 = 0$. We will focus on $\alpha \circ \gamma_1 = 0$. The argument for $\alpha \circ \gamma_2 = 0$ would be a straightforward generalization.

Take any $(d-1)$-dimensional system $c$ with coordinates $x_2, x_3, \ldots$ that is symmetric under $G$ and $x_2 \mapsto -x_2$. We have $[c]\in \SPT^{d-1}\paren{G \times\ZZZ_2^R, \phi}$. Applying the double-layer construction to $c$ along the $x_1$-axis, we obtain a $d$-dimensional system $b$ that is symmetric under $G$, $x_1 \mapsto x_1 + 1$, and $x_1 \mapsto -x_1$. We have $[b]\in \SPT^d\paren{G \times \paren{\ZZZ\rtimes \ZZZ_2^R}, \phi}$. We need to show that $b$ can be trivialized while respecting $G$, $x_1 \mapsto x_1 + 1$, and $x_1 \mapsto -x_1$.

We note that $b$ consists of copies of $c$ at $x_1 \in \mu + \ZZZ$ and copies of the image $\tilde c$ of $c$ under $x_1 \mapsto -x_1$ at $x_1 \in - \mu + \ZZZ$. We claim that a copy of $c$ at $x_1 = \mu$ can be continuously joined with a copy of $\tilde c$ at $x_1 = - \mu$ through wormholes while respecting $G$ and $x_1 \mapsto -x_1$. To see this, first envision a copy of $c$ on the $x_2 = -\epsilon$ plane that is rotated from a copy of $c$ on the $x_1 = 0$ plane, where $\epsilon>0$ is infinitesimal [Figure \ref{fig:wormhole}(a)]. Thanks to the reflection symmetry of $c$, the copy of $c$ on the $x_2 = -\epsilon$ plane is symmetric under $x_1 \mapsto -x_1$. Now, let us gradually bend the $|x_1| > \mu$ portions of this copy of $c$ downward in a fashion that respects $x_1 \mapsto -x_1$ [Figure \ref{fig:wormhole}(b)], until the sides are perpendicular to the $x_1$-axis. Now, we have an $\paren{x_1 \mapsto -x_1}$-symmetric state consisting of a copy of $c$ on the $x_1 = \mu$, $x_2<-\epsilon$ half plane and a copy of $\tilde c$ on the $x_1 = -\mu$, $x_2<-\epsilon$ half plane joined by a copy of $c$ in the $\bars{x_1} < \mu$, $x_2 = -\epsilon$ region [Figure \ref{fig:wormhole}(c)]. The reflection symmetry of $c$ is crucial for the existence of such an $\paren{x_1 \mapsto -x_1}$-symmetric state. Now, by the same token, there also exists an $\paren{x_1 \mapsto -x_1}$-symmetric state consisting of a copy of $c$ on the $x_1 = \mu$, $x_2> \epsilon$ half plane and a copy of $\tilde c$ on the $x_1 = -\mu$, $x_2>\epsilon$ half plane joined by a copy of $c$ in the $\bars{x_1} < \mu$, $x_2 = \epsilon$ region [Figure \ref{fig:wormhole}(d)]. Combining the two configurations, we obtain an $\paren{x_1 \mapsto -x_1}$-symmetric wormhole state connecting a copy of $c$ on the $x_1 = \mu$ plane to a copy of $\tilde c$ on the $x_1 = -\mu$ plane [Figure \ref{fig:wormhole}(e)].

\begin{figure}
\centering
\includegraphics[width=4in]{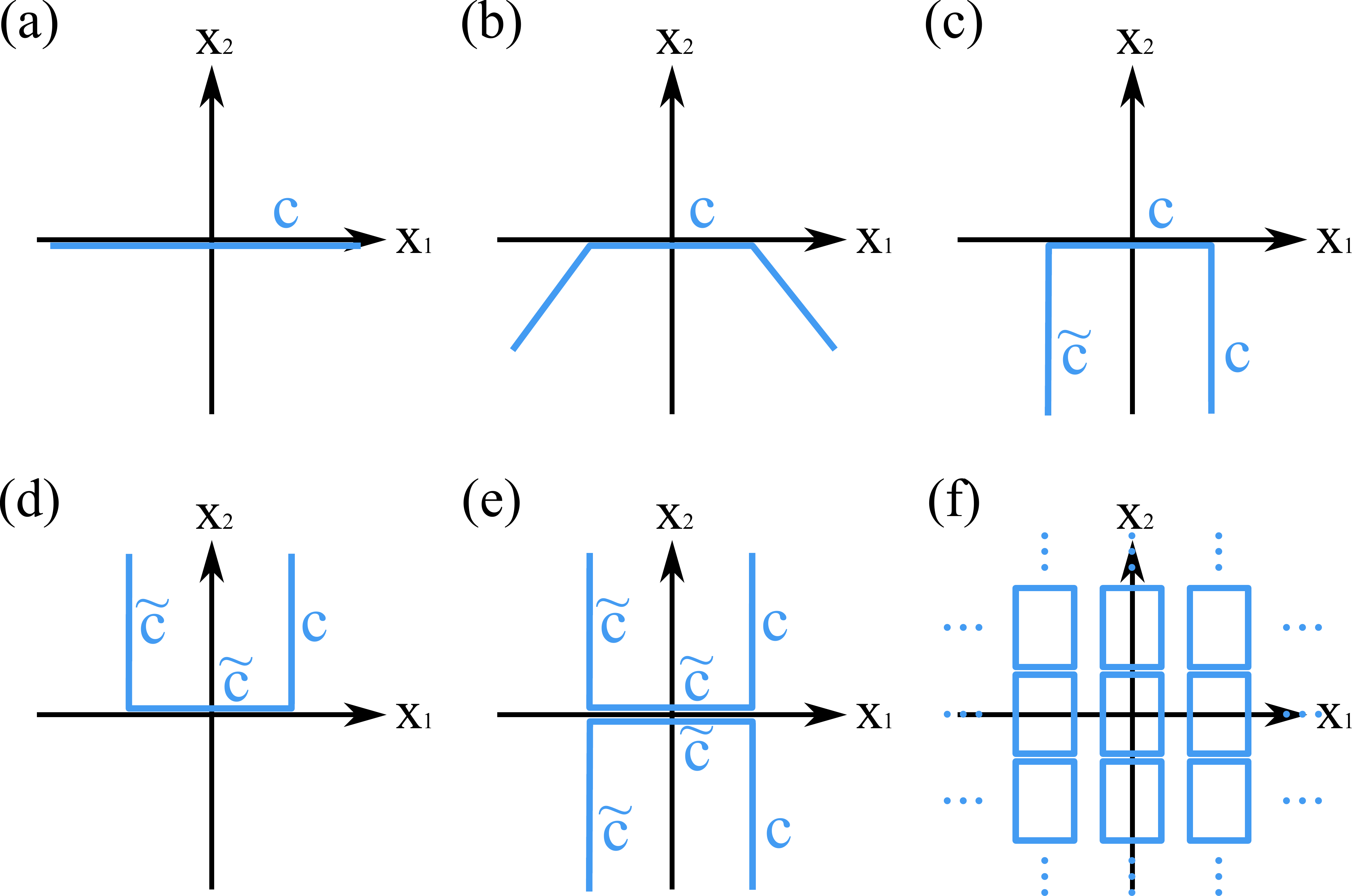}
\caption[The wormhole construction.]{The wormhole construction.}
\label{fig:wormhole}
\end{figure}

We can apply the wormhole construction not only between a copy of $c$ on the $x_1 = \mu$ plane and a copy of $\tilde c$ on the $x_1 = -\mu$ plane, but also between a copy of $c$ on the $x_1 = n + \mu$ plane and a copy of $\tilde c$ on the $x_1 = n-\mu$ plane for any nonzero integer $n$. Likewise, we can create wormholes at nonzero values of $x_2$ as well. It follows that we can create an array of wormholes to reduce $b$ to a collection of $c$ living on cylinders (circles if $d=2$, cylinders if $d=3$, hypercylinders if $d>3$) [Figure \ref{fig:wormhole}(f)]. We shall take it for granted that such wormholes can be continuously created without closing the many-body gap \cite{McGreevy_sSourcery}, and that a copy of $c$ on a cylinder can be shrunk to nothing while respecting $G$, $x_1 \mapsto -x_1$, and the energy gap. It follows that $b$ can be continuously deformed to a trivial system while respecting $G$, $x_1 \mapsto x_1 + 1$, $x_1 \mapsto - x_1$, and the energy gap. This completes the argument.

\paragraph{Intuition behind $\kernel \alpha \subset \image\paren{\gamma_1\oplus\gamma_2}$:}

We have broken down the exactness of sequence (\ref{LES_physical}) into six mathematical statements. We have given the physical intuition behind all but one of them. We can regard the existence of the physical arguments as circumstantial evidence for the correctness of the interpretation in Sec.\,\ref{sec:physical_hom_MV}. In fact, these arguments relate the interpretations of $\alpha$, $\beta_1$, $\beta_2$, $\gamma_1$, and $\gamma_2$ to one another. For instance, owing the arguments for $\image \alpha = \kernel \paren{\beta_1 \times \beta_2}$, the interpretation of $\beta_1$ and $\beta_2$ as translation-forgetting maps implies that $\alpha$ must be---or at least have the same image as---the double-layer construction. Conversely, the interpretation of $\alpha$ as the double-layer construction implies that $\beta_1 \times \beta_2$ must be---or have the same kernel as---the translation-forgetting maps. Similarly, the interpretation of $\beta_1\times\beta_2$ as being (resp.\,as having the same image as) translation forgetting is equivalent to the interpretation of $\gamma_1\oplus\gamma_2$ as having the same kernel as (resp.\,as being) reflection forgetting.

Thus, even though we do not have a simple physical argument for $\kernel \alpha \subset \image\paren{\gamma_1\oplus\gamma_2}$, we are confident that our physical interpretations of $\alpha$, $\beta_1$, $\beta_2$, $\gamma_1$, and $\gamma_2$ were correct. The long exact sequence (\ref{LES_physical}) then makes $\kernel \alpha \subset \image\paren{\gamma_1\oplus\gamma_2}$ a nontrivial prediction of the Mayer-Vietoris sequence. It states that if a $d$-dimensional, $G \times \paren{\ZZZ\rtimes \ZZZ_2^R}$-symmetric system obtained from the double-layer construction represents the trivial phase, then the $(d-1)$-dimensional, $G$-symmetric SPT phase used in the construction must have a $G \times \ZZZ_2^R$-symmetric representative.

\subsection{Example: bosonic SPT phases}
\label{subsec:MV_bSPT}

To demonstrate the applicability of the long exact sequence (\ref{LES_physical}), let us consider the example of bosonic SPT phases where $G$ is set to the trivial group.

\subsubsection{Fitting known classifications into the long exact sequence}
\label{subsubsec:fitting_known_classifications}

For $G = 0$ and bosonic SPT phases, we have the long exact sequence
\begin{equation}
\begin{tikzcd}[row sep=0.1in]
& & ~~~~~~~~~~~~~~~~~~~~~~~~~~~~~\cdots~ \arrow{r}{\gamma_1\oplus\gamma_2} & \SPT^{-1}_b\paren{0} \\
~ \arrow{r}{\alpha} & \SPT^0_b\paren{\ZZZ \rtimes \ZZZ_2^R} \arrow{r}{\beta_1\times\beta_2} & \SPT^0_b\paren{\ZZZ_2^R} \oplus \SPT^0_b\paren{\ZZZ_2^R} \arrow{r}{\gamma_1\oplus\gamma_2} & \SPT^0_b\paren{0} \\
~ \arrow{r}{\alpha} & \SPT^1_b\paren{\ZZZ \rtimes \ZZZ_2^R} \arrow{r}{\beta_1\times\beta_2} & \SPT^1_b\paren{\ZZZ_2^R} \oplus \SPT^1_b\paren{\ZZZ_2^R} \arrow{r}{\gamma_1\oplus\gamma_2} & \SPT^1_b\paren{0} \\
~ \arrow{r}{\alpha} & \SPT^2_b\paren{\ZZZ \rtimes \ZZZ_2^R} \arrow{r}{\beta_1\times\beta_2} & \SPT^2_b\paren{\ZZZ_2^R} \oplus \SPT^2_b\paren{\ZZZ_2^R} \arrow{r}{\gamma_1\oplus\gamma_2} & \SPT^2_b\paren{0} \\
~ \arrow{r}{\alpha} & \SPT^3_b\paren{\ZZZ \rtimes \ZZZ_2^R} \arrow{r}{\beta_1\times\beta_2} & \SPT^3_b\paren{\ZZZ_2^R} \oplus \SPT^3_b\paren{\ZZZ_2^R} \arrow{r}{\gamma_1\oplus\gamma_2} & \SPT^3_b\paren{0} \\
~ \arrow{r}{\alpha} & \cdots,\hspace{0.77in}
\end{tikzcd}\label{bSPT1}
\end{equation}
where we have suppressed $\phi$ in the notation. In dimensions $d < 1$, we interpret $\ZZZ_2^R$ as a time-reversal symmetry. In dimensions $d < 0$, we interpret $d$-dimensional $G$-symmetric SPT phases as the path-connected components of the space of $(-d)$-parameter families of $0$-dimensional $G$-symmetric short-range entangled states.

For $d \leq 0$, the classification of $d$-dimensional $\paren{G,\phi}$-symmetric SPT phases is given by $H^{\phi + d + 2}\paren{G; \ZZZ}$. In particular, we have
\begin{eqnarray}
\SPT^{-1}_b\paren{0} &=& 0, \\
\SPT^{0}_b\paren{0} &=& 0.
\end{eqnarray}
In what follows, we will focus on nonnegative dimensions.

\paragraph{0D:}

The classification for $\ZZZ_2^R$ can be found in Ref.\,\cite{Wen_Boson} to be trivial.

\paragraph{1D:}

The classification for $\ZZZ \rtimes \ZZZ_2^R$ can be found in Sec.\,III A of Ref.\,\cite{Wen_sgSPT_1d} to be $\ZZZ_2 \oplus \ZZZ_2$ by specializing to $G = 0$, which is also consistent with the counting in its Table I. The classification for $\ZZZ_2^R$ can be found in Ref.\,\cite{Hermele_torsor} to be $\ZZZ_2$. The classification in the absence of symmetry is trivial \cite{Wen_1d, Cirac}.

\paragraph{2D:}

The classification for $\ZZZ_2^R$ is trivial by Ref.\,\cite{Huang_dimensional_reduction}. The classification in the absence of symmetry is $\ZZZ$, which is generated by the $E_8$ phase \cite{Kitaev_honeycomb, 2dChiralBosonicSPT, 2dChiralBosonicSPT_erratum, Kitaev_KITP}.

\paragraph{3D:}

The classification for $\ZZZ_2^R$ was given as $\ZZZ_2$ in Ref.\,\cite{Wen_Boson}, but this misses the 3D $E_8$ phase \cite{Hermele_torsor}, which gives another factor of $\ZZZ_2$. Thus the complete classification is $\ZZZ_2^2$. The classification in the absence of symmetry is believed to be trivial.

The classification for $\ZZZ \rtimes \ZZZ_2^R$ can be argued to be $\ZZZ \oplus \ZZZ_2^3$ using the idea of Ref.\,\cite{Huang_dimensional_reduction}, as follows. Ref.\,\cite{Huang_dimensional_reduction} took a dimensional reduction approach in which 3D crystalline SPT phases were constructed using lower-dimensional ``block states," but the authors excluded block states of the $E_8$ phase for simplicity. For space group $Pm$, which has translation symmetries along $x_1$, $x_2$, and $x_3$ and the reflection symmetry $(x_1,x_2,x_3)\mapsto (-x_1,x_2,x_3)$, the classification was given as ${\mathcal C}_0 \oplus {\mathcal C}_2 = \ZZZ_2^2 \oplus \ZZZ_2^2$ in Ref.\,\cite{Huang_dimensional_reduction}. The ${\mathcal C}_0 = \ZZZ_2^2$ piece comes from 0D block states: if $a_0$ represents the nontrivial 0D SPT phase with unitary internal $\ZZZ_2$ symmetry, then the first (resp.\,second) generator of $\mathcal C_0$ is represented by the system with copies of $a_0$ at $(x_1, x_2, x_3)\in \ZZZ \times \ZZZ \times \ZZZ$ [resp.\,$\paren{\frac12+\ZZZ}\times\ZZZ\times\ZZZ$]. These phases are layered phases along the $x_2$- and $x_3$-axes and ought to be modded out in our classification. The ${\mathcal C}_2 = \ZZZ_2^2$ piece comes from 2D block states: if $a_2$ represents the nontrivial 2D conventional SPT phase with unitary onsite $\ZZZ_2$ symmetry, then the first (resp.\,second) generator of ${\mathcal C}_2$ is represented by the system consisting of copies of $a_2$ at $x_1\in \ZZZ$ (resp.\,$\frac12+\ZZZ$) planes. The inclusion of $E_8$ as a building block introduces another $\ZZZ \oplus \ZZZ_2$ in the classification: the $\ZZZ$ is generated by the system consisting of copies of $E_8$ at the $x_1 \in \ZZZ$ planes, whereas $\ZZZ_2$ is generated by the system consisting of copies of $E_8$ at the $x_1 \in \ZZZ$ planes and copies of $\overline{E_8}$ at the $x_1 \in \frac12+\ZZZ$ planes \cite{SongXiongHuang}.

\paragraph{Full sequence:}

Putting everything together, we arrive at the long exact sequence
\begin{equation}
\begin{tikzcd}[row sep=0.1in]
0 \arrow{r}             & \SPT^0_b\paren{\ZZZ \rtimes \ZZZ_2^R} \arrow{r}{\beta_1\times\beta_2} & 0 \oplus 0 \arrow{r}{\gamma_1\oplus\gamma_2} & 0 \\
~ \arrow{r}{\alpha} & \ZZZ_2^2 \arrow{r}{\beta_1\times\beta_2} & \ZZZ_2\oplus\ZZZ_2 \arrow{r}{\gamma_1\oplus\gamma_2} & 0 \\
~ \arrow{r}{\alpha} & \SPT^2_b\paren{\ZZZ \rtimes \ZZZ_2^R} \arrow{r}{\beta_1\times\beta_2} & 0 \oplus 0 \arrow{r}{\gamma_1\oplus\gamma_2} & \ZZZ \\
~ \arrow{r}{\alpha} & \ZZZ\oplus\ZZZ_2^3 \arrow{r}{\beta_1\times\beta_2} & \ZZZ_2^2 \oplus \ZZZ_2^2 \arrow{r}{\gamma_1\oplus\gamma_2} & 0.
\end{tikzcd}\label{bSPT2}
\end{equation}

\subsubsection{Analyzing the long exact sequence of classifications}
\label{subsubsec:analyzing_LES}

Let us analyze the long exact sequence (\ref{bSPT2}).

\paragraph{0D:}

By exactness, $\SPT^0_b\paren{\ZZZ \rtimes \ZZZ_2^R}$ must be trivial. Since there is no translation symmetry in 0D, one way to interpret this result would be to say there is no nontrivial Floquet time-reversal SPT phase in 0D. The triviality of $\SPT^0_b\paren{\ZZZ \rtimes \ZZZ_2^R}$ amounts to the nonexistence of a nontrivial ``equivalence class" of 1-dimensional representations of $\ZZZ \rtimes \ZZZ_2^R$. To verify this, we suppose the generator of $\ZZZ$ and $\ZZZ_2^R$ are represented by $u$ and $wK$, respectively, where $u, w\in U(1)$ and $K$ is the complex conjugation. The relation between translation and reflection, $(x_1\mapsto -x_1) \circ (x_1\mapsto x_1+1) \circ (x_1\mapsto -x_1) = (x_1\mapsto x_1+1)^{-1}$, is automatically satisfied:
\begin{equation}
(wK)(u)(wK) = u^{-1}.
\end{equation}
Conjugating by $\sqrt w$, we get an equivalent representation with $w = 1$. Insofar as we regard representations with different $u$ as representing the same equivalence class---this is not unreasonable since $u$ is a continuous parameter---we will have only one SPT phase with $\ZZZ \rtimes \ZZZ_2^R$ symmetry in 0D.

\paragraph{1D:}

Exactness implies we have an isomorphism
\begin{equation}
\beta_1\times\beta_2: \SPT^1_b\paren{\ZZZ \rtimes \ZZZ_2^R} \xfromto{\isomorphic} \SPT^1_b\paren{\ZZZ_2^R} \oplus \SPT^1_b\paren{\ZZZ_2^R},
\end{equation}
where $\SPT^1_b\paren{\ZZZ \rtimes \ZZZ_2^R} \isomorphic \ZZZ_2^2$ and $\SPT^1_b\paren{\ZZZ_2^R} \isomorphic \ZZZ_2$. According to Ref.\,\cite{Wen_sgSPT_1d}, the four phases in $\SPT^1_b\paren{\ZZZ \rtimes \ZZZ_2^R}$ can be constructed as in Figure \ref{fig:parity_chain} and distinguished by whether or not the blue and/or white entangled pairs gain a minus sign upon exchanging the two constituent spins. The black pair, on the other hand, gains no sign upon spin exchange. Following Ref.\,\cite{Wen_sgSPT_1d}, we introduce two indices, $\nu_1, \nu_2 \in \braces{\pm 1}$, to represent the parity of blue and white entangled pairs under spin exchange, respectively. As for $\SPT^1_b\paren{\ZZZ_2^R} \isomorphic \ZZZ_2$, according to Ref.\,\cite{Hermele_torsor}, the two phases can be represented by quasi 0-dimensional states at the reflection center as in Figure \ref{fig:parity_quasi_0D}, and distinguished by the parity of such a state under reflection.

\begin{figure}
\centering
\includegraphics[width=5in]{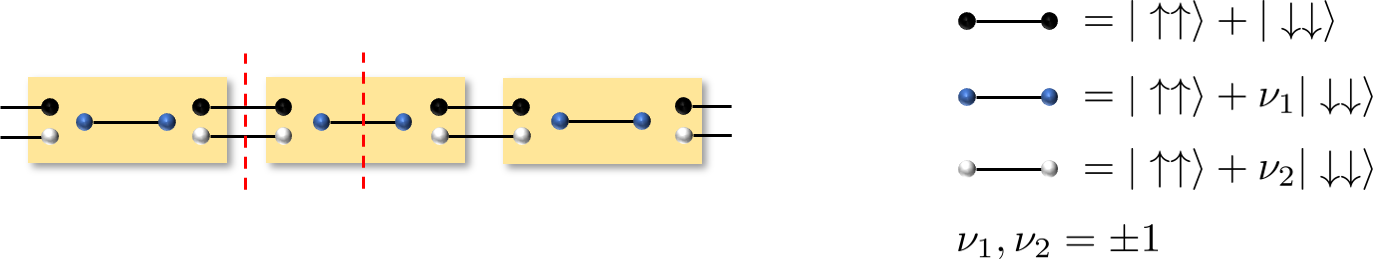}
\caption[Representative states of 1D bosonic $\ZZZ \rtimes \ZZZ_2^R$-symmetric SPT phases.]{Representative states of 1D bosonic $\ZZZ \rtimes \ZZZ_2^R$-symmetric SPT phases. Each box represents a lattice site, which contains six spins. Connected spins form entangled pairs.}
\label{fig:parity_chain}
\end{figure}

\begin{figure}
\centering
\includegraphics[width=5in]{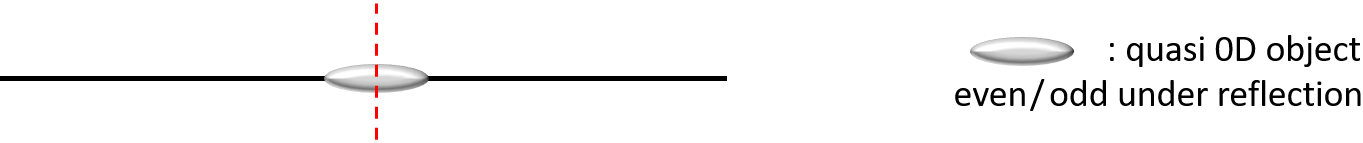}
\caption[Representative states of 1D bosonic $\ZZZ_2^R$-symmetric SPT phases.]{Representative states of 1D bosonic $\ZZZ_2^R$-symmetric SPT phases.}
\label{fig:parity_quasi_0D}
\end{figure}

It is then easy to see, under translation forgetting with respect to the reflection center of a blue entangled pair, that a 1D bosonic $\ZZZ \rtimes \ZZZ_2^R$-symmetric SPT phase becomes a 1D bosonic $\ZZZ_2^R$-symmetric SPT phase that is nontrivial iff $\nu_1 = -1$. Likewise, under translation forgetting with respect to the reflection center of a white entangled pair, a 1D bosonic $\ZZZ \rtimes \ZZZ_2^R$-symmetric SPT phase becomes a 1D bosonic $\ZZZ_2^R$-symmetric SPT phase that is nontrivial iff $\nu_2 = -1$.  Thus translation forgetting with respect to the two reflection centers indeed gives an isomorphism. This supports our interpretation of $\beta_1 \times \beta_2$ as translation forgetting.

\paragraph{2D:}

By exactness, we must have $\SPT^2_b\paren{\ZZZ \rtimes \ZZZ_2^R} = 0$. We can make sense of this using the dimensional reduction approach \cite{Hermele_torsor, Huang_dimensional_reduction, Lu_sgSPT, Xiong_Alexandradinata}. First, any 2D bosonic $\ZZZ \rtimes \ZZZ_2^R$-symmetric SPT phase must become trivial upon forgetting all symmetries; this follows from the fact that 2D invertible topological orders are generated by the $E_8$ state and that the $E_8$ state is not compatible with reflection. Consequently, given any 2D bosonic $\ZZZ \rtimes \ZZZ_2^R$-symmetric state, there must exists a finite-depth quantum circuit that would deform it to a trivial state. Applying the quantum circuit to regions away from the mirror line, we can reduce the system to a 1D system on the mirror line with internal unitary $\ZZZ_2$ symmetry. However, $\ZZZ_2$-symmetric SPT phases have a trivial classification in 1D, so this 1D system can be further trivialized.

\paragraph{3D:}

We have the short exact sequence
\begin{equation}
0 \fromto \SPT^2_b\paren{0} \xfromto{\alpha} \SPT^3_b\paren{\ZZZ \rtimes \ZZZ_2^R} \xfromto{\beta_1\times\beta_2} \SPT^3_b\paren{\ZZZ_2^R} \oplus \SPT^3_b\paren{\ZZZ_2^R} \fromto 0,
\end{equation}
which reads
\begin{equation}
0 \fromto \ZZZ \xfromto{\alpha} \ZZZ \oplus \ZZZ_2^3 \xfromto{\beta_1\times\beta_2} \ZZZ_2^2 \oplus\ZZZ_2^2 \fromto 0. \label{final_seq}
\end{equation}

The generator of the first term (from left, excluding the initial 0) of sequence (\ref{final_seq}) is represented by the $E_8$ state. Under $\alpha$, this generator maps to the doubly layered $E_8$ state shown in Figure \ref{fig:E8_map}(a). This state has a chiral central charge of $c_- = 16$ per unit length in the $x_1$ direction, and reflection exchanges the layers at $x_1 = \pm (n + \mu)$. The generator of $\ZZZ$ in the second term of sequence (\ref{final_seq}) is the layered $E_8$ state in Figure \ref{fig:E8_map}(b) \cite{SongXiongHuang}, which consists of copies of the $E_8$ state with internal unitary $\ZZZ_2$ symmetry on $x_1\in \ZZZ$ planes. This state has $c_- = 8$ per unit length. The remaining generators of the second term of sequence (\ref{final_seq}) have $c_- = 0$ per unit length. Comparing chiral central charge, we thus see that the doubly layered $E_8$ state must correspond to an element of the form $(2, n_0, n_1, n_2) \in \ZZZ \oplus \ZZZ_2^3$. That is, under the double-layer construction, we have
\begin{eqnarray}
\alpha: \ZZZ &\fromto& \ZZZ \oplus \ZZZ_2^3, \nonumber\\
1 &\mapsto& (2, n_0, n_1, n_2). \label{double_layer_SES}
\end{eqnarray}

\begin{figure}
\centering
\includegraphics[width=6in]{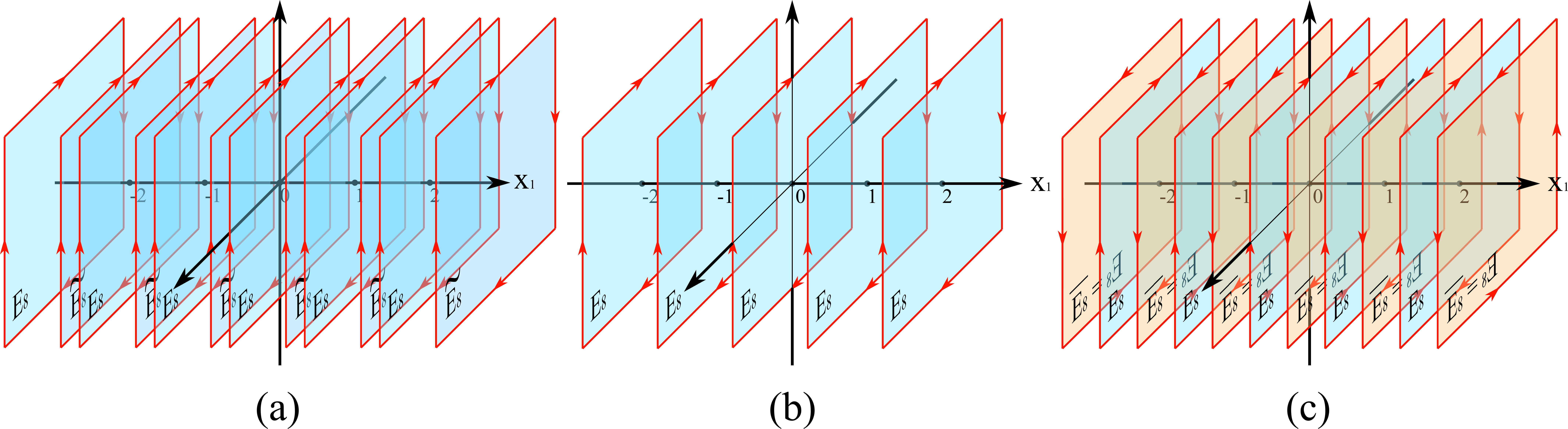}
\caption[The doubly layered $E_8$ state, layered $E_8$ state, and alternately layered $E_8$ state.]{(a) The doubly layered $E_8$ state. (b) The layered $E_8$ state. (c) The alternately layered $E_8$ state. Arrows indicate edge chirality.}
\label{fig:E8_map}
\end{figure}

On the other hand, under translation forgetting, we have
\begin{eqnarray}
\beta_1 \times \beta_2: \ZZZ \oplus \ZZZ_2^3 &\fromto& \ZZZ_2^2 \oplus \ZZZ_2^2, \nonumber \\
(1,0,0,0) &\mapsto& (0,1,0,0), \nonumber\\
(0,1,0,0) &\mapsto& (0,1,0,-1), \nonumber\\
(0,0,1,0) &\mapsto& (1,0,0,0), \nonumber\\
(0,0,0,1) &\mapsto& (0,0,1,0). \label{translation_forgetting_SES}
\end{eqnarray}
To see this, we simply need to recall the physical realizations of various generators in the classifications involved. The first $\ZZZ_2$ in the left-hand side of Eq.\,(\ref{translation_forgetting_SES}) is generated by the alternately layered $E_8$ phase in Figure \ref{fig:E8_map}(c) \cite{SongXiongHuang}. The other two $\ZZZ_2$'s in the left-hand side of Eq.\,(\ref{translation_forgetting_SES}) are generated by putting copies of the nontrivial element $[a] \in \SPT^2_b\paren{\ZZZ_2} \isomorphic \ZZZ_2$ on all $x_1 \in \ZZZ$ planes or all $x_1 \in \frac12 + \ZZZ$ planes, respectively \cite{Huang_dimensional_reduction}. The first and second $\ZZZ_2$'s in the right-hand side of Eq.\,(\ref{translation_forgetting_SES}) are generated by putting $[a]$ or $E_8$ on the $x_1 = 0$ plane, respectively. The third and fourth $\ZZZ_2$'s in the right-hand side of Eq.\,(\ref{translation_forgetting_SES}) are the same but with $x_1 = 0$ replaced by $x_1 = \frac12$.

By the physical arguments in Sec.\,\ref{sec:physical_exact_MV}, the composition of maps (\ref{double_layer_SES}) and (\ref{translation_forgetting_SES}) must be trivial. This forces
\begin{equation}
n_1 = n_2 = n_3 = 0.
\end{equation}

Now, we see that Eqs.\,(\ref{double_layer_SES}) and (\ref{translation_forgetting_SES}) indeed make the sequence $0 \fromto \ZZZ \fromto \ZZZ \oplus \ZZZ_2^3 \fromto \ZZZ_2^2 \oplus\ZZZ_2^2 \fromto 0$ exact. This supports our interpretation of $\alpha$ and $\beta_1\times\beta_2$ as the double-layer construction and translation forgetting, respectively.

\section{Atiyah-Hirzebruch decomposition of general crystalline SPT phases}
\label{sec:Atiyah-Hirzebruch}

In this section we will consider crystalline symmetries $G$ and the classification of crystalline SPT phases in general. We will use a homological formulation that we suspect is equivalent in some sense to the cohomological formulation of Chapter \ref{chap:minimalist} by some sort of Poincar\'e duality. Furthermore, we will allow for boundaries in a physical system by considering \emph{relative} generalized homology. Throughout the section, we assume $G$ is purely spatial. However, our discussions extend immediately to direct products $G' \times G$ for any $G'$ that acts internally.

\subsection{Relative SPT phenomena}
\label{sec:genehomo}

We will speak of \emph{SPT phenomena of degree $n$} for integer $n$, which we loosely define as follows. SPT phenomena of degree $0$ are short-range entangled states. SPT phenomena of degree $-1$ are anomalies of short-range entangled states. SPT phenomena of degree $-2$ are boundaries of anomalies of short-range entangled states, i.e.\,anomalies of anomalies. Inductively, we define SPT phenomena of degree $n$ for $n < 0$ as anomalies of SPT phenomena of degree $n + 1$. As for positive degrees, we define SPT phenomena of degree $1$ to be loops of short-range entangled states. We define SPT phenomena of degree $2$ to be loops of loops of short-range entangled states, i.e.\, families of short-range entangled states parameterized by the 2-sphere $\SS^2$. In general, we define SPT phenomena of degree $n$ for $n > 0$ to be families of short-range entangled states parameterized by $\SS^n$. We will not make it precise what the terms ``anomalies" or ``anomalies of anomalies" mean, but only demand that they satisfy certain constraints described below.\footnote{There may be a clear meaning to anomalies themselves, but we will need to know the topology of the space of anomalies, and for one thing, there is the question of whether a deformation means a deformation of anomaly itself or a deformation of the parent short-range entangled state.}

Now, let $G$ be a discrete group, $X$ be a $G$-CW complex, and $Y$ be a $G$-CW subcomplex. Let us consider the classification of deformation classes of short-range entangled states on $X$ that respect $G$ that can have anomalies on $Y$. We will denote this by
\begin{equation}
\SPT_0\paren{X, Y}, \label{SPT_0(X,Y)}
\end{equation}
and refer to it as the classification of \emph{degree-$0$ SPT phenomena on $X$ relative $Y$}. For example, if $X$ is a manifold with boundary and $Y$ is its boundary, then we can put any SPT phase on the entire $X$. By contrast, if $X$ is a 2-sphere and $Y$ is its equator, then we can put any SPT phase on one hemisphere of $X$ and an unrelated SPT phase on the other hemisphere. Note that, unlike in the notation $\SPT^d\paren{G, \phi}$ where symmetry and dimensionality are explicit, in Eq.\,(\ref{SPT_0(X,Y)}) they are implicit in $X$. We posit that $\SPT_0\paren{X, Y}$ can be given by the relative equivariant generalized homology group
\begin{equation}
\SPT_0(X, Y) \isomorphic h^G_0\paren{X, Y}, \label{homology_0}
\end{equation}
where $h^G_\bullet$ is the \emph{equivariant generalized homology theory} \cite{Whi,KonoTamaki} defined by the $\Omega$-spectrum of Chapter \ref{chap:minimalist}.

We can generalize this to other degrees. We will denote by
\begin{equation}
\SPT_n\paren{X, Y} \label{SPT_n(X,Y)}
\end{equation}
the classification of deformation classes of SPT phenomena of degree $n$ on $X$ that respect $G$ that can have SPT phenomena of degree $n-1$ on $Y$, or more succinctly, the classification of \emph{degree-$n$ SPT phenomena on $X$ relative $Y$}. We posit that this, in turn, can be given by the $n$th relative equivariant generalized homology group
\begin{equation}
\SPT_n(X, Y) \isomorphic h^G_n\paren{X, Y}. \label{homology_n}
\end{equation}

As mentioned earlier, we did not make precise the meaning of anomalies, anomalies of anomalies, etc., and thus the meaning of the classification of degree-$n$ SPT phenomena on $X$ relative $Y$ is accordingly vague. However, in order for Eqs.\,(\ref{homology_0})(\ref{homology_n}) to make sense, the meaning of these terms must satisfy the following constraint. Consider the case of trivial $G$, and let $X$ be the $d$-dimensional disk $D^d$ and $Y$ be boundary $\partial D^d$. If we write $\tilde h_\bullet$ and $\tilde h^\bullet$ for the \emph{reduced} version of $h_\bullet$ and $h^\bullet$, respectively, then we have \cite{Adams2}
\begin{equation}
h_n\paren{D^d, \partial D^d} = \tilde h_n\paren{\SS^d} \isomorphic \tilde h^{d-n}\paren{\SS^0} = h^{d-n}\paren{\pt} \isomorphic \SPT^{d-n}\paren{0},
\end{equation}
the classification of $(d-n)$-dimensional SPT phases without symmetry. In particular,
\begin{equation}
h_0\paren{D^d, \partial D^d} \isomorphic h^d\paren{\pt} \isomorphic \SPT^d\paren{0},
\end{equation}
the classification of $d$-dimensional SPT phases without symmetry. Therefore, the meaning of degree-$n$ SPT phenomena and the meaning of degree-$n$ SPT phenomena on $X$ relative $Y$ must be such that, for any integer $d$ and $n$, the classification of $(d-n)$-dimensional SPT phases without symmetry is exactly the same as the classification of degree-$n$ SPT phenomena on $D^d$ relative $\partial D^d$. For a few special cases,
\begin{itemize}
\item $n = 0$: the classification of $d$-dimensional SPT phases without symmetry must be the same as the classification of SPT phases on $D^d$ that can have anomalies on $\partial D^d$. This makes sense as these seem to correspond the same phase on an infinite system in one case and a finite system with boundary in the other.

\item $n = -1$: the classification of $d$-dimensional SPT phases without symmetry must be the same as the classification of anomalies on $D^{d-1}$ that can have anomalies of anomalies on $\partial D^{d-1}$. One can envision a $d$-dimensional SPT phase living on the $d$-dimensional ball $D^d$, and think of $D^{d-1}$ as the northern hemisphere of the ball and $\partial D^{d-1}$ as the equator. The idea is that the information of the SPT phase on $D^d$ ought to be encoded in the anomaly on its surface $\SS^{d-1}$, and knowing one hemisphere of $\SS^{d-1}$ ought to be as good as knowing the entire $\SS^{d-1}$.

\item $n = -2$: the classification of $d$-dimensional SPT phases without symmetry must be the same as the classification of anomalies of anomalies on $D^{d-2}$ that can have anomalies of anomalies of anomalies on $\partial D^{d-2}$. Similarly to the previous case, the idea is that knowing one half of the equator of $\SS^{d-1}$ ought to be as good as knowing the entire equator, which ought to be as good as knowing one hemisphere of $\SS^{d-1}$, which ought to be as good as knowing the entire sphere $\SS^{d-1}$, which ought to be as good as knowing an SPT phase on $D^d$ itself.

\item $n = 1$: the classification of $d$-dimensional SPT phases without symmetry must be the same as the classification of loops of short-range entangled states on $D^{d+1}$ that can pump SPT phases to $\partial D^{d+1}$. This makes sense according to the interpretation of $\Omega$-spectrum in Sec.\,\ref{subsubsec:pumping}.

\item $n = 2$: the classification of $d$-dimensional SPT phases without symmetry must be the same as the classification of loops of loops of short-range entangled states on $D^{d+2}$ that can pump loops of short-range entangled states to $\partial D^{d+2}$. This similarly makes sense according to Sec.\,\ref{subsubsec:pumping}.
\end{itemize}

In the more general case of nontrivial $G$ or $\paren{X, Y} \neq \paren{D^d, \partial D^d}$, we suspect, by some more general form of Poincar\'e duality, that it might be possible to write
\begin{equation}
h^{\phi + d}_G \paren{\pt} \label{co_formulation}
\end{equation}
as
\begin{equation}
h^G_0\paren{X, Y} \label{ho_formulation}
\end{equation}
for some $G$-CW complex $X$ and $G$-CW subcomplex $Y$. Note that, unlike Eq.\,(\ref{co_formulation}), which is twisted if $G$ does not preserve the spatial orientation, Eq.\,(\ref{ho_formulation}) is not twisted in any case.

\subsection{Bulk-boundary correspondence}
\label{subsec:bulk-boundary-correspondence}

Every generalized equivariant homology theory admits a long exact sequence of a pair:
\begin{equation}
\cdots \fromto h^G_{n+1}\paren{X, Y} \xfromto{\partial} h^G_n\paren{Y} \xfromto{\alpha} h^G_n\paren{X} \xfromto{\beta} h^G_n\paren{X, Y} \xfromto{\partial} h^G_{n-1}\paren{Y} \fromto \cdots
\end{equation}
Here, $\alpha$ and $\beta$ are induced by the inclusions $Y \fromto X$ and $(X, \emptyset) \fromto (X, Y)$, respectively. In the spirit of the functoriality of Chapter \ref{chap:minimalist}, we shall interpret these as physically being the embedding maps. To wit, we interpret $\alpha$ as embedding an SPT phenomena on $Y$ into $X$, and $\beta$ as reinterpreting an SPT phenomena on $X$ as an SPT phenomena on $X$ relative $Y$ (with the trivial anomaly on $Y$). Equivalently, for $\beta$, we can cut open the SPT phenomena on $X$ along $Y$ and define the resulting state representing as the element of $h^G_n\paren{X, Y}$ that $\beta$ maps to.

As for $\partial: h^G_n\paren{X, Y} \fromto h^G_{n-1}\paren{Y}$, we interpret it as the anomaly extraction map. Namely, if $[a]$ is a degree-$n$ SPT phenomena on $X$ relative $Y$, then $\partial\paren{[a]}$ will be the boundary anomaly of $[a]$. This is a form of the bulk-boundary correspondence.\footnote{In order for the interpretation of $h^G_n\paren{X} \coloneq h^G_n\paren{X, \emptyset}$ [resp.\,$h^G_n\paren{Y} \coloneq h^G_n\paren{Y, \emptyset}$] as degree-$n$ SPT phenomena on $X$ (resp.\,$Y$) relative $\emptyset$ to be meaningful, $X$ (resp.\,$Y$) must not have a ``boundary" itself. For example, $X$ could be a manifold without boundary and $Y$ could be a submanifold of $X$ without boundary. If either $X$ or $Y$ has a ``boundary," then one must clarify what $h^G_n\paren{X}$ or $h^G_n\paren{Y}$ represents physically.}

Every generalized equivariant homology theory also admits a long exact sequence of a triple $\paren{X, Y, Z}$:
\begin{equation}
\cdots \fromto h^G_{n+1}\paren{X, Y} \xfromto{\partial} h^G_n\paren{Y, Z} \xfromto{\alpha} h^G_n\paren{X, Z} \xfromto{\beta} h^G_n\paren{X, Y} \xfromto{\partial} h^G_{n-1}\paren{Y, Z} \fromto \cdots
\end{equation}
This can be interpreted in an analogous way.\footnote{In order for the physical interpretation of all terms to be meaningful, one must either ensure $Z$ contains the ``boundaries" of $X$ and $Y$ or else clarify the meaning of the terms.}

\subsection{Atiyah-Hirzebruch spectral sequence}
\label{sec:ahss}

The Atiyah-Hirzebruch spectral sequence is a tool for computing generalized homology and cohomology groups. Here we are interested in the relative equivariant generalized homology group $h^G_n(X,Y)$. The general structure of the Atiyah-Hirzebruch spectral sequence is as follows. We have abelian groups $E^r_{p, q}$ indexed by a positive integer $r$ (called \emph{page}), along with integers $p$ and $q$, which provide a double grading on $E^r$. There are doubly graded homomorphisms called \emph{differentials}:
\begin{equation}
d^r_{p, q}: E^r_{p, q} \fromto E^r_{p - r, q + r - 1}.
\end{equation}
In general, either the $E^1$-page or the $E^2$-page is given, and one defines higher pages inductively by
\begin{equation}
E^{r+1}_{p, q} \coloneq \kernel \paren{d^r_{p, q}} / \image \paren{d^r_{p + r, q - r + 1}}.
\end{equation}
That is, $E^{r+1}$ is the homology of $E^r$ with respect to the differential $d^r$. If there exists some non-negative integer $d$ such that $E^r_{p, q} = 0$ for all $p < 0$ or $p > d$, then the spectral sequence stabilizes; that is, $E^r = E^{r+1}$ for large enough $r$. Writing $E^\infty \coloneq \lim_{r\fromto \infty} E^r$, we then have a filtration
\begin{equation}
0 \eqcolon F_{-1} h^G_n(X,Y) \subset F_0 h^G_n(X,Y) \subset F_1 h^G_n(X,Y) \subset \cdots \subset F_d h^G_n(X,Y) \coloneq h^G_n(X,Y)
\end{equation}
such that
\begin{equation}
F_{p} h^G_n(X,Y) / F_{p - 1} h^G_n(X,Y) \isomorphic E^\infty_{p, n - p}.
\end{equation}
Equivalently, we have abelian group extensions
\begin{equation}
0 \fromto F_{p - 1} h^G_n(X,Y) \fromto F_{p} h^G_n(X,Y) \fromto E^\infty_{p, n - p} \fromto 0.
\end{equation}
By solving these abelian group extensions, one can determine $h^G_n(X,Y)$ itself.

In what follows we will present an Atiyah-Hirzebruch spectral sequence for $h^G_n(X,Y)$ and interpret it in terms of crystalline SPT phenomena. We will see that the first differentials underlie the constructions of Refs.\,\cite{Hermele_torsor, Huang_dimensional_reduction, TB18}. However, there are also higher differentials, and these have not been previously studied. The structure of a filtration of SPT phases in the presence of crystalline symmetries and its relation to higher-order SPT phases (see Sec.~IV F 2 of our work \cite{Shiozaki2018}) have been pointed out previously in Ref.\,\cite{TB18}.

\subsubsection{\texorpdfstring{$E^1$}{E1}-page}
\label{sec:e1}

Let us assume $X$ is finite-dimensional and let $d$ be its dimension. We denote by $I_p$ the set of inequivalent $p$-cells of $X$ that are not in $Y$. If $D_j^p$ is a $p$-cell, we denote by $G_{D_j^p}$ its little group, i.e.\,the group of elements of $G$ that fix any and hence all points in the interior of $D^p_j$. We shall consider the skeletons of $X$,
\begin{align}
X_0 \subset X_1 \subset \cdots \subset X_d = X, \label{skeleton}
\end{align}
where $X_p$ is the subcomplex made of all cells of dimension $\leq p$.

The Atiyah-Hirzebruch spectral sequence for $h^G_n(X,Y)$ associated to the filtration (\ref{skeleton}) has the following $E^1$-page:
\begin{align}
E^1_{p, q} 
= h_{p + q}^G(X_p \cup Y, X_{p-1} \cup Y).
\end{align}
Using the cellular structure of $X$ and $Y$, we can rewrite this as
\begin{equation}\begin{split}
E^1_{p, q} 
= \bigoplus_{j\in I_p} h^{G_{D^p_j}}_{p + q}(D^p_j,\p D^p_j)
= \bigoplus_{j \in I_p} h^{- q}_{G_{D^p_j}}(\pt).
\end{split}\end{equation}
The summand
\begin{equation}
h^{- q}_{G_{D^p_j}}(\pt) \isomorphic \SPT^{-q} \paren{ G_{D^p_j} }
\end{equation}
can be interpreted as ``$(-q)$-dimensional SPT phases on the $p$-cell $D^p_j$ with unitary internal symmetry $G_{D^p_j}$." The meaning of this when $p + q \neq 0$ as well as when $p + q = 0$ will be explained in subsequent sections. The $E^1$-page describes the local data of crystalline SPT phenomena. The correspondence between $E^1_{p, q}$ and SPT phases is summarized in Table \ref{table:E1}.

{\small
\begin{spacing}{\dnormalspacing}
\begin{longtable}[c]{c|ccccc}
\caption[Interpretation of $E^1$-page of Atiyah-Hirzebruch spectral sequence.]{The interpretation of the $E^1$-page of the Atiyah-Hirzebruch spectral sequence.\label{table:E1}} \\
\endfirsthead
\caption[]{(Continued).} \\
\endhead
\endfoot
$\vdots$ & $\vdots$                                & $\vdots$                                & ~           & $\vdots$ \\
$q=2$    & (-2)D SPT on $0$-cells & (-2)D SPT on $1$-cells & $\cdots$ & (-2)D SPT on $d$-cells \\
$q=1$    & (-1)D SPT on $0$-cells & (-1)D SPT on $1$-cells & $\cdots$ & (-1)D SPT on $d$-cells \\
$q=0$    & 0D SPT on $0$-cells     & 0D SPT on $1$-cells    & $\cdots$ & 0D SPT on $d$-cells \\
$q=-1$   & 1D SPT on $0$-cells     & 1D SPT on $1$-cells    & $\cdots$ & 1D SPT on $d$-cells \\
$q=-2$   & 2D SPT on $0$-cells     & 2D SPT on $1$-cells    & $\cdots$ & 2D SPT on $d$-cells \\
$\vdots$ & $\vdots$                                & $\vdots$                                & ~           & $\vdots$ \\
\nobreakhline 
$E^1_{p,q}$ & $p=0$                            & $p=1$                                   & $\cdots$  & $p=d$
\end{longtable}
\end{spacing}
}

\subsubsection{First differentials}
\label{sec:d1}

The first differentials
\begin{equation}
d^1_{p, q}: E^1_{p, q} \fromto E^1_{p - 1, q}
\end{equation}
map
\begin{equation}
d^1_{p, q}: \mbox{``$(-q)$D SPT phases on $p$-cells"} \fromto \mbox{``$(-q)$D SPT phases on $(p-1)$-cells"}.
\end{equation}
They relate the $E^1$- and $E^2$-pages via
\begin{align}
E^2_{p, q}:= \kernel \paren{d^1_{p, q}} / \image \paren{d^1_{p + 1, q}}.
\end{align}
This is the first step towards converging to $E^\infty$. Let us examine the physical meaning of $d^1_{p, q}$ in low degrees.

\paragraph{$d^1_{1, 0}$: ``0D SPT phases on 1-cells" $\fromto$ ``0D SPT phases on 0-cells"}

This is the process of creating a $0$D SPT phase and its conjugate in the middle of a $1$-cell and pumping them to the two ends of the $1$-cell. This leads to the equivalence relation $E^1_{0, 0} / \image \paren{d^1_{1, 0}}$ among $0$D SPT phases on $0$-cells.

\paragraph{$d^1_{2, -1}$: ``1D SPT phases on 2-cells" $\fromto$ ``1D SPT phases on 1-cells"}

This is the process of creating a $1$D SPT phase that lives on a small circle in the interior of a $2$-cell and expanding it to the edge. This leads to the equivalence relation $E^1_{1, -1} / \image \paren{d^1_{1, 0}}$ among $1$D SPT phases on $1$-cells.

\paragraph{$d^1_{1, -1}$: ``1D SPT phases on 1-cells" $\fromto$ ``1D SPT phases on 0-cells"}

This is the map that extracts the edge anomaly of a $1$D SPT phase that lives on a $1$-cell. This expresses the gluing condition $\kernel \paren{d^1_{1, -1}} \subset E^1_{1, -1}$ of $1$D SPT phases on $1$-cells.

\paragraph{$d^1_{2, -2}$: ``2D SPT phases on 2-cells" $\fromto$ ``2D SPT phases on 1-cells"}

This is the map that extracts the boundary anomaly of a $2$D SPT phase that lives on a $2$-cell. This expresses the gluing condition $\kernel \paren{d^1_{2, -2}} \subset E^1_{2, -2}$ of $2$D SPT phases on $2$-cells.

\paragraph{$d^1_{1, -2}$: ``2D SPT phases on 1-cells" $\fromto$ ``2D SPT phases on 0-cells"}

Given an SPT anomaly that lives on a $1$-cell, this is the map that extracts its boundary. In the classification of SPT anomalies on $X$, this differential expresses the gluing condition $\kernel \paren{d^1_{1, -2}} \subset E^1_{1, -2}$ of SPT anomalies on $1$-cells.

We thus see that the first differentials are the mathematical structure behind Refs.\,\cite{Hermele_torsor, Huang_dimensional_reduction, TB18}.

\subsubsection{Higher differentials}
\label{subsubsec:higher_differentials}

The higher differentials
\begin{equation}
d^r_{p, q}: E^r_{p, q} \fromto E^1_{p - r, q + r - 1}
\end{equation}
map
\begin{equation}
d^r_{p, q}: \mbox{``$(-q)$D SPT phases on $p$-cells"} \fromto \mbox{``$(- q - r + 1)$D SPT phases on $(p - r)$-cells"}.
\end{equation}
They relate the $E^r$- and $E^{r+1}$-pages via
\begin{align}
E^{r+1}_{p, q}:= \kernel \paren{d^r_{p, q}} / \image \paren{d^r_{p + r, q - r + 1}}.
\end{align}
Let us explore the physical meaning of these higher differentials.

\paragraph{$d^2_{2, -1}$: ``1D SPT phases on 2-cells" $\fromto$ ``0D SPT phases on 0-cells"}

This is the process of creating $1$D SPT phases on small circles in $2$-cells, expanding them to the boundaries of the $2$-cells, canceling them on the bounding $1$-cells, but possibly leaving behind some $0$D SPT phases on $0$-cells that border the $2$-cells. This leads to the equivalence relation $E^2_{0, 0} / \image\paren{d^2_{2, -1}}$ among $0$D SPT phases on $0$-cells.

Note that any element in the domain $E^2_{2, -1}$ of $d^2_{2, -1}$ must be representable by an element in the kernel $\kernel\paren{d^1_{2, -1}}$ of $d^1_{2, -1}$. Since $d^1_{2, -1}$ maps $1$D SPT phases on $2$-cells to $1$D SPT on $1$-cells, the subset $\kernel\paren{d^1_{2, -1}} \subset E^1_{2, -1}$ must represent those configurations of $1$D SPT phases on $2$-cells that satisfy the gluing condition along $1$-cells. Thus it is indeed possible to cancel elements of $E^2_{2, -1}$ on the $1$-cells as we proposed to do in the previous paragraph.

See Sec.~V B 2 of our work \cite{Shiozaki2018} for a nontrivial example of $d^2_{2,-1}$. 

\paragraph{$d^2_{2, -2}$: ``2D SPT phases on 2-cells" $\fromto$ ``1D SPT phases on 0-cells"}

Given a configuration of $2$D SPT phases on $2$-cells, this map extracts the possible SPT anomaly left behind on $0$-cells while gluing the them along the bounding $1$-cells. This gives the higher gluing condition $\kernel\paren{d^2_{2, -2}} \subset E^2_{2, -2}$.

Note that any element of $E^2_{2, -2}$ must be representable by an element of $\kernel \paren{d^1_{2, -2}} \subset E^1_{2, -2}$. Since $d^1_{2, -2}$ maps $2$D SPT phases on $2$-cells to $2$D SPT phases on $1$-cells, the subset $\kernel \paren{d^1_{2, -2}} \subset E^1_{2, -2}$ must consist of configurations that satisfy the gluing condition along $1$-cells. Thus it is indeed possible to glue elements of $E^2_{2, -2}$ along the $1$-cells as we proposed to do in the previous paragraph.

See Secs.~V B 2, V G and V H 1 of our work \cite{Shiozaki2018} for nontrivial examples of $d^2_{2,-2}$. 

\paragraph{$d^3_{3, -2}$: ``2D SPT phases on 3-cells" $\fromto$ ``0D SPT phases on 0-cells"}

This is the process of creating $2$D SPT phases on small spheres in $3$-cells, expanding them to the boundaries of the $3$-cells, canceling them on the the bounding $2$-cells and $1$-cells, but possibly leaving behind some $0$D SPT phases on $0$-cells that border the $3$-cells. This leads to the equivalence relation $E^3_{0, 0} / \image\paren{d^3_{3, -2}}$ among $0$D SPT phases on $0$-cells.

Note that any element of $E^3_{3, -2}$ must be representable by an element of $\kernel \paren{d^2_{3, -2}} \subset E^2_{3, -2}$ and an element of $\kernel \paren{d^1_{3, -2}} \subset E^1_{3, -2}$. Since $d^1_{3, -2}$ maps $2$D SPT phases on $3$-cells to $2$D SPT phases on $2$-cells, an element of $\kernel d^1_{3, -2} \subset E^1_{3, -2}$ must represent a configuration that satisfies the gluing condition along $2$-cells. Similarly, since $d^2_{3, -2}$ maps $2$D SPT phases on $3$-cells to $1$D SPT on $1$-cells, an element of $\kernel\paren{d^1_{2, -1}} \subset E^1_{2, -1}$ must represent a configuration that satisfies the gluing condition along $1$-cells. Thus it is indeed possible to cancel elements of $E^e_{3, -2}$ along the $2$-cells and $1$-cells as we proposed to do in the previous paragraph.

See Sec.~VI G of our work \cite{Shiozaki2018} for an example of nontrivial $d^2_{3,-2}$. 

\paragraph{$d^3_{3, -3}$: ``3D SPT phases on 3-cells" $\fromto$ ``1D SPT phases on 0-cells"}

Given a configuration of $3$D SPT phases on $3$-cells, this map extracts the possible SPT anomaly left behind on $0$-cells while gluing the them along $2$-cells and $1$-cells. This gives the higher gluing condition $\kernel\paren{d^3_{3, -3}} \subset E^3_{3, -3}$.

Note that any element of $E^3_{3, -3}$ must be representable by an element of $\kernel \paren{d^2_{3, -3}} \subset E^2_{3, -3}$ and an element of $\kernel \paren{d^1_{3, -3}} \subset E^1_{3, -3}$. Since $d^1_{3, -3}$ maps $3$D SPT phases on $3$-cells to SPT anomalies on $2$-cells, the subset $\kernel \paren{d^1_{2, -2}} \subset E^1_{2, -2}$ must consist of configurations that satisfy the gluing condition along $2$-cells. Similarly, since $d^2_{3, -3}$ maps $3$D SPT phases on $3$-cells to SPT anomalies $1$-cells, the subset $\kernel \paren{d^1_{3, -3}} \subset E^1_{3, -3}$ must consist of configurations that satisfy the gluing condition along $1$-cells. Thus it is indeed possible to glue elements of $E^3_{3, -3}$ along $2$-cells and $1$-cells as we proposed to do in the previous paragraph.

\subsubsection{\texorpdfstring{$E^\infty$-page}{E infty page}} 
\label{subsubsec:E_inf}

Since $X$ is $d$-dimensional, we have $E^r_{p, q} = 0$ for all $p < 0$ or $p > d$. Consequently,
\begin{equation}
d^r_{p, q} = 0 ~~ \forall r > d.
\end{equation}
Thus the spectral sequence stabilizes at the $(r+1)$st page:
\begin{equation}
E^{r+1}_{p, q} = E^{r+2}_{p, q} = \cdots \eqcolon E^\infty_{p, q}.
\end{equation}
Extending the analysis in Secs.\,\ref{sec:d1} and \ref{subsubsec:higher_differentials} to higher $r$ and general values of $p$ and $q$, we get that
\begin{eqnarray}
&& E^\infty_{p, q} \nonumber\\[20pt]
&\isomorphic& \braces{\parbox{4.5in}{configurations of $(-q)$-dimensional SPT phases on $p$-cells (i.e.\,cells in $I_p$) that satisfy all gluing conditions along lower-dimensional cells (i.e.\,cells in $I_0 \sqcup I_1 \sqcup \ldots \sqcup I_{p-1}$) modulo all equivalence relations due to higher-dimensional cells (i.e.\,cells in $I_{p+1} \sqcup I_{p+2} \sqcup \ldots \sqcup I_d$)}}.
\end{eqnarray}
Alternatively,
\begin{eqnarray}
&& E^\infty_{p, n - p} \nonumber\\[20pt]
&\isomorphic& \braces{\parbox{4.5in}{configurations of degree-$n$ SPT phenomena on $p$-cells (i.e.\,cells in $I_p$) that satisfy all gluing conditions along lower-dimensional cells (i.e.\,cells in $I_0 \sqcup I_1 \sqcup \ldots \sqcup I_{p-1}$) modulo all equivalence relations due to higher-dimensional cells (i.e.\,cells in $I_{p+1} \sqcup I_{p+2} \sqcup \ldots \sqcup I_d$)}}.
\end{eqnarray}

\subsubsection{Filtration and abelian group extension}
\label{sec:Group extension}

$h^G_n\paren{X, Y}$ classifies degree-$n$ SPT phenomena on $X$ relative $Y$. Restricting to the $p$-skeleton, we have that $h^G_n\paren{X_p, X_p \cap Y}$ classifies degree-$n$ SPT phenomena on $X_p$ relative $X_p \cap Y$. By definition, this is the classification of degree $n$-SPT phenomena on $X_p$ that are glued along cells in $I_0 \sqcup I_1 \sqcup \ldots \sqcup I_{p-1}$ modulo all equivalence relations due to cells in $I_1 \sqcup I_2 \sqcup \ldots \sqcup I_p$.

For $h^G_n\paren{X_p, X_p \cap Y}$, there is no equivalence relation due to cells of dimension $> p$. However, if we consider the homomorphism $h^G_n(X_p,X_p \cap Y) \to h^G_n(X,Y)$ induced by the inclusion map $\paren{X_p, X_p \cap Y} \fromto \paren{X, Y}$, then its image
\begin{align}
F_p h^G_n(X,Y):= \image \brackets{h^G_n(X_p,X_p \cap Y) \to h^G_n(X,Y)}
\end{align}
will have equivalence relations due to cells of dimension $> p$ incorporated. More explicitly,
\begin{eqnarray}
&& F_p h^G_n(X,Y) \nonumber\\[20pt]
&\isomorphic& \braces{\parbox{4.5in}{degree-$n$ SPT phenomena on $X_p$ that are glued along cells in $I_0 \sqcup I_1 \sqcup \ldots \sqcup I_{p-1}$ modulo all equivalence relations due to cells in $I_1 \sqcup I_2 \sqcup \ldots \sqcup I_d$}}. \label{F_p_meaning}
\end{eqnarray}
This gives a filtration
\begin{equation}
0 \eqcolon F_{-1} h^G_n(X,Y) \subset F_0 h^G_n(X,Y) \subset F_1 h^G_n(X,Y) \subset \cdots \subset F_d h^G_n(X,Y) \coloneq h^G_n(X, Y).
\end{equation}
Furthermore, Eq.\,(\ref{F_p_meaning}) implies that the quotient
\begin{eqnarray}
&& F_p h^G_n(X,Y) / F_{p - 1} h^G_n(X,Y) \nonumber\\[20pt]
&\isomorphic& \braces{\parbox{4.5in}{configurations of degree-$n$ SPT phenomena on cells in $I_p$ that satisfy all gluing conditions along cells in $I_0 \sqcup I_1 \sqcup \ldots \sqcup I_{p-1}$ modulo all equivalence relations due to cells in $I_p \sqcup I_{p+1} \sqcup \ldots \sqcup I_d$}}.
\end{eqnarray}
This explains the isomorphism
\begin{equation}
F_p h^G_n(X,Y) / F_{p-1} h^G_n(X,Y) \isomorphic E^\infty_{p, n - p}. \label{Fp/Fp-1=Einf}
\end{equation}

We can express Eq.\,(\ref{Fp/Fp-1=Einf}) as an abelian group extension,
\begin{align}
0 \fromto F_{p-1} h^G_n(X,Y) \fromto F_p h^G_n(X,Y) \fromto E^\infty_{p, n - p} \fromto 0.
\label{eq:ahss_extension}
\end{align}
The fact that an abelian group extension can be nontrivial reflects the fact that configurations of degree-$n$ SPT phenomena on $p$-cells classified by $E^\infty_{p, n - p}$ can be nontrivially extended by degree-$n$ SPT phenomena on $X_{p-1}$ classified by $F_{p-1} h^G_n\paren{X, Y}$. For example, consider $n = 0$ and suppose an SPT phenomenon $a$ on $X_p$ represents an element of $E^\infty_{p, -p}$ of order $k > 1$. This means a stack of $k$ copies of $a$ can be trivialized on the $p$-cells. However, in general it can remain nontrivial on $X_{p-1}$. See Sec.~V A 1 of our work \cite{Shiozaki2018} for an example where the abelian group extension is nontrivial.

\chapter{Conclusion}
\label{chap:conclusion}

In this thesis, we developed a new theoretical framework (Chapter \ref{chap:minimalist}) for the classification and construction of symmetry protected topological phases (Sec.\,\ref{sec:short-range entangled_SPT}), or SPT phases, which are the subset of topological phases of strongly interacting quantum many-body systems with symmetries that are short-range entangled. The framework was first proposed in a series of talks \cite{Kitaev_Stony_Brook_2011_SRE_1, Kitaev_Stony_Brook_2013_SRE, Kitaev_IPAM} by Kitaev; in my work \cite{Xiong}, my work \cite{Xiong_Alexandradinata} with A.\,Alexandradinata, my work \cite{Shiozaki2018} with Ken Shiozaki and Kiyonori Gomi, and my work \cite{SongXiongHuang} with Hao Song and Shengjie Huang, we developed it in full detail. The framework models the space of short-range entangled states by what are known as $\Omega$-spectra in the sense of algebraic topology (Sec.\,\ref{sec:short-range entangled_as_spectrum}), and the classification of SPT phases by the generalized cohomology theories associated with the $\Omega$-spectra (Sec.\,\ref{sec:from_spectrum_to_classification}). The defining relation, Eq.\,(\ref{Omega_spectrum}), of an $\Omega$-spectrum can be interpreted by considering either the process of pumping a lower-dimensional short-range entangled state to the boundary in a cyclic adiabatic evolution (Sec.\,\ref{subsubsec:pumping}) or equivalently a continuous spatial pattern of short-range entangled states (Sec.\,\ref{subsubsec:domain_wall}). We proposed a formulation for fermionic systems as well as bosonic systems by considering symmetries of the form $G_f \times G$ for some $G_f$ that contains the fermion parity symmetry (Sec.\,\ref{subsec:fermionic_SPT}). The central statement of the framework, which we referred to as the ``generalized cohomology hypothesis," posits that there exist natural isomorphisms between the classifications of SPT phases and the twisted equivariant generalized cohomology groups of some $\Omega$-spectrum (Sec.\,\ref{subsec:GCH} for bosonic SPT phases and Sec.\,\ref{subsec:fermionic_SPT} for fermionic SPT phases). The framework gives a unified view on various existing proposals for the classification of SPT phases, which correspond to different choices of the $\Omega$-spectrum (Sec.\,\ref{sec:existing_proposals}).

We demonstrated the power of the framework in a number of applications (Chapters \ref{chap:applications} and \ref{chap:advanced}). In our first application (Sec.\,\ref{sec:glide}), we derived from the generalized cohomology hypothesis a short exact sequence, Eq.\,(\ref{SES}), relating the classifications of SPT phases with and without glide-reflection symmetries. Using one implication of the short exact sequence, we argued that the recently discovered ``hourglass fermions" \cite{Ma_discoverhourglass, Hourglass, Cohomological} are robust to strong interactions (Sec.\,\ref{subsec:hourglass_fermions}). Furthermore, we derived the concrete Corollaries \ref{cor:classification_SPT_weak_wrt_glide}, \ref{cor:quad-chotomy}, and \ref{cor:direct_sum_decomposition} regarding the complete classification of SPT phases with glide symmetry (not just those SPT phases that become trivial upon glide forgetting). We broke down the short exact sequence into six physical statements and offered the physical intuition behind each of them (Sec.\,\ref{subsec:physical_picture}). Applying the short exact sequence to 3D fermionic SPT phases, we deduced that the complete classifications in the Wigner-Dyson classes A and AII with glide are $\ZZZ_2 \oplus \ZZZ_2$ and $\ZZZ_4\oplus \ZZZ_2 \oplus \ZZZ_2$, respectively (Sec.\,\ref{subsec:applications}). Applying the short exact sequence to bosonic SPT phases, we computed the complete classifications with glide in dimensions $\leq 3$ for a number of symmetry groups (\ref{subsec:bSPT_with_glide}). Finally, we compared the effect of glide symmetry with the effect of translation symmetry (Sec.\,\ref{sec:puretranslation}), and discussed the generalization to spatiotemporal glide reflections (Sec.\,\ref{subsec:temporalglide}). 

In our second application (Sec.\,\ref{sec:3D_beyond_group_cohomology}), we predicted from the generalized cohomology hypothesis that the classification of 3D bosonic SPT phases with space-group symmetry $G$ is $H^{\phi + 5}_G\paren{\pt; \ZZZ} \oplus H^{\phi + 1}_G\paren{\pt; \ZZZ}$ for all 230 space groups, where the twist $\phi: G \fromto \braces{\pm 1}$ of the cohomology groups is given by $\phi(g) = -1$ if $g$ reverses the spatial orientation and $1$ otherwise (Sec.\,\ref{subsec:prediction}). In particular, we predicted that the classification of 3D bosonic crystalline SPT phases beyond the group cohomology proposal is $H^{\phi + 1}_G\paren{\pt; \ZZZ}$. These beyond-cohomology phases are built from the $E_8$ state \cite{Kitaev_honeycomb, 2dChiralBosonicSPT, Kitaev_KITP}, which had been previously excluded for simplicity. The cohomology group $H^{\phi + 1}_G\paren{\pt; \ZZZ}$ can be easily read off from the international (Hermann-Mauguin) symbols of space groups (Theorem \ref{thm:H1_formula}):
\begin{eqnarray}
H^{\phi + 1}_G(\pt;\ZZZ) =
\begin{cases}
\ZZZ^{k}, & \mbox{space group preserves orientation},\\
\ZZZ^{k} \times \ZZZ_{2}, & \mbox{otherwise},
\end{cases}
\end{eqnarray}
where $k=0,1,3$ respectively if, in the international (Hermann-Mauguin) symbol, there is more than one symmetry direction listed, exactly one symmetry direction listed and it is not $1$ or $\bar{1}$, and exactly one symmetry direction listed and it is $1$ or $\bar{1}$. We worked out a few examples: space groups $P1$, $Pm$, $Pmm2$, $Pmmm$, $P\bar1$, and $Pc$ (Sec.\,\ref{subsec:examples}). We provided explicit models for the SPT phases in these cases using the layer construction (Figure \ref{fig:layerE8_y}), the alternating-layer construction (Figure \ref{fig:layerE8_y2}), and other more complicated cellular constructions (Figure \ref{fig:sg47_E8}). We argued for the nontriviality of the SPT phases that correspond to the generators of $\ZZZ^k$ by counting the chiral central charge per unit length (Sec.\,\ref{subsubsec:P1}). We argued for the nontriviality of the SPT phases that correspond to the generator of $\ZZZ_2$ by considering surface topological orders that cannot exist symmetrically in strictly two-dimensional systems (Figure \ref{fig:sg47_E8_STO}). Motivated by these examples, we went on and verified the classification $H^{\phi + 1}_G\paren{\pt; \ZZZ}$ in full generality, that is, for all 3D space groups (Sec.\,\ref{subsec:reduction_and_construction}). We used a combination of techniques, which included dimensional reduction (Sec.\,\ref{susbusbsec:dimensional_reduction}), surface anomaly (Sec.\,\ref{subsubsec:H1_invariance}), and explicit cellular construction (Sec.\,\ref{subsubsec:Construction}). We established a split short exact sequence, Eq.\,(\ref{eq:extension}), which implies that the classification of 3D bosonic SPT phases with space-group symmetry $G$ is indeed $H^{\phi + 5}_G\paren{\pt; \ZZZ} \oplus H^{\phi + 1}_G\paren{\pt; \ZZZ}$. Finally, we computed $H^{\phi + 1}_G(\pt;\ZZZ)$ for all 230 space groups (App.\,\ref{app:list_of_classifications}).

To facilitate the next two applications, we discussed the construction of the classifying space $BG$ of a general space group $G$ (Sec.\,\ref{sec:space_groups}). Such a $G$ could be either symmorphic or nonsymmorphic. However, we assumed the internal representation of $G$ to be unitary. We remarked that there are two models for $BG$, Eqs.\,(\ref{BS_1}) and (\ref{BS_2}), both of which took advantage of the relationship among space, translational, and point groups. These models allowed us to write the classification of SPT phases with space-group symmetries as a twisted generalized cohomology theory that is equivariant in the point group, as in Eq.\,(\ref{c+n_form}), instead of the space group. We also gave alternative expressions, Eqs.\,(\ref{c+n_form_trivial_phi}) and (\ref{fiber_bundle_form}), in the case that $\phi$ is trivial. Lastly, we remarked that one could replace $\pt$ by $\RRR^d$ in the twisted equivariant generalized cohomology theories, as in Eqs.\,(\ref{BS_form_Rd}) and (\ref{BS_form_Rd_trivial_phi}), thanks to the equivariant contractibility of $\RRR^d$.

Building upon the above discussion, we derived, in our third application (Sec.\,\ref{sec:Mayer-Vietoris}) of the framework, a relation between the classifications of SPT phases with and without reflection symmetries. We showed that there is a Mayer-Vietoris long exact sequence, Eq.\,(\ref{LES_physical}), that ties together the classifications of SPT phases with $G \times \paren{\ZZZ \rtimes \ZZZ_2^R}$, $G \times \ZZZ_2^R$, and $G$ symmetries, where $\ZZZ$ is generated by a translation, $\ZZZ_2^R$ is generated by a reflection, and $G$ is any symmetry group that commutes with the translation and reflection (Sec.\,\ref{subsec:MVLES}). We interpreted the homomorphisms in the Mayer-Vietoris sequence physically as the double-layer construction, translation-forgetting maps, and reflection-forgetting maps, respectively (Sec.\,\ref{sec:physical_hom_MV}). With respect to these interpretations, we explained the physical intuition behind the exactness of the sequence (Sec.\,\ref{sec:physical_exact_MV}). To demonstrate the applicability of the Mayer-Vietoris sequence, we considered the case of bosonic SPT phases where $G$ is set to the trivial group (Sec.\,\ref{subsec:MV_bSPT}). We fit known classifications into the Mayer-Vietoris sequence (Sec.\,\ref{subsubsec:fitting_known_classifications}), and verified the exactness of various segments of the resulting sequence by examining the physical realizations of various SPT phases (Sec.\,\ref{subsubsec:analyzing_LES}). In turn, this justified our interpretation of the homomorphisms in the Mayer-Vietoris sequence.

In our fourth and final application (Sec.\,\ref{sec:Atiyah-Hirzebruch}), we proposed a notion of relative crystalline SPT phenonema classified by relative equivariant generalized homology groups that we suspect is equivalent to the cohomological formulation of SPT phases by some sort of Poincar\'e duality (Sec.\,\ref{sec:genehomo}). We identified a bulk-boundary correspondence for crystalline SPT phenomena by considering the long exact sequence of a pair, which every equivariant generalized homology theory possesses (Sec.\,\ref{subsec:bulk-boundary-correspondence}). Furthermore, we derived an Atiyah-Hirzebruch spectral sequence for relative crystalline SPT phenomena using the skeletons $X_0 \subset X_1 \subset \cdots \subset X_d = X$ of the $G$-CW complex $X$ on which crystalline SPT phenomena reside (Sec.\,\ref{sec:ahss}). We proposed that the various terms on the $E^1$-page of the Atiyah-Hirzebruch spectral sequence correspond to the local data of crystalline SPT phenomena, namely, SPT phenomena of various degrees residing on cells of different dimensions with effective internal symmetries (Sec.\,\ref{sec:e1}). The first differentials of the Atiyah-Hirzebruch spectral sequence, which connect the $E^1$- and $E^2$-pages, describe first-order gluing conditions of and equivalence relations among the local data, and are the mathematical structure behind Refs.\,\cite{Hermele_torsor, Huang_dimensional_reduction, TB18} (Sec.\,\ref{sec:d1}). The higher differentials of the Atiyah-Hirzebruch spectral sequence, which connect the $E^r$- and $E^{r+1}$-pages for $r > 1$, correspond to higher-order gluing conditions and equivalence relations and have not been previously considered (Sec.\,\ref{subsubsec:higher_differentials}). The $E^\infty$-page of the Atiyah-Hirzebruch spectral sequence, which is the limit of $E^r$ as $r \fromto \infty$, corresponds to the classification of SPT phenomena on $p$-cells satisfying all gluing conditions modulo all equivalence relations for fixed values of $p$ (Sec.\,\ref{subsubsec:E_inf}). It gives a filtration of the classification of SPT phenomena on $X$ through a set of abelian group extensions (Sec.\,\ref{sec:Group extension}). We remarked that nontrivial abelian group extensions occur when stacking SPT phonemena on the $p$-skeleton $X_p$ gives an SPT phenomenon that can be trivialized on $p$-cells but remains nontrivial on $X_{p-1}$.

\setstretch{\dnormalspacing}

\begin{appendices}
    
\chapter{List of classifications of 3D bosonic crystalline SPT phases for all 230 space groups}
\label{app:list_of_classifications}

As proved in my work \cite{SongXiongHuang} with Hao Song and Shengjie Huang, $H^{\phi + 1}_{G} \paren{\pt; \ZZZ}$ can be determined from the international (Hermann-Mauguin) symbols of the space group $G$ as follows.

\begin{thm}
Let $G$ be a 3D space group and $\phi:G\fromto\braces{\pm1}$ be the homomorphism that sends orientation-preserving elements of $G$ to	$1$ and the rest to $-1$. Then
\begin{eqnarray}
H^{\phi + 1}_{G} \paren{\pt; \ZZZ} =
\begin{cases}
\ZZZ^{k}, & \mbox{space group preserves orientation},\\
\ZZZ^{k} \times \ZZZ_{2}, & \mbox{otherwise},
\end{cases}
\end{eqnarray}
where $k=0,1,3$ respectively if, in the international (Hermann-Mauguin) symbol, there is more than one symmetry direction listed, exactly one symmetry direction listed and it is not $1$ or	$\bar{1}$, and exactly one symmetry direction listed and it is $1$ or $\bar{1}$. \label{thm:H1_formula}
\end{thm}

Combined with the computations for $H^{\phi + 5}_{G} \paren{\pt; \ZZZ}$ in Refs.\,\cite{Huang_dimensional_reduction, ThorngrenElse}, we can thus determine the classifications \cite{SongXiongHuang}
\begin{equation}
\SPT^{3}(G,\phi) \cong H^{\phi + 5}_{G} \paren{\pt; \ZZZ} \oplus H^{\phi + 1}_{G} \paren{\pt; \ZZZ}
\end{equation}
of 3D bosonic crystalline SPT phases for all 230 space groups. The results were summarized in Ref.\,\cite{SongXiongHuang} and are reproduced in Table\,\ref{table:H1} below.

\begin{longtable}[c]{ccccc}
\caption[List of classifications of 3D bosonic crystalline SPT phases for all 230 space groups]{The classifications of 3D bosonic crystalline SPT phases for all 230 space groups. The first and second columns are the numbers and international (Hermann-Mauguin) symbols of the space group $G$. The third column is the classification $H^{\phi + 5}_G \paren{\pt; \ZZZ}$ of non-$E_8$-based phases, as in the group cohomology proposal \cite{Huang_dimensional_reduction, ThorngrenElse}. The fourth column is the classification $H^{\phi + 1}_{G} \paren{\pt; \ZZZ}$ of $E_8$-based phases \cite{SongXiongHuang}. The fifth column is the classification $\SPT^3\paren{G, \phi} \cong H^{\phi + 5}_G \paren{\pt; \ZZZ}\oplus H^{\phi + 1}_{G} \paren{\pt; \ZZZ}$ of all 3D bosonic crystalline SPT phases for the given space group \cite{SongXiongHuang}. Here $\phi:G\fromto\braces{\pm1}$ is the homomorphism that sends orientation-preserving elements of $G$ to $1$ and the rest to $-1$, and defines the twist of the twisted cohomology groups.\label{table:H1}}\\
\hline
\hline
\multirow{2}{*}{No.} & \multirow{2}{*}{~Symbol~} & \multicolumn{3}{c}{Classification of 3D bosonic crystalline SPT phases} \\
\cline{3-5}
&     & ~~~~~$H^{\phi + 5}_G \paren{\pt; \ZZZ}$~~~~~ & ~~~~~$H^{\phi + 1}_G \paren{\pt; \ZZZ}$~~~~~ & ~~~~Complete~~~~ \\
\hline
\endfirsthead
\caption[]{(Continued).}\\
\hline
\hline
\multirow{2}{*}{No.} & \multirow{2}{*}{~Symbol~} & \multicolumn{3}{c}{Classification of 3D bosonic crystalline SPT phases} \\
\cline{3-5}
&     & ~~~~~$H^{\phi + 5}_G \paren{\pt; \ZZZ}$~~~~~ & ~~~~~$H^{\phi + 1}_G \paren{\pt; \ZZZ}$~~~~~ & ~~~~Complete~~~~ \\
\hline
\endhead
\hline
\hline
\endfoot
1     & $P1$  & $0$ & {$\mathbb{Z}^{3}$} & $\mathbb{Z}^{3}$ \\
2     & $P\overline{1}$ & $\mathbb{Z}_2^{8}$ & {$\mathbb{Z}_2\times\mathbb{Z}^{3}$} & $\mathbb{Z}_2^{9}\times\mathbb{Z}^{3}$ \\
3     & $P2$  & $\mathbb{Z}_2^{4}$ & {$\mathbb{Z}$} & $\mathbb{Z}_2^{4}\times\mathbb{Z}$ \\
4     & $P2_{1}$ & $0$ & {$\mathbb{Z}$} & $\mathbb{Z}$ \\
5     & $C2$  & $\mathbb{Z}_2^{2}$ & {$\mathbb{Z}$} & $\mathbb{Z}_2^{2}\times\mathbb{Z}$ \\
6     & $Pm$  & $\mathbb{Z}_2^{4}$ & {$\mathbb{Z}_2\times\mathbb{Z}$} & $\mathbb{Z}_2^{5}\times\mathbb{Z}$ \\
7     & $Pc$  & $0$ & {$\mathbb{Z}_2\times\mathbb{Z}$} & $\mathbb{Z}_2\times\mathbb{Z}$ \\
8     & $Cm$  & $\mathbb{Z}_2^2$ & {$\mathbb{Z}_2\times\mathbb{Z}$} & $\mathbb{Z}_2^3\times\mathbb{Z}$ \\
9     & $Cc$  & $0$ & {$\mathbb{Z}_2\times\mathbb{Z}$} & $\mathbb{Z}_2\times\mathbb{Z}$ \\
10    & $P2/m$ & $\mathbb{Z}_2^{18}$ & {$\mathbb{Z}_2\times\mathbb{Z}$} & $\mathbb{Z}_2^{19}\times\mathbb{Z}$ \\
11    & $P2_{1}/m$ & $\mathbb{Z}_2^6$ & {$\mathbb{Z}_2\times\mathbb{Z}$} & $\mathbb{Z}_2^{7}\times\mathbb{Z}$ \\
12    & $C2/m$ & $\mathbb{Z}_2^{11}$ & {$\mathbb{Z}_2\times\mathbb{Z}$} & $\mathbb{Z}_2^{12}\times\mathbb{Z}$ \\
13    & $P2/c$ & $\mathbb{Z}_2^{6}$ & {$\mathbb{Z}_2\times\mathbb{Z}$} & $\mathbb{Z}_2^{7}\times\mathbb{Z}$ \\
14    & $P2_{1}/c$ & $\mathbb{Z}_2^{4}$ & {$\mathbb{Z}_2\times\mathbb{Z}$} & $\mathbb{Z}_2^{5}\times\mathbb{Z}$ \\
15    & $C2/c$ & $\mathbb{Z}_2^{5}$ & {$\mathbb{Z}_2\times\mathbb{Z}$} & $\mathbb{Z}_2^{6}\times\mathbb{Z}$ \\
16    & $P222$ & $\mathbb{Z}_2^{16}$ & $0$ & $\mathbb{Z}_2^{16}$ \\
17    & $P222_{1}$ & $\mathbb{Z}_2^{4}$ & $0$ & $\mathbb{Z}_2^{4}$ \\
18    & $P2_{1}2_{1}2$ & $\mathbb{Z}_2^{2}$ & $0$ & $\mathbb{Z}_2^{2}$ \\
19    & $P2_{1}2_{1}2_{1}$ & $0$ & $0$ & $0$ \\
20    & $C222_{1}$ & $\mathbb{Z}_2^{2}$ & $0$ & $\mathbb{Z}_2^{2}$ \\
21    & $C222$ & $\mathbb{Z}_2^{9}$ & $0$ & $\mathbb{Z}_2^{9}$ \\
22    & $F222$ & $\mathbb{Z}_2^{8}$ & $0$ & $\mathbb{Z}_2^{8}$ \\
23    & $I222$ & $\mathbb{Z}_2^{8}$ & $0$ & $\mathbb{Z}_2^{8}$ \\
24    & $I2_{1}2_{1}2_{1}$ & $\mathbb{Z}_2^{3}$ & $0$ & $\mathbb{Z}_2^{3}$ \\
25    & $Pmm2$ & $\mathbb{Z}_2^{16}$ & $\mathbb{Z}_2$      & $\mathbb{Z}_2^{17}$ \\
26    & $Pmc2_{1}$ & $\mathbb{Z}_2^{4}$ & $\mathbb{Z}_2$      & $\mathbb{Z}_2^{5}$ \\
27    & $Pcc2$ & $\mathbb{Z}_2^{4}$ & $\mathbb{Z}_2$      & $\mathbb{Z}_2^{5}$ \\
28    & $Pma2$ & $\mathbb{Z}_2^{4}$ & $\mathbb{Z}_2$      & $\mathbb{Z}_2^{5}$ \\
29    & $Pca2_{1}$ & $0$ & $\mathbb{Z}_2$      & $\mathbb{Z}_2$ \\
30    & $Pnc2$ & $\mathbb{Z}_2^{2}$ & $\mathbb{Z}_2$      & $\mathbb{Z}_2^{3}$ \\		
31    & $Pmn2_{1}$ & $\mathbb{Z}_2^2$ & $\mathbb{Z}_2$      & $\mathbb{Z}_2^3$ \\
32    & $Pba2$ & $\mathbb{Z}_2^{2}$ & $\mathbb{Z}_2$      & $\mathbb{Z}_2^{3}$ \\
33    & $Pna2_{1}$ & $0$ & $\mathbb{Z}_2$      & $\mathbb{Z}_2$ \\
34    & $Pnn2$ & $\mathbb{Z}_2^{2}$ & $\mathbb{Z}_2$      & $\mathbb{Z}_2^{3}$ \\
35    & $Cmm2$ & $\mathbb{Z}_2^{9}$ & $\mathbb{Z}_2$      & $\mathbb{Z}_2^{10}$ \\
36    & $Cmc2_{1}$ & $\mathbb{Z}_2^2$ & $\mathbb{Z}_2$      & $\mathbb{Z}_2^3$ \\
37    & $Ccc2$ & $\mathbb{Z}_2^{3}$ & $\mathbb{Z}_2$      & $\mathbb{Z}_2^{4}$ \\
38    & $Amm2$ & $\mathbb{Z}_2^{9}$ & $\mathbb{Z}_2$      & $\mathbb{Z}_2^{10}$ \\
39    & $Aem2$ & $\mathbb{Z}_2^{4}$ & $\mathbb{Z}_2$      & $\mathbb{Z}_2^{5}$ \\
40    & $Ama2$ & $\mathbb{Z}_2^{3}$ & $\mathbb{Z}_2$      & $\mathbb{Z}_2^{4}$ \\
41    & $Aea2$ & $\mathbb{Z}_2$ & $\mathbb{Z}_2$      & $\mathbb{Z}_2^2$ \\
42    & $Fmm2$ & $\mathbb{Z}_2^{6}$ & $\mathbb{Z}_2$      & $\mathbb{Z}_2^{7}$ \\
43    & $Fdd2$ & $\mathbb{Z}_2$ & $\mathbb{Z}_2$      & $\mathbb{Z}_2^2$ \\
44    & $Imm2$ & $\mathbb{Z}_2^{8}$ & $\mathbb{Z}_2$      & $\mathbb{Z}_2^{9}$ \\
45    & $Iba2$ & $\mathbb{Z}_2^{2}$ &  $\mathbb{Z}_2$     & $\mathbb{Z}_2^{3}$ \\
46    & $Ima2$ & $\mathbb{Z}_2^{3}$ &  $\mathbb{Z}_2$     & $\mathbb{Z}_2^{4}$ \\
47    & $Pmmm$ & $\mathbb{Z}_2^{42}$ & $\mathbb{Z}_2$      & $\mathbb{Z}_2^{43}$ \\
48    & $Pnnn$ & $\mathbb{Z}_2^{10}$ & $\mathbb{Z}_2$      & $\mathbb{Z}_2^{11}$ \\
49    & $Pccm$ & $\mathbb{Z}_2^{17}$ & $\mathbb{Z}_2$      & $\mathbb{Z}_2^{18}$ \\
50    & $Pban$ & $\mathbb{Z}_2^{10}$ & $\mathbb{Z}_2$      & $\mathbb{Z}_2^{11}$ \\
51    & $Pmma$ & $\mathbb{Z}_2^{17}$ & $\mathbb{Z}_2$      & $\mathbb{Z}_2^{18}$ \\
52    & $Pnna$ & $\mathbb{Z}_2^{4}$ & $\mathbb{Z}_2$      & $\mathbb{Z}_2^{5}$ \\
53    & $Pmna$ & $\mathbb{Z}_2^{10}$ & $\mathbb{Z}_2$      & $\mathbb{Z}_2^{11}$ \\
54    & $Pcca$ & $\mathbb{Z}_2^{5}$ & $\mathbb{Z}_2$      & $\mathbb{Z}_2^{6}$ \\
55    & $Pbam$ & $\mathbb{Z}_2^{10}$ &  $\mathbb{Z}_2$     & $\mathbb{Z}_2^{11}$ \\
56    & $Pccn$ & $\mathbb{Z}_2^{4}$ &$\mathbb{Z}_2$       & $\mathbb{Z}_2^{5}$ \\
57    & $Pbcm$ & $\mathbb{Z}_2^{5}$ &  $\mathbb{Z}_2$     & $\mathbb{Z}_2^{6}$ \\
58    & $Pnnm$ & $\mathbb{Z}_2^{9}$ &  $\mathbb{Z}_2$     & $\mathbb{Z}_2^{10}$ \\
59    & $Pmmn$ & $\mathbb{Z}_2^{10}$ & $\mathbb{Z}_2$      & $\mathbb{Z}_2^{11}$ \\
60    & $Pbcn$ & $\mathbb{Z}_2^{3}$ &  $\mathbb{Z}_2$     & $\mathbb{Z}_2^{4}$ \\
61    & $Pbca$ & $\mathbb{Z}_2^{2}$ &   $\mathbb{Z}_2$    & $\mathbb{Z}_2^{3}$ \\
62    & $Pnma$ & $\mathbb{Z}_2^{4}$ &   $\mathbb{Z}_2$    & $\mathbb{Z}_2^{5}$ \\
63    & $Cmcm$ & $\mathbb{Z}_2^{10}$ &  $\mathbb{Z}_2$     & $\mathbb{Z}_2^{11}$ \\
64    & $Cmca$ & $\mathbb{Z}_2^{7}$ &   $\mathbb{Z}_2$    & $\mathbb{Z}_2^{8}$ \\
65    & $Cmmm$ & $\mathbb{Z}_2^{26}$ & $\mathbb{Z}_2$      & $\mathbb{Z}_2^{27}$ \\
66    & $Cccm$ & $\mathbb{Z}_2^{13}$ &  $\mathbb{Z}_2$     & $\mathbb{Z}_2^{14}$ \\
67    & $Cmme$ & $\mathbb{Z}_2^{17}$ &  $\mathbb{Z}_2$     & $\mathbb{Z}_2^{18}$ \\
68    & $Ccce$ & $\mathbb{Z}_2^{7}$ &   $\mathbb{Z}_2$    & $\mathbb{Z}_2^{8}$ \\
69    & $Fmmm$ & $\mathbb{Z}_2^{20}$ &  $\mathbb{Z}_2$     & $\mathbb{Z}_2^{21}$ \\
70    & $Fddd$ & $\mathbb{Z}_2^{6}$ &   $\mathbb{Z}_2$    & $\mathbb{Z}_2^{7}$ \\
71    & $Immm$ & $\mathbb{Z}_2^{22}$ & $\mathbb{Z}_2$      & $\mathbb{Z}_2^{23}$ \\
72    & $Ibam$ & $\mathbb{Z}_2^{10}$ &  $\mathbb{Z}_2$     & $\mathbb{Z}_2^{11}$ \\
73    & $Ibca$ & $\mathbb{Z}_2^{5}$ & $\mathbb{Z}_2$      & $\mathbb{Z}_2^{6}$ \\
74    & $Imma$ & $\mathbb{Z}_2^{13}$ & $\mathbb{Z}_2$      & $\mathbb{Z}_2^{14}$ \\
75    & $P4$  & $\mathbb{Z}_{4}^{2}\times\mathbb{Z}_2$ & {$\mathbb{Z}$} & $\mathbb{Z}_{4}^{2}\times\mathbb{Z}_2\times\mathbb{Z}$ \\
76    & $P4_{1}$ & $0$ & {$\mathbb{Z}$} & $\mathbb{Z}$ \\
77    & $P4_{2}$ & $\mathbb{Z}_2^{3}$ & {$\mathbb{Z}$} & $\mathbb{Z}_2^{3}\times{\mathbb{Z}}$ \\
78    & $P4_{3}$ & $0$ & {$\mathbb{Z}$} & $\mathbb{Z}$ \\
79    & $I4$  & $\mathbb{Z}_{4}\times\mathbb{Z}_2$ & {$\mathbb{Z}$} & $\mathbb{Z}_{4}\times\mathbb{Z}_2\times\mathbb{Z}$ \\
80    & $I4_{1}$ & $\mathbb{Z}_2$ & {$\mathbb{Z}$} & $\mathbb{Z}_2\times\mathbb{Z}$ \\
81    & $P\overline{4}$ & $\mathbb{Z}_{4}^{2}\times\mathbb{Z}_2^{3}$ & {$\mathbb{Z}_2\times\mathbb{Z}$} & $\mathbb{Z}_{4}^{2}\times\mathbb{Z}_2^{4}\times\mathbb{Z}$ \\
82    & $I\overline{4}$ & $\mathbb{Z}_{4}^{2}\times\mathbb{Z}_2^{2}$ & {$\mathbb{Z}_2\times\mathbb{Z}$} & $\mathbb{Z}_{4}^{2}\times\mathbb{Z}_2^{3}\times\mathbb{Z}$ \\
83    & $P4/m$ & $\mathbb{\mathbb{Z}}_{4}^{2}\times\mathbb{Z}_2^{12}$ & {$\mathbb{Z}_2\times\mathbb{Z}$} & $\mathbb{\mathbb{Z}}_{4}^{2}\times\mathbb{Z}_2^{13}\times\mathbb{Z}$ \\
84    & $P4_{2}/m$ & $\mathbb{Z}_2^{11}$ & {$\mathbb{Z}_2\times\mathbb{Z}$} & $\mathbb{Z}_2^{12}\times\mathbb{Z}$ \\
85    & $P4/n$ & $\mathbb{Z}_{4}^{2}\times\mathbb{Z}_2^{3}$ & {$\mathbb{Z}_2\times\mathbb{Z}$} & $\mathbb{Z}_{4}^{2}\times\mathbb{Z}_2^{4}\times\mathbb{Z}$ \\
86    & $P4_{2}/n$ & $\mathbb{Z}_{4}\times\mathbb{Z}_2^{4}$ & {$\mathbb{Z}_2\times\mathbb{Z}$} & $\mathbb{Z}_{4}\times\mathbb{Z}_2^{5}\times\mathbb{Z}$ \\
87    & $I4/m$ & $\mathbb{Z}_{4}\times\mathbb{Z}_2^{8}$ & {$\mathbb{Z}_2\times\mathbb{Z}$} & $\mathbb{Z}_{4}\times\mathbb{Z}_2^{9}\times\mathbb{Z}$ \\
88    & $I4_{1}/a$ & $\mathbb{Z}_{4}\times\mathbb{Z}_2^{3}$ & {$\mathbb{Z}_2\times\mathbb{Z}$} & $\mathbb{Z}_{4}\times\mathbb{Z}_2^{4}\times\mathbb{Z}$ \\
89    & $P422$ & $\mathbb{Z}_2^{12}$ & $0$ & $\mathbb{Z}_2^{12}$ \\
90    & $P42_{1}2$ & $\mathbb{Z}_{4}\times\mathbb{Z}_2^{4}$ & $0$ & $\mathbb{Z}_{4}\times\mathbb{Z}_2^{4}$ \\
91    & $P4_{1}22$ & $\mathbb{Z}_2^{3}$ & $0$ & $\mathbb{Z}_2^{3}$ \\
92    & $P4_{1}2_{1}2$ & $\mathbb{Z}_2$ & $0$ & $\mathbb{Z}_2$ \\
93    & $P4_{2}22$ & $\mathbb{Z}_2^{12}$ & $0$ & $\mathbb{Z}_2^{12}$ \\
94    & $P4_{2}2_{1}2$ & $\mathbb{Z}_2^{5}$ & $0$ & $\mathbb{Z}_2^{5}$ \\
95    & $P4_{3}22$ & $\mathbb{Z}_2^{3}$ & $0$ & $\mathbb{Z}_2^{3}$ \\
96    & $P4_{3}2_{1}2$ & $\mathbb{Z}_2$ & $0$ & $\mathbb{Z}_2$ \\
97    & $I422$ & $\mathbb{Z}_2^{8}$ & $0$ & $\mathbb{Z}_2^{8}$ \\
98    & $I4_{1}22$ & $\mathbb{Z}_2^{5}$ & $0$ & $\mathbb{Z}_2^{5}$ \\
99    & $P4mm$ & $\mathbb{Z}_2^{12}$ & $\mathbb{Z}_2$      & $\mathbb{Z}_2^{13}$ \\
100   & $P4bm$ & $\mathbb{Z}_{4}\times\mathbb{Z}_2^{4}$ & $\mathbb{Z}_2$      & $\mathbb{Z}_{4}\times\mathbb{Z}_2^{5}$ \\
101   & $P4_{2}cm$ & $\mathbb{Z}_2^{6}$ & $\mathbb{Z}_2$      & $\mathbb{Z}_2^{7}$ \\
102   & $P4_{2}nm$ & $\mathbb{Z}_2^{5}$ &   $\mathbb{Z}_2$    & $\mathbb{Z}_2^{6}$ \\
103   & $P4cc$ & $\mathbb{Z}_2^{3}$ &   $\mathbb{Z}_2$    & $\mathbb{Z}_2^{4}$ \\
104   & $P4nc$ & $\mathbb{Z}_{4}\times\mathbb{Z}_2$ &  $\mathbb{Z}_2$     & $\mathbb{Z}_{4}\times\mathbb{Z}_2^2$ \\
105   & $P4_{2}mc$ & $\mathbb{Z}_2^{9}$ &   $\mathbb{Z}_2$    & $\mathbb{Z}_2^{10}$ \\
106   & $P4_{2}bc$ & $\mathbb{Z}_2^{2}$ &  $\mathbb{Z}_2$     & $\mathbb{Z}_2^{3}$ \\
107   & $I4mm$ & $\mathbb{Z}_2^{7}$ &    $\mathbb{Z}_2$   & $\mathbb{Z}_2^{8}$ \\
108   & $I4cm$ & $\mathbb{Z}_2^{4}$ &   $\mathbb{Z}_2$    & $\mathbb{Z}_2^{5}$ \\
109   & $I4_{1}md$ & $\mathbb{Z}_2^{4}$ &  $\mathbb{Z}_2$     & $\mathbb{Z}_2^{5}$ \\
110   & $I4_{1}cd$ & $\mathbb{Z}_2$ &    $\mathbb{Z}_2$   & $\mathbb{Z}_2^2$ \\
111   & $P\overline{4}2m$ & $\mathbb{Z}_2^{13}$ &  $\mathbb{Z}_2$     & $\mathbb{Z}_2^{14}$ \\
112   & $P\overline{4}2c$ & $\mathbb{Z}_2^{10}$ &  $\mathbb{Z}_2$     & $\mathbb{Z}_2^{11}$ \\
113   & $P\overline{4}2_{1}m$ & $\mathbb{Z}_{4}\times\mathbb{Z}_2^{5}$ &  $\mathbb{Z}_2$     & $\mathbb{Z}_{4}\times\mathbb{Z}_2^{6}$ \\
114   & $P\overline{4}2_{1}c$ & $\mathbb{Z}_{4}\times\mathbb{Z}_2^{2}$ & $\mathbb{Z}_2$      & $\mathbb{Z}_{4}\times\mathbb{Z}_2^{3}$ \\
115   & $P\overline{4}m2$ & $\mathbb{Z}_2^{13}$ & $\mathbb{Z}_2$      & $\mathbb{Z}_2^{14}$ \\
116   & $P\overline{4}c2$ & $\mathbb{Z}_2^{7}$ &    $\mathbb{Z}_2$   & $\mathbb{Z}_2^{8}$ \\
117   & $P\overline{4}b2$ & $\mathbb{Z}_{4}\times\mathbb{Z}_2^{5}$ & $\mathbb{Z}_2$      & $\mathbb{Z}_{4}\times\mathbb{Z}_2^{6}$ \\
118   & $P\overline{4}n2$ & $\mathbb{Z}_{4}\times\mathbb{Z}_2^{5}$ & $\mathbb{Z}_2$      & $\mathbb{Z}_{4}\times\mathbb{Z}_2^{6}$ \\
119   & $I\overline{4}m2$ & $\mathbb{Z}_2^{9}$ & $\mathbb{Z}_2$      & $\mathbb{Z}_2^{10}$ \\
120   & $I\overline{4}c2$ & $\mathbb{Z}_2^{6}$ & $\mathbb{Z}_2$      & $\mathbb{Z}_2^{7}$ \\
121   & $I\overline{4}2m$ & $\mathbb{Z}_2^{8}$ & $\mathbb{Z}_2$      & $\mathbb{Z}_2^{9}$ \\
122   & $I\overline{4}2d$ & $\mathbb{Z}_{4}\times\mathbb{Z}_2^{2}$ & $\mathbb{Z}_2$      & $\mathbb{Z}_{4}\times\mathbb{Z}_2^{3}$ \\
123   & $P4/mmm$ & $\mathbb{Z}_2^{32}$ &  $\mathbb{Z}_2$     & $\mathbb{Z}_2^{33}$ \\
124   & $P4/mcc$ & $\mathbb{Z}_2^{13}$ & $\mathbb{Z}_2$      & $\mathbb{Z}_2^{14}$ \\
125   & $P4/nbm$ & $\mathbb{Z}_2^{13}$ &  $\mathbb{Z}_2$     & $\mathbb{Z}_2^{14}$ \\
126   & $P4/nnc$ & $\mathbb{Z}_2^{8}$ &  $\mathbb{Z}_2$     & $\mathbb{Z}_2^{9}$ \\
127   & $P4/mbm$ & $\mathbb{Z}_{4}\times\mathbb{Z}_2^{15}$ &   $\mathbb{Z}_2$    & $\mathbb{Z}_{4}\times\mathbb{Z}_2^{16}$ \\
128   & $P4/mnc$ & $\mathbb{Z}_{4}\times\mathbb{Z}_2^{8}$ &  $\mathbb{Z}_2$     & $\mathbb{Z}_{4}\times\mathbb{Z}_2^{9}$ \\
129   & $P4/nmm$ & $\mathbb{Z}_2^{13}$ &  $\mathbb{Z}_2$     & $\mathbb{Z}_2^{14}$ \\
130   & $P4/ncc$ & $\mathbb{Z}_2^{5}$ &   $\mathbb{Z}_2$    & $\mathbb{Z}_2^{6}$ \\
131   & $P4_{2}/mmc$ & $\mathbb{Z}_2^{24}$ &   $\mathbb{Z}_2$    & $\mathbb{Z}_2^{25}$ \\
132   & $P4_{2}/mcm$ & $\mathbb{Z}_2^{18}$ &   $\mathbb{Z}_2$    & $\mathbb{Z}_2^{19}$ \\
133   & $P4_{2}/nbc$ & $\mathbb{Z}_2^{8}$ &    $\mathbb{Z}_2$   & $\mathbb{Z}_2^{9}$ \\
134   & $P4_{2}/nnm$ & $\mathbb{Z}_2^{13}$ &  $\mathbb{Z}_2$     & $\mathbb{Z}_2^{14}$ \\
135   & $P4_{2}/mbc$ & $\mathbb{Z}_2^{8}$ &  $\mathbb{Z}_2$     & $\mathbb{Z}_2^{9}$ \\
136   & $P4_{2}/mnm$ & $\mathbb{Z}_2^{14}$ &  $\mathbb{Z}_2$     & $\mathbb{Z}_2^{15}$ \\
137   & $P4_{2}/nmc$ & $\mathbb{Z}_2^{8}$ &   $\mathbb{Z}_2$    & $\mathbb{Z}_2^{9}$ \\
138   & $P4_{2}/ncm$ & $\mathbb{Z}_2^{10}$ &  $\mathbb{Z}_2$     & $\mathbb{Z}_2^{11}$ \\
139   & $I4/mmm$ & $\mathbb{Z}_2^{20}$ &   $\mathbb{Z}_2$    & $\mathbb{Z}_2^{21}$ \\
140   & $I4/mcm$ & $\mathbb{Z}_2^{14}$ &  $\mathbb{Z}_2$     & $\mathbb{Z}_2^{15}$ \\
141   & $I4_{1}/amd$ & $\mathbb{Z}_2^{9}$ &  $\mathbb{Z}_2$     & $\mathbb{Z}_2^{10}$ \\
142   & $I4_{1}/acd$ & $\mathbb{Z}_2^{5}$ & $\mathbb{Z}_2$      & $\mathbb{Z}_2^{6}$ \\
143   & $P3$  & $\mathbb{Z}_{3}^{3}$ & {$\mathbb{Z}$} & $\mathbb{Z}_{3}^{3}\times\mathbb{Z}$ \\
144   & $P3_{1}$ & $0$ & {$\mathbb{Z}$} & $\mathbb{Z}$ \\
145   & $P3_{2}$ & $0$ & {$\mathbb{Z}$} & $\mathbb{Z}$ \\
146   & $R3$  & $\mathbb{Z}_{3}$ & {$\mathbb{Z}$} & $\mathbb{Z}_{3}\times\mathbb{Z}$ \\
147   & $P\overline{3}$ & $\mathbb{Z}_{3}^{2}\times\mathbb{Z}_2^{4}$ & {$\mathbb{Z}_2\times\mathbb{Z}$} & $\mathbb{Z}_{3}^{2}\times\mathbb{Z}_2^{5}\times\mathbb{Z}$ \\
148   & $R\overline{3}$ & $\mathbb{Z}_{3}\times\mathbb{Z}_2^{4}$ & {$\mathbb{Z}_2\times\mathbb{Z}$} & $\mathbb{Z}_{3}\times\mathbb{Z}_2^{5}\times\mathbb{Z}$ \\
149   & $P312$ & $\mathbb{Z}_2^{2}$ & $0$ & $\mathbb{Z}_2^{2}$ \\
150   & $P321$ & $\mathbb{Z}_{3}\times\mathbb{Z}_2^{2}$ & $0$ & $\mathbb{Z}_{3}\times\mathbb{Z}_2^{2}$ \\
151   & $P3_{1}12$ & $\mathbb{Z}_2^{2}$ & $0$ & $\mathbb{Z}_2^{2}$ \\
152   & $P3_{1}21$ & $\mathbb{Z}_2^{2}$ & $0$ & $\mathbb{Z}_2^{2}$ \\
153   & $P3_{2}12$ & $\mathbb{Z}_2^{2}$ & $0$ & $\mathbb{Z}_2^{2}$ \\
154   & $P3_{2}21$ & $\mathbb{Z}_2^{2}$ & $0$ & $\mathbb{Z}_2^{2}$ \\
155   & $R32$ & $\mathbb{Z}_2^{2}$ & $0$ & $\mathbb{Z}_2^{2}$ \\
156   & $P3m1$ & $\mathbb{Z}_2^2$ & $\mathbb{Z}_2$      & $\mathbb{Z}_2^3$ \\
157   & $P31m$ & $\mathbb{Z}_{3}\times\mathbb{Z}_2^2$ & $\mathbb{Z}_2$      & $\mathbb{Z}_{3}\times\mathbb{Z}_2^3$ \\
158   & $P3c1$ & $0$ & $\mathbb{Z}_2$      &  $\mathbb{Z}_2$\\
159   & $P31c$ & $\mathbb{Z}_{3}$ & $\mathbb{Z}_2$      & $\mathbb{Z}_{3}\times\mathbb{Z}_2$ \\
160   & $R3m$ & $\mathbb{Z}_2^2$ &  $\mathbb{Z}_2$     & $\mathbb{Z}_2^3$ \\
161   & $R3c$ & $0$ & $\mathbb{Z}_2$      & $\mathbb{Z}_2$ \\
162   & $P\overline{3}1m$ & $\mathbb{Z}_2^{9}$ &  $\mathbb{Z}_2$     & $\mathbb{Z}_2^{10}$ \\
163   & $P\overline{3}1c$ & $\mathbb{Z}_2^{3}$ & $\mathbb{Z}_2$      & $\mathbb{Z}_2^{4}$ \\
164   & $P\overline{3}m1$ & $\mathbb{Z}_2^{9}$ & $\mathbb{Z}_2$      & $\mathbb{Z}_2^{10}$ \\
165   & $P\overline{3}c1$ & $\mathbb{Z}_2^{3}$ & $\mathbb{Z}_2$      & $\mathbb{Z}_2^{4}$ \\
166   & $R\overline{3}m$ & $\mathbb{Z}_2^{9}$ &  $\mathbb{Z}_2$     & $\mathbb{Z}_2^{10}$ \\
167   & $R\overline{3}c$ & $\mathbb{Z}_2^{3}$ &$\mathbb{Z}_2$       & $\mathbb{Z}_2^{4}$ \\
168   & $P6$  & $\mathbb{Z}_{3}^{2}\times\mathbb{Z}_2^{2}$ & {$\mathbb{Z}$} & $\mathbb{Z}_{3}^{2}\times\mathbb{Z}_2^{2}\times\mathbb{Z}$ \\
169   & $P6_{1}$ & $0$ & {$\mathbb{Z}$} & $\mathbb{Z}$ \\
170   & $P6_{5}$ & $0$ & {$\mathbb{Z}$} & $\mathbb{Z}$ \\
171   & $P6_{2}$ & $\mathbb{Z}_2^{2}$ & {$\mathbb{Z}$} & $\mathbb{Z}_2^{2}\times\mathbb{Z}$ \\
172   & $P6_{4}$ & $\mathbb{Z}_2^{2}$ & {$\mathbb{Z}$} & $\mathbb{Z}_2^{2}\times\mathbb{Z}$ \\
173   & $P6_{3}$ & $\mathbb{Z}_{3}^{2}$ & {$\mathbb{Z}$} & $\mathbb{Z}_{3}^{2}\times\mathbb{Z}$ \\
174   & $P\overline{6}$ & $\mathbb{Z}_{3}^{3}\times\mathbb{Z}_2^{4}$ & {$\mathbb{Z}_2\times\mathbb{Z}$} & $\mathbb{Z}_{3}^{3}\times\mathbb{Z}_2^{5}\times\mathbb{Z}$ \\
175   & $P6/m$ & $\mathbb{Z}_{3}^{2}\times\mathbb{Z}_2^{10}$ & {$\mathbb{Z}_2\times\mathbb{Z}$} & $\mathbb{Z}_{3}^{2}\times\mathbb{Z}_2^{11}\times\mathbb{Z}$ \\
176   & $P6_{3}/m$ & $\mathbb{Z}_{3}^{2}\times\mathbb{Z}_2^{4}$ & {$\mathbb{Z}_2\times\mathbb{Z}$} & $\mathbb{Z}_{3}^{2}\times\mathbb{Z}_2^{5}\times\mathbb{Z}$ \\
177   & $P622$ & $\mathbb{Z}_2^{8}$ & $0$ & $\mathbb{Z}_2^{8}$ \\
178   & $P6_{1}22$ & $\mathbb{Z}_2^{2}$ & $0$ & $\mathbb{Z}_2^{2}$ \\
179   & $P6_{5}22$ & $\mathbb{Z}_2^{2}$ & $0$ & $\mathbb{Z}_2^{2}$ \\
180   & $P6_{2}22$ & $\mathbb{Z}_2^{8}$ & $0$ & $\mathbb{Z}_2^{8}$ \\
181   & $P6_{4}22$ & $\mathbb{Z}_2^{8}$ & $0$ & $\mathbb{Z}_2^{8}$ \\
182   & $P6_{3}22$ & $\mathbb{Z}_2^{2}$ & $0$ & $\mathbb{Z}_2^{2}$ \\
183   & $P6mm$ & $\mathbb{Z}_2^{8}$ & $\mathbb{Z}_2$      & $\mathbb{Z}_2^{9}$ \\
184   & $P6cc$ & $\mathbb{Z}_2^{2}$ & $\mathbb{Z}_2$      & $\mathbb{Z}_2^{3}$ \\
185   & $P6_{3}cm$ & $\mathbb{Z}_2^2$ & $\mathbb{Z}_2$      & $\mathbb{Z}_2^3$ \\
186   & $P6_{3}mc$ & $\mathbb{Z}_2^2$ & $\mathbb{Z}_2$      & $\mathbb{Z}_2^3$ \\
187   & $P\overline{6}m2$ & $\mathbb{Z}_2^{9}$ &  $\mathbb{Z}_2$     & $\mathbb{Z}_2^{10}$ \\
188   & $P\overline{6}c2$ & $\mathbb{Z}_2^{3}$ & $\mathbb{Z}_2$      & $\mathbb{Z}_2^{4}$ \\
189   & $P\overline{6}2m$ & $\mathbb{Z}_{3}\times\mathbb{Z}_2^{9}$ & $\mathbb{Z}_2$      & $\mathbb{Z}_{3}\times\mathbb{Z}_2^{10}$ \\
190   & $P\overline{6}2c$ & $\mathbb{Z}_{3}\times\mathbb{Z}_2^{3}$ & $\mathbb{Z}_2$      & $\mathbb{Z}_{3}\times\mathbb{Z}_2^{4}$ \\
191   & $P6/mmm$ & $\mathbb{Z}_2^{22}$ & $\mathbb{Z}_2$      & $\mathbb{Z}_2^{23}$ \\
192   & $P6/mcc$ & $\mathbb{Z}_2^{9}$ & $\mathbb{Z}_2$      & $\mathbb{Z}_2^{10}$ \\
193   & $P6_{3}/mcm$ & $\mathbb{Z}_2^{9}$ &  $\mathbb{Z}_2$     & $\mathbb{Z}_2^{10}$ \\
194   & $P6_{3}/mmc$ & $\mathbb{Z}_2^{9}$ &  $\mathbb{Z}_2$     & $\mathbb{Z}_2^{10}$ \\
195   & $P23$ & $\mathbb{Z}_{3}\times\mathbb{Z}_2^{4}$ & $0$ & $\mathbb{Z}_{3}\times\mathbb{Z}_2^{4}$ \\
196   & $F23$ & $\mathbb{Z}_{3}$ & $0$ & $\mathbb{Z}_{3}$ \\
197   & $I23$ & $\mathbb{Z}_{3}\times\mathbb{Z}_2^{2}$ & $0$ & $\mathbb{Z}_{3}\times\mathbb{Z}_2^{2}$ \\
198   & $P2_{1}3$ & $\mathbb{Z}_{3}$ & $0$ & $\mathbb{Z}_{3}$ \\
199   & $I2_{1}3$ & $\mathbb{Z}_{3}\times\mathbb{Z}_2$ & $0$ & $\mathbb{Z}_{3}\times\mathbb{Z}_2$ \\
200   & $Pm\overline{3}$ & $\mathbb{Z}_{3}\times\mathbb{Z}_2^{14}$ & $\mathbb{Z}_2$      & $\mathbb{Z}_{3}\times\mathbb{Z}_2^{15}$ \\
201   & $Pn\overline{3}$ & $\mathbb{Z}_{3}\times\mathbb{Z}_2^{4}$ &  $\mathbb{Z}_2$     & $\mathbb{Z}_{3}\times\mathbb{Z}_2^{5}$ \\
202   & $Fm\overline{3}$ & $\mathbb{Z}_{3}\times\mathbb{Z}_2^{6}$ & $\mathbb{Z}_2$      & $\mathbb{Z}_{3}\times\mathbb{Z}_2^{7}$ \\
203   & $Fd\overline{3}$ & $\mathbb{Z}_{3}\times\mathbb{Z}_2^{2}$ & $\mathbb{Z}_2$      & $\mathbb{Z}_{3}\times\mathbb{Z}_2^{3}$ \\
204   & $Im\overline{3}$ & $\mathbb{Z}_{3}\times\mathbb{Z}_2^8$ &  $\mathbb{Z}_2$     & $\mathbb{Z}_{3}\times\mathbb{Z}_2^9$ \\
205   & $Pa\overline{3}$ & $\mathbb{Z}_{3}\times\mathbb{Z}_2^{2}$ &  $\mathbb{Z}_2$     & $\mathbb{Z}_{3}\times\mathbb{Z}_2^{3}$ \\
206   & $Ia\overline{3}$ & $\mathbb{Z}_{3}\times\mathbb{Z}_2^{3}$ &  $\mathbb{Z}_2$     & $\mathbb{Z}_{3}\times\mathbb{Z}_2^{4}$ \\
207   & $P432$ & $\mathbb{Z}_2^{6}$ & $0$ & $\mathbb{Z}_2^{6}$ \\
208   & $P4_{2}32$ & $\mathbb{Z}_2^{6}$ & $0$ & $\mathbb{Z}_2^{6}$ \\
209   & $F432$ & $\mathbb{Z}_2^{4}$ & $0$ & $\mathbb{Z}_2^{4}$ \\
210   & $F4_{1}32$ & $\mathbb{Z}_2$ & $0$ & $\mathbb{Z}_2$ \\
211   & $I432$ & $\mathbb{Z}_2^{5}$ & $0$ & $\mathbb{Z}_2^{5}$ \\
212   & $P4_{3}32$ & $\mathbb{Z}_2$ & $0$ & $\mathbb{Z}_2$ \\
213   & $P4_{1}32$ & $\mathbb{Z}_2$ & $0$ & $\mathbb{Z}_2$ \\
214   & $I4_{1}32$ & $\mathbb{Z}_2^{4}$ & $0$ & $\mathbb{Z}_2^{4}$ \\
215   & $P\overline{4}3m$ & $\mathbb{Z}_2^{7}$ &  $\mathbb{Z}_2$     & $\mathbb{Z}_2^8$ \\
216   & $F\overline{4}3m$ & $\mathbb{Z}_2^5$ & $\mathbb{Z}_2$      & $\mathbb{Z}_2^6$ \\
217   & $I\overline{4}3m$ & $\mathbb{Z}_2^5$ & $\mathbb{Z}_2$      & $\mathbb{Z}_2^6$ \\
218   & $P\overline{4}3n$ & $\mathbb{Z}_2^{4}$ & $\mathbb{Z}_2$      & $\mathbb{Z}_2^{5}$ \\
219   & $F\overline{4}3c$ & $\mathbb{Z}_2^{2}$ &   $\mathbb{Z}_2$    & $\mathbb{Z}_2^{3}$ \\
220   & $I\overline{4}3d$ & $\mathbb{Z}_{4}\times\mathbb{Z}_2$ & $\mathbb{Z}_2$      & $\mathbb{Z}_{4}\times\mathbb{Z}_2^2$ \\
221   & $Pm\overline{3}m$ & $\mathbb{Z}_2^{18}$ & $\mathbb{Z}_2$      & $\mathbb{Z}_2^{19}$ \\
222   & $Pn\overline{3}n$ & $\mathbb{Z}_2^{5}$ &   $\mathbb{Z}_2$    & $\mathbb{Z}_2^{6}$ \\
223   & $Pm\overline{3}n$ & $\mathbb{Z}_2^{10}$ &   $\mathbb{Z}_2$    & $\mathbb{Z}_2^{11}$ \\
224   & $Pn\overline{3}m$ & $\mathbb{Z}_2^{10}$ & $\mathbb{Z}_2$      & $\mathbb{Z}_2^{11}$ \\
225   & $Fm\overline{3}m$ & $\mathbb{Z}_2^{13}$ &  $\mathbb{Z}_2$     & $\mathbb{Z}_2^{14}$ \\
226   & $Fm\overline{3}c$ & $\mathbb{Z}_2^{7}$ & $\mathbb{Z}_2$      & $\mathbb{Z}_2^8$ \\
227   & $Fd\overline{3}m$ & $\mathbb{Z}_2^7$ &   $\mathbb{Z}_2$    & $\mathbb{Z}_2^8$ \\
228   & $Fd\overline{3}c$ & $\mathbb{Z}_2^{3}$ &   $\mathbb{Z}_2$    & $\mathbb{Z}_2^{4}$ \\
229   & $Im\overline{3}m$ & $\mathbb{Z}_2^{13}$ &  $\mathbb{Z}_2$     & $\mathbb{Z}_2^{14}$ \\
230   & $Ia\overline{3}d$ & $\mathbb{Z}_2^{4}$ &    $\mathbb{Z}_2$   & $\mathbb{Z}_2^{5}$
\end{longtable}
\end{appendices}

\backmatter


\begin{thebibliography}{100}

\bibitem{Klitzing_IQHE}
K.~v. Klitzing, G.~Dorda, and M.~Pepper.
\newblock {New Method for High-Accuracy Determination of the Fine-Structure
  Constant Based on Quantized Hall Resistance}.
\newblock {\em Phys. Rev. Lett.}, 45:494--497, Aug 1980.

\bibitem{FQHE_original}
D.~C. Tsui, H.~L. Stormer, and A.~C. Gossard.
\newblock {Two-Dimensional Magnetotransport in the Extreme Quantum Limit}.
\newblock {\em Phys. Rev. Lett.}, 48:1559--1562, May 1982.

\bibitem{TKNN}
D.~J. Thouless, M.~Kohmoto, M.~P. Nightingale, and M.~den Nijs.
\newblock {Quantized Hall Conductance in a Two-Dimensional Periodic Potential}.
\newblock {\em Phys. Rev. Lett.}, 49:405--408, Aug 1982.

\bibitem{Haldane1988}
F.~D.~M. Haldane.
\newblock {Model for a Quantum Hall Effect without Landau Levels:
  Condensed-Matter Realization of the "Parity Anomaly"}.
\newblock {\em Phys. Rev. Lett.}, 61:2015--2018, Oct 1988.

\bibitem{laughlin1981}
R.~B. Laughlin.
\newblock {Quantized Hall conductivity in two dimensions}.
\newblock {\em Phys. Rev. B}, 23:5632--5633, May 1981.

\bibitem{Laughlin_FQHE}
R.~B. Laughlin.
\newblock {Anomalous Quantum Hall Effect: An Incompressible Quantum Fluid with
  Fractionally Charged Excitations}.
\newblock {\em Phys. Rev. Lett.}, 50:1395--1398, May 1983.

\bibitem{AKLT}
Ian Affleck, Tom Kennedy, Elliott~H. Lieb, and Hal Tasaki.
\newblock Rigorous results on valence-bond ground states in antiferromagnets.
\newblock {\em Phys. Rev. Lett.}, 59:799--802, Aug 1987.

\bibitem{PhysRevLett.50.1153}
F.~D.~M. Haldane.
\newblock {Nonlinear Field Theory of Large-Spin Heisenberg Antiferromagnets:
  Semiclassically Quantized Solitons of the One-Dimensional Easy-Axis N\'eel
  State}.
\newblock {\em Phys. Rev. Lett.}, 50:1153--1156, Apr 1983.

\bibitem{Haldane_NLSM}
F.~D.~M. Haldane.
\newblock {Continuum dynamics of the 1-D Heisenberg antiferromagnet:
  Identification with the O(3) nonlinear sigma model}.
\newblock {\em Physics Letters A}, 93(9):464 -- 468, 1983.

\bibitem{Affleck_Haldane}
Ian Affleck and F.~D.~M. Haldane.
\newblock Critical theory of quantum spin chains.
\newblock {\em Phys. Rev. B}, 36:5291--5300, Oct 1987.

\bibitem{Haldane_gap}
Ian Affleck.
\newblock {Quantum spin chains and the Haldane gap}.
\newblock {\em Journal of Physics: Condensed Matter}, 1(19):3047, 1989.

\bibitem{Wen_1d}
Xie Chen, Zheng-Cheng Gu, and Xiao-Gang Wen.
\newblock {Classification of gapped symmetric phases in one-dimensional spin
  systems}.
\newblock {\em Phys. Rev. B}, 83:035107, Jan 2011.

\bibitem{Cirac}
Norbert Schuch, David P\'erez-Garc\'{\i}a, and Ignacio Cirac.
\newblock {Classifying quantum phases using matrix product states and projected
  entangled pair states}.
\newblock {\em Phys. Rev. B}, 84:165139, Oct 2011.

\bibitem{zhang2010exact}
J.~M. Zhang and R.~X. Dong.
\newblock {Exact diagonalization: the Bose{\textendash}Hubbard model as an
  example}.
\newblock {\em European Journal of Physics}, 31(3):591--602, Apr 2010.

\bibitem{weisse2008exact}
Alexander Wei{\ss}e and Holger Fehske.
\newblock {\em {Exact Diagonalization Techniques}}, pages 529--544.
\newblock Springer Berlin Heidelberg, Berlin, Heidelberg, 2008.

\bibitem{Wen_Definition}
Xie Chen, Zheng-Cheng Gu, and Xiao-Gang Wen.
\newblock {Local unitary transformation, long-range quantum entanglement, wave
  function renormalization, and topological order}.
\newblock {\em Phys. Rev. B}, 82:155138, Oct 2010.

\bibitem{Kitaev_TI}
Alexei Kitaev.
\newblock {Periodic table for topological insulators and superconductors}.
\newblock {\em AIP Conference Proceedings}, 1134(1):22--30, 2009.

\bibitem{Hasan_Kane}
M.~Z. Hasan and C.~L. Kane.
\newblock {Colloquium: Topological insulators}.
\newblock {\em Rev. Mod. Phys.}, 82:3045--3067, Nov 2010.

\bibitem{Kitaev_honeycomb}
Alexei Kitaev.
\newblock Anyons in an exactly solved model and beyond.
\newblock {\em Annals of Physics}, 321(1):2--111, 2006.
\newblock January Special Issue.

\bibitem{bakalov2001lectures}
Bojko Bakalov and Alexander~A Kirillov.
\newblock {\em Lectures on tensor categories and modular functors}.
\newblock American Mathematical Society, 2001.

\bibitem{rowell2009classification}
Eric Rowell, Richard Stong, and Zhenghan Wang.
\newblock {On Classification of Modular Tensor Categories}.
\newblock {\em Communications in Mathematical Physics}, 292(2):343--389, Dec
  2009.

\bibitem{Wen_Boson}
Xie Chen, Zheng-Cheng Gu, Zheng-Xin Liu, and Xiao-Gang Wen.
\newblock {Symmetry protected topological orders and the group cohomology of
  their symmetry group}.
\newblock {\em Phys. Rev. B}, 87:155114, Apr 2013.

\bibitem{Kapustin_Boson}
Anton Kapustin.
\newblock {Symmetry Protected Topological Phases, Anomalies, and Cobordisms:
  Beyond Group Cohomology}.
\newblock {\em arXiv:1403.1467}, 2014.

\bibitem{Wen_Fermion}
Zheng-Cheng Gu and Xiao-Gang Wen.
\newblock {Symmetry-protected topological orders for interacting fermions:
  Fermionic topological nonlinear $\ensuremath{\sigma}$ models and a special
  group supercohomology theory}.
\newblock {\em Phys. Rev. B}, 90:115141, Sep 2014.

\bibitem{Kapustin_Fermion}
Anton Kapustin, Ryan Thorngren, Alex Turzillo, and Zitao Wang.
\newblock Fermionic symmetry protected topological phases and cobordisms.
\newblock {\em Journal of High Energy Physics}, 2015(12):1--21, Dec 2015.

\bibitem{Freed_SRE_iTQFT}
Daniel~S. Freed.
\newblock {Short-range entanglement and invertible field theories}.
\newblock {\em arXiv:1406.7278}, 2014.

\bibitem{Freed_ReflectionPositivity}
Daniel~S. Freed and Michael~J. Hopkins.
\newblock {Reflection positivity and invertible topological phases}.
\newblock {\em arXiv:1604.06527}, 2016.

\bibitem{wang2018towards}
Qing-Rui Wang and Zheng-Cheng Gu.
\newblock {Towards a Complete Classification of Symmetry-Protected Topological
  Phases for Interacting Fermions in Three Dimensions and a General Group
  Supercohomology Theory}.
\newblock {\em Phys. Rev. X}, 8:011055, Mar 2018.

\bibitem{Huang_dimensional_reduction}
Sheng-Jie Huang, Hao Song, Yi-Ping Huang, and Michael Hermele.
\newblock Building crystalline topological phases from lower-dimensional
  states.
\newblock {\em Phys. Rev. B}, 96:205106, Nov 2017.

\bibitem{Kitaev_Stony_Brook_2011_SRE_1}
A.~Kitaev.
\newblock {Toward a topological classification of many-body quantum states with
  short-range entanglement}.
\newblock In {\em {Topological Quantum Computing Workshop}}, Stony Brook, New
  York, Sep 2011. Simons Center for Geometry and Physics, Stony Brook
  University.

\bibitem{Kitaev_Stony_Brook_2013_SRE}
A.~Kitaev.
\newblock {On the Classification of Short-Range Entangled States}.
\newblock In {\em {Topological Phases of Matter Program Seminar}}, Stony Brook,
  New York, Jun 2013. Simons Center for Geometry and Physics, Stony Brook
  University.

\bibitem{Kitaev_IPAM}
A.~Kitaev.
\newblock {Homotopy-theoretic approach to {SPT} phases in action: $Z_{16}$
  classification of three-dimensional superconductors}.
\newblock In {\em Symmetry and Topology in Quantum Matter Workshop}, Los
  Angeles, California, Jan 2015. Institute for Pure \& Applied Mathematics,
  University of California.

\bibitem{Xiong}
Charles~Zhaoxi Xiong.
\newblock Minimalist approach to the classification of symmetry protected
  topological phases.
\newblock {\em Journal of Physics A: Mathematical and Theoretical},
  51(44):445001, 2018.

\bibitem{Xiong_Alexandradinata}
Charles~Zhaoxi Xiong and A.~Alexandradinata.
\newblock Organizing symmetry-protected topological phases by layering and
  symmetry reduction: A minimalist perspective.
\newblock {\em Phys. Rev. B}, 97:115153, Mar 2018.

\bibitem{Shiozaki2018}
Ken {Shiozaki}, Charles~Zhaoxi {Xiong}, and Kiyonori {Gomi}.
\newblock {Generalized homology and Atiyah-Hirzebruch spectral sequence in
  crystalline symmetry protected topological phenomena}.
\newblock {\em arXiv:1810.00801}, Oct 2018.

\bibitem{SongXiongHuang}
Hao Song, Charles~Zhaoxi Xiong, and Sheng-Jie Huang.
\newblock {Bosonic Crystalline Symmetry Protected Topological Phases Beyond the
  Group Cohomology Proposal}.
\newblock {\em arXiv:1811.06558}, 2018.

\bibitem{Ma_discoverhourglass}
Junzhang Ma, Changjiang Yi, Baiqing Lv, ZhiJun Wang, Simin Nie, Le~Wang,
  Lingyuan Kong, Yaobo Huang, Pierre Richard, Peng Zhang, Koichiro Yaji, Kenta
  Kuroda, Shik Shin, Hongming Weng, Bogdan~Andrei Bernevig, Youguo Shi, Tian
  Qian, and Hong Ding.
\newblock {Experimental evidence of hourglass fermion in the candidate
  nonsymmorphic topological insulator KHgSb}.
\newblock {\em Science Advances}, 3(5):e1602415, 2017.

\bibitem{Hourglass}
Zhijun Wang, A.~Alexandradinata, R.~J. Cava, and B.~Andrei Bernevig.
\newblock Hourglass fermions.
\newblock {\em Nature}, 532:189--194, Apr 2016.

\bibitem{Cohomological}
A.~Alexandradinata, Zhijun Wang, and B.~Andrei Bernevig.
\newblock Topological insulators from group cohomology.
\newblock {\em Phys. Rev. X}, 6:021008, Apr 2016.

\bibitem{2dChiralBosonicSPT}
Yuan-Ming Lu and Ashvin Vishwanath.
\newblock {Theory and classification of interacting integer topological phases
  in two dimensions: A Chern-Simons approach}.
\newblock {\em Phys. Rev. B}, 86:125119, Sep 2012.

\bibitem{Kitaev_KITP}
A.~Kitaev.
\newblock {Toward Topological Classification of Phases with Short-range
  Entanglement}.
\newblock In {\em {Topological Insulators and Superconductors Workshop}}, Santa
  Barbara, California, 2011. Kavli Institute for Theoretical Physics,
  University of California.

\bibitem{fundamental_domain}
E.~Moln\'ar, I.~Prok, and J.~Szirmai.
\newblock {D-V cells and fundamental domains for crystallographic groups,
  algorithms, and graphic realizations}.
\newblock {\em Mathematical and Computer Modelling}, 38(7):929 -- 943, 2003.
\newblock Hungarian Applied Mathematics.

\bibitem{MPS_fundamental_theorem}
J~Ignacio Cirac, D~Perez-Garcia, Norbert Schuch, and F~Verstraete.
\newblock Matrix product density operators: Renormalization fixed points and
  boundary theories.
\newblock {\em Annals of Physics}, 378:100--149, 2017.

\bibitem{MPS_geometry}
Jutho Haegeman, Micha{\"e}l Mari{\"e}n, Tobias~J Osborne, and Frank Verstraete.
\newblock Geometry of matrix product states: Metric, parallel transport, and
  curvature.
\newblock {\em Journal of Mathematical Physics}, 55(2):021902, 2014.

\bibitem{Kitaev_Stony_Brook_2011_SRE_2}
A.~Kitaev.
\newblock {Conclusion: Toward a topological classification of many-body quantum
  states with short-range entanglement}.
\newblock In {\em {Topological Quantum Computing Workshop}}, Stony Brook, New
  York, Sep 2011. Simons Center for Geometry and Physics, Stony Brook
  University.

\bibitem{McGreevy_sSourcery}
Brian Swingle and John McGreevy.
\newblock Renormalization group constructions of topological quantum liquids
  and beyond.
\newblock {\em Phys. Rev. B}, 93:045127, Jan 2016.

\bibitem{Majorana_chain}
A.~Yu Kitaev.
\newblock {Unpaired Majorana fermions in quantum wires}.
\newblock {\em Physics-Uspekhi}, 44(10S):131, 2001.

\bibitem{Volovik_p+ip}
G.~E. Volovik.
\newblock Fermion zero modes on vortices in chiral superconductors.
\newblock {\em Journal of Experimental and Theoretical Physics Letters},
  70(9):609--614, 1999.

\bibitem{Read_p+ip}
N.~Read and Dmitry Green.
\newblock Paired states of fermions in two dimensions with breaking of parity
  and time-reversal symmetries and the fractional quantum hall effect.
\newblock {\em Phys. Rev. B}, 61:10267--10297, Apr 2000.

\bibitem{Ivanov_p+ip}
D.~A. Ivanov.
\newblock {Non-Abelian Statistics of Half-Quantum Vortices in $\mathit{p}$-Wave
  Superconductors}.
\newblock {\em Phys. Rev. Lett.}, 86:268--271, Jan 2001.

\bibitem{Hatcher}
Allen Hatcher.
\newblock {\em Algebraic Topology}.
\newblock Cambridge University Press, Cambridge, 2002.

\bibitem{Adams1}
John~Frank Adams.
\newblock {\em {Infinite Loop Spaces (AM-90): Hermann Weyl Lectures, The
  Institute for Advanced Study.(AM-90)}}, volume~90.
\newblock Princeton University Press, Princeton, 1978.

\bibitem{Adams2}
John~Frank Adams.
\newblock {\em {Stable Homotopy and Generalised Homology}}.
\newblock University of Chicago press, Chicago, 1995.

\bibitem{May}
J~Peter May, LG~Lewis, Robert~John Piacenza, M~Cole, G~Comezana, S~Costenoble,
  Anthony~D Elmendorf, and JPC Greenlees.
\newblock {\em Equivariant homotopy and cohomology theory: Dedicated to the
  memory of Robert J. Piacenza}.
\newblock Number~91. American Mathematical Soc., 1996.

\bibitem{AdemMilgram}
A.~Adem and R.~J. Milgram.
\newblock {\em {Cohomology of Finite Groups}}.
\newblock Springer-Verlag Berlin Heidelberg, New York, 2 edition, 2004.
\newblock {Chap.\ II.}

\bibitem{decorated_domain_walls}
Xie Chen, Yuan-Ming Lu, and Ashvin Vishwanath.
\newblock Symmetry-protected topological phases from decorated domain walls.
\newblock {\em Nature communications}, 5:3507, 2014.

\bibitem{Gaiotto_Johnson-Freyd}
Davide Gaiotto and Theo Johnson-Freyd.
\newblock {Symmetry Protected Topological phases and Generalized Cohomology}.
\newblock {\em arXiv:1712.07950}, 2017.

\bibitem{ThorngrenElse}
Ryan Thorngren and Dominic~V. Else.
\newblock Gauging spatial symmetries and the classification of topological
  crystalline phases.
\newblock {\em Phys. Rev. X}, 8:011040, Mar 2018.

\bibitem{Keyserlingk_Floquet}
C.~W. von Keyserlingk and S.~L. Sondhi.
\newblock Phase structure of one-dimensional interacting floquet systems. i.
  abelian symmetry-protected topological phases.
\newblock {\em Phys. Rev. B}, 93:245145, Jun 2016.

\bibitem{Else_Floquet}
Dominic~V. Else and Chetan Nayak.
\newblock Classification of topological phases in periodically driven
  interacting systems.
\newblock {\em Phys. Rev. B}, 93:201103, May 2016.

\bibitem{Potter_Floquet}
Andrew~C. Potter, Takahiro Morimoto, and Ashvin Vishwanath.
\newblock Classification of interacting topological floquet phases in one
  dimension.
\newblock {\em Phys. Rev. X}, 6:041001, Oct 2016.

\bibitem{Jiang_sgSPT}
Shenghan Jiang and Ying Ran.
\newblock Anyon condensation and a generic tensor-network construction for
  symmetry-protected topological phases.
\newblock {\em Phys. Rev. B}, 95:125107, Mar 2017.

\bibitem{Wen_sgSPT_1d}
Xie Chen, Zheng-Cheng Gu, and Xiao-Gang Wen.
\newblock Complete classification of one-dimensional gapped quantum phases in
  interacting spin systems.
\newblock {\em Phys. Rev. B}, 84:235128, Dec 2011.

\bibitem{SPt}
Yohei Fuji, Frank Pollmann, and Masaki Oshikawa.
\newblock {Distinct Trivial Phases Protected by a Point-Group Symmetry in
  Quantum Spin Chains}.
\newblock {\em Phys. Rev. Lett.}, 114:177204, May 2015.

\bibitem{You_sgSPT}
Yi-Zhuang You and Cenke Xu.
\newblock Symmetry-protected topological states of interacting fermions and
  bosons.
\newblock {\em Phys. Rev. B}, 90:245120, Dec 2014.

\bibitem{Yoshida_sgSPT}
Tsuneya Yoshida, Takahiro Morimoto, and Akira Furusaki.
\newblock Bosonic symmetry-protected topological phases with reflection
  symmetry.
\newblock {\em Phys. Rev. B}, 92:245122, Dec 2015.

\bibitem{Hsieh_sgSPT}
Chang-Tse Hsieh, Olabode~Mayodele Sule, Gil~Young Cho, Shinsei Ryu, and
  Robert~G. Leigh.
\newblock Symmetry-protected topological phases, generalized laughlin argument,
  and orientifolds.
\newblock {\em Phys. Rev. B}, 90:165134, Oct 2014.

\bibitem{Cho_sgSPT}
Gil~Young Cho, Chang-Tse Hsieh, Takahiro Morimoto, and Shinsei Ryu.
\newblock Topological phases protected by reflection symmetry and cross-cap
  states.
\newblock {\em Phys. Rev. B}, 91:195142, May 2015.

\bibitem{Hermele_torsor}
Hao Song, Sheng-Jie Huang, Liang Fu, and Michael Hermele.
\newblock Topological phases protected by point group symmetry.
\newblock {\em Phys. Rev. X}, 7:011020, Feb 2017.

\bibitem{Fidkowski_Kitaev_1}
Lukasz Fidkowski and Alexei Kitaev.
\newblock Effects of interactions on the topological classification of free
  fermion systems.
\newblock {\em Phys. Rev. B}, 81:134509, Apr 2010.

\bibitem{Fidkowski_Kitaev_2}
Lukasz Fidkowski and Alexei Kitaev.
\newblock Topological phases of fermions in one dimension.
\newblock {\em Phys. Rev. B}, 83:075103, Feb 2011.

\bibitem{WangChong_3DSPTAII}
Chong Wang, Andrew~C. Potter, and T.~Senthil.
\newblock Classification of interacting electronic topological insulators in
  three dimensions.
\newblock {\em Science}, 343(6171):629--631, 2014.

\bibitem{3dBTScVishwanathSenthil}
Ashvin Vishwanath and T.~Senthil.
\newblock {Physics of Three-Dimensional Bosonic Topological Insulators:
  Surface-Deconfined Criticality and Quantized Magnetoelectric Effect}.
\newblock {\em Phys. Rev. X}, 3:011016, Feb 2013.

\bibitem{3dBTScWangSenthil}
Chong Wang and T.~Senthil.
\newblock {Boson topological insulators: A window into highly entangled quantum
  phases}.
\newblock {\em Phys. Rev. B}, 87:235122, Jun 2013.

\bibitem{3dBTScBurnell}
F.~J. Burnell, Xie Chen, Lukasz Fidkowski, and Ashvin Vishwanath.
\newblock {Exactly soluble model of a three-dimensional symmetry-protected
  topological phase of bosons with surface topological order}.
\newblock {\em Phys. Rev. B}, 90:245122, Dec 2014.

\bibitem{ChaoxingNonsymm}
Chao-Xing Liu, Rui-Xing Zhang, and Brian~K. VanLeeuwen.
\newblock Topological nonsymmorphic crystalline insulators.
\newblock {\em Phys. Rev. B}, 90:085304, 2014.

\bibitem{unpinned}
Chen Fang and Liang Fu.
\newblock New classes of three-dimensional topological crystalline insulators:
  Nonsymmorphic and magnetic.
\newblock {\em Phys. Rev. B}, 91:161105, 2015.

\bibitem{Shiozaki2015}
Ken Shiozaki, Masatoshi Sato, and Kiyonori Gomi.
\newblock ${Z}_{2}$ topology in nonsymmorphic crystalline insulators: M\"obius
  twist in surface states.
\newblock {\em Phys. Rev. B}, 91:155120, 2015.

\bibitem{Nonsymm_Shiozaki}
Ken Shiozaki, Masatoshi Sato, and Kiyonori Gomi.
\newblock Topology of nonsymmorphic crystalline insulators and superconductors.
\newblock {\em Phys. Rev. B}, 93:195413, May 2016.

\bibitem{Poyao_mobiuskondo}
Po-Yao Chang, Onur Erten, and Piers Coleman.
\newblock Möbius kondo insulators.
\newblock {\em Nature Physics (2017) doi:10.1038/nphys4092}.

\bibitem{Ezawa_hourglass}
Motohiko Ezawa.
\newblock Hourglass fermion surface states in stacked topological insulators
  with nonsymmorphic symmetry.
\newblock {\em Phys. Rev. B}, 94:155148, Oct 2016.

\bibitem{shiozaki_review}
Ken Shiozaki, Masatoshi Sato, and Kiyonori Gomi.
\newblock Topological crystalline materials: General formulation, module
  structure, and wallpaper groups.
\newblock {\em Phys. Rev. B}, 95:235425, Jun 2017.

\bibitem{singlediraccone}
Ling Lu, Chen Fang, Liang Fu, Steven~G. Johnson, John~D. Joannopoulos, and
  Marin Soljacic.
\newblock Symmetry-protected topological photonic crystal in three dimensions.
\newblock {\em Nature Physics (2016) doi:10.1038/nphys3611}.

\bibitem{hiller1986crystallography}
Howard Hiller.
\newblock Crystallography and cohomology of groups.
\newblock {\em The American Mathematical Monthly}, 93(10):765--779, 1986.

\bibitem{SPT_origin}
Zheng-Cheng Gu and Xiao-Gang Wen.
\newblock Tensor-entanglement-filtering renormalization approach and
  symmetry-protected topological order.
\newblock {\em Phys. Rev. B}, 80:155131, Oct 2009.

\bibitem{kane2005B}
\textrm{C. L. Kane} and \textrm{E. J. Mele}.
\newblock \textrm{$Z_2$ topological order and the quantum spin Hall effect}.
\newblock {\em Phys. Rev. Lett.}, 95:146802, 2005.

\bibitem{kane2005A}
\textrm{C. L. Kane} and \textrm{E. J. Mele}.
\newblock \textrm{Quantum spin Hall effect in graphene}.
\newblock {\em Phys. Rev. Lett.}, 95:226801, 2005.

\bibitem{moore2007}
J.~E. Moore and L.~Balents.
\newblock Topological invariants of time-reversal-invariant band structures.
\newblock {\em Phys. Rev. B}, 75(12):121306, 2007.

\bibitem{fu2007b}
Liang Fu, C.~L. Kane, and E.~J. Mele.
\newblock Topological insulators in three dimensions.
\newblock {\em Phys. Rev. Lett.}, 98(10):106803, 2007.

\bibitem{Rahul_3DTI}
Rahul Roy.
\newblock Topological phases and the quantum spin hall effect in three
  dimensions.
\newblock {\em Phys. Rev. B}, 79:195322, May 2009.

\bibitem{2dChiralBosonicSPT_erratum}
Yuan-Ming Lu and Ashvin Vishwanath.
\newblock {Erratum: Theory and classification of interacting integer
  topological phases in two dimensions: A Chern-Simons approach [{P}hys. {R}ev.
  {B} \textbf{86} , 125119 (2012)]}.
\newblock {\em Phys. Rev. B}, 89:199903(E), May 2014.

\bibitem{qi_spincharge}
Xiao-Liang Qi and Shou-Cheng Zhang.
\newblock Spin-charge separation in the quantum spin hall state.
\newblock {\em Phys. Rev. Lett.}, 101:086802, Aug 2008.

\bibitem{essin2009}
Andrew~M. Essin, Joel~E. Moore, and David Vanderbilt.
\newblock Magnetoelectric polarizability and axion electrodynamics in
  crystalline insulators.
\newblock {\em Phys. Rev. Lett.}, 102:146805, 2009.

\bibitem{Lu_sgSPT}
Fuyan Lu, Bowen Shi, and Yuan-Ming Lu.
\newblock Classification and surface anomaly of glide symmetry protected
  topological phases in three dimensions.
\newblock {\em New Journal of Physics}, 19(7):073002, 2017.

\bibitem{Inversion_Fu}
Liang Fu and C.~L. Kane.
\newblock Topological insulators with inversion symmetry.
\newblock {\em Phys. Rev. B}, 76(4):045302, 2007.

\bibitem{AA_Z4}
Aris Alexandradinata, Zhijun Wang, B~Andrei Bernevig, and Michael Zaletel.
\newblock Glide-resolved photoemission spectroscopy: measuring topological
  invariants in nonsymmorphic space groups.
\newblock {\em arXiv:1902.10722}, 2019.

\bibitem{Qi_Hughes_Zhang}
Xiao-Liang Qi, Taylor~L. Hughes, and Shou-Cheng Zhang.
\newblock {Topological field theory of time-reversal invariant insulators}.
\newblock {\em Phys. Rev. B}, 78:195424, Nov 2008.

\bibitem{Wilczek_axion}
Frank Wilczek.
\newblock Two applications of axion electrodynamics.
\newblock {\em Phys. Rev. Lett.}, 58:1799--1802, May 1987.

\bibitem{connectivityMichelZak}
L.~Michel and J.~Zak.
\newblock Connectivity of energy bands in crystals.
\newblock {\em Phys. Rev. B}, 59:5998--6001, 1999.

\bibitem{lee2008}
S.-S. Lee and S.~Ryu.
\newblock {Many-body generalization of the $\ensuremath{\mathbb Z}_2$
  topological invariant for the quantum spin hall effect}.
\newblock {\em Phys. Rev. Lett.}, 100:186807, 2008.

\bibitem{aa2011}
A.~{Alexandradinata}, T.~L. {Hughes}, and B.~A. {Bernevig}.
\newblock Trace index and spectral flow in the entanglement spectrum of
  topological insulators.
\newblock {\em Phys. Rev. B}, 84:195103, 2011.

\bibitem{fu2006}
Liang Fu and C.~L. Kane.
\newblock {Time reversal polarization and a $\ensuremath{Z}_2$ adiabatic spin
  pump}.
\newblock {\em Phys. Rev. B}, 74(19):195312, 2006.

\bibitem{hastings2006spectral}
Matthew~B Hastings and Tohru Koma.
\newblock Spectral gap and exponential decay of correlations.
\newblock {\em Communications in mathematical physics}, 265(3):781--804, 2006.

\bibitem{Niu_twistedBC}
Qian Niu, D.~J. Thouless, and Yong-Shi Wu.
\newblock Quantized hall conductance as a topological invariant.
\newblock {\em Phys. Rev. B}, 31:3372--3377, Mar 1985.

\bibitem{Wen_review_2016}
Xiao-Gang Wen.
\newblock Colloquium: Zoo of quantum-topological phases of matter.
\newblock {\em Rev. Mod. Phys.}, 89:041004, Dec 2017.

\bibitem{BIQH}
T.~Senthil and Michael Levin.
\newblock Integer quantum hall effect for bosons.
\newblock {\em Phys. Rev. Lett.}, 110:046801, Jan 2013.

\bibitem{Teo_AF_TRS}
Sharmistha Sahoo, Zhao Zhang, and Jeffrey C.~Y. Teo.
\newblock Coupled wire model of symmetric majorana surfaces of topological
  superconductors.
\newblock {\em Phys. Rev. B}, 94:165142, Oct 2016.

\bibitem{Schnyder}
Andreas~P. Schnyder, Shinsei Ryu, Akira Furusaki, and Andreas W.~W. Ludwig.
\newblock {Classification of topological insulators and superconductors in
  three spatial dimensions}.
\newblock {\em Phys. Rev. B}, 78:195125, Nov 2008.

\bibitem{Qi_Zhang}
Xiao-Liang Qi and Shou-Cheng Zhang.
\newblock {Topological insulators and superconductors}.
\newblock {\em Rev. Mod. Phys.}, 83:1057--1110, Oct 2011.

\bibitem{Chiu_Teo_Schnyder_Ryu}
Ching-Kai Chiu, Jeffrey C.~Y. Teo, Andreas~P. Schnyder, and Shinsei Ryu.
\newblock Classification of topological quantum matter with symmetries.
\newblock {\em Rev. Mod. Phys.}, 88:035005, Aug 2016.

\bibitem{pollmann2010}
Frank Pollmann, Ari~M. Turner, Erez Berg, and Masaki Oshikawa.
\newblock Entanglement spectrum of a topological phase in one dimension.
\newblock {\em Phys. Rev. B}, 81:064439, Feb 2010.

\bibitem{Wen_2d}
Xie Chen, Zheng-Xin Liu, and Xiao-Gang Wen.
\newblock {Two-dimensional symmetry-protected topological orders and their
  protected gapless edge excitations}.
\newblock {\em Phys. Rev. B}, 84:235141, Dec 2011.

\bibitem{levin2012}
Michael Levin and Zheng-Cheng Gu.
\newblock Braiding statistics approach to symmetry-protected topological
  phases.
\newblock {\em Phys. Rev. B}, 86:115109, Sep 2012.

\bibitem{Kapustin_equivariant}
Anton Kapustin and Alex Turzillo.
\newblock Equivariant topological quantum field theory and symmetry protected
  topological phases.
\newblock {\em Journal of High Energy Physics}, 2017(3):6, Mar 2017.

\bibitem{Husain}
A.~Husain.
\newblock {$\ensuremath{G}$-Extensions of Quantum Group Categories and
  Functorial SPT}.
\newblock {\em arXiv:1605.08398}, 2016.

\bibitem{3dFTScWangSenthil_2}
Chong Wang and T.~Senthil.
\newblock {Interacting fermionic topological insulators/superconductors in
  three dimensions}.
\newblock {\em Phys. Rev. B}, 89:195124, May 2014.

\bibitem{3dFTScWangSenthil_2_erratum}
Chong Wang and T.~Senthil.
\newblock {Erratum: Interacting fermionic topological
  insulators/superconductors in three dimensions [Phys. Rev. B \textbf{89} ,
  195124 (2014)]}.
\newblock {\em Phys. Rev. B}, 91:239902, Jun 2015.

\bibitem{2dFermionGExtension}
Meng Cheng, Zhen Bi, Yi-Zhuang You, and Zheng-Cheng Gu.
\newblock Classification of symmetry-protected phases for interacting fermions
  in two dimensions.
\newblock {\em Phys. Rev. B}, 97:205109, May 2018.

\bibitem{Lan_Kong_Wen_1}
Tian Lan, Liang Kong, and Xiao-Gang Wen.
\newblock Theory of (2+1)-dimensional fermionic topological orders and
  fermionic/bosonic topological orders with symmetries.
\newblock {\em Phys. Rev. B}, 94:155113, Oct 2016.

\bibitem{Lan_Kong_Wen_2}
Tian Lan, Liang Kong, and Xiao-Gang Wen.
\newblock Classification of (2+1)-dimensional topological order and
  symmetry-protected topological order for bosonic and fermionic systems with
  on-site symmetries.
\newblock {\em Phys. Rev. B}, 95:235140, Jun 2017.

\bibitem{SOinfty}
Xiao-Gang Wen.
\newblock {Construction of bosonic symmetry-protected-trivial states and their
  topological invariants via
  $G\ifmmode\times\else\texttimes\fi{}SO(\ensuremath{\infty})$ nonlinear
  $\ensuremath{\sigma}$ models}.
\newblock {\em Phys. Rev. B}, 91:205101, May 2015.

\bibitem{Else_edge}
Dominic~V. Else and Chetan Nayak.
\newblock Classifying symmetry-protected topological phases through the
  anomalous action of the symmetry on the edge.
\newblock {\em Phys. Rev. B}, 90:235137, Dec 2014.

\bibitem{Wang_Levin_invariants}
Chenjie Wang and Michael Levin.
\newblock Topological invariants for gauge theories and symmetry-protected
  topological phases.
\newblock {\em Phys. Rev. B}, 91:165119, Apr 2015.

\bibitem{Wang_intrinsic_fermionic}
Chenjie Wang, Chien-Hung Lin, and Zheng-Cheng Gu.
\newblock Interacting fermionic symmetry-protected topological phases in two
  dimensions.
\newblock {\em Phys. Rev. B}, 95:195147, May 2017.

\bibitem{DavisKirk}
James~F. Davis and Paul Kirk.
\newblock {\em {Lecture Notes in Algebraic Topology}}.
\newblock American Mathematical Society, Providence, 2001.
\newblock {Sec.\ 8.8}.

\bibitem{TCI_Fu}
Liang Fu.
\newblock Topological crystalline insulators.
\newblock {\em Phys. Rev. Lett.}, 106:106802, Mar 2011.

\bibitem{Else2018}
D.~V. {Else} and R.~{Thorngren}.
\newblock {Crystalline topological phases as defect networks}.
\newblock {\em arXiv:1810.10539}, Oct 2018.

\bibitem{Song2018}
Z.~{Song}, C.~{Fang}, and Y.~{Qi}.
\newblock {Real-space recipes for general topological crystalline states}.
\newblock {\em arXiv:1810.11013}, Oct 2018.

\bibitem{ITA2006}
Th. Hahn, editor.
\newblock {\em International Tables for Crystallography}.
\newblock International Union of Crystallography, Oct 2006.

\bibitem{Elements_AT}
James~R. Munkres.
\newblock {\em Elements Of Algebraic Topology}.
\newblock Westview Press, 1996.

\bibitem{Gomi_comm}
Kiyonori Gomi.
\newblock {\em Private communication}, 2018.

\bibitem{Whi}
George~W. Whitehead.
\newblock Generalized homology theories.
\newblock {\em Transactions of the American Mathematical Society},
  102(2):227--283, 1962.

\bibitem{KonoTamaki}
Akira K{\=o}no and Dai Tamaki.
\newblock {\em Generalized cohomology}.
\newblock American Mathematical Soc., 2006.

\bibitem{TB18}
Luka Trifunovic and Piet Brouwer.
\newblock Higher-order bulk-boundary correspondence for topological crystalline
  phases.
\newblock {\em arXiv:1805.02598}, 2018.

\end{thebibliography}


\end{document}